\documentclass{article}

\usepackage[lmargin=2.8cm,rmargin=2.8cm,bmargin=2cm,tmargin=2cm]{geometry}
\usepackage[utf8]{inputenc}
\usepackage[T1]{fontenc}
\usepackage{lmodern}
\usepackage{microtype}
\usepackage{setspace}
\usepackage{amsmath,amssymb,amsthm}
\usepackage{mathtools}
\usepackage{bm}
\usepackage{bbm}

\usepackage{tocbibind}

\usepackage{longtable}
\usepackage{booktabs}
\usepackage{enumitem}
\usepackage{xcolor}
\usepackage{tikz}

\usepackage[authoryear]{natbib}

\usepackage[unicode=true,pdfusetitle,bookmarks=true,
  bookmarksnumbered=false,bookmarksopen=false,
  breaklinks=true,colorlinks=true]{hyperref}
\hypersetup{linkcolor=blue,citecolor=blue,urlcolor=blue}

\newif\iffinalversion
\finalversionfalse
\iffinalversion
  \newcommand{\authnote}[1]{}
\else
  \newcommand{\authnote}[1]{\textcolor{red}{#1}}
\fi

\makeatletter
\renewcommand\paragraph{\@startsection{paragraph}{4}{\z@}%
            {-2.5ex\@plus -1ex \@minus -.25ex}%
            {1.25ex \@plus .25ex}%
            {\normalfont\normalsize\bfseries}}
\makeatother

\theoremstyle{plain}
\newtheorem{theorem}{Theorem}[section]
\newtheorem{lemma}[theorem]{Lemma}
\newtheorem{proposition}[theorem]{Proposition}
\newtheorem{corollary}[theorem]{Corollary}
\theoremstyle{definition}

\newtheorem{example}[theorem]{Example}
\theoremstyle{remark}
\newtheorem{remark}[theorem]{Remark}

\theoremstyle{plain}
\newtheorem{theoremintro}{Theorem}

\newcommand{\setX}{\mathsf{E}}
\newcommand{\setZ}{\mathsf{Z}}
\newcommand{\F}{\mathcal{F}}
\newcommand{\dd}{\mathrm{d}}
\newcommand{\Alt}{\mathrm{Alt}}
\newcommand{\R}{\mathbb{R}}
\renewcommand{\P}{\mathbb{P}}
\newcommand{\E}{\mathbb{E}}
\newcommand{\ind}{\mathbbm{1}}
\newcommand\iid{\stackrel{\mathclap{\normalfont\mbox{\tiny{iid}}}}{\sim}}

\title{A perturbation theory for Feynman–Kac models and their particle interpretations}

\author{Pierre Del Moral\thanks{Centre de Recherche Inria Bordeaux
Sud-Ouest, \texttt{pierre.del-moral@inria.fr}}
\and
Mathieu Gerber\thanks{School of Mathematics, University of Bristol,
\texttt{mathieu.gerber@bristol.ac.uk}}}
\date{}

\begin{document}

\maketitle

\begin{abstract}
We develop a perturbation analysis of Feynman--Kac models,
describing how local errors introduced at each step of the
Feynman--Kac recursion propagate over time. Under general conditions
on the potential functions and the Markov transitions, which are shown to hold for a large class of state-space models, 
we prove that the global error for any time horizon $T\geq 1$ is bounded 
by $C\,\Delta_T\log(1+\Delta_T^{-1})$, where $\Delta_T$ is the
largest local error; this bound holds both for the prediction and the updated
models. These conditions require no Lyapunov function and no
stability of the Markov transitions, and allow for
non-compact state spaces: confinement is provided by the
potential functions. Our perturbation bounds hold  in total variation and
in the Kolmogorov norm; for the latter, we extend the
Hardy--Krause variation and the Koksma--Hlawka inequality to
unbounded domains. Our perturbation theory is applied to derive time-uniform bounds
for linearized Kalman filters, Euler schemes for partially
observed diffusions and nudged state-space models, as well as
high-probability error bounds, uniform over whole
trajectories, for particle filters and, more generally, for
mean field particle methods such as sequential Monte Carlo
and diffusion Monte Carlo algorithms.


\end{abstract}
\setcounter{tocdepth}{2}
\begingroup
\makeatletter
\renewcommand*\l@section{\@dottedtocline{1}{0em}{1.8em}}
\makeatother
\tableofcontents
\endgroup
\section{Introduction}

\subsection{Feynman--Kac models and objectives\label{sub:setup}}

\subsubsection{General notation}

Throughout the article, $(\setX,\mathcal{E})$ stands
for some measurable space; when $\setX\subseteq\R^d$ is
a Borel set, $\mathcal{E}$ is the Borel
$\sigma$-algebra on $\setX$. We denote by
$\mathcal{P}(\setX)$ the set of probability measures on
$\mathcal{E}$, by $\mathcal{M}(\setX)$ the set of finite
signed measures on $\mathcal{E}$, and by
$\mathcal{B}_1(\setX)$ the set of measurable functions
$f\colon\setX\to[0,1]$.

For any signed measure $\nu\in\mathcal{M}(\setX)$ and any measurable
function $f\colon\setX\to\R$, we set
$
  \nu(f)\coloneqq\int_{\setX}f(x)\,\nu(\dd x)
$
whenever the integral is well defined. When $f=1_A$ is
the indicator function of some set $A\in\mathcal{E}$, we
slightly abuse notation and write $\nu(A)$ instead of
$\nu(1_A)$; we also write $\nu(1)$ instead of
$\nu(\setX)$. For any bounded, and possibly
unnormalized, integral operator $Q(x,\dd y)$ on
$\setX$, we denote by $Q(f)$ the function defined by
$Q(f)(x)\coloneqq\int_{\setX}Q(x,\dd y)\,f(y)$, and by
$\nu Q$ the signed measure defined by
$(\nu Q)(f)\coloneqq\nu(Q(f))$.

Unless otherwise stated, all the random variables are
defined on some probability space $(\Omega,\F,\P)$ and
$\E[\cdot]$ stands for the expectation with respect to
$\P$. For a random variable $X$ taking values in
$\setX$ and $\mu\in\mathcal{P}(\setX)$, we write
$X\sim\mu$ when $X$ has distribution $\mu$, and
$Y\,|\,X\sim K(X,\dd y)$ when the conditional
distribution of $Y$ given $X$ is $K(X,\cdot)$, for some
Markov transition $K$ on $\setX$. Finally, we use the
conventions $\sum_{\emptyset}=0$ and
$\prod_{\emptyset}=1$, and, for any $a,b\in\R$, we set
$a\wedge b\coloneqq\min\{a,b\}$ and
$a\vee b\coloneqq\max\{a,b\}$. A sequence indexed by the time
parameter is denoted by its generic term; for
instance, $\eta_t$ stands for the sequence
$(\eta_t)_{t\geq 1}$ when no confusion may arise.

\subsubsection{Feynman--Kac models\label{sub:FK-model}}

Let $G_t$ be a sequence of functions in
$\mathcal{B}_1(\setX)$ and let $M_{t+1}$ be
a sequence of Markov transitions on $\setX$. Given a probability measure $\eta_1\in\mathcal{P}(\setX)$, let $X_t$ be a
Markov chain such that
$X_1\sim\eta_1$ and $X_{t+1}\,|\,X_t\sim
M_{t+1}(X_t,\dd x')$ for any $t\geq 1$. The Feynman--Kac flow $\eta_t\in\mathcal{P}(\setX)$ is
defined, for any $t\geq 1$ and
$f\in\mathcal{B}_1(\setX)$, by
\begin{align}\label{eq:def_eta_t}
\eta_t(f)\coloneqq\frac{\gamma_t(f)}{\gamma_t(1)}
\qquad\text{with}\qquad
\gamma_t(f)\coloneqq
\E\Big[f(X_t)\prod_{s=1}^{t-1}G_s(X_s)\Big].
\end{align}
The unnormalized measures $\gamma_t$, with
$\gamma_1=\eta_1$, satisfy the linear evolution
equation
\begin{align}\label{eq:gamma_linear}
\gamma_{t+1}=\gamma_tQ_{t+1}
\qquad\text{with the integral operator}\qquad
Q_{t+1}(x,\dd x')\coloneqq G_t(x)\,M_{t+1}(x,\dd x') 
\end{align}
and throughout this work we assume that the following   condition holds.
\begin{enumerate}[label=\textbf{(A{{\arabic*}})}]
\setcounter{enumi}{-1}
\item\label{A0} There exists a sequence  $W_t\in \mathcal{E}$ such that
$Q_{t+1}(x,W_{t+1})>0$ for any $x\in W_t$ and any.
$t\geq 1$ 
\end{enumerate}
Then, for any $t\geq 1$, we let
\[
  \mathcal{P}_{W_t}(\setX)\coloneqq
  \big\{\mu\in\mathcal{P}(\setX):\ \mu(W_t)>0\big\} 
\]
and note that, under \ref{A0}, the functions $G_t$ and
$Q_{t+1}(\ind_{W_{t+1}})$ are positive on $W_t$, so
that $\mu(G_t)>0$ for any
$\mu\in\mathcal{P}_{W_t}(\setX)$. We can therefore define
the mappings
$\Psi_t\colon\mathcal{P}_{W_t}(\setX)\to
\mathcal{P}_{W_t}(\setX)$ and
$\Phi_{t+1}\colon\mathcal{P}_{W_t}(\setX)\to
\mathcal{P}_{W_{t+1}}(\setX)$ by
\begin{align}\label{eq:def_Phi}
  \Psi_t(\mu)(\dd x)\coloneqq
  \frac{G_t(x)}{\mu(G_t)}\,\mu(\dd x)
  \qquad\text{and}\qquad
  \Phi_{t+1}(\mu)\coloneqq\Psi_t(\mu)M_{t+1}
  =\frac{\mu Q_{t+1}}{\mu Q_{t+1}(1)}.
\end{align}
When  $G_t>0$ on
$\setX$ for any $t\geq 1$, \ref{A0} holds with
$W_t=\setX$, in which case
$\mathcal{P}_{W_t}(\setX)=\mathcal{P}(\setX)$. Observe also that the
expressions in \eqref{eq:def_Phi} make sense as soon as
$\mu(G_t)>0$, but only their restrictions to
$\mathcal{P}_{W_t}(\setX)$ are used in the following.

Throughout the article, we assume that the initial distribution
$\eta_1$ belongs to $\mathcal{P}_{W_1}(\setX)$. In this context, $\eta_t\in\mathcal{P}_{W_t}(\setX)$
for any $t\geq 1$, and the normalized measures
$\eta_t$ satisfy the nonlinear evolution
equation
\begin{equation}\label{def-ref-eta}
\eta_{t+1}=\Phi_{t+1}(\eta_t),
\qquad t\geq 1.
\end{equation}
Following \citet[][Definition~2.3.2]{DelMoral:book}, the
flow $\eta_t$ is called the Feynman--Kac
\emph{prediction model} and the functions
$G_t$ are called the \emph{potential
functions}. The flow $\widehat{\eta}_t$
defined by $\widehat{\eta}_t\coloneqq\Psi_t(\eta_t)$ is
called the \emph{updated model}, and the pair
$\big(\eta_t,\widehat{\eta}_t\big)$
the \emph{Feynman--Kac model}. In the filtering problem
discussed in Section~\ref{sub:motive_example}, the
measures $\eta_t$ and $\widehat{\eta}_t$ coincide with
the predictive and the filtering distributions of the
hidden state.

We use the convention that $\mathcal{P}_{W_0}(\setX)=\mathcal{P}(\setX)$, and that $\Phi_1$ and $Q_1$ are the
identity operators, that is, $\Phi_1(\mu)\coloneqq\mu$
and $Q_1(x,\dd x')\coloneqq\delta_x(\dd x')$, so that
$\mu Q_1=\mu$ for any $\mu\in\mathcal{M}(\setX)$. For any
integers $t\geq s\geq 1$, we set
\[
\Phi_{s,t}\coloneqq\Phi_t\circ\cdots\circ\Phi_{s+1}
\quad\text{and}\quad
Q_{s,t}\coloneqq Q_{s+1}\cdots Q_t,
\]
with the convention $\Phi_{s,s}\coloneqq\Phi_1$ and
$Q_{s,s}\coloneqq Q_1$. Note that, by \ref{A0},
$\Phi_{s,t}(\mu)$ is well defined and belongs to
$\mathcal{P}_{W_t}(\setX)$ for any $t\geq s\geq 1$ and any
$\mu\in\mathcal{P}_{W_s}(\setX)$. For any integer $s\geq 1$, and any probability measures
$\mu\in\mathcal{P}_{W_{s-1}}(\setX)$ and
$\nu\in\mathcal{P}(\setX)$, we also set
\begin{equation}\label{eq:delta_s}
\delta_s(\mu,\nu)\coloneqq\nu-\Phi_s(\mu);
\end{equation}
in particular, $\delta_1(\mu,\nu)=\nu-\mu$. With this notation in place, we check
that, for any $t\geq s\geq 1$,
\begin{align*}
\Phi_{s,t}(\mu)=\frac{\mu Q_{s,t}}{\mu Q_{s,t}(1)}
\quad\text{and}\quad
\eta_t=\Phi_{s,t}(\eta_s),
\end{align*}
for any $\mu\in\mathcal{P}_{W_s}(\setX)$.

Feynman--Kac models arise in a variety of areas,
including applied mathematics
\citep{beskos2017multilevel}, statistics
\citep{del2006sequential}, physics
\citep{caffarel1988development,simon2005functional,
caffarel2025mathematical}, nonlinear filtering and
control theory \citep{del2000branching} and
quantitative finance \citep{pham2009continuous}; see
also the books \citet{DelMoral:book,del2013mean}. In state-space models, they describe the conditional
distributions of a hidden Markov chain given the
observations (Section~\ref{sub:motive_example}).

\subsubsection{Objectives of the article}

In practice, the flows $\eta_t$ and $\hat{\eta}_t$ are rarely computable, and various methods to approximate them have been proposed in the literature. Obtaining   error bounds for these methods that  hold under weak assumptions on the Feynman-Kac model (e.g.~non-compact state spaces and unstable Markov kernels $M_{t+1}$) is one motivation of this
article.

To state more precisely the purpose of this article, let $\mu_t\in\mathcal{P}_{W_t}(\setX)$ be an approximation 
of the flow $\eta_t$ defined in \eqref{def-ref-eta}. By
definition, the sequence $\mu_t$ follows the
Feynman--Kac recursion up to the \emph{local
perturbations}
$\delta_{t+1}(\mu_t,\mu_{t+1})=\mu_{t+1}-\Phi_{t+1}(\mu_t)$.  With the
convention $\mu_0\coloneqq\eta_1$, at any time $t\geq 1$ the the global error
can be decomposed as follows:
\begin{align}\label{eq:decomp_mu}
\mu_t-\eta_t=\sum_{s=1}^{t}\Big(
\Phi_{s,t}\big(\Phi_s(\mu_{s-1})+\delta_s(\mu_{s-1},\mu_s)\big)
-\Phi_{s,t}\big(\Phi_s(\mu_{s-1})\big)\Big),
\end{align}
where all the terms are well defined by \ref{A0}, since
$\mu_s$ and $\Phi_s(\mu_{s-1})$ belong to
$\mathcal{P}_{W_s}(\setX)$ for any $s\geq 1$. In \eqref{eq:decomp_mu}, the $s$-th
term is the propagation up to time $t$ of the local
perturbation $\delta_s(\mu_{s-1},\mu_s)$. These local errors
may thus accumulate over time, and the main purpose of this article is to give general conditions on the model and on the approximation $\mu_t$ of the flow $\eta_t$ under
which they do not. These general conditions are then used to obtain error bounds for some popular  methods for approximating the flow $\eta_t$, including  linearized Kalman filter and particle filter algorithms (see Section \ref{sub:main_results} for a summary of our main results).

\subsection{Filtering   with state-space models\label{sub:motive_example}}

Given a \textit{hidden} process $\tilde{X}_t$ and an \textit{observation} process $\tilde{Y}_t$, learning about the random variable $\tilde{X}_t$ from an observation  $(y_1,\dots,y_t)$ of $(\tilde{Y}_1,\dots,\tilde{Y}_t)$ is known as (Bayesian) filtering. State-space models are popular tools for performing this task, and ar, for instance, used in robotics \citep{barfoot2024state}, neuroscience
\citep{paninski2010new} and ecology \citep{newman2023state}. Although not limited to this setting, filtering with state-space models is the  main motivation for the perturbation theory developed in this article.  

In the following,  we assume that $\tilde{X}_t$ takes its values in $\R^d$ and $\tilde{Y}_t$ its values in $\R^r$ for some integers $d\geq 1$ and $r\geq 1$.

\subsubsection{The model\label{sub:model}}

Let  
$\Sigma_{G_t}$ and
$\Sigma_{M_{t+1}}$ be some sequences of
symmetric positive definite matrices of size
$r\times r$ and $d\times d$, respectively, and let
$\eta_1\in\mathcal{P}(\R^d)$. For any $t\geq 1$ we
also consider some measurable functions
\begin{equation}\label{eq:functions}
H_t\colon\R^d\to\R^r,\qquad
F_{t+1}\colon\R^d\to\R^d,\qquad
\zeta_{G_t}\colon[0,\infty)\to(0,1],\qquad
\zeta_{M_{t+1}}\colon[0,\infty)\to(0,\infty)
\end{equation}
such that
\begin{align*}
C_{G_t}\coloneqq\int_{\R^r}\zeta_{G_t}(\|u\|^2)\,
\lambda_r(\dd u)<\infty
\qquad\text{and}\qquad
C_{M_{t+1}}\coloneqq\int_{\R^d}\zeta_{M_{t+1}}(\|v\|^2)\,
\lambda_d(\dd v)<\infty,
\end{align*}
where $\lambda_k$ stands for the Lebesgue measure on
$\R^k$. We then consider the state-space model defined by
\begin{align}\label{eq:SSM}
X_1\sim\eta_1,
\qquad
Y_t=H_t(\tilde{X}_t)+\Sigma_{G_t}^{1/2}\,U_t,
\qquad
X_{t+1}=F_{t+1}(X_t)+\Sigma_{M_{t+1}}^{1/2}\,V_{t+1},
\qquad t\geq 1,
\end{align}
where $X_1$, $U_t$ and
$V_{t+1}$ are independent, and $U_t$ and
$V_{t+1}$ have the spherically symmetric densities
$u\mapsto\zeta_{G_t}(\|u\|^2)/C_{G_t}$ and
$v\mapsto\zeta_{M_{t+1}}(\|v\|^2)/C_{M_{t+1}}$ with
respect to $\lambda_r$ and $\lambda_d$, respectively. This class of models contains the linear Gaussian
state-space model (choosing $\eta_1$ Gaussian, $H_t$
and $F_{t+1}$ linear and
$\zeta_{G_t}(z)=\zeta_{M_{t+1}}(z)=e^{-z/2}$), as well
as nonlinear models with possibly heavy-tailed (e.g.\
Student) noises; see Section~\ref{sec:examples}.
Throughout, capital letters $X_t$ and $Y_t$ refer to the two processes defined by \eqref{eq:SSM}, while lowercase letters $y_t$ denote a deterministic sequence in $\mathbb{R}^r$.

Given some sequence of observations
$y_t$ in $\R^r$, we let $\widehat{p}_t$ be
the conditional distribution of $X_t$ given
$(Y_1,\dots,Y_t)=(y_1,\dots,y_t)$, and $p_{t+1}$ be the
conditional distribution of $X_{t+1}$ given
$(Y_1,\dots,Y_t)=(y_1,\dots,y_t)$, with the convention
$p_1=\eta_1$. These are the \emph{filtering} and the
\emph{predictive} distributions of the model; except for linear
Gaussian models (Kalman filter), they are not
computable.

As indicated by the notation, the two processes defined in \eqref{eq:SSM} may, and in practice typically do, differ from the hidden process of interest $\tilde{X}_t$ and the process $\tilde{Y}_t$ for which observations $y_t$ are available; that is, model \eqref{eq:SSM} may be, and in practice typically is, misspecified.

\subsubsection{Feynman--Kac representation}

For any $t\geq 1$ and $x,x'\in\R^d$, consider the
potential functions
\begin{align}\label{eq:Gt_def}
G_t(x)\coloneqq\zeta_{G_t}\Big(\big(y_t-H_t(x)\big)^\top
\Sigma_{G_t}^{-1}\big(y_t-H_t(x)\big)\Big)
\end{align}
and the Markov transitions
\begin{align}\label{eq:Mt_def}
M_{t+1}(x,\dd x')\coloneqq m_{t+1}(x,x')\,
\lambda_d(\dd x')
\end{align}
with the densities
\begin{align}\label{eq:Mt_def2}
m_{t+1}(x,x')\coloneqq
\frac{1}{C_{M_{t+1}}\det(\Sigma_{M_{t+1}})^{1/2}}\,
\zeta_{M_{t+1}}\Big(\big(x'-F_{t+1}(x)\big)^\top
\Sigma_{M_{t+1}}^{-1}\big(x'-F_{t+1}(x)\big)\Big).
\end{align}

Up to a multiplicative constant that does not depend on
$x$, the function $G_t(x)$ coincides with the
likelihood of the observation $y_t$ under the distribution of $Y_t|X_t=x$, and
$M_{t+1}$ is the Markov transition of the hidden state $X_t$.
In this situation, Bayes' formula yields $\eta_t=p_t$
and $\widehat{\eta}_t=\widehat{p}_t$ for any
$t\geq 1$; that is, the predictive and the filtering
distributions coincide with the prediction and the
updated models defined in Section~\ref{sub:setup}, and below we will
refer to the Feynman--Kac model defined by
\eqref{eq:Gt_def}--\eqref{eq:Mt_def2} as the
\emph{Feynman--Kac representation} of the state-space
model \eqref{eq:SSM} associated with the observations
$y_t$. Our main results, summarized in     in  Section~\ref{sub:main_results}, therefore apply
to any approximation of the  distributions $p_t$ and $\hat{p}_t$.  The
potential functions  defined in \eqref{eq:Gt_def} depend on the observations $y_t$,
treated as a fixed sequence, and applying our theory to the Feynman--Kac model \eqref{eq:Gt_def}--\eqref{eq:Mt_def2}  therefore require to impose some conditions on this sequence (see Section~\ref{sec:examples}). For instance, for linear Gaussian state-space
models with $r\geq d$, these conditions reduce to an explicit
condition on $y_t$ (which does not necessarily implies that $\sup_{t\geq 1}\|y_t\|<\infty$), with no stability condition on the
  dynamics of $X_t$.

\subsubsection{Beyond filtering}

The Feynman--Kac model \eqref{eq:Gt_def}--\eqref{eq:Mt_def2}
is well defined for any deterministic sequence $y_t$ in
$\R^r$, whether or not it is a realization $y_t=Y_t(\omega)$
of the observation process of \eqref{eq:SSM}. Filtering,
where the $y_t$ are realizations of some observation process $\tilde{Y}_t$, is
only one instance of this framework. To introduce another instance of this framework  observe that, since $G_t$ takes
values in $(0,1]$, $\eta_t$ is the distribution of 
  the Markov chain $X'_t$, with transitions $M_{t+1}$ and killed at
each time $s$ with probability $1-G_s(X'_s)$, conditionally
on its survival up to time $t-1$. In this context, $\gamma_t(1)$ is the
survival probability of $X'_t$.

For instance, let $r=d$, $H_t(x)=x$,
$\zeta_{G_t}(z)=e^{-z/2}$, and let  $y_t$ be a reference trajectory for the Markov chain $\tilde{X}_t$. In addition, write $G_t=e^{-V_t}$ with
\[
V_t(x)=\tfrac{1}{2}\,(x-y_t)^\top\Sigma_{G_t}^{-1}(x-y_t).
\]
The potential $V_t$ is a quadratic well centred at $y_t$,
which moves deterministically in time. The chain $X'_t$ is therefore
killed at a rate that increases with its distance to the
reference trajectory $y_t$, and the
Feynman--Kac model can be seen as a time-inhomogeneous
harmonic oscillator. With linear Gaussian transitions and
$y_t=0$, one recovers the time-homogeneous harmonic
oscillators studied in \citet{delmoral2023coupled} and
\citet{caffarel2025mathematical}. For time-homogeneous
models, the fixed points $\eta_\infty=\Phi(\eta_\infty)$ of
the semigroup are the quasi-stationary distributions of the
killed chain
\citep{del2000branching,delmoral2003particle,delmoral2004particle,collet2013quasi,champagnat2016exponential},
and are related to the ground states of
Feynman--Kac--Schr\"odinger operators
\citep{DelMoral:book,caffarel2025mathematical}.

\subsection{Main results\label{sub:main_results}}
 
Below we state our main results informally; precise
statements are given in Sections~\ref{sec:main}, \ref{sec:compareFK}
and~\ref{sec:PF}.  To this aim, we let   $\|\cdot\|$ be some norm on $\mathcal{M}(\setX)$ and $\mu_t\in\mathcal{P}_{W_t}(\setX)$ be an approximation 
of the flow $\eta_t$ defined in \eqref{def-ref-eta}. In addition, setting $\bm{\mu}\coloneqq(\mu_t)_{t\geq 1}$, for any
time horizon $T\geq 1$  we let
\begin{align*}
\Delta_{T,\bm{\mu},\|\cdot\|}\coloneqq
\max_{1\leq t\leq T}\big\|\delta_t(\mu_{t-1},\mu_t)\big\|\qquad\text{ with $\delta_1(\mu_0,\mu_1)=\mu_1-\eta_1$}
\end{align*}
be the largest local perturbation up to time $T$. Lastly, we let 
$\vartheta\colon[0,\infty)\to[0,\infty)$ be the continuous and increasing function defined by
\begin{align}\label{eq:vartheta}
\vartheta(z)\coloneqq z\log\big(1+z^{-1}\big),
\quad z\in(0,\infty),
\qquad\vartheta(0)\coloneqq 0.
\end{align}

The next theorem is an informal version of
Theorems~\ref{thm:main_TV} and~\ref{thm:main} and of
Propositions~\ref{prop:filter_TV}
and~\ref{prop:filter}.

\begin{theoremintro}\label{thmA}
Assume that the potential functions satisfy the conditions of Section~\ref{sub:primitive}, and let $\|\cdot\|$ be the
total variation norm or, when $\setX\subseteq\R^d$, the
Kolmogorov norm. Then, there exists a constant $C_{\eta_1}\in[1,\infty)$ that does not depend
 on $\bm{\mu}$ such that, for any sequence
$\bm{\mu}=(\mu_t)_{t\geq 1}$ with
$\mu_t\in\mathcal{P}_{W_t}(\setX)$, and any $T\geq 1$, we
have the   estimate
\begin{align}\label{eq:summaryMain}
\max_{1\leq t\leq T}\|\mu_t-\eta_t\|\vee \max_{1\leq t\leq T}\|\widehat{\mu}_t-\widehat{\eta}_t\|
\leq C_{\eta_1}\,
\vartheta\big(\Delta_{T,\bm{\mu},\|\cdot\|}\big).
\end{align}
\end{theoremintro}

For the updated model, \eqref{eq:summaryMain} also holds for some class of
unbounded test functions when the local perturbations are small enough.
For the state-space models introduced in Section~\ref{sec:examples}, this class
contains the coordinates of $H_t$, and hence the posterior mean
$\int_\setX x\,\widehat{\eta}_t(\dd x)$ when $H_t(x)=Ax$ with $A$ of
full column rank matrix (e.g.~$H_t(x)=x$; see Propositions~\ref{prop:filter_TV} and~\ref{prop:filter}).

To the best of our knowledge, these are the first
perturbation  bounds for Feynman--Kac models that
hold  under general
conditions. In particular, these bounds hold for  non-compact state spaces and unstable Markov kernels $M_{t+1}$, and, in the context of the state-space model introduced in Section \ref{sub:motive_example}, for observation sequences  compatible with the
dynamics of the model (in a sense that will be made clear in Section \ref{subsub:O_condition}); for instance, if  $X_t$ is random walk then $y_t$ may be unbounded but must have
unbounded sequence  $y_t$ with bounded increments. Still in the context of  Section \ref{sub:motive_example}, if the
observations generated by the state-space model, that is, if
realizations $y_t=Y_t(\omega)$,
then compatibility condition on the observations only holds on finite time
horizons
(see Section~\ref{sub:LG_SSM}).

Up to a logarithmic factor, Theorem \ref{thmA} shows that the global error for  
controlled by the largest local error. Since the constant $C_{\eta_1}$ in the theorem does not depend on $\bm{\mu}$, the perturbation bound
\eqref{eq:summaryMain} holds for deterministic
approximations and, pathwise, for random ones such as
particle filter estimates. By \eqref{eq:summaryMain}, one step of an approximation scheme
is all that needs to be controlled to control its global error.
Section~\ref{sec:main} also allows weighted local
errors, used in Section~\ref{sec:compareFK}.

\subsubsection{Comparison of Feynman--Kac models}

In total variation, Theorem~\ref{thmA} can be used to compare two
Feynman--Kac models. Let $G^\star_t$ and
$M^\star_{t+1}$ be potential functions and Markov
transitions defining a model easier to compute or to
sample. We denote by
$Q^\star_{t+1}(x,\dd x')\coloneqq
G^\star_t(x)\,M^\star_{t+1}(x,\dd x')$ the
corresponding integral operators, and by
$\eta^\star_t$ and
$\widehat{\eta}^\star_t$ the corresponding
prediction and updated models, with
$\eta^\star_1=\eta_1$. We also consider the local
discrepancy parameter
\begin{align}\label{eq:delta}
\Delta_{Q|Q^\star}\coloneqq
\sup_{t> 1}\sup_{x\in\setX}
\big\|\delta_xQ^\star_{t}-\delta_xQ_{t}\big\|_{\mathrm{tv}} 
\end{align}
where $\|\cdot\|_{\mathrm{tv}}$ denote the total variation norm on $\mathcal{M}(\setX)$.

Choosing $\mu_t=\eta^\star_t$ in Theorem~\ref{thmA} yields
  the following result, which is  an informal version of
Theorem~\ref{thm:compare}.

\begin{theoremintro}\label{thmB}
Assume that the total variation conditions of
Theorem~\ref{thmA} are satisfied by the model
$(G_t,M_{t+1})$. Then, there exists a constant $C_{\eta_1}\in[1,\infty)$ that only
depends on $G_t$, $M_{t+1}$ and
$\eta_1$ such that, for any approximating model
$(G^\star_t,M^\star_{t+1})$, we have the
time-uniform estimate
\begin{align}\label{eq:delta_Q}
\sup_{t\geq 1}\|\eta^\star_t-\eta_t\|_{\mathrm{tv}}
\vee
\sup_{t\geq 1}
\|\widehat{\eta}^{\,\star}_t-\widehat{\eta}_t\|_{\mathrm{tv}}
\leq C_{\eta_1}\,\vartheta\big(\Delta_{Q|Q^\star}\big).
\end{align}
\end{theoremintro}

The left-hand side of \eqref{eq:delta_Q} involves the whole nonlinear flow on
an infinite horizon, while $\Delta_{Q|Q^\star}$ compares one
step of the two models, pointwise. Since
$\delta_xQ_{t+1}=G_t(x)\,\delta_xM_{t+1}$, discrepancies between the Markov transitions only count
where $G_t$ is not negligible, and $\Delta_{Q|Q^\star}$
can often be bounded by hand. For instance, consider
the state-space model
\begin{align*}
X_{t+1}=F(X_t)+V_{t+1},\qquad Y_t=AX_t+U_t,
\qquad X_1\sim\mathcal{N}_d(w_0,\Sigma_0),
\end{align*}
with some independent noises
$U_t\sim\mathcal{N}_r(0,\Sigma_U)$ and
$V_{t+1}\sim\mathcal{N}_d(0,\Sigma_V)$, some matrix
$A\in\R^{r\times d}$ with rank $d$, and some
continuously differentiable function
$F\colon\R^d\to\R^d$. The linearized Kalman filter
replaces, at each time step, the function $F$ by its
first order Taylor expansion around some point
$x^\star_{t+1}$, and the resulting Gaussian approximations
$(\eta_t^\star,\widehat{\eta}^\star_t)$ of $(\eta_t,\widehat{\eta}_t)$ are  
computed exactly by the Kalman recursions. The  accuracy of these approximations
depends on the nonlinearity of $F$, measured by the parameter 
\[
\kappa_F\coloneqq\sup_{x,x'\in\R^d}
\|J_F(x)-J_F(x')\|,
\]
 where $J_F$ stands for the Jacobian
matrix of $F$; note that   $\kappa_F=0$ if and only if $F$ is affine. Theorem~\ref{thmB} suggests linearizing $F$ where $G_t$
is maximal; that is, at 
$x^\star_{t+1}=Ky_t$, with the
matrix
$K\coloneqq(A^\top\Sigma_U^{-1}A)^{-1}A^\top\Sigma_U^{-1}$.
Specifically, we obtain the following result, which is an informal version of
Proposition~\ref{prop:extended}.

\begin{theoremintro}\label{thmC}
For any $L\geq 1$ and $\Xi\geq 0$, there exists some
constant $C$ that only depends on the parameters
$(A,\Sigma_U,\Sigma_V,w_0,\Sigma_0)$ and $(L,\Xi)$, such
that, for any continuously differentiable and
$L$-Lipschitz function $F$ and any sequence of
observations such that 
\begin{align}\label{eq:obs_cond}
\|Ky_1-w_0\|\leq\Xi\quad\mbox{and}
\quad\sup_{t\geq 1}\|Ky_{t+1}-F(Ky_t)\|\leq\Xi,
\end{align}
 we have the
time-uniform estimate $\sup_{t\geq 1}
\|\widehat{\eta}^\star_t-\widehat{\eta}_t\|_{\mathrm{tv}}
\leq C\,\vartheta(\kappa_F)$ with $\vartheta$  as in \eqref{eq:vartheta}. In addition, when $\kappa_F\leq 1/C$ we have
\begin{align*}
\sup_{t\geq 1}\Big\lVert\int_{\R^d}x\,
\widehat{\eta}^\star_t(\dd x)
-\int_{\R^d}x\,\widehat{\eta}_t(\dd x)\Big\rVert
\leq C\,\vartheta(\kappa_F).
\end{align*}
\end{theoremintro}

In the above theorem, the constant $C$ depends neither
on the time horizon, nor on the sequence of
observations  if \eqref{eq:obs_cond} holds, nor on the function $F$ (if its
Lipschitz constant is bounded by $L$). Observe that no stability of $F$ is required in the theorem; hence,  $F(x)=x+\beta\,g(x)$,
with $g$ smooth with bounded Jacobian and $\beta$
small, is allowed. 

In this article, Theorem \ref{thmB} is also applied to study the Euler--Maruyama time discretizations
of partially observed diffusions and the nudged
state-space models recently introduced in
\citet{gonzalez2025nudging} to handle model
misspecification; see Table~\ref{tab:compare}. To the best of our
knowledge, the time-uniform estimates presented in Table~\ref{tab:compare} are new.

\begin{table}[ht]
\centering
\small
\begin{tabular}{@{}p{0.40\textwidth}p{0.24\textwidth}p{0.26\textwidth}@{}}
\toprule
Approximation & Local discrepancy
$\Delta_{Q|Q^\star}$ & Time-uniform error bound\\
\midrule
Linearized Kalman filter
(Proposition~\ref{prop:extended})
& $\lesssim\kappa_F$
& $\lesssim\vartheta(\kappa_F)$\\
Euler scheme with $l$ steps per observation
(Proposition~\ref{prop:compare_diff})
& $\lesssim 1/l^\star$
& $\lesssim\vartheta(1/l)$\\
Nudging with parameter $\chi$
(Proposition~\ref{prop:compare_nudge})
& $\lesssim\chi^\alpha$
& $\lesssim\vartheta\big(\chi^\alpha\big)$\\
\bottomrule
\end{tabular}
\caption{Time-uniform estimates of
$\sup_{t\geq 1}\|\widehat{\eta}^\star_t
-\widehat{\eta}_t\|_{\mathrm{tv}}$ derived from
Theorem~\ref{thmB} (see Section~\ref{sec:compareFK}).
The constants depend neither on the time parameter nor
on the approximation parameters $\kappa_F$, $l$
and $\chi$. The parameter $\alpha\in(0,1]$ stands
for a H\"older exponent of the function $F$ of the
state-space model.}
\label{tab:compare}
\end{table}

\subsubsection{Particle filters}

Let $\mu_t$ be a particle filter approximation of the flow $\eta_t$. In this context,  $\mu_t$ is supported on finite
sets, and thus, when $\eta_t$ has a density with respect to the Lebesgue measure, the resulting local
perturbations have total variation norm one. For particle filters, the
total variation version of \eqref{eq:summaryMain} is therefore usually
void. On the other hand, in the Kolmogorov norm  $\|\cdot\|_{\mathrm{K}}$  (defined in Section~\ref{sub:basic_notation}), the particle filters local errors
are of order $N^{-1/2}$  by
Dvoretzky--Kiefer--Wolfowitz-type inequalities. Let
$\eta^N_t$ and $\widehat{\eta}^N_t$ be the bootstrap
particle approximations of $\eta_t$ and
$\widehat{\eta}_t$ with $N$ particles. In this notation, the next theorem
is an informal version of Theorem~\ref{thm:PF}.

\begin{theoremintro}\label{thmD}
Assume that the Kolmogorov conditions of
Theorem~\ref{thmA} are satisfied, with
$\setX\subseteq\R^d$. Then, there exists a constant $C_{\eta_1}$   such that, for any $q\in(0,1]$, any $N\geq 1$ and any $T\geq 1$,
  we have, with probability at least $1-q$,
the estimate
\begin{align*}
\max_{1\leq t\leq T}\|\eta^N_t-\eta_t\|_{\mathrm{K}}
\vee
\max_{1\leq t\leq T}
\|\widehat{\eta}^N_t-\widehat{\eta}_t\|_{\mathrm{K}}
\leq C_{\eta_1}\,\vartheta\big(\varepsilon_{N,T,q}\big)
\quad\text{with}\quad
\varepsilon_{N,T,q}\coloneqq
\sqrt{\frac{1+\log(T/q)}{N}}.
\end{align*}
\end{theoremintro}

By Theorem \ref{thmD}, for $\rho\in(0,1)$, a particle approximation $(\eta^N_t,\widehat{\eta}^N_t)$
is within Kolmogorov distance $\rho$ of   $(\eta_t,\hat{\eta}_t)$
at all times $t\in\{1,\dots,T\}$, and with probability at
least $1-q$, if
\begin{align}\label{eq:Min_N}
N\gtrsim\rho^{-2}\log^2(1/\rho)\,\big(1+\log(T/q)\big).
\end{align}
The lower bound on   $N$ given in \eqref{eq:Min_N} only grows logarithmically with the time
horizon, and this dependence to $T$ cannot be improved (see Section \ref{sub:discussionPF}). Similar estimates hold for fixed test
functions and for some unbounded ones, such as  
the identity function $f(x)=x$. To the
best of our knowledge, these are the first
high-probability error bounds for particle filters that
hold uniformly on whole trajectories under general
conditions. In particular, these bounds hold for  non-compact state spaces and unstable Markov kernels $M_{t+1}$. 
 
\subsubsection{A dual norm}
When
$\setX\subseteq\R^d$ is a (not necessarily bounded)
hyper-rectangle, such as $\setX=\R^d$, we show that, for
a suitable Banach space
$(\mathcal{B}_{\mathrm{HK}}(\setX),
\|\cdot\|_{\mathrm{HK}})$ of functions with bounded
Hardy--Krause variation (see
Section~\ref{sec:preliminaries}), the mapping
\begin{align*}
  \|\nu\|_{\mathcal{V}_\setX}\coloneqq
  \sup\Big\{|\nu(f)|:\
  f\in\mathcal{B}_{\mathrm{HK}}(\setX),\
  \|f\|_{\mathrm{HK}}\leq 1\Big\},
  \qquad\nu\in\mathcal{M}(\setX),
\end{align*}
is a norm on $\mathcal{M}(\setX)$. We also prove that the
metric induced by this norm on $\mathcal{P}(\setX)$ is
equivalent to, and bounded from above by, the
Kolmogorov metric. In addition to be of independent interest, this result plays a key role in this work. Specifically, it allows us to measure local errors in a weak norm
while using the stability of the semigroup in total
variation.

\subsection{Related work and novelty\label{sub:literature}}

Error bounds for approximations of Feynman--Kac models
are often derived using  the interpolation (a.k.a.\
telescoping sum)  formula \eqref{eq:decomp_mu}, which describes the
propagation of a local error introduced at some given
time \citep[][Chapter~7]{DelMoral:book}.  This approach yields time-uniform estimates whenever the semigroup $\Phi_{s,t}$ does not amplify
local errors, that is, when $\Phi_{s,t}(\mu)$ and $\Phi_{s,t}(\mu')$ become close as $t-s$ grows,
uniformly in $\mu,\mu'$. Stability results of this type for Feynman--Kac semigroups are surveyed in
\citet{DelMoral:book,del2013mean}; see \citet{del2023stability} for a more recent account and further
references. Under uniform
(Dobrushin-type) minorization conditions, one obtains
exponential forgetting properties and Lipschitz
estimates for $\Phi_{s,t}$ in total variation or with
respect to the Hilbert projective metric. These yield stability of the filter and time-uniform
convergence of particle filters
\citep{del2000branching,delmoral2001stability,legland2004stability,heine2008uniform},
including exponential concentration inequalities
\citep{delmoral2011concentration,delmoral2012concentration,
del2013mean}.
However, these conditions are essentially
restricted to compact or uniformly mixing models.

For non-compact state spaces, Lyapunov approaches
combine a Foster--Lyapunov drift condition on the
integral operators $Q_{t+1}$ with two-sided density estimates on the compact sub-level sets of some
Lyapunov function $V$
\citep[see][]{Nick,whiteley2012linear,del2023stability};
see also the review article \citet{arnaudon2024lyapunov}.
They give local Lipschitz estimates for
$\Phi_{s,t}$ with respect to some weighted norm
$\|\cdot\|_V$, with exponentially decaying constants. A
related local Doeblin approach is used in
\citet{DoucMoulinesOlsson2014} to derive time-uniform
estimates of the asymptotic variance of particle
filters, and \citet{vanhandel2009uniform} proves the
consistency of particle filters for time averages under
weak ergodicity conditions.

Two recent articles are closer to ours. First, the perturbation bounds \eqref{eq:summaryMain}
 with respect to the
Kolmogorov norm are proved in
\citet[][Lemma~4.1]{gerber2026safety} for the
one-dimensional harmonic oscillator model; that is, for
a one-dimensional linear Gaussian state-space model with
observations $y_t=0$ for any $t\geq 1$, in which case $\eta_t$ and $\widehat{\eta}_t$ are explicit. Second, \citet{gonzalez2026operator}
consider the expected total variation distance between
the optimal filter $\widehat{p}_t$ of a well-specified
state-space model and the optimal filter
$\widehat{p}^\star_t$ associated with some misspecified
Markov transitions. When the well-specified Markov
transitions are uniformly contracting, the authors
provide an estimate of
$\E\big[\|\widehat{p}_t-\widehat{p}^\star_t\|_{\mathrm{tv}}\big]$
that tends to zero with the misspecification error,
where the expectation is taken with respect to the
distribution of the observations under the
well-specified model
\citep[][Corollary~4.7]{gonzalez2026operator}. 


Our approach differs  from existing ones  in  three respects.
\begin{enumerate}[label=(\roman*)]
\item \emph{Confinement by the potential functions.}
Our conditions \ref{A1}--\ref{A3} imposed on the Feynman-Kac model
(Section~\ref{sub:primitive}) do not rely on any
Lyapunov function nor on any drift condition: the local
minorization condition \ref{A1b} is stated on the
super-level sets $\{G_t\geq\varrho\}$ of the potential
functions, and the non-vanishing property \ref{A2} can
be checked using these sets only
(Proposition~\ref{prop:A1p-A2}). In addition, the exponential forgetting property of the semigroup $\Phi_{s,t}$ is derived
(Lemma~\ref{lemma:forgetting}) and not assumed, and the Markov
transitions need be neither contracting nor stable. Lastly, our assumptions allow $\eta_t(A)\to0$ for every bounded set $A\in\mathcal{E}$: in the linear
Gaussian model of Section~\ref{sub:LG_SSM} with $d=r=1$, $H_t(x)=F_{t+1}(x)=x$ and $y_t=t$, the predictive mean
drifts to infinity while the predictive variance remains bounded. This is excluded by approaches
based on a fixed Lyapunov function $V$ with bounded sublevel sets and for which $\sup_{t\geq1}\eta_t(V)<\infty$ \citep[see e.g.][]{Nick,del2023stability,arnaudon2024lyapunov}, since,
 by Markov's inequality, if $\sup_{t\geq1}\eta_t(V)<\infty$ then 
$\inf_{t\geq1}\eta_t(\{V\leq c\})>0$ for $c$ large enough.
 
\item \emph{Local errors in a weak norm.} Weighted $V$-norms and the total variation norm cannot
measure the local errors of particle filters. Our
abstract condition
\ref{Astar} (Section~\ref{sub:assumptions_gen})
decouples the distance used to measure the local
perturbations from the distance used to express the
forgetting properties of the semigroup. To use the
Kolmogorov norm in this framework, we extend the
Hardy--Krause variation and the Koksma--Hlawka-type
inequality of \citet{gotz2002discrepancy} from
$[0,1]^d$ to possibly unbounded hyper-rectangles
(Section~\ref{sec:preliminaries}).

\item \emph{Arbitrary local errors.} We study the propagation of arbitrary, possibly random,
local errors, rather than a specific scheme. As a consequence,   a
single result allows to cover a vast range of method used to approximate Feynman--Kac models, as
illustrated in  Sections~\ref{sec:compareFK}
and~\ref{sec:PF}.

\end{enumerate}

\subsection{Organization of the article\label{sub:organization}}

Section~\ref{sec:notation} fixes the notation and
recalls the variations of Vitali and of Hardy and
Krause. Section~\ref{sec:preliminaries} extends the
Hardy and
Krause variation and the Koksma--Hlawka inequality of
\citet{gotz2002discrepancy} to unbounded
hyper-rectangles 
and introduces the norm
$\|\cdot\|_{\mathcal{V}_\setX}$. Section~\ref{sec:main}
contains the perturbation analysis, first under an abstract
condition \ref{Astar},  and then under primitive conditions \ref{A1}--\ref{A3} and
\ref{B1} which allow us to derive perturbation bounds expressed in term of the total variation and Kolmogorov  metrics. Section~\ref{sec:examples} verifies
these primitive conditions on a large class of state-space models,
 while Sections~\ref{sec:compareFK}
and~\ref{sec:PF} contain the applications. Section \ref{sec:concludion} proposes some concluding remarks and, for the convenience of the reader,  an index of notation is
given in Appendix~\ref{sec:notation_index}. The proofs
not given in the text are provided in Appendix \ref{app:proofs}, 
following the order of the sections. 

\paragraph{Reading guide.}
Section~\ref{sec:preliminaries} is self-contained, and
only used for the Kolmogorov norm estimates of
Section~\ref{sub:K_estimates} and in
Section~\ref{sec:PF}. Readers interested in applications can read Section~\ref{sub:primitive} for
the regularity conditions, the statements of the results presented in
Sections~\ref{sub:TV_estimates}
and~\ref{sub:K_estimates}, and then
Sections~\ref{sec:examples} to~\ref{sec:PF}.
Section~\ref{sec:compareFK} only uses the total
variation estimates, while Section~\ref{sec:PF} only uses the
Kolmogorov norm estimates. 

\section{Notation and definitions\label{sec:notation}}

Section~\ref{sub:basic_notation} fixes the notation and
Section~\ref{sub:HK-var} recalls the variations of
Vitali and of Hardy and Krause.   An index of notation is given in Appendix~\ref{sec:notation_index}.

\subsection{Preliminary definitions\label{sub:basic_notation}}

\subsubsection{Vectors and matrices}

For any integer $d\geq 1$, we write $x=(x_1,\dots,x_d)$
the coordinates of a vector $x\in\R^d$, and
$x_n=(x_{n,1},\dots,x_{n,d})$ the coordinates of a
vector $x_n\in\R^d$ indexed by some parameter $n$. For
any $x\in\R^d$ and $c\in\R$, we set
$x+c\coloneqq(x_1+c,\dots,x_d+c)$ and we denote by
$\|x\|$ the Euclidean norm of $x$.

We denote by $I_d$ the $d\times d$ identity matrix and
by $\mathcal{S}_d^+$ the set of $d\times d$ symmetric
positive definite matrices. For any matrix $\Sigma\in\mathcal{S}_d^+$, we denote by
$\Sigma^{1/2}\in\mathcal{S}_d^+$ the unique matrix such
that $\Sigma=\Sigma^{1/2}\Sigma^{1/2}$. For any real
matrix $B$, we denote by $\|B\|$ its spectral norm and,
when $B$ is symmetric, by $\sigma_{\min}(B)$ and
$\sigma_{\max}(B)$ its smallest and largest eigenvalues.
For any $m\in\R^d$ and $\Sigma\in\mathcal{S}_d^+$, we
denote by $\mathcal{N}_d(m,\Sigma)$ the Gaussian
distribution on $\R^d$ with mean $m$ and covariance
matrix $\Sigma$.

Consider some (possibly empty) subset of coordinates
$u\subseteq\{1,\dots,d\}$ and set
$-u\coloneqq\{1,\dots,d\}\setminus u$.
When $u\neq\emptyset$ we denote by $1\leq
i_1<\dots<i_{|u|}\leq d$ its elements and, for any
vector $x\in\R^d$, we set
$x^u\coloneqq(x_{i_1},\dots,x_{i_{|u|}})\in\R^{|u|}$.
For any $z\in\R^{|u|}$ and $y\in\R^d$, we denote by
$\langle z,y\rangle_u$ the vector $x\in\R^d$ such that
$x^u=z$ and $x^{-u}=y^{-u}$; that is, $\langle
z,y\rangle_u$ is obtained by replacing in $y$ the
coordinates in $u$ by the ones of $z$. We use the
conventions $\langle z,y\rangle_u=z$ when
$u=\{1,\dots,d\}$ and $\langle z,y\rangle_u=y$ when
$u=\emptyset$.

\subsubsection{Hyper-rectangles}

Let $\bar{\R}\coloneqq\R\cup\{-\infty,\infty\}$.
Then, for any $a,b\in\bar{\R}^d$, we write $a\leq b$ (resp.\
$a<b$) when $a_i\leq b_i$ (resp.\ $a_i<b_i$)
for any $i\in\{1,\dots,d\}$, and we use the notation
\begin{align*}
[a,b]\coloneqq\times_{i=1}^d[a_i,b_i],
\qquad
(-\infty,b]\coloneqq\times_{i=1}^d(-\infty,b_i].
\end{align*}
Unless otherwise stated, hyper-rectangles
$[a,b]\subseteq\bar{\R}^d$ are only assumed to satisfy
$a\leq b$; that is, some of their sides may be reduced
to a point.

Let $\setZ$ be some set of
the form $\setZ\coloneqq\R^d\cap[\underline{z},\bar{z}]$
for some hyper-rectangle
$[\underline{z},\bar{z}]\subseteq\bar{\R}^d$;
for instance, we have $\setZ=\R^d$ when
$[\underline{z},\bar{z}]=\bar{\R}^d$. To  approximate $\setZ$ from inside by bounded
hyper-rectangles, we set
\begin{align*}
\mathcal{I}_\setZ\coloneqq
\{i\in\{1,\dots,d\}:\ \bar{z}_i=\infty\} 
\end{align*}
and we denote by $\mathcal{U}(\setZ)$ the set of
sequences $b_n$ in $\R^d$ such that, for any $n\geq 1$, we have $b_{n,i}=\bar{z}_i$
for any $i\notin\mathcal{I}_\setZ$, and
$b_{n,i}\to\infty$ as $n\to\infty$ for any
$i\in\mathcal{I}_\setZ$. Similarly, we denote by
$\mathcal{L}(\setZ)$ the set of sequences
$a_n$ in $\R^d$ such that, for any $n\geq 1$, we have
$a_{n,i}=\underline{z}_i$ for any $i$ with
$\underline{z}_i\in\R$, and $a_{n,i}\to-\infty$ as
$n\to\infty$ for any $i$ with
$\underline{z}_i=-\infty$.
Observe that, for any $(a_n)_{n\geq 1}\in\mathcal{L}(\setZ)$ and
$(b_n)_{n\geq 1}\in\mathcal{U}(\setZ)$, we have $a_n\leq b_n$ for any sufficiently large $n$.
Note also that $[a_n,b_n]$ is a bounded hyper-rectangle contained
in $\setZ$, and that $\bigcup_{n\geq m}[a_n,b_n]=\setZ$ for
any $m\geq 1$ if the sequences $a_n$ and $b_n$ are monotone; only such
 sequences  are considered below.
Finally, for any non-empty subset
$u\subseteq\{1,\dots,d\}$, we set
$\setZ_u\coloneqq
\R^{|u|}\cap[\underline{z}^u,\bar{z}^u]$.

\subsubsection{Total variation and Kolmogorov norms}

The total variation norm of a finite signed measure
$\nu\in\mathcal{M}(\setX)$ is defined by
\begin{align*}
\|\nu\|_{\mathrm{tv}}
\coloneqq\sup_{A\in\mathcal{E}}|\nu(A)|
=\sup_{f\in\mathcal{B}_1(\setX)}|\nu(f)|,
\end{align*}
where we recall that $\mathcal{B}_1(\setX)$ stands for
the set of measurable functions $f\colon\setX\to[0,1]$
and $\nu(f)=\int_\setX f\,\dd\nu$. Note that the quantity
$\|\mu_1-\mu_2\|_{\mathrm{tv}} \in[0,1]$ coincides with
the usual total variation distance between the
probability measures
$\mu_1,\mu_2\in\mathcal{P}(\setX)$. In addition, for
any $\nu\in\mathcal{M}(\setX)$ with total variation
measure $|\nu|$, we have
$\tfrac12|\nu|(\setX)\leq\|\nu\|_{\mathrm{tv}}
\leq|\nu|(\setX)$.

When $\setX\subseteq\R^d$, the Kolmogorov norm of a
finite signed measure $\nu\in\mathcal{M}(\setX)$ is
defined by
\begin{align}\label{eq:K-norm}
\|\nu\|_{\mathrm{K}}\coloneqq
\sup_{b\in\setX}\big|\nu\big((-\infty,b]\cap\setX\big)\big|.
\end{align}
For probability measures $\mu_1,\mu_2\in\mathcal{P}(\setX)$, the quanity
$\|\mu_1-\mu_2\|_{\mathrm{K}}$ coincides with the
supremum distance between the multivariate cumulative
distribution functions of $\mu_1$ and $\mu_2$. The
supremum in \eqref{eq:K-norm} being taken over a subset
of $\mathcal{E}$, we   have
$\|\nu\|_{\mathrm{K}}\leq\|\nu\|_{\mathrm{tv}}$ for any
$\nu\in\mathcal{M}(\setX)$.

\subsubsection{Functions and function spaces}

For any function $f\colon\setX\to\R$, we set
$\|f\|_\infty\coloneqq\sup_{x\in\setX}|f(x)|$ and below we introduces some functions spaces that will be considered in this work.

\paragraph{Pseudo-decreasing functions.}
For any $\varepsilon\in(0,1]$, a function
$f\colon[0,\infty)\to(0,\infty)$ is said to be
\emph{$\varepsilon$-pseudo-decreasing} if
\begin{align*}
f(z)\geq\varepsilon\,f(z')
\qquad\text{for any } 0\leq z\leq z',
\end{align*}
and \emph{pseudo-decreasing} if it is
$\varepsilon$-pseudo-decreasing for some
$\varepsilon\in(0,1]$. Any non-increasing
function is $1$-pseudo-decreasing. More generally, when
$f_1\colon[0,\infty)\to(0,\infty)$ is non-increasing
and $f_2\colon[0,\infty)\to[c,C]$ for some constants
$0<c\leq C<\infty$, the product $f_1f_2$ is
$(c/C)$-pseudo-decreasing.

\paragraph{Mixed partial derivatives.}
When $\setX\subseteq\R^d$, we denote by
$\mathcal{C}_{\mathrm{mixed}}(\setX)$ the set of
functions $f\colon\setX\to\R$ with continuous mixed
partial derivatives of order one in each coordinate;
that is, for any non-empty subset
$u\subseteq\{1,\dots,d\}$ with elements
$i_1<\dots<i_{|u|}$, the function
\begin{align*}
x\in\setX\mapsto
\partial^u f(x)\coloneqq
\frac{\partial^{|u|}f}
{\partial x_{i_1}\cdots\partial x_{i_{|u|}}}(x)
\end{align*}
is well defined and continuous. For any function
$f\colon\setX^2\to\R$ such that
$f(\cdot,x')\in\mathcal{C}_{\mathrm{mixed}}(\setX)$ for
any $x'\in\setX$, we denote by $\partial_1^u f(x,x')$
the derivative $\partial^u$ with respect to the first
coordinate $x$.

For any integer  $r\geq 1$, and any function
$g\colon\R^d\to[0,\infty]$, we denote by
$\mathcal{C}_{\mathrm{mixed}}(\R^d;r,g)$ the set of
functions $f=(f_1,\dots,f_r)\colon\R^d\to\R^r$ such
that, for any $i\in\{1,\dots,r\}$, we have
$f_i\in\mathcal{C}_{\mathrm{mixed}}(\R^d)$ and
\begin{align*}
|\partial^u f_i(x)|\leq g(x),
\qquad\forall x\in\R^d,
\quad\forall\,\emptyset\neq u\subseteq\{1,\dots,d\}.
\end{align*}
For such functions, we set $\partial^u f\coloneqq(\partial^u f_1,\dots,\partial^u f_r)$.
When $g\equiv C$ is a constant function, for some
$C\in(0,\infty]$, we simply write
$\mathcal{C}_{\mathrm{mixed}}(\R^d;r,C)$.

\paragraph{Bi-Lipschitz functions.}
For any $L\in[1,\infty)$ and $r\geq 1$, a function
$f\colon\R^d\to\R^r$ is said to be
\emph{$L$-bi-Lipschitz} if
\begin{align*}
\frac{1}{L}\|x-x'\|\leq\|f(x)-f(x')\|\leq L\|x-x'\|,
\qquad\forall x,x'\in\R^d,
\end{align*}
and \emph{bi-Lipschitz} if it is
$L$-bi-Lipschitz for some $L\in[1,\infty)$.
Equivalently, $f$ is bi-Lipschitz if and only if it is
injective and Lipschitz continuous, with a Lipschitz
continuous inverse $f^{-1}\colon f(\R^d)\to\R^d$. Bi-Lipschitz functions $f\colon\R^d\to\R^r$ only
exist when $r\geq d$, and they are necessarily
unbounded.

\subsection{The Hardy--Krause variation\label{sub:HK-var}}

The Hardy--Krause variation (a.k.a.\ HK-variation)
extends the notion of total variation to functions of several
variables, through the variation in the sense of
Vitali. 

\subsubsection{Variation in the sense of Vitali}

Let $k\geq 1$ and   $[\tilde{a},\tilde{b}]\subset\R^k$
be some bounded hyper-rectangle with $\tilde{a}<\tilde{b}$. A \emph{grid
splitting}
of $[\tilde{a},\tilde{b}]$ is a finite collection of
hyper-rectangles of the form
\begin{align*}
P=\Big\{\textstyle\times_{i=1}^k
\big[x^{(i)}_{j_i},x^{(i)}_{j_i+1}\big]:\
j_i\in\{0,\dots,m_i-1\},\ i\in\{1,\dots,k\}\Big\},
\end{align*}
where, for each $i$, we have $m_i\geq 1$ and
$\tilde{a}_i=x^{(i)}_0<\dots<x^{(i)}_{m_i}=\tilde{b}_i$.
In what follows, we denote by $\mathsf{P}([\tilde{a},\tilde{b}])$ the
set of grid splittings of $[\tilde{a},\tilde{b}]$; note
that $\bigcup_{R\in P}R=[\tilde{a},\tilde{b}]$ for any
$P\in\mathsf{P}([\tilde{a},\tilde{b}])$. Next, for any
function $f\colon\R^k\to\R$, any hyper-rectangle
$R=[\hat{a},\hat{b}]\subset\R^k$ with $\hat{a}\leq\hat{b}$ and any subset $w\subseteq\{1,\dots,k\}$, we let
\begin{align}\label{eq:alternating}
\Alt_w(f;R)\coloneqq
\sum_{u\subseteq w}(-1)^{|u|}\,
f\big(\langle\hat{a}^u,\hat{b}\rangle_u\big)
\end{align}
denote the alternating sum of the values of $f$ at the
vertices of the face $\{x\in R:\ x^{-w}=\hat{b}^{-w}\}$ of $R$ that contains its upper corner $\hat{b}$.
In particular, $\Alt_\emptyset(f;R)=f(\hat{b})$, and we write
$\Alt(f;R)\coloneqq\Alt_{\{1,\dots,k\}}(f;R)$ for the alternating sum of the values of $f$ at the
vertices of $R$.

With this notation in place, the variation of $f$ on
$[\tilde{a},\tilde{b}]$ in the sense of Vitali is then defined by
\begin{align*}
\tilde{\mathcal{V}}_{[\tilde{a},\tilde{b}]}(f)
\coloneqq
\sup_{P\in\mathsf{P}([\tilde{a},\tilde{b}])}
\sum_{R\in P}\big|\Alt(f;R)\big|.
\end{align*}
When $[\tilde{a},\tilde{b}]\subset\R^k$ is a bounded hyper-rectangle with some side reduced to
a point, so that no grid splitting of $[\tilde{a},\tilde{b}]$ exists, we set
$\tilde{\mathcal{V}}_{[\tilde{a},\tilde{b}]}(f)
\coloneqq0$;  this convention is consistent with the fact that
$\Alt(f;R)=0$ for any hyper-rectangle $R$ with a
side reduced to a point.
When the mixed derivative
$\partial^{\{1,\dots,k\}}f$ is well defined and continuous on
$[\tilde{a},\tilde{b}]$ we have
\citep[][Proposition~14]{owen2005multidimensional}
\begin{align}\label{eq:smoothV}
\tilde{\mathcal{V}}_{[\tilde{a},\tilde{b}]}(f)
=\int_{[\tilde{a},\tilde{b}]}
\big|\partial^{\{1,\dots,k\}}f(x)\big|\,\dd x.
\end{align}

\subsubsection{Variation in the sense of Hardy and Krause.}
The Vitali variation of a function that does not depend
on all its coordinates is null  \citep[][Proposition 8]{owen2005multidimensional}. To avoid this problem, the HK-variation adds
the Vitali variations of the restrictions of $f$ to the
faces containing the upper corner. Specifically, for any $d\geq 1$ and any bounded
hyper-rectangle $[a,b]\subset\R^d$, the HK-variation of
$f\colon\R^d\to\R$ on $[a,b]$ anchored at $b$ is
defined by
\begin{align}\label{eq:HK}
\mathcal{V}_{[a,b]}(f)\coloneqq
\sum_{\emptyset\neq u\subseteq\{1,\dots,d\}}
\tilde{\mathcal{V}}_{[a^u,b^u]}
\big(f(\langle\cdot,b\rangle_u)\big).
\end{align}
When $d=1$,  the quantity $\mathcal{V}_{[a,b]}(f)$ is the total
variation of $f$ on $[a,b]$. The anchor $b$ is chosen
so that indicators of orthants have HK-variation at
most one; this links the HK-variation $\mathcal{V}$ to the Kolmogorov
norm.

\begin{example}\label{ex:HK}
Let $d=2$, $[a,b]=[0,1]^2$, $c\in(0,1)^2$ and
$f\coloneqq\ind_{(-\infty,c]}$. The two terms of
\eqref{eq:HK} with $|u|=1$ vanish, since
$z\mapsto f(z,1)$ and $z\mapsto f(1,z)$ are
identically zero. For $R=[\hat{a},\hat{b}]$,
\begin{align*}
\Alt(f;R)
=\big(\ind\{\hat{b}_1\leq c_1\}-\ind\{\hat{a}_1\leq c_1\}\big)
\big(\ind\{\hat{b}_2\leq c_2\}-\ind\{\hat{a}_2\leq c_2\}\big)
=\prod_{i=1}^2\ind\{\hat{a}_i\leq c_i<\hat{b}_i\},
\end{align*}
and every grid splitting of $[0,1]^2$ contains
exactly one rectangle with $\hat{a}\leq c<\hat{b}$.
Hence $\mathcal{V}_{[0,1]^2}(f)=1$, in accordance
with part~\ref{HK_6} of Proposition~\ref{prop:V_HK}.
Anchored at $a=(0,0)$ instead, the terms with
$|u|=1$ involve $z\mapsto\ind\{z\leq c_1\}$ and
$z\mapsto\ind\{z\leq c_2\}$, each with a jump of size
one, and the variation equals $3$; in dimension $d$,
it equals $2^d-1$.
\end{example}

In this work, we will use the following properties of the HK-variation.
\begin{proposition}\label{prop:order}
Let $[a,b]\subset\R^d$ be a bounded hyper-rectangle
and let $f\colon\R^d\to\R$. Then:
\begin{enumerate}[label=(\roman*)]
\item\label{basic_HK1} $\mathcal{V}_{[\tilde{a},\tilde{b}]}(f)
\leq\mathcal{V}_{[a,b]}(f)$ for any hyper-rectangle
$[\tilde{a},\tilde{b}]\subseteq[a,b]$, where, by definition, $\mathcal{V}_{[\tilde{a},\tilde{b}]}(f)$ is anchored at $\tilde{b}$;
\item\label{basic_HK2} when there exists a non-empty subset
$v\subsetneq\{1,\dots,d\}$ such that
$f(x)=f(\langle x^v,b\rangle_v)$ for any $x\in[a,b]$ we have
$\mathcal{V}_{[a,b]}(f)=
\mathcal{V}_{[a^v,b^v]}\big(f(\langle\cdot,b\rangle_v)\big)$, where the right-hand side is the HK-variation of the function
$f(\langle\cdot,b\rangle_v)\colon\R^{|v|}\to\R$ on $[a^v,b^v]$, anchored at $b^v$;
\item\label{basic_HK3} $\mathcal{V}_{[a,b]}(f)=0$ if and only if $f$ is
constant on $[a,b]$. 
\end{enumerate}
\end{proposition}

\begin{proof}
See Appendix~\ref{p-prop:order}.
\end{proof}

\section{A Hardy--Krause dual norm\label{sec:preliminaries}}

The Koksma--Hlawka-type inequality of
\citet{gotz2002discrepancy} bounds
$|\mu_1(f)-\mu_2(f)|$ by the product of the
HK-variation of $f$ and the Kolmogorov distance between
$\mu_1$ and $\mu_2$, for probability measures on
$[0,1]^d$ (Lemma~\ref{lemma:aistleitner}). In the following, we extend the
HK-variation   and this
inequality   to unbounded
hyper-rectangles, and introduce the norm
$\|\cdot\|_{\mathcal{V}_\setX}$.

Throughout this section, we fix some $d\geq 1$ and we
set
\begin{align}\label{eq:setX}
  \setX\coloneqq\R^d\cap[\underline{x},\bar{x}]
\end{align}
for some hyper-rectangle
$[\underline{x},\bar{x}]\subseteq\bar{\R}^d$; for instance $\setX=\R^d$.

\subsection{The Koksma--Hlawka inequality\label{sub:KH}}

The Kolmogorov metric and the HK-variation are
connected by the following result, due to
\citet[][Corollary~2.4]{gotz2002discrepancy}.

\begin{lemma}\label{lemma:aistleitner}
For any $\mu_1,\mu_2\in\mathcal{P}([0,1]^d)$ and any
measurable function $f\colon[0,1]^d\to\R$ we have
\begin{align*}
  |\mu_1(f)-\mu_2(f)|
  \leq
  \mathcal{V}_{[0,1]^d}(f)\,\|\mu_1-\mu_2\|_{\mathrm{K}}.
\end{align*}
\end{lemma}

In the scenario where $\mu_1$ is the uniform distribution on $[0,1]^d$ and $\mu_2$ is
finitely supported, the inequality stated in Lemma \ref{lemma:aistleitner} goes back to Koksma (1942/1943) for
$d=1$ and to Hlawka (1961) for $d\geq 1$. We  note that the result of the lemma when $\mu_2$ is an  empirical measures  has also been obtained by
\citet{aistleitner2014functions}.

\subsection{The HK-variation on unbounded
hyper-rectangles\label{sub:HK_gen}}

\subsubsection{Definition and main properties}

For any function $f\colon\setX\to\R$, the HK-variation
of $f$ on $\setX$ is defined by
\begin{align}\label{eq:genHK}
  \mathcal{V}_\setX(f)\coloneqq
  \lim_{n\to\infty}\mathcal{V}_{[a_n,b_n]}(f)
  \in[0,\infty],
\end{align}
where $(a_n)_{n\geq 1}\in\mathcal{L}(\setX)$ and
$(b_n)_{n\geq 1}\in\mathcal{U}(\setX)$ are as in Section~\ref{sub:basic_notation}, and
$\mathcal{V}_{[a_n,b_n]}(f)$ is given by \eqref{eq:HK}.
As shown in the following lemma, the limit in \eqref{eq:genHK} exists and does not depend on the sequences.

\begin{lemma}\label{lemma:V_well_def}
Let $f\colon\setX\to\R$, and consider some sequences
$(a_n)_{n\geq 1}\in\mathcal{L}(\setX)$ and
$(b_n)_{n\geq 1}\in\mathcal{U}(\setX)$. Then
\begin{align*}
  l\coloneqq
  \limsup_{n\to\infty}\mathcal{V}_{[a_n,b_n]}(f)
  =\liminf_{n\to\infty}\mathcal{V}_{[a_n,b_n]}(f),
\end{align*}
and $\lim_{n\to\infty}\mathcal{V}_{[a'_n,b'_n]}(f)=l$
for any sequences $(a'_n)_{n\geq 1}\in\mathcal{L}(\setX)$ and
$(b'_n)_{n\geq 1}\in\mathcal{U}(\setX)$.
\end{lemma}
\begin{proof}
See Appendix~\ref{p-lemma:V_well_def}.
\end{proof}

We observe that when the $\setX=[\underline{x},\bar{x}]$ is bounded, the definition
\eqref{eq:HK} and \eqref{eq:genHK} of $\mathcal{V}_\setX(f)$ coincide. Below, we denote
by
$\mathcal{B}_{\mathrm{HK}}(\setX)$ the set of functions
$f\colon\setX\to\R$ such that
$\mathcal{V}_\setX(f)<\infty$. 

\begin{proposition}\label{prop:set_BHK}
Any function $f\in\mathcal{B}_{\mathrm{HK}}(\setX)$ is
bounded and measurable.
\end{proposition}
\begin{proof}
When $\setX=[0,1]^d$, the result is proved in
\citet[][Theorem~3.1]{aistleitner2017functions}. For the general definition \eqref{eq:setX} of  the set $\setX$, see
Appendix~\ref{p-prop:set_BHK}.
\end{proof}

\begin{proposition}\label{prop:V_HK}
The following hold.
\begin{enumerate}[label=(\roman*)]
\item\label{HK_1}
  $\mathcal{V}_\setX(f+g)\leq
  \mathcal{V}_\setX(f)+\mathcal{V}_\setX(g)$ for any
functions $f,g\colon\setX\to\R$.
\item\label{HK_2}
  $\mathcal{V}_\setX(cf)=|c|\,\mathcal{V}_\setX(f)$
  for any $c\in\R$ and any function
  $f\colon\setX\to\R$.
\item\label{HK_3}
  $\mathcal{V}_\setX(f)=0$ if and only if the function
  $f\colon\setX\to\R$ is constant on $\setX$.
\item\label{HK_4}
  $\mathcal{V}_\setX(\ind_{[a,b]})\leq 3^d-1$ for any bounded hyper-rectangle $[a,b]\subseteq\setX$.
\item\label{HK_5}
  $\mathcal{V}_\setX(fg)\leq\frac{1}{3}
  \big(3\mathcal{V}_\setX(f)+\|f\|_\infty\big)
  \big(3\mathcal{V}_\setX(g)+\|g\|_\infty\big)$ for any
functions $f,g\colon\setX\to\R$.
\item\label{HK_6}
  $\mathcal{V}_\setX(\ind_{\setX\cap(-\infty,b]})=1$
  for any $b\in\setX$ such that
  $\setX\cap(-\infty,b]\subsetneq\setX$.
\item\label{HK_7}
  Let $\nu$ be some $\sigma$-finite measure on $\setX$
and let $f\colon\setX\times\setX\to\R$ be some
measurable function such that $g(x)\coloneqq\int_\setX
f(x,x')\,\nu(\dd x')\in\R$ for any $x\in\setX$. Then
the function
$x'\mapsto\mathcal{V}_\setX\big(f(\cdot,x')\big)$ is
universally measurable and, letting $\bar{\nu}$ denote the completion of $\nu$,
  \begin{align*}
    \mathcal{V}_\setX(g)\leq
    \int_\setX\mathcal{V}_\setX\big(f(\cdot,x')\big)
    \,\bar{\nu}(\dd x').
  \end{align*}
\end{enumerate}
\end{proposition}
\begin{proof}
Parts~\ref{HK_1}--\ref{HK_2} follow from the definition \eqref{eq:genHK}, and part \ref{HK_3}
follows from part \ref{basic_HK3} of Proposition~\ref{prop:order}
combined with \eqref{eq:genHK}.  To check part
\ref{HK_4}, observe that
$\mathcal{V}_\setX(\ind_{[a,b]})=
\mathcal{V}_{[\tilde{a},\tilde{b}]}(\ind_{[a,b]})$
for some bounded hyper-rectangle
$[\tilde{a},\tilde{b}]\subseteq\setX$ containing $[a,b]$; the estimate then follows from
\citet[][Proposition~20]{owen2005multidimensional}.
When $\setX=[0,1]^d$, part \ref{HK_5} is a consequence
of
\citet[][Corollary~4.3 and
Lemma~5.2]{aistleitner2017functions}. The extension of
part \ref{HK_5} to the set $\setX$ defined in
\eqref{eq:setX}, as well as the proofs of parts \ref{HK_6}--\ref{HK_7},
are given in Appendix~\ref{p-prop:V_HK}.
\end{proof}

By part~\ref{HK_6} of Proposition \ref{prop:V_HK}, the indicator functions of the sets
$\setX\cap(-\infty,b]$, which define the Kolmogorov norm, have HK-variation one; this relies on the
choice of the upper corner as the anchor, as illustrated in Example \ref{ex:HK}.

\subsubsection{A Banach space of functions}

By Proposition~\ref{prop:set_BHK}, we have
$\|f\|_\infty<\infty$ for any
$f\in\mathcal{B}_{\mathrm{HK}}(\setX)$, and we then set
\begin{align}\label{eq:HK_norm}
  \|f\|_{\mathrm{HK}}\coloneqq
  \|f\|_\infty+\mathcal{V}_\setX(f),
  \qquad f\in\mathcal{B}_{\mathrm{HK}}(\setX).
\end{align}
Observe that, by parts \ref{HK_1}--\ref{HK_3} of
Proposition~\ref{prop:V_HK}, the mapping $\|\cdot\|_{\mathrm{HK}}$ is a norm on
$\mathcal{B}_{\mathrm{HK}}(\setX)$.

\begin{proposition}\label{prop:Banach}
The normed space
$(\mathcal{B}_{\mathrm{HK}}(\setX),\|\cdot\|_{\mathrm{HK}})$
is a Banach space.
\end{proposition}
\begin{proof}
When $\setX=[0,1]^d$, the result follows from
\citet[][Corollary~4.3 and
Lemma~5.1]{aistleitner2017functions}; the extension of this result to
the set $\setX$ defined in \eqref{eq:setX} is given in Appendix~\ref{p-prop:Banach}.
\end{proof}

Proposition~\ref{prop:Banach} is not used below. However, its
shows that the quantity $\|\nu\|_{\mathcal{V}_\setX}$ defined in
Section~\ref{sub:New_norm} is the norm of the linear
functional $f\mapsto\nu(f)$ on this Banach space,
whence the title of this section.

\subsubsection{Smooth functions\label{sub:Hardy_smooth}}

For any function $f\in\mathcal{C}_{\mathrm{mixed}}(\setX)$ we
set
\begin{align*}
  V_\setX(f)\coloneqq
  \sum_{\substack{\mathcal{I}_\setX\subseteq u
  \subseteq\{1,\dots,d\}\\ u\neq\emptyset}}
  \int_{\setX_u}
  \big|\partial^u f\big(\langle z,\bar{x}\rangle_u\big)\big|
  \,\lambda_{|u|}(\dd z),
\end{align*}
where we recall that $\mathcal{I}_\setX$ stands for the
set of coordinates in which $\setX$ is not bounded from
above. Only subsets $u$ containing $\mathcal{I}_\setX$ appear in the definition of  $V_\setX(f)$, and thus $\langle z,\bar{x}\rangle_u\in\setX$ for all $z\in\setX_u$. Below, we denote by
$\mathcal{C}_{\mathrm{HK}}(\setX)$ the set of functions
$f\in\mathcal{C}_{\mathrm{mixed}}(\setX)$ such that
$V_\setX(f)<\infty$ and
$\lim_{n\to\infty}V_{\setX\cap(-\infty,b_n]}(f)
=V_\setX(f)$ for some sequence
$(b_n)_{n\geq 1}\in\mathcal{U}(\setX)$.

When $\setX$ is a bounded hyper-rectangle,
by \eqref{eq:smoothV} and \eqref{eq:HK} we have
$\mathcal{V}_\setX(f)=V_\setX(f)$ for any
$f\in\mathcal{C}_{\mathrm{HK}}(\setX)$. As shown in the next result, this property still holds when $\setX$ is unbounded.

\begin{proposition}\label{prop:variationHK}
We have
$\mathcal{C}_{\mathrm{HK}}(\setX)\subset
\mathcal{B}_{\mathrm{HK}}(\setX)$, and
$\mathcal{V}_\setX(f)=V_\setX(f)$ for any
$f\in\mathcal{C}_{\mathrm{HK}}(\setX)$.
\end{proposition}

\begin{proof} See Appendix~\ref{p-prop:variationHK}. \end{proof}

For instance, when $\setX=\R^d$,
Proposition~\ref{prop:variationHK} yields the formula
\begin{align*}
  \mathcal{V}_{\R^d}(f)=
  \int_{\R^d}
  \Big|\frac{\partial^d f}{\partial x_1\cdots\partial x_d}(x)\Big|
  \,\lambda_d(\dd x),
  \qquad\forall f\in\mathcal{C}_{\mathrm{HK}}(\R^d).
\end{align*}

\subsection{A Koksma--Hlawka inequality on unbounded
hyper-rectangles\label{sub:KH_gen}}

Recall from Proposition~\ref{prop:set_BHK} that $\mu(f)$ is well
defined for any $\mu\in\mathcal{P}(\setX)$ and
$f\in\mathcal{B}_{\mathrm{HK}}(\setX)$.

\begin{proposition}\label{prop:Multivariate_KH1}
For any $\mu_1,\mu_2\in\mathcal{P}(\setX)$ and any
$f\in\mathcal{B}_{\mathrm{HK}}(\setX)$, we have
\begin{align*}
  \big|\mu_1(f)-\mu_2(f)\big|
  \leq
  \mathcal{V}_\setX(f)\,\|\mu_1-\mu_2\|_{\mathrm{K}}.
\end{align*}
\end{proposition}
\begin{proof}
The proof is based on Lemma~\ref{lemma:aistleitner} and
is given in Appendix~\ref{p-prop:Multivariate_KH1}.
\end{proof}

Observe that, for $\setX=[0,1]^d$, Proposition \ref{prop:Multivariate_KH1} coincides with
Lemma~\ref{lemma:aistleitner}. Since the bound is
attained by the indicator functions of the sets
$\setX\cap(-\infty,b]$ (part~\ref{HK_6} of
Proposition~\ref{prop:V_HK}), we obtain:

\begin{corollary}\label{cor:Kolmogorov_sup}
For any $\mu_1,\mu_2\in\mathcal{P}(\setX)$, we have
\begin{align*} \|\mu_1-\mu_2\|_{\mathrm{K}}
  =\sup_{f\colon\setX\to\R\,:\,
  \mathcal{V}_\setX(f)\leq 1}|\mu_1(f)-\mu_2(f)|
  =\sup_{f\in\mathcal{B}_1(\setX)\,:\,
  \mathcal{V}_\setX(f)\leq 1}|\mu_1(f)-\mu_2(f)|.
\end{align*}
\end{corollary}

\subsection{A norm on finite signed
measures\label{sub:New_norm}}

For any finite signed-measure $\nu\in\mathcal{M}(\setX)$ we set
\begin{align*}
  \|\nu\|_{\mathcal{V}_\setX}\coloneqq
  \sup_{f\in\mathcal{B}_{\mathrm{HK}}(\setX)\,:\,
  \|f\|_{\mathrm{HK}}\leq 1}|\nu(f)|,
\end{align*}
with the norm $\|\cdot\|_{\mathrm{HK}}$ defined in
\eqref{eq:HK_norm}; note that this quantity is finite, since
$|\nu(f)|\leq\|f\|_\infty|\nu|(\setX)$.

\begin{theorem}\label{thm:norm}
The mapping
$\|\cdot\|_{\mathcal{V}_\setX}\colon
\mathcal{M}(\setX)\to[0,\infty)$ is a norm on $\mathcal{M}(\setX)$ satisfying the
following properties:
\begin{enumerate}[label=(\roman*)]
\item\label{norm_i}
$\frac{1}{2}\|\nu\|_{\mathrm{K}}\leq
\|\nu\|_{\mathcal{V}_\setX}$ for any
$\nu\in\mathcal{M}(\setX)$;
\item\label{norm_ii}
$\|\mu_1-\mu_2\|_{\mathcal{V}_\setX}\leq
\|\mu_1-\mu_2\|_{\mathrm{K}}$ for any
$\mu_1,\mu_2\in\mathcal{P}(\setX)$.
\end{enumerate}
In particular, for any
$\mu_1,\mu_2\in\mathcal{P}(\setX)$, we have
\begin{align*}
\frac{1}{2}\,\|\mu_1-\mu_2\|_{\mathrm{K}}
\leq
\|\mu_1-\mu_2\|_{\mathcal{V}_\setX}
\leq
\|\mu_1-\mu_2\|_{\mathrm{K}}.
\end{align*}
\end{theorem}
\begin{proof}
Homogeneity and the triangle inequality follow from the
definition of
$\|\cdot\|_{\mathcal{V}_\setX}$. To check \ref{norm_i}, consider some
$\nu\in\mathcal{M}(\setX)$ and $b\in\setX$, and set
$f\coloneqq\frac12\ind_{\setX\cap(-\infty,b]}$. Then $\|f\|_\infty\leq\frac12$
and, by parts \ref{HK_2},
\ref{HK_3} and \ref{HK_6} of
Proposition~\ref{prop:V_HK},
$\mathcal{V}_\setX(f)\leq\frac12$. This implies that
$\|f\|_{\mathrm{HK}}\leq 1$ and therefore
$\|\nu\|_{\mathcal{V}_\setX}\geq|\nu(f)|
=\frac12|\nu(\setX\cap(-\infty,b])|$. Taking the supremum over $b\in\setX$ we check
\ref{norm_i}. Since a finite signed measure on $\setX$
is characterized by its values on the sets $\setX\cap(-\infty,b]$,
$b\in\setX$, the mapping $\|\cdot\|_{\mathrm{K}}$ is a
norm on $\mathcal{M}(\setX)$; thus, by \ref{norm_i},
$\|\nu\|_{\mathcal{V}_\setX}=0$ implies $\nu=0$.
Finally, \ref{norm_ii} follows from
Proposition~\ref{prop:Multivariate_KH1}, recalling that
$\mathcal{V}_\setX(f)\leq\|f\|_{\mathrm{HK}}\leq 1$
for any function $f$ in the supremum defining
$\|\mu_1-\mu_2\|_{\mathcal{V}_\setX}$.
\end{proof}

\section{Perturbation analysis\label{sec:main}}

Sections~\ref{sub:assumptions_gen}
and~\ref{sub:perturbation_gen} treat an abstract
framework, where the local errors are measured with a
pseudo-metric $d_2$, the stability of the semigroup is
expressed with a pseudo-metric $d_1$, and condition
\ref{Astar} ensures that local errors do not
accumulate over time.
Sections~\ref{sub:primitive}--\ref{sub:K_estimates}
give primitive conditions on the model, derive the
stability of the semigroup and check \ref{Astar} in
total variation  and in
the Kolmogorov norm,
which yields \eqref{eq:summaryMain}.

\subsection{An abstract framework\label{sub:assumptions_gen}}

Let $\mathcal{F}_2(\setX)\subseteq\mathcal{F}_1(\setX)$
be two non-empty sets of measurable functions from
$\setX$ to $\R$. For $i\in\{1,2\}$, let
$\mathcal{P}_i(\setX)$ be the set of probability
measures $\mu\in\mathcal{P}(\setX)$ such that
$\mu(f)\in\R$ for any $f\in\F_i(\setX)$, and let
$d_i\colon\mathcal{P}_i(\setX)\times
\mathcal{P}_i(\setX)\to[0,\infty]$ be defined, for any
$\mu_1,\mu_2\in\mathcal{P}_i(\setX)$, by
\begin{align*}
d_i(\mu_1,\mu_2)\coloneqq
\sup_{f\in\mathcal{F}_i(\setX)}|\mu_1(f)-\mu_2(f)|.
\end{align*}
Here $d_i(\mu_1,\mu_2)=0$ does not imply $\mu_1=\mu_2$,
and the functions in $\F_1(\setX)$ and $\F_2(\setX)$
may be unbounded. For any $t\geq 1$ we let
$\mathcal{P}_{W_t|1}(\setX)\coloneqq
\mathcal{P}_{W_t}(\setX)\cap\mathcal{P}_1(\setX)$, with
the sets $W_t$ as in \ref{A0} and with $\mathcal{P}_{W_t}(\setX)$ as defined in Section \ref{sub:FK-model}, and we assume that the set
$\F_1(\setX)$ is chosen so that
$\mathcal{P}_{W_1|1}(\setX)\neq\emptyset$ and, for any $t\geq 1$,
\begin{align}\label{eq:regul_d1}
\Phi_{t+1}(\mu)\in\mathcal{P}_1(\setX)
\quad\text{for any }\mu\in\mathcal{P}_{W_t|1}(\setX).
\end{align}
By \ref{A0}, $\Phi_{t+1}$ then maps
$\mathcal{P}_{W_t|1}(\setX)$ into
$\mathcal{P}_{W_{t+1}|1}(\setX)$ for any $t\geq 1$.
We then consider the Feynman--Kac model
\begin{align}\label{eq:FK_pred_model}
\eta_{t+1}=\Phi_{t+1}(\eta_t)
\quad\text{and}\quad
\widehat{\eta}_t=\Psi_t(\eta_t),
\quad t\geq 1,
\qquad\text{with}\quad
\eta_1\in\mathcal{P}_{W_1|1}(\setX).
\end{align}

Our analysis relies on the following condition.
\begin{enumerate}[label=\textbf{(A$^\star$)}]
\setcounter{enumi}{-1}
\item\label{Astar} There exist
(i) a sequence $f_t$ of non-negative
functions in $\mathcal{F}_2(\setX)$,
(ii) a constant $\delta\in(0,1)$,
(iii) a function
$\varphi_1\colon[0,\infty)^2\to[0,\infty]$,
non-increasing in each coordinate, such that
$\varphi_1(z,z')=\infty\Leftrightarrow\min\{z,z'\}=0$,
and
(iv) a non-increasing function
$\varphi_2\colon[0,\infty)\to[0,\infty]$ such that
$\varphi_2(z)=\infty\Leftrightarrow z=0$,
such that, for any $t>s\geq 1$ and
$\mu_1,\mu_2\in\mathcal{P}_{W_s|1}(\setX)$, we have the
exponential forgetting property
\begin{align}\label{eq:forget_past}
d_1\big(\Phi_{s,t}(\mu_1),\Phi_{s,t}(\mu_2)\big)
\leq\varphi_1\big(\mu_1(f_s),\mu_2(f_s)\big)
\,e^{-\delta(t-s)}
\end{align}
and the local Lipschitz property
\begin{align}\label{eq:lipschitz_gen}
d_1\big(\Phi_{s,t}(\mu_1),\Phi_{s,t}(\mu_2)\big)
\leq\varphi_2\big(\mu_1(f_s)\big)\,
d_2(\mu_1,\mu_2).
\end{align}
In addition, for any
$\mu\in\mathcal{P}_{W_1|1}(\setX)$, we have the
implication
\begin{align}\label{eq:lower_gen}
\mu(f_1)>0\implies
\inf_{t\geq 2}\Phi_{1,t}(\mu)(f_t)>0.
\end{align}
\end{enumerate}

\subsection{A general perturbation bound\label{sub:perturbation_gen}}

Let $\bm{\mu}=(\mu_t)_{t\geq 1}$ be a sequence with
$\mu_t\in\mathcal{P}_{W_t|1}(\setX)$, and let
$\eta\in\mathcal{P}_{W_1|1}(\setX)$. We let, for any $T\geq 1$,
\begin{align}\label{eq:maxPerturbation}
\Delta_{T,\bm{\mu}|\eta}\coloneqq
d_2(\mu_1,\eta)\vee\max_{1<t\leq T}
d_2\big(\mu_t,\Phi_t(\mu_{t-1})\big)
\end{align}
denote the largest local perturbation up to time $T$.
When \ref{Astar} is satisfied for some functions
$f_t$, we also set
\begin{align*}
\iota_{T,\bm{\mu}|\eta}\coloneqq
\eta(f_1)\wedge\min_{1\leq t<T}
\Big(\mu_t(f_t)\wedge\Phi_{t+1}(\mu_t)(f_{t+1})\Big).
\end{align*}
Finally, we consider the function
$\varphi\colon[0,\infty)\to[1,\infty]$ defined by
\begin{align}\label{eq:def_varphi}
\varphi(z)\coloneqq\frac{4}{\log(2)\,\delta(1-e^{-\delta})}
\max\big\{1,\varphi_1(z,z),\varphi_2(z)\big\},
\qquad z\in[0,\infty),
\end{align}
where $\delta$, $\varphi_1$ and $\varphi_2$ are as in
\ref{Astar}.

\begin{lemma}\label{lemma:main_Gen}
Let $\eta_t$ be as defined in
\eqref{eq:FK_pred_model} and assume that
\ref{Astar} holds. Then, for any sequences
$\bm{\mu}=(\mu_t)_{t\geq 1}$ with
$\mu_t\in\mathcal{P}_{W_t|1}(\setX)$, and any
$T\geq 1$ such that $\Delta_{T,\bm{\mu}|\eta_1}\leq 1$, 
\begin{align*}
\max_{1\leq t\leq T}d_2(\mu_t,\eta_t)
\leq\varphi\big(\iota_{T,\bm{\mu}|\eta_1}\big)\,
\vartheta\big(\Delta_{T,\bm{\mu}|\eta_1}\big),
\end{align*}
with the function $\vartheta$ as in \eqref{eq:vartheta}.
In addition, for any $f\in\mathcal{F}_1(\setX)$ and
$1<t\leq T$,
\begin{align*}
\big|\mu_t(f)-\eta_t(f)\big|
\leq\big|\mu_t(f)-\Phi_t(\mu_{t-1})(f)\big|
+\varphi\big(\iota_{T,\bm{\mu}|\eta_1}\big)\,
\vartheta\big(\Delta_{T,\bm{\mu}|\eta_1}\big).
\end{align*}
\end{lemma}
\begin{proof}
The proof combines \eqref{eq:forget_past} and
\eqref{eq:lipschitz_gen} with the telescoping sum
formula \eqref{eq:decomp_mu}; see
Appendix~\ref{p-lemma:main_Gen}.
\end{proof}

\begin{remark}
The condition $\Delta_{T,\bm{\mu}|\eta_1}\leq 1$ holds
in the two settings of Section~\ref{sub:primitive},
since $\|\mu_1-\mu_2\|_{\mathrm{tv}}\leq 1$ and
$\|\mu_1-\mu_2\|_{\mathcal{V}_\setX}\leq
\|\mu_1-\mu_2\|_{\mathrm{K}}\leq 1$ for any
$\mu_1,\mu_2\in\mathcal{P}(\setX)$.
\end{remark}

The error bounds provided in Lemma \ref{lemma:main_Gen} depend  on $\bm{\mu}$ not only through the largest local perturbation $\Delta_{T,\bm{\mu}|\eta_1}$, but also through  the constant $\varphi(\iota_{T,\bm{\mu}|\eta_1})$. Our aim now is to  remove the dependence to this constant, and, motivated by the applications of Section~\ref{sec:compareFK}, to present a perturbation bound expressed in term of weighted local perturbations.

To this aim, we let $\varpi\colon[0,\infty)\to[0,1]$ be some function and,
with $\bm{\mu}$ and $\eta$ as in \eqref{eq:maxPerturbation}, for any $T\geq 1$, we set
\begin{align}\label{eq:Delta_phi}
\Delta_{T,\bm{\mu}|\eta,\varpi}\coloneqq
\Big(d_2(\mu_1,\eta)\,\varpi\big(\eta(f_1)\big)\Big)
\vee\max_{1\leq t<T}
\Big(d_2\big(\mu_{t+1},\Phi_{t+1}(\mu_t)\big)\,
\varpi\big(\mu_t(f_t)\big)\Big).
\end{align}
Observe that when $\varpi\equiv 1$ we have
$\Delta_{T,\bm{\mu}|\eta,\varpi}=\Delta_{T,\bm{\mu}|\eta}$.

With this notation in place, we have the following theorem.
  
\begin{theorem}\label{thm:main_Gen}
Let $\eta_t$ be as defined in
\eqref{eq:FK_pred_model}, and let
$\varpi\colon[0,\infty)\to[0,1]$ be a non-decreasing
function such that $z>0\implies \varpi(z)>0$.
Assume that \ref{Astar} holds and that
$\eta_1(f_1)>0$, with $f_t$ as in
\ref{Astar}. Then, there exist constants
$(c_{\eta_1,\varpi},C_{\eta_1,\varpi})
\in(0,\infty)^2$ such that, for any sequences
$\bm{\mu}=(\mu_t)_{t\geq 1}$ with
$\mu_t\in\mathcal{P}_{W_t|1}(\setX)$ and any
$T\geq 1$,
\begin{align*}
\Delta_{T,\bm{\mu}|\eta_1,\varpi}\leq
c_{\eta_1,\varpi}
\implies
\max_{1\leq t\leq T}d_2(\mu_t,\eta_t)
\leq C_{\eta_1,\varpi}\,
\vartheta\big(\Delta_{T,\bm{\mu}|\eta_1,\varpi}\big).
\end{align*}
In addition, for any $f\in\mathcal{F}_1(\setX)$ and
$1<t\leq T$, the condition
$\Delta_{T,\bm{\mu}|\eta_1,\varpi}\leq
c_{\eta_1,\varpi}$ implies
\begin{align*}
|\mu_t(f)-\eta_t(f)|
\leq|\mu_t(f)-\Phi_t(\mu_{t-1})(f)|
+C_{\eta_1,\varpi}\,
\vartheta\big(\Delta_{T,\bm{\mu}|\eta_1,\varpi}\big).
\end{align*}
\end{theorem}
\begin{proof}
The proof uses three facts, namely (i)
in Lemma~\ref{lemma:main_Gen}, the quantity $\iota_{T,\bm{\mu}|\eta_1}$
does not depend on $\mu_T$,   (ii) for any $t\geq 1$, we have
$f_t\in\mathcal{F}_2(\setX)$ under \ref{Astar}, and (iii) the implication stated in
\eqref{eq:lower_gen}. See
Appendix~\ref{p-thm:main_Gen}.
\end{proof}

\begin{remark}\label{rem:log_factor}
The logarithmic factor in $\vartheta$ comes from the
structure of \ref{Astar}: the forgetting property
\eqref{eq:forget_past} decays exponentially in $t-s$ but
does not scale with the distance between the initial
distributions, while the Lipschitz property
\eqref{eq:lipschitz_gen} scales with this distance but
does not decay with $t-s$. In the telescoping
decomposition  \eqref{eq:decomp_mu} used in the proof of
Lemma~\ref{lemma:main_Gen}, the contribution of the
local perturbation introduced at time $p$ is therefore
bounded, up to constants, by
$\min\{\Delta,e^{-\delta(t-p)}\}$, where
$\Delta\in(0,1)$ is the largest local perturbation. Let
$K\coloneqq\lceil\log(1/\Delta)/\delta\rceil$, so that
$K<1+\log(1/\Delta)/\delta$ and $e^{-\delta K}\leq\Delta$.
Bounding the first $K$ terms by $\Delta$ and the
remaining ones by $e^{-\delta k}$, we obtain
\begin{align*}
\sum_{k\geq 0}\min\big\{\Delta,e^{-\delta k}\big\}
\leq K\Delta+\frac{e^{-\delta K}}{1-e^{-\delta}}
\leq\Delta\Big(1+\frac{\log(1/\Delta)}{\delta}\Big)
+\frac{\Delta}{1-e^{-\delta}}.
\end{align*}
This bound is sharp up to constants: since
$\min\{\Delta,e^{-\delta k}\}=\Delta$ for any integer
$0\leq k\leq\log(1/\Delta)/\delta$, the sum is at least
$\Delta\log(1/\Delta)/\delta$. In contrast, if the Lipschitz constant
decayed exponentially, that is, if
\[
d_2(\Phi_{s,t}(\mu_1),\Phi_{s,t}(\mu_2))\leq
C\,e^{-\delta(t-s)}\,d_2(\mu_1,\mu_2),
\qquad t\geq s\geq 1,
\]
each contribution
would be bounded by $C\Delta e^{-\delta(t-p)}$, and the
global error by $C\Delta/(1-e^{-\delta})$. Such estimates
hold under uniform minorization conditions on the Markov
transitions and bounded potential functions
\citep{delmoral2001stability,DelMoral:book}, but not
under the conditions of Section~\ref{sub:primitive},
where the Lipschitz constant $\varphi_2(\mu(f_s))$
depends on t  $\mu(f_s)$. We do not
know whether the logarithmic factor can be removed under
these conditions.
\end{remark}

\begin{remark}
The conclusions of Theorem~\ref{thm:main_Gen} remain
valid for other choices of
$\Delta_{T,\bm{\mu}|\eta,\varpi}$. First, they remain
valid when, in \eqref{eq:Delta_phi}, the functions
$f_t$ are replaced by a sequence
$g_t$ of non-negative functions in
$\mathcal{F}_2(\setX)$ such that
$\mu(f_1)\mu(g_1)>0\implies
\inf_{t\geq 1}\Phi_{1,t}(\mu)(g_t)>0$ for any
$\mu\in\mathcal{P}_{W_1|1}(\setX)$, provided that $\eta_1(g_1)>0$. Second, when $g_t$ is a
sequence of functions in $\mathcal{F}_2(\setX)$,
bounded below by one and such that
$\mu(f_1)\mu(g_1)\in(0,\infty)\implies
\sup_{t\geq 1}\Phi_{1,t}(\mu)(g_t)<\infty$ for any
$\mu\in\mathcal{P}_{W_1|1}(\setX)$, then, if $\eta_1(g_1)<\infty$, the conclusions of
Theorem~\ref{thm:main_Gen} hold with
\begin{align*}
\Delta_{T,\bm{\mu}|\eta_1,\varpi}\coloneqq
\max_{1\leq t\leq T}
\frac{d_2\big(\mu_t,\Phi_t(\mu_{t-1})\big)}
{\varpi\big(\mu_{t-1}(g_{t-1})\big)},
\qquad\text{with }\mu_0=\eta_1
\text{ and }g_0=f_1.
\end{align*}
\end{remark}

For the updated model $\widehat{\eta}_t$, we have:

\begin{lemma}\label{lemma:filter_Gen}
Let $\eta_t$ and
$\widehat{\eta}_t$ be as defined
in \eqref{eq:FK_pred_model}. Assume that
$\inf_{t\geq 1}\eta_t(G_t)>0$ and that
there exists a constant
$c_\star\in(0,\infty)$ such that
$c_\star G_t\in\mathcal{F}_2(\setX)$
for any $t\geq 1$.
Then, there exist constants
$(c_{\eta_1},C_{\eta_1})
\in(0,\infty)^2$
such that, for any sequences
$\mu_t$ with
$\mu_t\in\mathcal{P}_{W_t|1}(\setX)$,
the following two statements hold:
\begin{enumerate}[label=(\roman*)]
\item for any $t\geq 1$ and
$f\in\mathcal{B}_1(\setX)$,
\begin{align*}
\big|\Psi_t(\mu_t)(f)
-\widehat{\eta}_t(f)\big|
\leq C_{\eta_1}
\Big(\big|\mu_t(G_tf)
-\eta_t(G_tf)\big|
+d_2(\mu_t,\eta_t)\Big);
\end{align*}
\item for any $t\geq 1$ and every
measurable function
$f\colon\setX\to\R$ for which there
exists a constant $a_t\in\R$ such that
$g_t\coloneqq(f+a_t)G_t
\in\mathcal{B}_1(\setX)$,
\begin{align*}
d_2(\mu_t,\eta_t)\leq c_{\eta_1}
\ \implies\
\big|\Psi_t(\mu_t)(f)
-\widehat{\eta}_t(f)\big|
\leq C_{\eta_1}
\Big(\big|\mu_t(g_t)
-\eta_t(g_t)\big|
+d_2(\mu_t,\eta_t)\Big).
\end{align*}
\end{enumerate}
\end{lemma}
\begin{proof}
See Appendix~\ref{p-lemma:filter_Gen}.
\end{proof}

\subsection{Primitive conditions on the model\label{sub:primitive}}

Below, we consider
the Feynman--Kac model defined by
\begin{align}\label{eq:FK_pred_model2}
\eta_{t+1}=\Phi_{t+1}(\eta_t)
\quad\text{and}\quad
\widehat{\eta}_t=\Psi_t(\eta_t),
\quad t\geq 1,
\qquad\text{with}\quad
\eta_1\in\mathcal{P}_{W_1}(\setX),
\end{align}
with the sets $W_t$ as in \ref{A0}. Recall that \ref{A0}
holds with $W_t=\setX$ when $G_t>0$ on $\setX$ for any
$t\geq 1$, as for the models  presented in
Sections~\ref{sec:examples}--\ref{sec:PF}.

In this subsection, the local perturbations are measured with the total variation norm
$\|\cdot\|_{\mathrm{tv}}$ or with the norm
$\|\cdot\|_{\mathcal{V}_\setX}$ introduced
Section~\ref{sub:New_norm}, which is equivalent to the Kolmogorov norm $\|\cdot\|_{\mathrm{K}}$ on
$\mathcal{P}(\setX)$ 
(Theorem~\ref{thm:norm}). In the framework of Section~\ref{sub:assumptions_gen},
the total variation estimates are obtained by applying
Theorem~\ref{thm:main_Gen} and
Lemma~\ref{lemma:filter_Gen} in the setting
\begin{align}\label{eq:scenarioA}
\mathcal{F}_1(\setX)=\mathcal{F}_2(\setX)
=\mathcal{B}_1(\setX),
\quad\text{so that}\quad
d_1(\mu_1,\mu_2)=d_2(\mu_1,\mu_2)
=\|\mu_1-\mu_2\|_{\mathrm{tv}},
\tag{TV}
\end{align}
while the estimates in terms of
$\|\cdot\|_{\mathcal{V}_\setX}$ are obtained by
applying the same results when $\setX\subseteq\R^d$
is a hyper-rectangle (e.g.\ $\setX=\R^d$) and
\begin{equation}\label{eq:scenarioB}
\begin{split}
&\mathcal{F}_1(\setX)=\mathcal{B}_1(\setX),
\quad\text{so that}\quad
d_1(\mu_1,\mu_2)=\|\mu_1-\mu_2\|_{\mathrm{tv}},\\
&\mathcal{F}_2(\setX)=
\{f\in\mathcal{B}_{\mathrm{HK}}(\setX):\
\|f\|_{\mathrm{HK}}\leq 1\},
\quad\text{so that}\quad
d_2(\mu_1,\mu_2)=\|\mu_1-\mu_2\|_{\mathcal{V}_\setX},
\end{split}
\tag{$\mathcal{V}_\setX$}
\end{equation}
for any $\mu_1,\mu_2\in\mathcal{P}(\setX)$, and with the
Banach space
$(\mathcal{B}_{\mathrm{HK}}(\setX),
\|\cdot\|_{\mathrm{HK}})$ defined in
Section~\ref{sub:HK_gen}. In both settings,
$\mathcal{P}_1(\setX)=\mathcal{P}(\setX)$  and
\eqref{eq:regul_d1} trivially holds. In addition, under the primitive conditions on the model introduced in what follows, \ref{Astar} holds with
$f_t=c\,\ind_{D_t}$ for any $t\geq 1$, for some constant $c\in(0,1]$ and with the sets
$D_t$ introduced in the conditions below.

For $t\geq 1$ and $\varrho\in[0,1]$, we denote by
\begin{align*}
C_{t,\varrho}\coloneqq
\{x\in\setX:\ G_t(x)\geq\varrho\}
\end{align*}
the super-level sets of the potential functions, and we
assume a local minorization on these sets:
\begin{enumerate}[label=\textbf{(A{{\arabic*}})}]
\item\label{A1} There exist a sequence
$D_t\in  \mathcal{E}$ 
and constants
$(\underline{q},\xi)\in (0,1)\times(1,\infty)$ and
$\varrho\in(0,\underline{q}^2)$ such that the
following two conditions hold:
\begin{enumerate}[label=\textbf{\alph*.},ref=\theenumi\textbf{\alph*}]
\item\label{A1a}
$\inf_{t\geq 1}\inf_{x\in D_t}
Q_{t+1}(x,D_{t+1})\geq\underline{q}$;
\item\label{A1b} for any $t\geq 1$, there exists a
measure $\nu_{t+1}$ on $\setX$ such that, for any
$A\in\mathcal{E}$ and $x\in C_{t,\varrho}$,
\begin{align*}
\frac{1}{\xi}\,\nu_{t+1}(A\cap C_{t+1,\varrho})
\leq Q_{t+1}\big(x,A\cap C_{t+1,\varrho}\big)
\leq\xi\,\nu_{t+1}(A\cap C_{t+1,\varrho}).
\end{align*}
\end{enumerate}
\end{enumerate}
Under \ref{A1}, assumption \ref{A0} holds with
$W_t=D_t$ for any $t\geq 1$. More generally, if
\ref{A0} holds for some sets $W_t$ and \ref{A1a} holds
for some sets $D_t$, then \ref{A0} also holds with
$W_t$ replaced by $W_t\cup D_t$, since
$Q_{t+1}(x,W_{t+1}\cup D_{t+1})\geq Q_{t+1}(x,W_{t+1})>0$
for any $x\in W_t$ and
$Q_{t+1}(x,W_{t+1}\cup D_{t+1})\geq\underline{q}$ for any
$x\in D_t$. Since this substitution only enlarges the sets
$\mathcal{P}_{W_t}(\setX)$, in what follows we assume
without loss of generality that $D_t\subseteq W_t$ for any
$t\geq 1$, with $D_t$ the sets appearing in the
condition~\ref{A1}, \ref{A2} or \ref{A3} under
consideration.

The non-vanishing property \eqref{eq:lower_gen}
requires a stronger condition.
\begin{enumerate}[label=\textbf{(A{{\arabic*}})}]
\setcounter{enumi}{1}
\item\label{A2} Condition~\ref{A1} holds for some sets
$D_t$ for which there exists
$\mu^*\in\mathcal{P}_{W_1}(\setX)$ such that
\[\inf_{t\geq 1}\Phi_{1,t}(\mu^*)(D_t)>0.
\]
\end{enumerate}
Proposition~\ref{prop:A1p-A2} gives a sufficient
condition for \ref{A2}.

\begin{proposition}\label{prop:A1p-A2}
Assume that, for some $\bar{\varrho}\in(0,1)$, the
constant
\begin{align*}
c_{\bar{\varrho}}\coloneqq
\inf_{t\geq 1}\inf_{x\in C_{t,\bar{\varrho}}}
Q_{t+1}(x,C_{t+1,\bar{\varrho}})
\end{align*}
is positive, and that there exists
$\varrho\in(0,c_{\bar{\varrho}})$ such that
$\inf_{t\geq 1}\inf_{x\in C_{t,\varrho}}
Q_{t+1}(x,C_{t+1,\bar{\varrho}})>0$. Then, for any
$\upsilon\in(0,1]$, there exist constants
$(\varepsilon_\upsilon,\kappa_\upsilon)\in(0,1)^2$
such that, for any $\mu\in\mathcal{P}_{W_1}(\setX)$,
\begin{align*}
\mu(C_{1,\bar{\varrho}})\geq\upsilon
\implies
\inf_{t\geq 1}\Phi_{1,t}(\mu)
(C_{t,\varepsilon_\upsilon})\geq\kappa_\upsilon.
\end{align*}
In addition, if $D_t$ is a sequence of
subsets of $\setX$ such that
$\inf_{t\geq 1}\inf_{x\in C_{t,\varepsilon_\upsilon}}
Q_{t+1}(x,D_{t+1})>0$ for some $\upsilon\in(0,1)$,
then there exists $\mu\in\mathcal{P}_{W_1}(\setX)$
such that 
\[
\inf_{t\geq 1}\Phi_{1,t}(\mu)(D_t)>0.
\]

\end{proposition}
\begin{proof}
See Appendix~\ref{p-prop:A1p-A2}.
\end{proof}

In the setting \eqref{eq:scenarioB}, the functions
$f_t=c\,\ind_{D_t}$ must belong to
$\mathcal{F}_2(\setX)$, whence:
\begin{enumerate}[label=\textbf{(A{{\arabic*}})}]
\setcounter{enumi}{2}
\item\label{A3} Condition~\ref{A2} holds for some
measurable sets $D_t$ such that
$\sup_{t\geq 1}\mathcal{V}_\setX(\ind_{D_t})<\infty$.
\end{enumerate}

The Lipschitz property \eqref{eq:lipschitz_gen} in
the setting \eqref{eq:scenarioB} uses the following assumption:
\begin{enumerate}[label=\textbf{(B{{\arabic*}})}]
\item\label{B1} $\setX=\R^d\cap[\underline{x},\bar{x}]$
for some $d\geq 1$ and some hyper-rectangle
$[\underline{x},\bar{x}]\subseteq\bar{\R}^d$ (e.g.\
$\setX=\R^d$). In addition, for any $t\geq 1$ and
$x\in\setX$,
\begin{align*}
M_{t+1}(x,\dd x')=m_{t+1}(x,x')\,
\varsigma_{t+1}(\dd x')
\end{align*}
for some measurable function
$m_{t+1}\colon\setX^2\to[0,\infty)$ and some
$\sigma$-finite measure $\varsigma_{t+1}$ on $\setX$. Lastly, for any $t\geq 1$,  letting $\bar{\varsigma}_{t+1}$ denote the completion of $\varsigma_{t+1}$ and
$q_{t+1}(x,x')\coloneqq G_t(x)\,m_{t+1}(x,x')$, we have
\begin{align*}
\sup_{t\geq 1}\int_\setX
\mathcal{V}_\setX\big(q_{t+1}(\cdot,x')\big)
\,\bar{\varsigma}_{t+1}(\dd x')<\infty.
\end{align*}
\end{enumerate}
By part~\ref{HK_7} of Proposition~\ref{prop:V_HK}, for any $t\geq 1$
the function
$x'\mapsto\mathcal{V}_\setX\big(q_{t+1}(\cdot,x')\big)$
is universally measurable, and thus,  
the integral in \ref{B1}  is well defined.

\subsection{Stability properties of the Feynman--Kac
semigroup\label{sub:stability}}

We derive the stability properties required in
\ref{Astar}, starting with the  exponential forgetting condition \eqref{eq:forget_past}.

\begin{lemma}\label{lemma:forgetting}
Assume that \ref{A1} holds for some sets
$D_t$. Then, tere exist constants
$\delta\in(0,1)$ and $c\in(0,\infty)$ such that, for
any $t\geq s\geq 1$ and
$\mu_1,\mu_2\in\mathcal{P}_{W_s}(\setX)$,
\begin{align*}
\|\Phi_{s,t}(\mu_1)-\Phi_{s,t}(\mu_2)\|_{\mathrm{tv}}
\leq\frac{c\,e^{-\delta(t-s)}}
{\mu_1(D_s)\,\mu_2(D_s)}.
\end{align*}
\end{lemma}
\begin{proof}
The proof follows Lemma~13 and Proposition~5 in
\citet{MLE_SSM}; see Appendix~\ref{p-lemma:forgetting}.
\end{proof}

Under \ref{A2}, the mass of $D_t$ under
$\Phi_{1,t}(\mu)$ does not vanish, which yields the non-vanishing property
\eqref{eq:lower_gen}.

\begin{lemma}\label{lemma:lower_phi}
Assume that \ref{A2} holds for some sets
$D_t$. Then, for any
$\mu\in\mathcal{P}_{W_1}(\setX)$,
\begin{align*}
\mu(D_1)>0\implies
\inf_{t\geq 1}\Phi_{1,t}(\mu)(D_t)>0\implies
\inf_{t\geq 1}\Phi_{1,t}(\mu)(G_t)>0.
\end{align*}
\end{lemma}
\begin{proof}
See Appendix~\ref{p-lemma:lower_phi}.
\end{proof}

Using Lemmas~\ref{lemma:forgetting} and~\ref{lemma:lower_phi}, we can
bound the ratio of the normalizing constants of $\Phi_{s,t}(\mu_1)$
and $\Phi_{s,t}(\mu_2)$.

\begin{lemma}\label{lemma:ratio}
Assume that \ref{A2} holds for some sets
$D_t$. Then, there exists a non-increasing
function $\varphi\colon[0,1]\to[1,\infty]$ such that
$\varphi(z)=\infty\Leftrightarrow z=0$ and, for any
$t\geq s\geq 1$, $\mu_1\in\mathcal{P}(\setX)$ and
$\mu_2\in\mathcal{P}_{W_s}(\setX)$,
\begin{align*}
\frac{\mu_1Q_{s,t}(1)}{\mu_2Q_{s,t}(1)}
\leq\varphi\big(\mu_2(D_s)\big).
\end{align*}
\end{lemma}
\begin{proof}
See Appendix~\ref{p-lemma:ratio}.
\end{proof}

The next result, Lemma~\ref{lemma:Lipschitz_tv}, yields the Lipschitz condition \eqref{eq:lipschitz_gen}
in the setting \eqref{eq:scenarioA}:

\begin{lemma}\label{lemma:Lipschitz_tv}
Assume that \ref{A2} holds for some sets
$D_t$. Then, there exists a non-increasing
function $\varphi\colon[0,1]\to(0,\infty]$ such that
$\varphi(z)=\infty\Leftrightarrow z=0$ and, for any
$t\geq s\geq 1$ and
$\mu_1,\mu_2\in\mathcal{P}_{W_s}(\setX)$,
\begin{align*}
\|\Phi_{s,t}(\mu_1)-\Phi_{s,t}(\mu_2)\|_{\mathrm{tv}}
\leq\varphi\big(\mu_2(D_s)\big)\,
\|\mu_1-\mu_2\|_{\mathrm{tv}}.
\end{align*}
\end{lemma}
\begin{proof}
See Appendix~\ref{p-lemma:Lipschitz_tv}.
\end{proof}

Under \ref{B1}, the total variation norm in the bound of Lemma~\ref{lemma:ratio}  can be replaced by the norm
$\|\cdot\|_{\mathcal{V}_\setX}$, which yields the Lipschitz condition
\eqref{eq:lipschitz_gen} in the setting
\eqref{eq:scenarioB}:

\begin{lemma}\label{lemma:Lipschitz}
Assume that \ref{A2} and \ref{B1} hold for some sets
$D_t$. Then, there exists a
non-increasing function
$\varphi\colon[0,1]\to(0,\infty]$ such that
$\varphi(z)=\infty\Leftrightarrow z=0$ and, for any
$t>s\geq 1$ and
$\mu_1,\mu_2\in\mathcal{P}_{W_s}(\setX)$,
\begin{align*}
\|\Phi_{s,t}(\mu_1)-\Phi_{s,t}(\mu_2)\|_{\mathrm{tv}}
\leq\varphi\big(\mu_2(D_s)\big)\,
\|\mu_1-\mu_2\|_{\mathcal{V}_\setX}.
\end{align*}
\end{lemma}
\begin{proof}
See Appendix~\ref{p-lemma:Lipschitz}.
\end{proof}

Unlike Lemma~\ref{lemma:Lipschitz_tv},
Lemma~\ref{lemma:Lipschitz} fails for $t=s$: one step
of the semigroup is needed to smooth $\mu_1-\mu_2$.

\subsection{Total variation estimates\label{sub:TV_estimates}}

\begin{theorem}\label{thm:main_TV}
Let the flow $\eta_t$ be as defined in
\eqref{eq:FK_pred_model2}, and let
$\varpi\colon[0,\infty)\to[0,1]$ be a non-decreasing
function such that $z>0\implies\varpi(z)>0$.
Assume that \ref{A2} holds for some sets
$D_t$ with $\eta_1(D_1)>0$. Then, there
exists a constant $C_{\eta_1,\varpi}\in(0,\infty)$
such that, for any sequences
$\bm{\mu}=(\mu_t)_{t\geq 1}$ with
$\mu_t\in\mathcal{P}_{W_t}(\setX)$ and any
$T\geq 1$,
\begin{align*}
\max_{1\leq t\leq T}\|\mu_t-\eta_t\|_{\mathrm{tv}}
\leq C_{\eta_1,\varpi}\,
\vartheta\big(\Delta_{T,\bm{\mu}|\eta_1,\varpi}\big),
\end{align*}
where
\begin{align*}
\Delta_{T,\bm{\mu}|\eta_1,\varpi}\coloneqq
\Big(\|\mu_1-\eta_1\|_{\mathrm{tv}}\,
\varpi\big(\eta_1(D_1)\big)\Big)
\vee\max_{1\leq t<T}
\Big(\|\mu_{t+1}-\Phi_{t+1}(\mu_t)\|_{\mathrm{tv}}\,
\varpi\big(\mu_t(D_t)\big)\Big).
\end{align*}
\end{theorem}
\begin{proof}
By Lemmas~\ref{lemma:forgetting},
\ref{lemma:lower_phi} and~\ref{lemma:Lipschitz_tv},
in the setting \eqref{eq:scenarioA} and under
\ref{A2}, condition \ref{Astar} holds with
$f_t=\ind_{D_t}$ for any $t\geq 1$. The result of the theorem then follows from
Theorem~\ref{thm:main_Gen}, recalling that $\|\mu_1-\mu_2\|_{\mathrm{tv}}\leq 1$ for
any $\mu_1,\mu_2\in\mathcal{P}(\setX)$, so that the
condition $\Delta_{T,\bm{\mu}|\eta_1,\varpi}\leq
c_{\eta_1,\varpi}$ in Theorem~\ref{thm:main_Gen} can
be removed by increasing the constant
$C_{\eta_1,\varpi}$.
\end{proof}

We recall that the quantity $\Delta_{T,\bm{\mu}|\eta_1,\varpi}$ is the largest weighted local perturbation up to time
$T$; when
$\varpi\equiv 1$ we have
$\Delta_{T,\bm{\mu}|\eta_1,\varpi}
=\max_{1\leq t\leq T}
\|\mu_t-\Phi_t(\mu_{t-1})\|_{\mathrm{tv}}$ with
$\mu_0=\eta_1$. Non-constant weight functions are used in
Theorem~\ref{thm:compare}.

The next result allows to convert the perturbation bounds obtained in Theorem~\ref{thm:main_TV} for the prediction model $\eta_t$   into perturbation bounds  for the updated model $\hat{\eta}_t$; note that part~\ref{tv_unbounded} below allows some
unbounded functions $f$.

\begin{proposition}\label{prop:filter_TV}
Let the flows $\eta_t$ and
$\widehat{\eta}_t$ be as defined in
\eqref{eq:FK_pred_model2}. Assume that \ref{A2} holds
for some sets $D_t$ with
$\eta_1(D_1)>0$. Then,  there exist constants
$(c_{\eta_1},C_{\eta_1})\in(0,\infty)^2$ such that,
for any sequences $\mu_t$ with
$\mu_t\in\mathcal{P}_{W_t}(\setX)$, the following two
statements hold:
\begin{enumerate}[label=(\roman*)]
\item\label{tv_bounded} for any $t\geq 1$,
\begin{align*}
\|\Psi_t(\mu_t)-\widehat{\eta}_t\|_{\mathrm{tv}}
\leq C_{\eta_1}\,\|\mu_t-\eta_t\|_{\mathrm{tv}};
\end{align*}
\item\label{tv_unbounded} for any $t\geq 1$ and every measurable
function $f\colon\setX\to\R$ for which there exists
a constant $a_t\in\R$ such that
$g_t\coloneqq(f+a_t)G_t\in\mathcal{B}_1(\setX)$,
\begin{align*}
\|\mu_t-\eta_t\|_{\mathrm{tv}}\leq c_{\eta_1}
\implies
|\Psi_t(\mu_t)(f)-\widehat{\eta}_t(f)|
\leq C_{\eta_1}\,\|\mu_t-\eta_t\|_{\mathrm{tv}}.
\end{align*}
\end{enumerate}
\end{proposition}
\begin{proof}
As shown in the proof of Theorem~\ref{thm:main_TV},
in the setting \eqref{eq:scenarioA} and under
\ref{A2}, condition \ref{Astar} holds with
$f_t=\ind_{D_t}$ for any $t\geq 1$. Moreover, since
$\eta_1(D_1)>0$, Lemma~\ref{lemma:lower_phi} gives
$\inf_{t\geq 1}\eta_t(G_t)>0$. Since, in the setting
\eqref{eq:scenarioA}, $fG_t\in\mathcal{B}_1(\setX)$
for any $f\in\mathcal{B}_1(\setX)$ and $t\geq 1$, and
$\|\mu_1-\mu_2\|_{\mathrm{tv}}\leq 1$ for any
$\mu_1,\mu_2\in\mathcal{P}(\setX)$, the result of the proposition follows from
Lemma~\ref{lemma:filter_Gen}.
\end{proof}

\subsection{Kolmogorov norm estimates\label{sub:K_estimates}}

\begin{theorem}\label{thm:main}
Let  the flow $\eta_t$ be as defined in
\eqref{eq:FK_pred_model2}, and let
$\varpi\colon[0,\infty)\to[0,1]$ be a non-decreasing
function such that $z>0\implies \varpi(z)>0$.
Assume that \ref{A3} and \ref{B1} hold and that
$\eta_1(D_1)>0$, with the sets $D_t$ as in
\ref{A3}. Then, there exists a constant
$C_{\eta_1,\varpi}\in(0,\infty)$ such that, for any
sequences $\bm{\mu}=(\mu_t)_{t\geq 1}$ with
$\mu_t\in\mathcal{P}_{W_t}(\setX)$ and any
$T\geq 1$,
\begin{align*}
\max_{1\leq t\leq T}
\|\mu_t-\eta_t\|_{\mathcal{V}_\setX}
\leq C_{\eta_1,\varpi}\,
\vartheta\big(\tilde{\Delta}_{T,\bm{\mu}|\eta_1,\varpi}\big)
\end{align*}
with $\vartheta$ as in \eqref{eq:vartheta}
and, with $\mu_0\coloneqq\eta_1$, for any
$f\in\mathcal{B}_1(\setX)$ and
$t\in\{1,\dots,T\}$,
\begin{align*}
|\mu_t(f)-\eta_t(f)|
\leq|\mu_t(f)-\Phi_t(\mu_{t-1})(f)|
+C_{\eta_1,\varpi}\,
\vartheta\big(\tilde{\Delta}_{T,\bm{\mu}|\eta_1,\varpi}\big),
\end{align*}
where
\begin{align*}
\tilde{\Delta}_{T,\bm{\mu}|\eta_1,\varpi}\coloneqq
\Big(\|\mu_1-\eta_1\|_{\mathcal{V}_\setX}\,
\varpi\big(\eta_1(D_1)\big)\Big)
\vee\max_{1\leq t<T}
\Big(\|\mu_{t+1}-\Phi_{t+1}(\mu_t)\|_{\mathcal{V}_\setX}\,
\varpi\big(\mu_t(D_t)\big)\Big).
\end{align*}
\end{theorem}
\begin{proof}
By Lemmas~\ref{lemma:forgetting},
\ref{lemma:lower_phi} and~\ref{lemma:Lipschitz}, in
the setting \eqref{eq:scenarioB} and under \ref{A3}
and \ref{B1}, condition \ref{Astar} holds with
$f_t=(1/c)\ind_{D_t}$ for any $t\geq 1$, where
$c\coloneqq 1+\sup_{t\geq 1}
\mathcal{V}_\setX(\ind_{D_t})$, so that
$\|f_t\|_{\mathrm{HK}}\leq 1$. Since $\varpi$ is
non-decreasing and $c\geq 1$, we have
$\varpi(\mu_t(f_t))\leq\varpi(\mu_t(D_t))$, so that the
quantity \eqref{eq:Delta_phi} is bounded by
$\tilde{\Delta}_{T,\bm{\mu}|\eta_1,\varpi}$. The result of the theorem then
follows from Theorem~\ref{thm:main_Gen}, recalling that,
by Theorem~\ref{thm:norm},
$\|\mu_1-\mu_2\|_{\mathcal{V}_\setX}\leq 1$ for any
$\mu_1,\mu_2\in\mathcal{P}(\setX)$.
\end{proof}

The next proposition allows to convert the perturbation bounds obtained in Theorem~\ref{thm:main_TV} for the prediction model $\eta_t$   into perturbation bounds  for the updated model $\hat{\eta}_t$; note that part~\ref{K_unbounded} of the proposition allows some
unbounded functions $f$.

\begin{proposition}\label{prop:filter}
Let $\eta_t$ and
$\widehat{\eta}_t$ be as defined in
\eqref{eq:FK_pred_model2}. Assume that \ref{A3} and
\ref{B1} hold and that $\eta_1(D_1)>0$, with the sets
$D_t$ as in \ref{A3}. Then, there exist
constants $(c_{\eta_1},C_{\eta_1})\in(0,\infty)^2$
such that, for any sequences $\mu_t$ with
$\mu_t\in\mathcal{P}_{W_t}(\setX)$, the following
three statements hold:
\begin{enumerate}[label=(\roman*)]
\item\label{i-prop:filter} for any $t\geq 1$,
\begin{align*}
\|\Psi_t(\mu_t)-\widehat{\eta}_t\|_{\mathcal{V}_\setX}
\leq C_{\eta_1}\,\|\mu_t-\eta_t\|_{\mathcal{V}_\setX};
\end{align*}
\item\label{ii-prop:filter} for any $t\geq 1$ and
$f\in\mathcal{B}_1(\setX)$,
\begin{align*}
|\Psi_t(\mu_t)(f)-\widehat{\eta}_t(f)|
\leq C_{\eta_1}\Big(|\mu_t(G_tf)-\eta_t(G_tf)|
+\|\mu_t-\eta_t\|_{\mathcal{V}_\setX}\Big);
\end{align*}
\item\label{K_unbounded} for any $t\geq 1$ and every measurable
function $f\colon\setX\to\R$ for which there exists
a constant $a_t\in\R$ such that
$g_t\coloneqq(f+a_t)G_t\in\mathcal{B}_1(\setX)$,
\begin{align*}
\|\mu_t-\eta_t\|_{\mathcal{V}_\setX}\leq c_{\eta_1}
\implies
|\Psi_t(\mu_t)(f)-\widehat{\eta}_t(f)|
\leq C_{\eta_1}\Big(|\mu_t(g_t)-\eta_t(g_t)|
+\|\mu_t-\eta_t\|_{\mathcal{V}_\setX}\Big).
\end{align*}
\end{enumerate}
\end{proposition}
\begin{proof}
As shown in the proof of Theorem~\ref{thm:main}, in
the setting \eqref{eq:scenarioB} and under \ref{A3}
and \ref{B1}, condition \ref{Astar} holds with
$f_t=(1/c)\ind_{D_t}$ for any $t\geq 1$, for some
constant $c\in[1,\infty)$. Since $\eta_1(D_1)>0$,
Lemma~\ref{lemma:lower_phi} gives
$\inf_{t\geq 1}\eta_t(G_t)>0$. By
Lemma~\ref{lemma:Q_smooth} and
Remark~\ref{rem:Q_smooth} (Appendix~\ref{p-lemma:Lipschitz}),
$C_G\coloneqq 1+\sup_{t\geq 1}\mathcal{V}_\setX(G_t)
<\infty$, so that $\|G_t\|_{\mathrm{HK}}\leq C_G$ and
$(G_t/C_G)\in\F_2(\setX)$ for any
$t\geq 1$ in the setting \eqref{eq:scenarioB}.
Moreover, by part~\ref{HK_5} of
Proposition~\ref{prop:V_HK}, $fg/3\in\F_2(\setX)
\subseteq\mathcal{B}_1(\setX)$ for any
$f,g\in\F_2(\setX)$ in this setting. The result of the proposition then follows from
Lemma~\ref{lemma:filter_Gen}, recalling that, by Theorem~\ref{thm:norm},
$\|\mu_1-\mu_2\|_{\mathcal{V}_\setX}\leq
\|\mu_1-\mu_2\|_{\mathrm{K}}\leq 1$ for any
$\mu_1,\mu_2\in\mathcal{P}(\setX)$.
\end{proof}

\section{A class of state-space models\label{sec:examples}}

In this section, we give conditions   under which the Feynman--Kac model introduced in  
Section~\ref{sub:motive_example}, with state space
$\setX=\R^d$, potential functions \eqref{eq:Gt_def} and
Markov transitions
\eqref{eq:Mt_def}--\eqref{eq:Mt_def2}, satisfies
\ref{A3} and \ref{B1}.

\subsection{The model and the assumptions\label{sub:model_assumptions}}

The model depends on a sequence of
observations $y_t$ in $\R^r$ and on the
following parameters:
\begin{equation}\label{eq:fn_model}
\left\{\begin{aligned}
&\text{the functions $\zeta_{G_t}$,
$\zeta_{M_{t+1}}$, $H_t$ and
$F_{t+1}$ of \eqref{eq:functions};}\\
&\text{the matrices $\Sigma_{G_t}$ in
$\mathcal{S}_r^+$ and $\Sigma_{M_{t+1}}$
in $\mathcal{S}_d^+$.}
\end{aligned}\right.
\end{equation}

\subsubsection{Assumptions on the parameters}

\begin{enumerate}[label=\textbf{(P{{\arabic*}})}]
\item\label{Model1} There exist constants
$(\varepsilon,\nu_G,\nu_M,L)\in(0,1]\times(-d,\infty)
\times(0,\infty)\times[1,\infty)$ and two continuous
and $\varepsilon$-pseudo-decreasing functions
$\zeta_G\colon[0,\infty)\to(0,1]$ and
$\zeta_M\colon[0,\infty)\to(0,\infty)$ such that,
for any $t\geq 1$:
\begin{enumerate}[label=\textbf{\alph*.},ref=\theenumi\textbf{\alph*}]
\item\label{mat}
$1/L\leq\sigma_{\min}(\Sigma_{G_t})\leq
\sigma_{\max}(\Sigma_{G_t})\leq L$ and
$1/L\leq\sigma_{\min}(\Sigma_{M_{t+1}})\leq
\sigma_{\max}(\Sigma_{M_{t+1}})\leq L$;
\item\label{lips} the function $H_t$ is
$L$-bi-Lipschitz;
\item\label{zeta}
$\zeta_G(z)\leq\zeta_{G_t}(z)$ and
$\zeta_M(z)\leq\zeta_{M_{t+1}}(z)$ for any
$z\in[0,\infty)$;
\item\label{upper_GtMt}
$\zeta_{G_t}(z)\leq L(1+z)^{-(d+\nu_G)/2}$ and
$\zeta_{M_{t+1}}(z)\leq L(1+z)^{-(d+\nu_M)/2}$ for
any $z\in[0,\infty)$.
\end{enumerate}
\end{enumerate}

\begin{enumerate}[label=\textbf{(P{{\arabic*}})}]
\setcounter{enumi}{1}
\item\label{Model} There exist constants
$(\varepsilon,\nu_G,\nu_M,L)\in(0,1]\times(0,\infty)^2
\times[1,\infty)$ and two continuous and
$\varepsilon$-pseudo-decreasing functions
$\zeta_G\colon[0,\infty)\to(0,1]$ and
$\zeta_M\colon[0,\infty)\to(0,\infty)$ such that
conditions \ref{mat}--\ref{upper_GtMt} hold and, for
any $t\geq 1$:
\begin{enumerate}[label=\textbf{\alph*.},ref=\theenumi\textbf{\alph*}]
\item $H_t\in\mathcal{C}_{\mathrm{mixed}}(\R^d;r,L)$;
\item\label{cond_zeta} $\zeta_{G_t}$ and
$\zeta_{M_{t+1}}$ are $d$ times continuously
differentiable and, for any $z\in[0,\infty)$ and
$i\in\{1,\dots,d\}$,
\begin{align*}
(1+z)^{i/2}\Big|\frac{\dd^i}{\dd z^i}\zeta_{G_t}(z)\Big|
\leq\frac{L}{(1+z)^{(d+\nu_G)/2}}
\quad\text{and}\quad
(1+z)^{i/2}\Big|\frac{\dd^i}{\dd z^i}\zeta_{M_{t+1}}(z)\Big|
\leq\frac{L}{(1+z)^{(d+\nu_M)/2}}.
\end{align*}
\end{enumerate}
\end{enumerate}

Let $\varphi(\cdot;\kappa)$ denote the density of the Student $t$-distribution $t_\kappa(0,1)$ with
$\kappa\in(0,\infty]$ degrees of freedom, with the convention $t_\infty(0,1)=\mathcal{N}_1(0,1)$.
Conditions \ref{zeta} and \ref{upper_GtMt} are met, e.g., when
$\zeta_{G_t}(z)\coloneqq\varphi(z^{1/2};\kappa_G)$ and $\zeta_{M_{t+1}}(z)\coloneqq\varphi(z^{1/2};\kappa_M)$
for all $t\geq1$, with $\kappa_G\in(0,\infty]$ and $\kappa_M\in(d-1,\infty]$. Noting that
$(1+z)^{i/2}|\zeta^{(i)}(z)|\leq C(1+z)^{-(\kappa+2)/2}$ for $\zeta(z)=\varphi(z^{1/2};\kappa)$ and
$i\geq1$, condition \ref{cond_zeta} also holds in this scenario, and the constraint $\nu_G>0$ in \ref{Model}  is met if
$\kappa_G\in(d-1,\infty]$.

Condition \ref{lips}, which requires the functions $h_t$   to be bi-Lipschitz, is the most restrictive assumption  imposed in \ref{Model1} and \ref{Model}. In particular, under \ref{lips}, the dimension $r$ of the observations $y_t$ is larger or equal to the dimension $d$ of the state space $\setX$, and the functions $h_t$ are all unbounded. (Recall that a function $h:\R^d\rightarrow\R^r$ can only  be bi-Lipschitz if $r\geq d$ and if $\|h\|_\infty=\infty$.)  Condition \ref{lips} is, for instance, met if $r=d$ and $h_t(x):=x$. Such a scenario is common in the context of the state-space models, where in many applications the observations are noisy measurements of a hidden process of interest.  More generally,  if $f:\R^d\rightarrow\R^d$ is a  bi-Lipschitz function and  $A\in\R^{r\times d}$ is a matrix of rank $d$, then  the function  $h(x):=A f(x)$   is  bi-Lipschitz.  We also note  that, if   $r>d$ and $g:\R^d\rightarrow\R^{r-d}$ is a Lipschitz function, then, with $f$ as above, the    function   $h(x):=(f(x), g(x))$    is bi-Lipschitz. Lastly, it is worth mentioning that, when $r=d$, the class of  bi-Lipschitz functions $h:\R^d\rightarrow\R^r$ includes      functions of the form $h(x):=x+g(x)$ with $g$ a contraction mapping (i.e.~$g$ is an $l$-Lipschitz function for some constant $l<1$).

\begin{remark}
Conditions \ref{Model1} and \ref{Model} do not impose that $\int_{\R^r}\zeta_{G_t}(\|u\|^2)\,
\lambda_r(\dd u)<\infty$. Consequently, under these two conditions, the Feynman--Kac model with potential functions \eqref{eq:Gt_def} and
Markov transitions
\eqref{eq:Mt_def}--\eqref{eq:Mt_def2}   is not necessarily the Feynman--Kac representation of a state-space model.
\end{remark}

\subsubsection{Assumptions on the observations\label{subsub:O_condition}}

\begin{enumerate}[label=\textbf{(O{{\arabic*}})}]
\setcounter{enumi}{0}
\item\label{Obs} The observations $y_t$
are such that:
\begin{enumerate}[label=\textbf{\alph*.},ref=\theenumi\textbf{\alph*}]
\item\label{Obs1} $y_t\in H_t(\R^d)$ for any
$t\geq 1$;
\item\label{Obs2} there exists a constant
$\Xi\in[0,\infty)$ such that
$\sup_{t\geq 1}\big\|H^{-1}_{t+1}(y_{t+1})
-F_{t+1}\big(H_t^{-1}(y_t)\big)\big\|\leq\Xi$.
\end{enumerate}
\end{enumerate}

Condition~\ref{Obs1} can be interpreted as an \textit{observability}
assumption: under  state-space model~\eqref{eq:SSM}, $Y_t$ is a noisy
version of $H_t(X_t)$, and \ref{Obs1} requires that each $y_t$ belongs
to the set of values taken by the quantity of which it is an
observation. As shown in Section~\ref{sub:compatibility}, this
condition can often be dropped when $H_t(\R^d)$ is a $d$-dimensional
linear subspace for any $t\geq 1$.

Condition~\ref{Obs2} requires the observations to be
\emph{compatible with the dynamics} of the model. Under~\ref{Obs1}, for any
$t\geq 1$ there exists $x_{y_t}\in\R^d$ such that $H_t(x_{y_t})=y_t$,
and \ref{Obs2} is equivalent to
$\sup_{t\geq 1}\|x_{y_{t+1}}-F_{t+1}(x_{y_t})\|<\infty$; that is, the states
suggested by the observations follow the noiseless dynamics
$x_{t+1}=F_{t+1}(x_t)$ up to bounded errors. Under~\ref{Model1}, this
implies that $\inf_{t\geq 1}m_{t+1}(x_{y_t},x_{y_{t+1}})>0$, with the
densities $m_{t+1}$ defined in~\eqref{eq:Mt_def2}; without this
property, the transition from $X_t=x_{y_t}$ to $X_{t+1}=x_{y_{t+1}}$
under $M_{t+1}$ can be arbitrarily unlikely. By
Lemma~\ref{lemma:Obs2} provided below, \ref{Obs2} holds for bounded observations
under mild conditions on the functions $H_t$ and $F_{t+1}$.

\subsection{Verification of the
assumptions\label{sub:sufficient}}

\begin{proposition}\label{prop:example}
Let $G_t$ and $M_{t+1}$ be as
defined in \eqref{eq:Gt_def}--\eqref{eq:Mt_def2}, and
assume that \ref{Obs} holds.
\begin{enumerate}[label=(\roman*)]
\item\label{i-prop:example} If \ref{Model1} holds and there exists a
non-decreasing function
$\varphi\colon[0,\infty)\to[0,\infty)$ such that
\begin{align*}
\|F_{t+1}(x_1)-F_{t+1}(x_2)\|
\leq\varphi\big(\|x_1-x_2\|\big),
\qquad\forall x_1,x_2\in\R^d,
\quad\forall t\geq 1,
\end{align*}
then \ref{A3} holds.
\item\label{ii-prop:example} If \ref{Model} holds and
$F_{t+1}\in\mathcal{C}_{\mathrm{mixed}}(\R^d;d,L)$ for
any $t\geq 1$, with $L$ as in \ref{Model}, then
\ref{A3} and \ref{B1} hold.
\end{enumerate}
\end{proposition}
\begin{proof}
See Appendix~\ref{p-prop:example}.
\end{proof}

For bounded observations, we have the following result which relies on weaker assumptions on the function $F_{t+1}$.

\begin{proposition}\label{prop:example2}
Let $G_t$ and $M_{t+1}$ be as
defined in \eqref{eq:Gt_def}--\eqref{eq:Mt_def2}, and
assume that \ref{Obs} holds and that
$\sup_{t\geq 1}\|y_t\|\leq\Xi$, with $\Xi$ as in
\ref{Obs}.
\begin{enumerate}[label=(\roman*)]
\item\label{i-prop:example2} If \ref{Model1} holds and there exists a
coordinate-wise non-decreasing function
$\varphi\colon[0,\infty)^2\to[0,\infty)$ such that
\begin{align}\label{eq:bound_derive_ht0}
\|F_{t+1}(x_1)-F_{t+1}(x_2)\|
\leq\varphi\big(\|x_1-x_2\|,\|H_t(x_1)\|\big),
\qquad\forall x_1,x_2\in\R^d,
\quad\forall t\geq 1,
\end{align}
then \ref{A3} holds.
\item\label{ii-prop:example2} If \ref{Model} holds and, for any $t\geq 1$,
$F_{t+1}\in\mathcal{C}_{\mathrm{mixed}}(\R^d;d,\infty)$
and, for any non-empty sets
$u\subseteq\{1,\dots,d\}$,
\begin{align}\label{eq:bound_derive_ht}
\|\partial^uF_{t+1}(x)\|
\leq L\big(1+\|H_t(x)\|^{\nu_G(1-1/L)/d}\big),
\qquad\forall x\in\R^d,
\end{align}
with $(\nu_G,L)$ as in \ref{Model}, then \ref{A3} and
\ref{B1} hold.
\end{enumerate}
\end{proposition}
\begin{proof}
See Appendix~\ref{p-prop:example}.
\end{proof}

\begin{remark}
 Under
\eqref{eq:bound_derive_ht}, condition
\eqref{eq:bound_derive_ht0} holds with
$\varphi(z_1,z_2)=2Ldz_1(1+z_2+Lz_1)^{\nu_G(1-1/L)/d}$
(see Lemma~\ref{lemma:tech_th}).

\end{remark}

\begin{remark}\label{rem:unbounded}
For any $t\geq 1$, write
$H_t=(H_{t,1},\dots,H_{t,r})$. Then, together with the results presented in Section \ref{sec:main}, Propositions \ref{prop:example} and  \ref{prop:example2} allow to obtain perturbation bounds for approximation of $\hat{\eta}_t(H_{t,i})$ for any $i\in\{1,\dots,r\}$. To see this note that,  under \ref{Model1},
for any $i\in\{1,\dots,r\}$, any $x\in\R^d$ and any
$t\geq 1$, we have
\begin{align*}
|H_{t,i}(x)-y_{t,i}|\,G_t(x)
\leq\|y_t-H_t(x)\|\,G_t(x)
\leq\sup_{s\geq 1}\sup_{x'\in\R^d}
\big(\|y_s-H_s(x')\|\,G_s(x')\big)=:K<\infty.
\end{align*}
We now let $i\in\{1,\dots,r\}$ and $t\geq 1$ be fixed, and $f^\pm_{t,i}\coloneqq
(H_{t,i}-y_{t,i})^\pm/K$, where $a^\pm$ denotes the
positive and negative parts of $a\in\R$. Then,
$f^\pm_{t,i}G_t\in\mathcal{B}_1(\setX)$, so that the
condition imposed on $f$ in part~\ref{tv_unbounded} of
Proposition~\ref{prop:filter_TV} and in part~\ref{K_unbounded} of
Proposition~\ref{prop:filter} holds for the functions $f^+_{t,i}$
and $f^-_{t,i}$, with $a_t=0$. Since
$H_{t,i}=y_{t,i}+K(f^+_{t,i}-f^-_{t,i})$ and the map
$f\mapsto\Psi_t(\mu_t)(f)-\widehat{\eta}_t(f)$ is linear
and vanishes on constant functions, the bounds of
these propositions extend, up to a factor $2K$, to the
unbounded functions $f=H_{t,i}$. In particular, when   $H_t(x)=Ax$ for some full rank column matrix $A$ (e.g.~$H_t(x)=x$), these
estimates apply to the filtering mean.

\end{remark}

\begin{remark}\label{rem:constant}
The constants in the perturbation bounds of
Section~\ref{sec:main} are uniform with respect to some
classes of models discussed in this section. To expand on this point, let $a\in(0,1]$, for any $t\geq 1$   let
$E_t\coloneqq\{x\in\R^d:\ \|x-H_t^{-1}(y_t)\|_\infty\leq 1\}$,
  and   let $\varpi$ be as in
Theorems~\ref{thm:main_TV} and~\ref{thm:main}.
\begin{enumerate}[label=(\roman*)]
\item\label{rem1} Under the assumptions of part~\ref{i-prop:example} of
Proposition~\ref{prop:example} or of part~\ref{i-prop:example2} of
Proposition~\ref{prop:example2}, and if
$\eta_1(E_1)\geq a$, \ref{A2} holds with the sets
$D_t=E_t$, and the constants of
Theorem~\ref{thm:main_TV} and
Proposition~\ref{prop:filter_TV}, the constant
$\underline{q}$ in \ref{A1a} and a positive lower bound
on $\inf_{t\geq 1}\eta_t(G_t)$ can be chosen to depend
on $G_t$, $M_{t+1}$ and
$\eta_1$ only through $d$, $r$, the constants
$(\varepsilon,\nu_G,\nu_M,L,\Xi)$, the functions
$\zeta_G$ and $\zeta_M$ in \ref{Model1}, the function
$\varphi$ in these propositions, and $a$ (and on
$\varpi$).
\item\label{rem2} Under the assumptions of part~\ref{ii-prop:example} of
Proposition~\ref{prop:example} or of part~\ref{ii-prop:example2} of
Proposition~\ref{prop:example2}, and if
$\eta_1(E_1)\geq a$, the same holds for the constants
of Theorems~\ref{thm:main_TV} and~\ref{thm:main} and
of Propositions~\ref{prop:filter_TV}
and~\ref{prop:filter}, with a dependence only on $d$,
$r$, $(\varepsilon,\nu_G,\nu_M,L,\Xi)$, $\zeta_G$,
$\zeta_M$ and $a$ (and on $\varpi$).
\end{enumerate}
 A proof of the above assertions is given in Appendix~\ref{p-rem:constant}.
\end{remark}

The text result is useful to check that condition \ref{Obs2} holds.

\begin{lemma}\label{lemma:Obs2}
Let $G_t$ and $M_{t+1}$ be as
defined in \eqref{eq:Gt_def}--\eqref{eq:Mt_def2}, and
assume that \ref{Model1} and \ref{Obs1} hold and that
the functions $F_{t+1}$ satisfy one of the
four conditions stated in
Propositions~\ref{prop:example}
and~\ref{prop:example2}. Assume in addition that
\begin{enumerate}[label=(\roman*)]
\item $\sup_{t\geq 1}\|y_t\|<\infty$;
\item $y_1\in H_{t}(\R^d)$ for any $t\geq 2$, and
$\sup_{t\geq 1}\|H_{t+1}^{-1}(y_1)\|<\infty$;
\item $\sup_{t\geq 1}\|F_{t+1}(H_t^{-1}(y_1))\|<\infty$.
\end{enumerate}
Then \ref{Obs2} holds.
\end{lemma}
\begin{proof}
Under the assumptions of the lemma, there exist a
coordinate-wise non-decreasing function
$\varphi\colon[0,\infty)^2\to[0,\infty)$ such that
\eqref{eq:bound_derive_ht0} holds and a constant
$L\in[1,\infty)$ such that $H_t^{-1}$ is $L$-Lipschitz
on $H_t(\R^d)$ for any $t\geq 1$. Since
$y_t\in H_t(\R^d)$ for any $t\geq 1$ by \ref{Obs1},
we obtain, for any $t\geq 1$,
\begin{align*}
\big\|H^{-1}_{t+1}(y_{t+1})-F_{t+1}\big(H_t^{-1}(y_t)\big)\big\|
\leq{}&\|H^{-1}_{t+1}(y_1)\|
+\big\|F_{t+1}\big(H_t^{-1}(y_1)\big)\big\|\\
&+L\|y_{t+1}-y_1\|
+\varphi\big(L\|y_1-y_t\|,\|y_1\|\big),
\end{align*}
and the result follows from (i)--(iii).
\end{proof}

\subsection{Relaxing the observability
condition\label{sub:compatibility}}

\subsubsection{Linear observation spaces\label{sub:model_obs}}

Assume that $H_t(\R^d)$ is a $d$-dimensional linear
subspace of $\R^r$ for any $t\geq 1$; for instance,
$H_t(x)=Ag(x)$ with $A\in\R^{r\times d}$ a matrix of rank $d$
and $g\colon\R^d\to\R^d$ a surjective function (e.g.\
$g(x)=x$). In this scenario, as explained below, condition \ref{Obs1} can often be dropped. To distinguish
the original parameters and observations from those of an equivalent
representation of the model, we use the following notation. Let $z_t$
be a sequence in $\R^r$, $S_{G_t}$ a
sequence in $\mathcal{S}_r^+$ and
$S_{M_{t+1}}$ a sequence in
$\mathcal{S}_d^+$ and, for any $t\geq 1$, let
$f_t\colon\R^d\to\R^r$,
$\xi_{G_t}\colon[0,\infty)\to(0,1]$,
$\xi_{M_{t+1}}\colon[0,\infty)\to(0,\infty)$ and
$\tilde{f}_{t+1}\colon\R^d\to\R^d$ be measurable
functions, with $f_t(\R^d)$ a $d$-dimensional linear
subspace. Throughout this subsection,
$G_t$ and $M_{t+1}$ are
defined by \eqref{eq:Gt_def}--\eqref{eq:Mt_def2} with
\begin{align}\label{eq:def_model1}
\zeta_{G_t}=\xi_{G_t},\quad y_t=z_t,\quad
H_t=f_t,\quad\Sigma_{G_t}=S_{G_t},
\qquad\forall t\geq 1,
\end{align}
and
\begin{align}\label{eq:def_model1_M}
\zeta_{M_{t+1}}=\xi_{M_{t+1}},\quad
F_{t+1}=\tilde{f}_{t+1},\quad
\Sigma_{M_{t+1}}=S_{M_{t+1}},
\qquad\forall t\geq 1.
\end{align}
We assume that these parameters satisfy \ref{Model1}
(respectively, \ref{Model}) as well as the conditions on
$F_{t+1}$ in part~\ref{i-prop:example} (respectively,
part~\ref{ii-prop:example}) of Proposition~\ref{prop:example}. If
$\sup_{t\geq 1}\|z_t\|<\infty$, the weaker conditions
of Proposition~\ref{prop:example2} can be used
instead, in which case Proposition~\ref{prop:example}
should be replaced by Proposition~\ref{prop:example2}
in what follows.

\subsubsection{An equivalent parametrization via
projections\label{sub:gen}}

For $t\geq 1$, the set
$V_t\coloneqq\{S_{G_t}^{-1/2}f_t(x):\ x\in\R^d\}$ is
a $d$-dimensional linear subspace of $\R^r$, and we let
$P_t\in\R^{r\times r}$ be the orthogonal projection
onto $V_t$, and  
\begin{align*}
\widehat{z}_t\coloneqq P_tS_{G_t}^{-1/2}z_t,\qquad
w_t\coloneqq(I_r-P_t)S_{G_t}^{-1/2}z_t,\qquad
\widehat{f}_t\coloneqq S_{G_t}^{-1/2}f_t.
\end{align*}
Since $\widehat{z}_t-\widehat{f}_t(x)\in V_t$ and
$w_t\in V_t^\perp$, Pythagoras' theorem gives
\begin{align}\label{eq:GH}
G_t(x)=\xi_{G_t}\Big(\|\widehat{z}_t-\widehat{f}_t(x)\|^2
+\|w_t\|^2\Big).
\end{align}
Hence, letting
$\widehat{\zeta}_{G_t}(z)\coloneqq\xi_{G_t}(z+\|w_t\|^2)$
for $z\in[0,\infty)$, the functions
$G_t$ are also given by
\eqref{eq:Gt_def} with
\begin{align}\label{eq:model_redef}
\zeta_{G_t}=\widehat{\zeta}_{G_t},\quad
y_t=\widehat{z}_t,\quad
H_t=\widehat{f}_t,\quad
\Sigma_{G_t}=I_r,
\qquad\forall t\geq 1.
\end{align}
By construction, $\widehat{z}_t\in\widehat{f}_t(\R^d)$
for any $t\geq 1$, so that \ref{Obs1} holds for this
representation of the model. This representation also satisfies \ref{Model1}
(respectively, \ref{Model}) if the components $w_t$ of
the observations orthogonal to the observation spaces
are bounded:

\begin{proposition}\label{prop:equiv_Model}
Let $G_t$ be as defined by
\eqref{eq:Gt_def} and \eqref{eq:model_redef}, and let
$M_{t+1}$ be as defined by
\eqref{eq:Mt_def}--\eqref{eq:Mt_def2} and
\eqref{eq:def_model1_M}. If
$W\coloneqq\sup_{t\geq 1}\|w_t\|<\infty$, then
\ref{Model1} (respectively, \ref{Model}) holds.
\end{proposition}
\begin{proof}
Let $(\varepsilon,\nu_G,\nu_M,L)$, $\zeta_G$ and
$\zeta_M$ be as in \ref{Model1} (respectively,
\ref{Model}) for the original parameters
\eqref{eq:def_model1}--\eqref{eq:def_model1_M}. The
conditions on $\Sigma_{M_{t+1}}$ and
$\zeta_{M_{t+1}}$ are unchanged. Since
$\Sigma_{G_t}=I_r$ in \eqref{eq:model_redef}, the
first half of \ref{mat} holds. Since
$1/L\leq\sigma_{\min}(S_{G_t})\leq
\sigma_{\max}(S_{G_t})\leq L$, the functions
$\widehat{f}_t=S_{G_t}^{-1/2}f_t$ are
$L^{3/2}$-bi-Lipschitz and thus, under \ref{Model},
belong to $\mathcal{C}_{\mathrm{mixed}}(\R^d;r,L^{3/2}r)$.
Next, since $\zeta_G$ is $\varepsilon$-pseudo-decreasing,
\begin{align*}
\widehat{\zeta}_{G_t}(z)=\xi_{G_t}(z+\|w_t\|^2)
\geq\zeta_G(z+\|w_t\|^2)
\geq\varepsilon\,\zeta_G(z+W^2),
\qquad\forall z\in[0,\infty),
\end{align*}
and the function $z\mapsto\varepsilon\zeta_G(z+W^2)$
is continuous and $\varepsilon$-pseudo-decreasing.
Moreover, we have
\[\widehat{\zeta}_{G_t}(z)\leq L(1+z+\|w_t\|^2)^{-(d+\nu_G)/2}
\leq L(1+z)^{-(d+\nu_G)/2},\]
since $d+\nu_G>0$, and,
under \ref{Model}, for any $i\in\{1,\dots,d\}$ we have
\begin{align*}
(1+z)^{i/2}\Big|\frac{\dd^i}{\dd z^i}
\widehat{\zeta}_{G_t}(z)\Big|
\leq(1+z+\|w_t\|^2)^{i/2}
\Big|\frac{\dd^i\xi_{G_t}}{\dd z^i}(z+\|w_t\|^2)\Big|
\leq\frac{L}{(1+z)^{(d+\nu_G)/2}}.
\end{align*}
Hence \ref{Model1} (respectively, \ref{Model}) holds
for the representation \eqref{eq:model_redef}, with
$L$ replaced by $L^{3/2}r$ and with $\zeta_G$ replaced by
$\varepsilon\zeta_G(\cdot+W^2)$.
\end{proof}

By Propositions~~\ref{prop:example} and \ref{prop:equiv_Model}, to check that \ref{A3} holds
(respectively, that \ref{A3} and \ref{B1} hold), it
is therefore enough to check \ref{Obs2} for the
representation \eqref{eq:model_redef}, that is,
\begin{align}\label{eq:O2_cond}
\sup_{t\geq 1}\Big\|
\widehat{f}^{-1}_{t+1}(\widehat{z}_{t+1})
-\tilde{f}_{t+1}\big(\widehat{f}_t^{-1}(\widehat{z}_t)\big)
\Big\|<\infty,
\qquad\text{where}\quad
\widehat{f}_t^{-1}(\widehat{z}_t)
=f_t^{-1}\big(S_{G_t}^{1/2}P_tS_{G_t}^{-1/2}z_t\big).
\end{align}
When $f_t(x)=A_tx$ for some matrix
$A_t\in\R^{r\times d}$ of rank $d$, condition
\eqref{eq:O2_cond} reads
\begin{align*}
\sup_{t\geq 1}\Big\|
\big(A_{t+1}^\top S^{-1}_{G_{t+1}}A_{t+1}\big)^{-1}
A_{t+1}^\top S^{-1}_{G_{t+1}}z_{t+1}
-\tilde{f}_{t+1}\Big(\big(A_t^\top S^{-1}_{G_t}A_t\big)^{-1}
A_t^\top S^{-1}_{G_t}z_t\Big)\Big\|<\infty.
\end{align*}

\subsubsection{Exponential potential functions\label{sub:exp}}

The boundedness condition on $w_t$ is not
needed when $\xi_{G_t}(z)=\exp(-z)$ for any $t\geq 1$. Before proving this claim, 
observe that the Feynman--Kac model does not change when the
potential functions are multiplied by positive
constants. Below we say that two
sequences of functions $G_{1,t}$ and
$G_{2,t}$ are \emph{equivalent} if, for
any $t\geq 1$, there exists $c_t\in(0,\infty)$ such
that $G_{1,t}=c_tG_{2,t}$; equivalent sequences of
potential functions define the same distributions
$\eta_t$ and $\widehat{\eta}_t$.

When $\xi_{G_t}(z)=\exp(-z)$ for any $t\geq 1$, the result stated in \eqref{eq:GH} becomes
\begin{align*}
G_t(x)=\exp\big(-\|\widehat{z}_t-\widehat{f}_t(x)\|^2\big)
\,\exp\big(-\|w_t\|^2\big),
\end{align*}
so that $G_t$ is equivalent to the
sequence of functions defined by \eqref{eq:Gt_def}
with
\begin{align}\label{eq:model_redef2}
\zeta_{G_t}=\xi_{G_t},\quad
y_t=\widehat{z}_t,\quad
H_t=\widehat{f}_t,\quad
\Sigma_{G_t}=I_r.
\end{align}
In \eqref{eq:model_redef2}, the functions
$\zeta_{G_t}$ no longer depend on
$w_t$. Applying
Proposition~\ref{prop:equiv_Model} with $w_t=0$ for
any $t\geq 1$ shows that \ref{Model1} (respectively,
\ref{Model}) holds for the representation
\eqref{eq:model_redef2}, even when
$\sup_{t\geq 1}\|w_t\|=\infty$. If, in addition,
\eqref{eq:O2_cond} holds, then, by
Proposition~\ref{prop:example}, \ref{A3}
(respectively, \ref{A3} and \ref{B1}) holds for this
representation, and hence the results of
Section~\ref{sec:main} apply to the original model.
Thus, for exponential potential functions and linear observation
spaces, \ref{A3} (respectively,
\ref{A3} and \ref{B1}) holds if
\eqref{eq:O2_cond} holds, whether or not the
observations are bounded.

\subsection{Extensions}

\subsubsection{General transition
kernels\label{sub:extension}}

Propositions~\ref{prop:example} and~\ref{prop:example2}
extend, under less explicit conditions, to Markov
transitions not of the form
\eqref{eq:Mt_def2}. Specificity,  let $G_t$ be as in
\eqref{eq:Gt_def} and $M_{t+1}$ as in
\eqref{eq:Mt_def}, for some measurable functions
$m_{t+1}\colon\R^d\times\R^d\to(0,\infty)$. The
model then depends on the observations
$y_t$ and on the parameters
$\zeta_{G_t}$, $H_t$,
$m_{t+1}$ and $\Sigma_{G_t}$,
for which we consider the following two conditions.
\begin{enumerate}[label=\textbf{(P{{\arabic*}'})}]
\setcounter{enumi}{0}
\item\label{Model11} There exist constants
$(\varepsilon,\nu_G,\nu_M,L)\in(0,1]\times(-d,\infty)
\times(0,\infty)\times[1,\infty)$ and two continuous
and $\varepsilon$-pseudo-decreasing functions
$\zeta_G\colon[0,\infty)\to(0,1]$ and
$\zeta_M\colon[0,\infty)\to(0,\infty)$ such that:
\begin{enumerate}[label=\textbf{\alph*.},ref=\theenumi\textbf{\alph*}]
\item $\Sigma_{G_t}$, $H_t$ and $\zeta_{G_t}$ satisfy
the conditions of \ref{Model1} for any $t\geq 1$;
\item for any $t\geq 1$, there exists a measurable
function $F_{t+1}\colon\R^d\to\R^d$ such that, for
any $x,x'\in\R^d$,
\begin{align*}
\frac{1}{L}\,\zeta_M\big(L\|x'-F_{t+1}(x)\|^2\big)
\leq m_{t+1}(x,x')
\leq L\big(1+\|x'-F_{t+1}(x)\|^2\big)^{-\frac{d+\nu_M}{2}}.
\end{align*}
\end{enumerate}
\item\label{Model3} There exist constants
$(\varepsilon,\nu_G,\nu_M,L)\in(0,1]\times(0,\infty)^2
\times[1,\infty)$, two continuous and
$\varepsilon$-pseudo-decreasing functions
$\zeta_G\colon[0,\infty)\to(0,1]$ and
$\zeta_M\colon[0,\infty)\to(0,\infty)$, and
measurable functions $F_{t+1}\colon\R^d\to\R^d$,
$t\geq 1$, such that \ref{Model11} holds and, for any
$t\geq 1$:
\begin{enumerate}[label=\textbf{\alph*.},ref=\theenumi\textbf{\alph*}]
\item $H_t$ and $\zeta_{G_t}$ satisfy the conditions
of \ref{Model};
\item for any $x'\in\R^d$,
$m_{t+1}(\cdot,x')\in\mathcal{C}_{\mathrm{mixed}}(\R^d)$
and, for any non-empty sets
$u\subseteq\{1,\dots,d\}$,
\begin{align*}
\big|\partial^u_1m_{t+1}(x,x')\big|
\leq L\big(1+\|x'-F_{t+1}(x)\|^2\big)^{-\frac{d+\nu_M}{2}},
\qquad\forall x\in\R^d.
\end{align*}
\end{enumerate}
\end{enumerate}
Propositions~\ref{prop:example} and~\ref{prop:example2}
remain valid for this larger class of models if
\ref{Model1} is replaced by \ref{Model11} and
\ref{Model} by \ref{Model3}, the functions
$F_{t+1}$ in \ref{Obs2} and in these
propositions being those of \ref{Model11}. Indeed, in
their proofs (see Appendix~\ref{p-prop:example}), the
Markov transitions enter only through: a lower bound
on $m_{t+1}(x,x')$ when $\|x'-F_{t+1}(x)\|$ is bounded
(Lemma~\ref{lemma:Prelim1_A1}), which follows from the
first inequality of \ref{Model11} and the
pseudo-monotonicity of $\zeta_M$; the bound
$\sup_{x,x'}m_{t+1}(x,x')\leq L$
(Lemma~\ref{lemma:proofA1}); and, for \ref{B1}, the
bound $|\partial_1^u m_{t+1}(x,x')|\leq
L(1+\|x'-F_{t+1}(x)\|^2)^{-(d+\nu_M)/2}$ for any
$u\subseteq\{1,\dots,d\}$
(Lemma~\ref{lemma:Prop1_showB1_ex2}), which is
assumed in \ref{Model3} for $u\neq\emptyset$ and in
\ref{Model11} for $u=\emptyset$. Under \ref{Model3},
this bound replaces the application of
Lemma~\ref{lemma:derivative} to $m_{t+1}$, so that no
smoothness of $F_{t+1}$ is needed for \ref{B1}; only
the constants change. This
extension is used in Section~\ref{sub:SDE}.

\subsubsection{Bounded perturbations of the potential
functions}

The next lemma, used in Section \ref{sub:multi_mod}, allows to establish that conditions \ref{A3} and \ref{B1} hold beyond the class of potential function $G_t$ of the form   \eqref{eq:Gt_def}. 

\begin{lemma}\label{lemma:bounded_pert}
Let $g_t$ be functions of the form
\eqref{eq:Gt_def}, that is,
$g_t(x)=\zeta_{G_t}\big((y_t-H_t(x))^\top
\Sigma_{G_t}^{-1}(y_t-H_t(x))\big)$ for any
$t\geq 1$ and $x\in\R^d$, and let
$M_{t+1}$ be as defined in
\eqref{eq:Mt_def}--\eqref{eq:Mt_def2}. Assume that
the Feynman--Kac model with potential functions
$g_t$ and Markov transitions
$M_{t+1}$ satisfies the assumptions of
part~\ref{i-prop:example} of Proposition~\ref{prop:example} or of
Proposition~\ref{prop:example2}. Let
$b_t\colon\R^d\to[c,1]$, $t\geq 1$, be measurable
functions, for some constant $c\in(0,1]$, and
consider the perturbed potential functions
$G_t\coloneqq b_t\,g_t$, $t\geq 1$. Then:
\begin{enumerate}[label=(\roman*)]
\item the Feynman--Kac model with potential
functions $G_t$ and Markov transitions
$M_{t+1}$ satisfies \ref{A3};
\item if, in addition, the assumptions of
part~\ref{ii-prop:example} of Proposition~\ref{prop:example} or of
part~\ref{ii-prop:example2} of Proposition~\ref{prop:example2} hold, and
$b_t\in\mathcal{C}_{\mathrm{mixed}}(\R^d;1,L_b)$ for any $t\geq 1$ and some constant
$L_b\in[1,\infty)$, then the Feynman--Kac model with potential
functions $G_t$ and Markov transitions
$M_{t+1}$ also satisfies \ref{B1}.
\end{enumerate}
\end{lemma}
\begin{proof}
See Appendix \ref{p-lemma:bounded_pert}.
\end{proof}

\subsection{Examples\label{sub:expl_ex}}

Below, we assume that the
initial distribution $\eta_1$ has a positive density
with respect to the Lebesgue measure, so that
$\eta_1(D_1)>0$ for any set $D_1$ of positive
Lebesgue measure.

\subsubsection{Linear Gaussian state-space
models\label{sub:LG_SSM}}

Let $r\geq d\geq 1$, let $A\in\R^{r\times d}$ have
rank $d$, let $B\in\R^{d\times d}$, and let
$\Sigma_G\in\mathcal{S}_r^+$ and
$\Sigma_M\in\mathcal{S}_d^+$. Consider the
state-space model \eqref{eq:SSM} with, for any
$t\geq 1$,
\begin{align}\label{eq:SSM2}
H_t(x)=Ax,\quad F_{t+1}(x)=Bx,\quad
\Sigma_{G_t}=\Sigma_G,\quad\Sigma_{M_{t+1}}=\Sigma_M,
\quad U_t\iid\mathcal{N}_r(0,I_r),\quad
V_{t+1}\iid\mathcal{N}_d(0,I_d),
\end{align}
and let $y_t$ be a sequence in $\R^r$. As
recalled in Section~\ref{sub:motive_example}, the
predictive and filtering distributions of this model
are the prediction and updated models of the
Feynman--Kac model with
\begin{align*}
G_t(x)=\exp\Big(-\frac{1}{2}(y_t-Ax)^\top
\Sigma_G^{-1}(y_t-Ax)\Big),
\qquad
\delta_xM_{t+1}=\mathcal{N}_d(Bx,\Sigma_M),
\qquad\forall x\in\R^d,\quad\forall t\geq 1.
\end{align*}
Since
$G_t(x)=\exp\big(-(y_t-Ax)^\top(2\Sigma_G)^{-1}(y_t-Ax)\big)$, the discussion in
Section~\ref{sub:exp} applies with $\xi_{G_t}(z)=e^{-z}$,
$f_t(x)=Ax$ and $S_{G_t}=2\Sigma_G$. Hence, all the
perturbation bounds of Section~\ref{sec:main} hold for
this model if
\begin{align}\label{eq:cond_LG}
\sup_{t\geq 1}\big\|\Upsilon y_{t+1}-B\Upsilon y_t\big\|
<\infty,
\qquad\text{where}\quad
\Upsilon\coloneqq(A^\top\Sigma_G^{-1}A)^{-1}
A^\top\Sigma_G^{-1}.
\end{align}
For $B=I_d$ and $r=d$, \eqref{eq:cond_LG} reduces
to $\sup_{t\geq 1}\|y_{t+1}-y_t\|<\infty$; the
observations may be unbounded. Condition
\eqref{eq:cond_LG} involves no stability assumption on
$B$: it allows $B=I_d$ (random walk) and matrices with
spectral radius larger than one.

We now consider observations generated by the model,
that is, $y_t=Y_t(\omega)$ for $t\geq 1$, where
$(X_t,Y_t)_{t\geq 1}$ is the process defined by
\eqref{eq:SSM} and \eqref{eq:SSM2}. Since $\Upsilon A=I_d$, the random vectors
\[
Z_t\coloneqq\Upsilon Y_{t+1}-B\Upsilon Y_t
=\Sigma_M^{1/2}V_{t+1}+\Upsilon\Sigma_G^{1/2}U_{t+1}
-B\Upsilon\Sigma_G^{1/2}U_t,\qquad t\geq 1,
\]
form a stationary, $1$-dependent, centered Gaussian
sequence with covariance matrix
$\Gamma\coloneqq\Sigma_M+\Sigma_\Upsilon+B\Sigma_\Upsilon B^\top$,
where $\Sigma_\Upsilon\coloneqq(A^\top\Sigma_G^{-1}A)^{-1}
=\Upsilon\Sigma_G\Upsilon^\top$. Hence
$\sup_{t\geq 1}\|Z_t\|=\infty$ almost surely, and condition
\eqref{eq:cond_LG} fails for almost every realization of
the observations. However, for any $T\geq 2$ and
$q\in(0,1]$, the event
\begin{align*}
\Omega_{T,q}\coloneqq\Big\{\max_{1\leq t<T}\|Z_t\|
\leq\sqrt{2d\,\sigma_{\max}(\Gamma)\log(2dT/q)}\Big\}
\end{align*}
has probability at least $1-q$. For any
$\omega\in\Omega_{T,q}$, the deterministic sequence
$y_t=Y_t(\omega)$ therefore satisfies the condition in
\eqref{eq:cond_LG} on the horizon $\{1,\dots,T\}$, and our
bounds hold on this horizon for this sequence, with
constants depending on $T$ through the above quantity,
which grows like $\|B\|\sqrt{\log(T/q)}$, since
$$\sigma_{\min}(\Sigma_\Upsilon)\|B\|^2\leq\sigma_{\max}(\Gamma)\leq
\sigma_{\max}(\Sigma_M)+\sigma_{\max}(\Sigma_\Upsilon)(1+\|B\|^2),$$
where $\|B\|$ denotes the operator norm of $B$.

To check this claim, note that each coordinate $Z_{t,i}$
of $Z_t$ is a
centered Gaussian random variable with variance at most
$\sigma_{\max}(\Gamma)$. Thus, for any $u\geq 0$ we have
\[
\P\big(|Z_{t,i}|>u\,\sigma_{\max}(\Gamma)^{1/2}\big)\leq
2e^{-u^2/2}.
\]
 A union bound over
$i\in\{1,\dots,d\}$ and $t\in\{1,\dots,T-1\}$, with
$u^2=2\log(2dT/q)$, and the estimate
$\|Z_t\|\leq\sqrt{d}\max_{i}|Z_{t,i}|$ give the claim.
We stress that the factor $\sqrt{\log T}$ cannot be removed, since
$\max_{1\leq t\leq T}\|Z_t\|/\sqrt{\log T}$ has a
positive almost sure limit as $T\to\infty$.

Lastly, we note that for $y_t=0$ the supremum in \eqref{eq:cond_LG} vanishes
for any matrix $B$, including unstable ones.
The resulting model is the so-called harmonic oscillator model of
\citet{delmoral2023coupled}, for which all our
perturbation bounds hold; see also
\citet{gerber2026safety} for the one-dimensional case.

\subsubsection{A tracking model\label{sub:track_ex}}

We consider the 2-DOF (two-degrees-of-freedom)
helicopter model used in
\citet{kanthalakshmi2022active} for fault detection
and diagnosis with particle filters. This laboratory prototype is mounted on a fixed base;
its front propeller controls the pitch axis and its rear
propeller the yaw axis. The pitch angle $\theta_t$
and the yaw angle $\psi_t$ are the two degrees of
freedom of the system, and high-resolution encoders
provide measurements $z_{t,1}$ of $\theta_t$ and
$z_{t,2}$ of $\psi_t$. The system has two inputs,
the pitch motor voltage $u_t$ and the yaw motor
voltage $v_t$. We denote by $\dot{\theta}_t$ and
$\dot{\psi}_t$ the pitch and yaw velocities, and let
$x_t=(\theta_t,\psi_t,\dot{\theta}_t,\dot{\psi}_t)$
and $z_t=(z_{t,1},z_{t,2})$.

To estimate $x_t$ from
$z_t$, \citet{kanthalakshmi2022active}
use the Feynman--Kac model with, for any $t\geq 1$
and $x\in\R^4$,
\begin{align}\label{eq:Hic1}
G_t(x)=\exp\Big(-\big(z_t-H(x)\big)^\top\Sigma_G^{-1}
\big(z_t-H(x)\big)\Big)
\quad\text{and}\quad
\delta_xM_{t+1}=\mathcal{N}_4\big(F_{t+1}(x),\Sigma_M\big),
\end{align}
for some matrices $\Sigma_G\in\mathcal{S}^+_2$ and
$\Sigma_M\in\mathcal{S}_4^+$, where $H(x)=(x_1,x_2)$
and
\[
F_{t+1}(x)=\big(x_1+x_3,\,x_2+x_4,\,
F_{t+1,3}(x),\,F_{t+1,4}(x)\big),
\]
with
\begin{align*}
F_{t+1,3}(x)&=-\frac{a_4x_3+a_2a_3^2x_4^2\sin(x_1)\cos(x_1)
+a_2a_3\cos(x_1)-a_5u_t-a_6v_t}{a_1+a_2a_3^2},\\
F_{t+1,4}(x)&=-\frac{a_8x_4+2a_2a_3^2x_3x_4\sin(x_1)\cos(x_1)
-a_9u_t-a_{10}v_t}{a_7+a_2a_3^2\cos(x_1)^2},
\end{align*}
and $a\in(0,\infty)^{10}$ a vector of parameters.
Observe that this is the Feynman--Kac representation of the
state-space model \eqref{eq:SSM},   with
\[\Sigma_{G_t}^{1/2}U_t\iid\mathcal{N}_2(0,\Sigma_G/2)\quad
\text{and} \quad\Sigma_{M_{t+1}}^{1/2}V_{t+1}\iid
\mathcal{N}_4(0,\Sigma_M).
\]

Since the observations are two-dimensional while the
state is four-dimensional, the function  $H$ is not bi-Lipschitz,
so that \ref{Model1} and \ref{Model} fail and our
results do not apply to \eqref{eq:Hic1}. We can however augment the observations with the finite differences of the
measured angles: for any
$t\geq 1$, let $z_{t,3}\coloneqq z_{t+1,1}-z_{t,1}$ and
$z_{t,4}\coloneqq z_{t+1,2}-z_{t,2}$, which estimate
$\dot{\theta}_t$ and $\dot{\psi}_t$, and let
$y_t\coloneqq(z_{t,1},z_{t,2},z_{t,3},z_{t,4})$.
Given $\tilde{\Sigma}_G\in\mathcal{S}_4^+$, we then
consider the Feynman--Kac model with, for any
$t\geq 1$ and $x\in\R^4$,
\begin{align}\label{eq:FK_helicopter2}
G_t(x)=\exp\Big(-(y_t-x)^\top\tilde{\Sigma}_G^{-1}
(y_t-x)\Big)
\quad\text{and}\quad
\delta_xM_{t+1}=\mathcal{N}_4\big(F_{t+1}(x),\Sigma_M\big).
\end{align}
Since $y_t$ depends on $z_{t+1}$, the distribution
$\widehat{\eta}_t$ of this model is computed with a
delay of one time step, that is, once $z_{t+1}$ is
available. The voltages $u_t$ and
$v_t$ and the observations
$y_t$ are bounded in practice, owing to
physical constraints. Using Lemma~\ref{lemma:Obs2}, we check that the model
\eqref{eq:FK_helicopter2} satisfies all the
assumptions of Proposition~\ref{prop:example2}, while
the functions $F_{t+1}$ do not satisfy
those of Proposition~\ref{prop:example}. Hence, all
the perturbation bounds of Section~\ref{sec:main}
hold for \eqref{eq:FK_helicopter2}.

\subsubsection{A multisensory perception
model\label{sub:multi_mod}}

The next model, from \citet{kutschireiter2017nonlinear},
is a state-space model of the form \eqref{eq:SSM} with
nonlinear observation and transition functions. A frog sitting
between two branches tracks the position
$x_t\in\R$ of an insect flying between them, using
an observation $z_t\in\R^2$ whose two components are
measurements of $x_t$ obtained through its visual and
auditory channels. \citet{kutschireiter2017nonlinear}
use the Feynman--Kac model with, for any $t\geq 1$
and $x\in\R$,
\begin{align}\label{eq:frog}
G_t(x)=\exp\Big(-\big(y_t-\tilde{H}(x)\big)^\top\Sigma_G^{-1}
\big(y_t-\tilde{H}(x)\big)\Big)
\quad\text{and}\quad
\delta_xM_{t+1}=\mathcal{N}_1\big(F(x),\sigma^2_M\big),
\end{align}
for some matrix $\Sigma_G\in\mathcal{S}_2^+$ and constant
$\sigma_M\in(0,\infty)$,   where
$y_t\coloneqq z_{t+1}-z_t$ is the increment of the
measurements (the model being a discretization of a
continuous-time model), and where  $\tilde{H}(x)=(x,\tanh(2x))$ while
$F(x)=x+3x(1-x^2)$. Note that \eqref{eq:frog} is the Feynman--Kac
representation of the state-space model
\eqref{eq:SSM} with
\[\Sigma_{G_t}^{1/2}U_t\iid\mathcal{N}_2(0,\Sigma_G/2)\quad
\text{and}
\quad\Sigma_{M_{t+1}}^{1/2}V_{t+1}\iid
\mathcal{N}_1(0,\sigma_M^2).
\]

The function $\tilde{H}$ is bi-Lipschitz, but $H(\R)$ is the
graph of $x\mapsto\tanh(2x)$, which is a curve in $\R^2$.
Condition~\ref{Obs1} therefore requires
$y_{t,2}=\tanh(2y_{t,1})$ for any $t\geq 1$, which
does not hold for real data. However,  as we now show, by  interpreting the contribution of the auditory channels to the model  as a bounded
perturbation of the potential functions associated to the visual channels, our results  apply for this model.

Assume that $\Sigma_G=\mathrm{diag}(s_1,s_2)$ for some
$s_1,s_2\in(0,\infty)$, (independent sensory channels), and that the observations are bounded, i.e.\
$\sup_{t\geq 1}\|y_t\|\leq\Xi$ for some
$\Xi\in[0,\infty)$. Then $G_t=b_t\,g_t$ for any
$t\geq 1$, where
\begin{align*}
g_t(x)\coloneqq\exp\big(-(y_{t,1}-x)^2/s_1\big)
\quad\text{and}\quad
b_t(x)\coloneqq\exp\big(-(y_{t,2}-\tanh(2x))^2/s_2\big),
\qquad x\in\R.
\end{align*}
Since $|\tanh|\leq 1$, we have
$b_t(x)\in[c,1]$ with
$c\coloneqq\exp\big(-(\Xi+1)^2/s_2\big)$, and
$\sup_{t\geq 1}\|b_t'\|_\infty\leq
4(\Xi+1)/s_2$.

The potential functions $g_t$ are of the
form \eqref{eq:Gt_def} with the observation function
$H_t(x)= x$, for which \ref{Obs1} trivially holds,
and with a Gaussian function $\zeta_{G_t}$, so that
\ref{Model} holds with an arbitrarily large constant
$\nu_G$. The function $F(x)=4x-3x^3$ is not globally
Lipschitz, so that it does not satisfy the
assumptions of Proposition~\ref{prop:example}.
However, $|F'(x)|=|4-9x^2|$ grows quadratically;
since $\nu_G$ can be taken arbitrarily large for
Gaussian potential functions, $F$ satisfies
\eqref{eq:bound_derive_ht0} and
\eqref{eq:bound_derive_ht}. Moreover, since the
observations are bounded and the model is
homogeneous, \ref{Obs2} holds by
Lemma~\ref{lemma:Obs2}. Hence, the Feynman--Kac model
with potential functions $g_t$ satisfies
the assumptions of both parts of
Proposition~\ref{prop:example2}, and, by
Lemma~\ref{lemma:bounded_pert}, the model
\eqref{eq:frog} satisfies \ref{A3} and \ref{B1}. All the perturbation bounds of Section~\ref{sec:main}
therefore hold for this model.
\subsubsection{Partially observed stochastic
differential equations\label{sub:SDE}}

Let $k\geq d\geq 1$, let
$b\colon\R^d\to\R^d$ and
$\sigma\colon\R^d\to\R^{d\times k}$ be measurable
functions, and consider the stochastic differential
equation
\begin{align*}
X_\tau=X_0+\int_0^\tau b(X_s)\,\dd s
+\int_0^\tau\sigma(X_s)\,\dd B_s,
\qquad\tau\in[0,\infty),
\end{align*}
where $X_0$ is an $\R^d$-valued random variable and
$(B_s)_{s\geq 0}$ is a $k$-dimensional Wiener process
independent of $X_0$. The process is observed at
integer times, and we let $M$ be the Markov transition of
the discrete-time process $X_t$ and, for a
sequence $y_t$ in $\R^d$, we let
\begin{align}\label{eq:Model_SDE}
G_t(x)=\exp\big(-\|y_t-x\|^2\big)
\quad\text{and}\quad
M_{t+1}=M,
\qquad\forall x\in\R^d,\quad\forall t\geq 1.
\end{align}
We write $b_i$ for the $i$-th component of $b$ and
$\sigma_{ij}$ for the $(i,j)$-th entry of $\sigma$,
and we assume that:
\begin{enumerate}[label=\textbf{(D{{\arabic*}})}]
\item\label{diff1} for any $i\in\{1,\dots,d\}$ and
$j\in\{1,\dots,k\}$, the functions $b_i$ and
$\sigma_{ij}$ are bounded, and the functions
$\sigma_{ij}$ are H\"older continuous;
\item\label{diff2} there exists a constant $c_\sigma>0$ such that, for any $x_1,x_2\in\R^d$, we have
\[x_1^\top\sigma(x_2)\sigma(x_2)^\top x_1\geq
c_\sigma\|x_1\|^2.\]
\end{enumerate}
Under these assumptions,
$M(x,\dd x')=m(x,x')\,\lambda_d(\dd x')$ for any
$x\in\R^d$, where the transition density $m$
satisfies two-sided Gaussian bounds: for some
constants $(c_1,c_2)\in(0,1)^2$,
\begin{align}\label{eq:densSDE}
c_1\exp\Big(-\frac{\|x'-x\|^2}{c_1}\Big)
\leq m(x,x')
\leq\frac{1}{c_2}\exp\big(-c_2\|x'-x\|^2\big),
\qquad\forall(x,x')\in\R^d\times\R^d;
\end{align}
see \citet[][Chapter~9]{friedman1964partial}
and \citet[][Theorem~1.2]{menozzi2021density},
where these bounds are established by the
parametrix method; see also
\citet{aronson1967bounds} for the analogous bounds
for operators in divergence form.

By \eqref{eq:densSDE} and the definition
\eqref{eq:Model_SDE} of $G_t$, condition
\ref{Model11} holds with $F_{t+1}(x)=x$ for any
$t\geq 1$. Hence, by Section~\ref{sub:extension},
\ref{A3} holds if
$\sup_{t\geq 1}\|y_{t+1}-y_t\|<\infty$, and the total
variation bounds of Section~\ref{sub:TV_estimates}
hold for the model \eqref{eq:Model_SDE}; notice that the observations may be unbounded.

Assume now that, in addition:
\begin{enumerate}[label=\textbf{(D{{\arabic*}})}]
\setcounter{enumi}{2}
\item\label{diff3} for any $i\in\{1,\dots,d\}$ and
$j\in\{1,\dots,k\}$, the functions $b_i$ and
$\sigma_{ij}$ are infinitely differentiable, with
bounded derivatives of all orders.
\end{enumerate}
Then, by \citet[][Theorem~6]{guyon2006euler}, there
exists a constant $c\in(0,\infty)$ such that, for
any non-empty sets $u\subseteq\{1,\dots,d\}$,
\begin{align}\label{eq:derive_sde}
\big|\partial^u_1m(x,x')\big|
\leq\frac{1}{c}\exp\big(-c\|x'-x\|^2\big),
\qquad\forall(x,x')\in\R^d\times\R^d.
\end{align}
Together with \eqref{eq:densSDE}, this shows that
\ref{Model3} holds. Therefore, under conditions
\ref{diff1}--\ref{diff3}, \ref{A3} and \ref{B1} hold
for the model \eqref{eq:Model_SDE} if
$\sup_{t\geq 1}\|y_{t+1}-y_t\|<\infty$, and all the
perturbation bounds of Section~\ref{sec:main} are valid.

\section{Approximate Feynman--Kac
models\label{sec:compareFK}}

\subsection{Context and general result}

Throughout this section, $G^\star_t$ stands
for some sequence of functions in
$\mathcal{B}_1(\setX)$ and $M^\star_{t+1}$
for some sequence of Markov transitions on $\setX$. To
simplify the presentation, we assume that $G_t(x)>0$ and $G^\star_t(x)>0$ for any
$t\geq 1$ and $x\in\setX$, so that \ref{A0} holds with
$W_t=\setX$. For any $t\geq 1$, we let
\[
Q^\star_{t+1}(x,\dd x')\coloneqq
G^\star_t(x)\,M^\star_{t+1}(x,\dd x'),
\]
 and we define
the mappings
$\Psi^\star_t\colon\mathcal{P}(\setX)\to\mathcal{P}(\setX)$
and
$\Phi^\star_{t+1}\colon\mathcal{P}(\setX)\to\mathcal{P}(\setX)$
by
\begin{align*}
\Psi^\star_t(\mu)(\dd x)\coloneqq
\frac{G^\star_t(x)}{\mu(G^\star_t)}\,\mu(\dd x)
\qquad\text{and}\qquad
\Phi^\star_{t+1}(\mu)\coloneqq\Psi^\star_t(\mu)M^\star_{t+1}.
\end{align*}
We consider the two Feynman--Kac prediction models
$\eta_t$ and $\eta^\star_t$
defined  by
\begin{align}\label{eq:model3}
\eta_{t+1}=\Phi_{t+1}(\eta_t),
\qquad\text{and}\qquad
\eta^\star_{t+1}=\Phi^\star_{t+1}(\eta^\star_t)
\qquad\text{with}\quad
\eta^\star_1=\eta_1\in\mathcal{P}(\setX),
\end{align}
and the corresponding updated models
\begin{align}\label{eq:model4}
\widehat{\eta}_t=\Psi_t(\eta_t)
\qquad\text{and}\qquad
\widehat{\eta}^\star_t=\Psi^\star_t(\eta^\star_t).
\end{align}
Thus, $\eta_t$ and
$\widehat{\eta}_t$ are the prediction and
updated models of Section~\ref{sub:setup}, while
$\eta^\star_t$ and
$\widehat{\eta}^\star_t$ are the ones of the Feynman--Kac model with potential
functions $G^\star_t$ and Markov
transitions $M^\star_{t+1}$. In this section, we interpret  the latter model as an approximation of the former model, which is  easier to compute or to
sample.

The next theorem   bounds the distances between $\eta^\star_t$ and
$\eta_t$, and between $\widehat{\eta}^\star_t$ and
$\widehat{\eta}_t$, uniformly in time, in terms of the
function $\vartheta$ defined  in \eqref{eq:vartheta} and the
quantity $\Delta_{Q|Q^\star}\in[0,1]$ defined in  \eqref{eq:delta}.

\begin{theorem}\label{thm:compare}
Let $\eta_t$ and
$\widehat{\eta}_t$ be as defined in
\eqref{eq:model3}--\eqref{eq:model4}. Assume that
\ref{A2} holds for some sets $D_t$ with
$\eta_1(D_1)>0$. Then, there exists a constant
$C_{\eta_1}\in[1,\infty)$, which depends on
$G_t$, $M_{t+1}$ and $\eta_1$
but not on $G^\star_t$ and
$M^\star_{t+1}$, such that
\begin{align*}
\sup_{t\geq 1}\|\eta^\star_t-\eta_t\|_{\mathrm{tv}}
\vee
\sup_{t\geq 1}
\|\widehat{\eta}^\star_t-\widehat{\eta}_t\|_{\mathrm{tv}}
\leq C_{\eta_1}\,\vartheta\big(\Delta_{Q|Q^\star}\big).
\end{align*}
\end{theorem}

\begin{proof}
 Recall  that \ref{A1} holds under \ref{A2} and let $\underline{q}\in(0,1)$ be as in
\ref{A1a}. Then, note that since
$Q_{t+1}(x,D_{t+1})\leq G_t(x)$, it follows from \ref{A1a} that we have $G_t\geq\underline{q}$ on $D_t$,
and therefore
\begin{align}\label{eq:lower_mu_G}
\mu(G_t)\geq\underline{q}\,\mu(D_t),
\qquad\forall\mu\in\mathcal{P}(\setX),
\quad\forall t\geq 1.
\end{align}

\emph{Step 1:   one-step estimate.} Let $t\geq 1$,
$\mu\in\mathcal{P}(\setX)$ and
$f\in\mathcal{B}_1(\setX)$, and let
$c\coloneqq\Phi^\star_{t+1}(\mu)(f)\in[0,1]$. Since
$\mu Q^\star_{t+1}(f)=c\,\mu Q^\star_{t+1}(1)$ and
$\mu Q_{t+1}(1)=\mu(G_t)$, we have
\begin{align*}
\Phi^\star_{t+1}(\mu)(f)-\Phi_{t+1}(\mu)(f)
=\frac{(\mu Q^\star_{t+1}-\mu Q_{t+1})(f-c)}{\mu(G_t)}.
\end{align*}
As $\|f-c\|_\infty\leq 1$,
$|\nu(g)|\leq 2\|\nu\|_{\mathrm{tv}}\|g\|_\infty$ for
any $\nu\in\mathcal{M}(\setX)$ and bounded measurable
$g$, and
$\|\mu Q^\star_{t+1}-\mu Q_{t+1}\|_{\mathrm{tv}}\leq
\sup_{x\in\setX}\|\delta_xQ^\star_{t+1}
-\delta_xQ_{t+1}\|_{\mathrm{tv}}\leq\Delta_{Q|Q^\star}$,
we check that
\begin{align}\label{eq:one_step_Phi}
\big\|\Phi^\star_{t+1}(\mu)-\Phi_{t+1}(\mu)\big\|_{\mathrm{tv}}
\leq\frac{2\,\Delta_{Q|Q^\star}}{\mu(G_t)}.
\end{align}
Similarly, with
$c'\coloneqq\Psi^\star_t(\mu)(f)$, we have
$\Psi^\star_t(\mu)(f)-\Psi_t(\mu)(f)
=\mu\big((G^\star_t-G_t)(f-c')\big)/\mu(G_t)$ and,
since
$|G^\star_t(x)-G_t(x)|
=|(\delta_xQ^\star_{t+1}-\delta_xQ_{t+1})(\setX)|
\leq\Delta_{Q|Q^\star}$,
\begin{align}\label{eq:one_step_Psi}
\big\|\Psi^\star_t(\mu)-\Psi_t(\mu)\big\|_{\mathrm{tv}}
\leq\frac{\Delta_{Q|Q^\star}}{\mu(G_t)}.
\end{align}

\emph{Step 2: bound for the prediction model.} We apply
Theorem~\ref{thm:main_TV} with
$\bm{\mu}=(\eta^\star_t)_{t\geq 1}$ and
$\varpi(z)\coloneqq z\wedge 1$. Since
$\eta^\star_1=\eta_1$ and
$\eta^\star_{t+1}=\Phi^\star_{t+1}(\eta^\star_t)$,
\eqref{eq:lower_mu_G} and \eqref{eq:one_step_Phi}
give, for any $T\geq 1$,
\begin{align*}
\Delta_{T,\bm{\mu}|\eta_1,\varpi}
=\max_{1\leq t<T}
\big\|\Phi^\star_{t+1}(\eta^\star_t)
-\Phi_{t+1}(\eta^\star_t)\big\|_{\mathrm{tv}}\,
\eta^\star_t(D_t)
\leq\frac{2}{\underline{q}}\,\Delta_{Q|Q^\star}.
\end{align*}
By Theorem~\ref{thm:main_TV}, \eqref{eq:vartheta_scale}
and the monotonicity of $\vartheta$, there exists
a constant $C_1\in[1,\infty)$, depending only on
$G_t$, $M_{t+1}$ and $\eta_1$,
such that
\begin{align}\label{eq:compare_pred}
\sup_{t\geq 1}\|\eta^\star_t-\eta_t\|_{\mathrm{tv}}
\leq C_1\,\vartheta\big(\Delta_{Q|Q^\star}\big).
\end{align}

\emph{Step 3: bound for the updated model.} Observe that, for any $c\in[1,\infty)$ and $z\in(0,\infty)$, we have
\begin{align}\label{eq:vartheta_scale}
\vartheta(cz)\leq c\,\vartheta(z),\quad  \vartheta(z)\geq z\log 2
\end{align}
Then, for any $t\geq 1$,
\begin{align*}
\|\widehat{\eta}^\star_t-\widehat{\eta}_t\|_{\mathrm{tv}}
\leq
\big\|\Psi^\star_t(\eta^\star_t)
-\Psi_t(\eta^\star_t)\big\|_{\mathrm{tv}}
+\big\|\Psi_t(\eta^\star_t)
-\Psi_t(\eta_t)\big\|_{\mathrm{tv}}.
\end{align*}
where, by part~\ref{tv_bounded} of Proposition~\ref{prop:filter_TV} and
\eqref{eq:compare_pred}, the second term is bounded by
$C_2C_1\vartheta(\Delta_{Q|Q^\star})$, where $C_2$
denotes the constant $C_{\eta_1}$ of that
proposition. The first term is null when $G^\star_t=G_t$; in
general, we estimate this term using \eqref{eq:one_step_Psi}. By
Lemma~\ref{lemma:lower_phi},
$g_\star\coloneqq\inf_{t\geq 1}\eta_t(G_t)>0$ and,
since $G_t\in\mathcal{B}_1(\setX)$, we have
\[
\eta^\star_t(G_t)\geq
g_\star-\|\eta^\star_t-\eta_t\|_{\mathrm{tv}}.
\]
 If
$C_1\vartheta(\Delta_{Q|Q^\star})\leq g_\star/2$, then
\eqref{eq:compare_pred} gives
$\eta^\star_t(G_t)\geq g_\star/2$ for any $t\geq 1$,
so that, by \eqref{eq:one_step_Psi},
\begin{align*}
\big\|\Psi^\star_t(\eta^\star_t)
-\Psi_t(\eta^\star_t)\big\|_{\mathrm{tv}}
\leq\frac{2\,\Delta_{Q|Q^\star}}{g_\star}
\leq\frac{2}{g_\star\log 2}\,
\vartheta\big(\Delta_{Q|Q^\star}\big).
\end{align*}
Otherwise, recalling that the total variation distance
between two probability measures is at most one, this
term is bounded by
$1<2C_1\vartheta(\Delta_{Q|Q^\star})/g_\star$. In both situations, this term is bounded by
$2C_1\vartheta(\Delta_{Q|Q^\star})/(g_\star\log 2)$.
This ends the proof of the theorem with the constant
\[
C_{\eta_1}\coloneqq C_1\Big(C_2+\frac{2}{g_\star\log 2}\Big),
\]
which only depends on $G_t$,
$M_{t+1}$ and $\eta_1$.
\end{proof}

When $G^\star_t=G_t$ for any $t\geq 1$, as in
Sections~\ref{sub:EKF} and~\ref{sub:Euler},
$\widehat{\eta}^\star_t=\Psi_t(\eta^\star_t)$, and part~\ref{tv_unbounded} of Proposition~\ref{prop:filter_TV} also
provides estimates of
$|\widehat{\eta}^\star_t(f)-\widehat{\eta}_t(f)|$ for
some unbounded functions $f$, such as the coordinates
of $x$.

\begin{remark}
When $\setX\subseteq\R^d$ is a hyper-rectangle, a
bound similar to that of Theorem~\ref{thm:compare} can
be obtained, under stronger assumptions on the model,
for the norm $\|\cdot\|_{\mathcal{V}_\setX}$ of
Section~\ref{sub:New_norm}, and hence for the
Kolmogorov norm. We do not develop this: for models with absolutely
continuous distributions, the total variation norm is
the strongest choice.
\end{remark}

We illustrate Theorem~\ref{thm:compare} with three
applications, without seeking generality: linearized
Kalman filters, Euler approximations of partially
observed diffusions and nudged state-space models. To the best of our knowledge, the
time-uniform estimates presented below are new.

\subsection{Linearized Kalman filters\label{sub:EKF}}

\subsubsection{Description of the models}

Let $r\geq d\geq 1$ be integers, let
$A\in\R^{r\times d}$ be a matrix of rank $d$, let
$\Sigma_G\in\mathcal{S}_r^+$ and
$\Sigma_M\in\mathcal{S}_d^+$, and let
$y_t$ be a sequence in $\R^r$. Let
$F\colon\R^d\to\R^d$ be a continuously differentiable
function, and denote by $J_F(x)\in\R^{d\times d}$ its
Jacobian matrix at $x\in\R^d$. In this section, for
any $t\geq 1$ and $x\in\R^d$,
\begin{align}\label{eq:Model_ekl}
G_t(x)\coloneqq
\exp\Big(-(y_t-Ax)^\top\Sigma_G^{-1}(y_t-Ax)\Big)
\qquad\text{and}\qquad
\delta_xM_{t+1}\coloneqq
\mathcal{N}_d\big(F(x),\Sigma_M\big),
\end{align}
and $\eta_1\coloneqq\mathcal{N}_d(w_0,\Sigma_0)$ for
some $w_0\in\R^d$ and $\Sigma_0\in\mathcal{S}_d^+$.
This model coincides with the Feynman--Kac
representation of the state-space model \eqref{eq:SSM}
with Gaussian noises,
$H_t(x)=Ax$, $F_{t+1}=F$, $\Sigma_{G_t}=\Sigma_G/2$
and $\Sigma_{M_{t+1}}=\Sigma_M$ for any $t\geq 1$.

When the function $F$ is linear, the distributions $\eta_t$ and
$\widehat{\eta}_t$ are Gaussian, and their means and
covariance matrices are given by the Kalman filter.
When $F$ is nonlinear, one can linearizes $F$ around some sequence of points
$x^\star_{t+1}$ in $\R^d$, and let, for any $t\geq 1$ and $x\in\R^d$,
\begin{align}\label{eq:M_starEK}
G^\star_t(x)\coloneqq G_t(x)
\qquad\text{and}\qquad
\delta_xM^\star_{t+1}\coloneqq
\mathcal{N}_d\Big(F(x^\star_{t+1})
+J_F(x^\star_{t+1})\,(x-x^\star_{t+1}),\,
\Sigma_M\Big).
\end{align}
The resulting distributions $\eta^\star_t$ and
$\widehat{\eta}^\star_t$ are Gaussian and can therefore be computed
exactly by the Kalman recursions. When
$x^\star_{t+1}$ is fixed in advance, each
$x^\star_{t+1}$ being a plausible value of
$X_t\sim\widehat{\eta}_t$, this procedure is the
\emph{linearized Kalman filter}. By contrast, the \emph{extended
Kalman filter}  chooses the points adaptively,
by letting $x^\star_{t+1}\coloneqq
\int_{\R^d}x\,\widehat{\eta}^\star_t(\dd x)$. In what follows, we use Theorem~\ref{thm:compare} to choose
$x^\star_{t+1}$ with a time-uniform error control.

\subsubsection{Local perturbation analysis}

Since $G^\star_t=G_t$ and the total variation norm is
positively homogeneous, the quantity defined in
\eqref{eq:delta} is
\begin{align}\label{eq:ineq_Delta}
\Delta_{Q|Q^\star}
=\sup_{t\geq 1}\sup_{x\in\R^d}
\Big(G_t(x)\,
\big\|\delta_xM^\star_{t+1}-\delta_xM_{t+1}\big\|_{\mathrm{tv}}\Big).
\end{align}
The nonlinearity of $F$ is measured by the oscillation
of its Jacobian
\begin{align}\label{def-kappaF}
\kappa_F\coloneqq
\sup_{x,x'\in\R^d}\big\|J_F(x)-J_F(x')\big\|
\in[0,\infty].
\end{align}
We have $\kappa_F=0$ if and only if $F$ is affine, and
$\kappa_F\leq 2L$ if $F$ is $L$-Lipschitz. For
instance, if $F(x)=Bx+\beta\,g(x)$ with
$B\in\R^{d\times d}$ and $g$ continuously
differentiable with bounded Jacobian, then
$\kappa_F\leq 2\beta\sup_{x\in\R^d}\|J_g(x)\|$.
Writing
\[
F(x)-F(x_0)-J_F(x_0)(x-x_0)
=\int_0^1\big(J_F(x_0+s(x-x_0))-J_F(x_0)\big)(x-x_0)\,\dd s,
\]
we obtain that
\begin{align}\label{eq:taylor}
\big\|F(x)-F(x_0)-J_F(x_0)(x-x_0)\big\|
\leq\kappa_F\,\|x-x_0\|,
\qquad\forall x,x_0\in\R^d.
\end{align}

We also use the following estimate: for any
$m_1,m_2\in\R^d$ and $\Sigma\in\mathcal{S}_d^+$,
\begin{align}\label{eq:TV_gauss}
\big\|\mathcal{N}_d(m_1,\Sigma)
-\mathcal{N}_d(m_2,\Sigma)\big\|_{\mathrm{tv}}
\leq\frac{\|m_1-m_2\|}{\sqrt{2\pi\,\sigma_{\min}(\Sigma)}}.
\end{align}
Indeed, the total variation
distance is invariant under the bijection $z\mapsto\Sigma^{-1/2}(z-m_2)$ and
under rotations, so that the left-hand side of \eqref{eq:TV_gauss} equals
\begin{align*}
\|\mathcal{N}_1(\rho,1)-\mathcal{N}_1(0,1)\|_{\mathrm{tv}}
&=2\Phi_{\mathcal{N}}(\rho/2)-1\leq\rho/\sqrt{2\pi}
\\
\text{with}\quad
\rho&\coloneqq\|\Sigma^{-1/2}(m_1-m_2)\|
\leq\sigma_{\min}(\Sigma)^{-1/2}\|m_1-m_2\|,
\end{align*}
and with
$\Phi_{\mathcal{N}}$ is the standard Gaussian
distribution function; see also
\citet[][Theorem~1]{barsov1987estimates}. We set
\[
C_{\mathcal{N}}\coloneqq(2\pi\,\sigma_{\min}(\Sigma_M))^{-1/2}.
\]
Then, by using   \eqref{eq:taylor} and
\eqref{eq:TV_gauss}, for any $t\geq 1$ and
$x\in\R^d$, we have
\begin{align}\label{eq:bound_TV_normal}
\big\|\delta_xM^\star_{t+1}-\delta_xM_{t+1}\big\|_{\mathrm{tv}}
\leq C_{\mathcal{N}}\,
\big\|F(x)-F(x^\star_{t+1})
-J_F(x^\star_{t+1})(x-x^\star_{t+1})\big\|
\leq\kappa_F\,C_{\mathcal{N}}\,\|x-x^\star_{t+1}\|.
\end{align}

To take advantage of the decay of $G_t$, we set
\[
K\coloneqq(A^\top\Sigma_G^{-1}A)^{-1}A^\top\Sigma_G^{-1},
\]
so that $Ky_t$ is  the maximizer of
$G_t$. Since $A^\top\Sigma_G^{-1}(y_t-AKy_t)=0$, writing
\[
y_t-Ax=(y_t-AKy_t)+A(Ky_t-x)
\]
 gives
\begin{align}\label{eq:New_G}
G_t(x)=c_t\,\widetilde{G}_t(x),
\qquad
\widetilde{G}_t(x)\coloneqq
\exp\Big(-(Ky_t-x)^\top A^\top\Sigma_G^{-1}A\,(Ky_t-x)\Big),
\qquad x\in\R^d,
\end{align}
with
\[
c_t\coloneqq\exp\big(-(y_t-AKy_t)^\top\Sigma_G^{-1}(y_t-AKy_t)\big)\in(0,1].
\]

Multiplying $G_t$ and $G^\star_t$ by the same positive
constant does not change the mappings $\Psi_t$,
$\Phi_{t+1}$, $\Psi^\star_t$ and $\Phi^\star_{t+1}$.
Thus, the models \eqref{eq:model3}--\eqref{eq:model4}
do not change when
$G_t$ and $G^\star_t$ are both replaced by
$\widetilde{G}_t$; we denote by
$\widetilde{\Delta}_{Q|Q^\star}$ the quantity
\eqref{eq:ineq_Delta} computed with $\widetilde{G}_t$
in place of $G_t$. Since $c_t\leq 1$, we have
$\Delta_{Q|Q^\star}\leq\widetilde{\Delta}_{Q|Q^\star}$ and, with
$l_A\coloneqq\sigma_{\min}(A^\top\Sigma_G^{-1}A)>0$,  
\begin{align}\label{eq:A_D}
\widetilde{G}_t(x)\leq
\exp\big(-l_A\,\|Ky_t-x\|^2\big).
\end{align}

The above calculations yield the following result.
\begin{lemma}\label{lemma:extended}
Let $G_t$ and $M_{t+1}$ be as
defined in \eqref{eq:Model_ekl}, and let
$G^\star_t$ and
$M^\star_{t+1}$ be as defined in
\eqref{eq:M_starEK} with $x^\star_{t+1}=Ky_t$ for any
$t\geq 1$. Then,
\begin{align*}
\Delta_{Q|Q^\star}\leq\widetilde{\Delta}_{Q|Q^\star}
\leq C_{A,M}\,\kappa_F,
\qquad
C_{A,M}\coloneqq
\big(4\pi e\,l_A\,\sigma_{\min}(\Sigma_M)\big)^{-1/2}.
\end{align*}
\end{lemma}

\begin{proof}
By \eqref{eq:bound_TV_normal} and \eqref{eq:A_D},
with $x^\star_{t+1}=Ky_t$,
\begin{align*}
\widetilde{G}_t(x)\,
\big\|\delta_xM^\star_{t+1}-\delta_xM_{t+1}\big\|_{\mathrm{tv}}
\leq\kappa_F\,C_{\mathcal{N}}\,\|x-Ky_t\|\,
e^{-l_A\|x-Ky_t\|^2}
\leq\kappa_F\,C_{\mathcal{N}}\sup_{z\geq 0}z\,e^{-l_Az^2}
=\frac{\kappa_F\,C_{\mathcal{N}}}{\sqrt{2e l_A}}
\end{align*}
for any $t\geq 1$ and $x\in\R^d$, and
$C_{\mathcal{N}}/\sqrt{2e l_A}=C_{A,M}$.
\end{proof}

With $x^\star_{t+1}=Ky_t$, the function $F$ is linearized where
$G_t$, the weight of the linearization error, is
maximal. This choice only depends on the observations,
and can be computed before running the Kalman
recursions.

\subsubsection{Time-uniform error bounds}

\begin{proposition}\label{prop:extended}
Let $G_t$ and $M_{t+1}$ be as
defined in \eqref{eq:Model_ekl}, let
$G^\star_t$ and
$M^\star_{t+1}$ be as defined in
\eqref{eq:M_starEK} with $x^\star_{t+1}=Ky_t$ for any
$t\geq 1$, and let
$\eta_1=\mathcal{N}_d(w_0,\Sigma_0)$. Let
$L\in[1,\infty)$ and $\Xi\in[0,\infty)$. Then, there exists a constant
$C\in[1,\infty)$, that only depends on
$(A,\Sigma_G,\Sigma_M,w_0,\Sigma_0,L,\Xi)$, such that
the following estimates hold for any continuously
differentiable and $L$-Lipschitz function $F$ and
every sequence $y_t$ such that
$\|Ky_1-w_0\|\leq\Xi$ and
$\sup_{t\geq 1}\|Ky_{t+1}-F(Ky_t)\|\leq\Xi$:
\begin{align*}
\sup_{t\geq 1}\|\eta^\star_t-\eta_t\|_{\mathrm{tv}}
\vee
\sup_{t\geq 1}
\|\widehat{\eta}^\star_t-\widehat{\eta}_t\|_{\mathrm{tv}}
\leq C\,\vartheta(\kappa_F)
\end{align*}
with $\vartheta$ as in \eqref{eq:vartheta}
and, if $\kappa_F\leq 1/C$, then, for any $t\geq 1$,
\begin{align*}
\Big\|\int_{\R^d}x\,\widehat{\eta}^\star_t(\dd x)
-\int_{\R^d}x\,\widehat{\eta}_t(\dd x)\Big\|
\leq C\,\vartheta(\kappa_F).
\end{align*}
\end{proposition}

\begin{proof}
As explained above, there is no loss of generality in
replacing $G_t$ and $G^\star_t$ by $\widetilde{G}_t$. With
$\widetilde{\Sigma}\coloneqq(A^\top\Sigma_G^{-1}A)^{-1}$,
the model $(\widetilde{G}_t,M_{t+1})$ is of
the form \eqref{eq:Gt_def}--\eqref{eq:Mt_def2}, with
observations $Ky_t$ in $\R^d$,
$H_t(x)=x$, $\Sigma_{G_t}=\widetilde{\Sigma}$,
$\zeta_{G_t}(z)=e^{-z}$, $F_{t+1}=F$,
$\Sigma_{M_{t+1}}=\Sigma_M$ and
$\zeta_{M_{t+1}}(z)=e^{-z/2}$ (see also
Section~\ref{sub:exp}). For this model, \ref{Model1}
holds with 
\[
\zeta_G(z)=e^{-z},\quad \zeta_M(z)=e^{-z/2},\quad
\varepsilon=\nu_G=\nu_M=1
\]
 and a constant $L'\geq L$
depending only on $(d,A,\Sigma_G,\Sigma_M)$;
\ref{Obs1} holds trivially and \ref{Obs2} holds with
the constant $\Xi$. Since $F$ is $L$-Lipschitz,
part~\ref{i-prop:example} of Proposition~\ref{prop:example} applies
with $\varphi(z)=Lz$, so that \ref{A3}, and thus
\ref{A2}, holds. Theorem~\ref{thm:compare}, Lemma~\ref{lemma:extended}
and \eqref{eq:vartheta_scale} then give the first
bound, with the constant
$C_{\eta_1}\max\{1,C_{A,M}\}$, with $C_{\eta_1}$ as in  Theorem~\ref{thm:compare}.

 In particular, from  Remark~\ref{rem:constant} (proved in
Appendix~\ref{p-rem:constant}), condition \ref{A2} holds with the sets
\[
D_t=\{x\in\R^d:\ \|x-Ky_t\|_\infty\leq 1\}
\]
and the condition
$\|Ky_1-w_0\|\leq\Xi$ implies that
$\eta_1(D_1)\geq a$ for some $a\in(0,1]$ depending
only on $(d,w_0,\Sigma_0,\Xi)$. Together with the part~\ref{rem1} of
Remark~\ref{rem:constant}, this gives
the claimed dependence of $C$.

For the second bound, note that
$\widehat{\eta}^\star_t=\Psi_t(\eta^\star_t)$ since
$G^\star_t=G_t$. In the representation above,
$H_t(x)=x$, so that, by the discussion following
Proposition~\ref{prop:example}, part~\ref{ii-prop:example} of
Proposition~\ref{prop:filter_TV} applies to the
coordinate functions $f(x)=x_i$,
$i\in\{1,\dots,d\}$. Let $c_{\eta_1}$ be as in that
proposition and let $\delta_0\in(0,1]$ be such that
\[
C_{\eta_1}\max\{1,C_{A,M}\}\vartheta(\delta_0)\leq
c_{\eta_1}.
\]
 Increasing $C$ if necessary, we may
assume that $1/C\leq\delta_0$. Then $\kappa_F\leq 1/C$
implies $\|\eta^\star_t-\eta_t\|_{\mathrm{tv}}\leq
c_{\eta_1}$ for any $t\geq 1$, and the second bound
follows, after a further increase of $C$, from
part~\ref{tv_unbounded} of Proposition~\ref{prop:filter_TV} applied
to each coordinate.
\end{proof}

By Proposition \ref{prop:extended}, the closer $F$ is to an affine function, the
smaller the time-uniform error of the linearized Kalman
filter. Recall that
$\lim_{z\downarrow 0}\vartheta(z)=0$ and that
$\kappa_F=0$ when $F$ is affine, in which case the
approximation is exact.

\subsection{Euler approximation of partially observed
diffusions\label{sub:Euler}}

\subsubsection{Description of the models}

We consider the stochastic differential equation of
Section~\ref{sub:SDE}, where the functions $b$ and $\sigma$ satisfy
conditions~\ref{diff1}--\ref{diff3}, and we let $M$
denote the Markov transition of the discrete-time
process $X_t$. Recall that
$M(x,\dd x')=m(x,x')\,\lambda_d(\dd x')$ for any
$x\in\R^d$. Given a sequence $y_t$ in
$\R^d$, we let, as in \eqref{eq:Model_SDE},
\begin{align}\label{eq:G_diff}
G_t(x)\coloneqq\exp\big(-\|y_t-x\|^2\big)
\qquad\text{and}\qquad
M_{t+1}\coloneqq M,
\qquad x\in\R^d,\quad t\geq 1.
\end{align}

Here, the function  $m$ is intractable and sampling from $M(x,\dd x')$
is not trivial \citep[see][for exact
simulation methods]{beskos2006exact}. Particle filters designed for this problem exist
\citep{fearnhead2008particle}, with the drawback to have a random running
time. A common alternative is to run a particle filter (Section~\ref{sec:PF}) on a
time-discretization of the model; see for instance
\citet{delmoral2001monte} for particle approximations
of the filter based on Euler schemes.

For an integer $l\geq 1$, the Euler--Maruyama
approximation $(X^l_\tau)_{\tau\geq 0}$ of the
diffusion is defined by $X^l_0=X_0$ and, for any
$i\in\mathbb{N}$ and
$\tau\in[\frac{i-1}{l},\frac{i}{l}]$,
\begin{align*}
X^l_\tau=X^l_{\frac{i-1}{l}}
+\Big(\tau-\frac{i-1}{l}\Big)\,
b\big(X^l_{\frac{i-1}{l}}\big)
+\sigma\big(X^l_{\frac{i-1}{l}}\big)
\big(B_\tau-B_{\frac{i-1}{l}}\big).
\end{align*}
We denote by $\widehat{M}_l$ the Markov transition of
$X^l_t$; that is, one step of this chain consists of $l$ Euler
steps. For a fixed integer
$l^\star\geq 1$, we let
\begin{align}\label{eq:star_diff}
G^\star_t\coloneqq G_t
\qquad\text{and}\qquad
M^\star_{t+1}\coloneqq\widehat{M}_{l^\star},
\qquad t\geq 1.
\end{align}
One can easly sample from $\widehat{M}_{l^\star}(x,\dd x')$, and therefore
  $\eta^\star_t$
and $\widehat{\eta}^\star_t$ can be
approximated by a bootstrap particle filter.

\subsubsection{Local perturbation analysis}

Under \ref{diff1}--\ref{diff3}, the transition
$\widehat{M}_l$ has a density $\widehat{m}_l$ with
respect to $\lambda_d$ (see the remark at the end of
this section), and there exists a constant
$c\in(0,1]$ such that, for any $l\geq 1$ and
$x,x'\in\R^d$,
\begin{align*}
\big|m(x,x')-\widehat{m}_l(x,x')\big|
\leq\frac{1}{c\,l}\,\exp\big(-c\,\|x-x'\|^2\big);
\end{align*}
see \citet[][Theorem~6]{guyon2006euler} and
\citet{bally1996law}. Since
$\|\mu-\nu\|_{\mathrm{tv}}
=\frac{1}{2}\int|p_\mu-p_\nu|\,\dd\lambda_d$ for
probability measures $\mu$ and $\nu$ with densities
$p_\mu$ and $p_\nu$, integration gives
\begin{align}\label{eq:diff_l}
\sup_{x\in\R^d}
\big\|\delta_xM-\delta_x\widehat{M}_l\big\|_{\mathrm{tv}}
\leq\frac{C_{\mathrm{E}}}{l},
\qquad
C_{\mathrm{E}}\coloneqq\frac{\pi^{d/2}}{2\,c^{1+d/2}},
\qquad\forall l\geq 1.
\end{align}

\begin{lemma}\label{lemma:loc_diff}
Let $G_t$, $M_{t+1}$,
$G^\star_t$ and
$M^\star_{t+1}$ be as defined in
\eqref{eq:G_diff} and \eqref{eq:star_diff}, and
assume that $b$ and $\sigma$ satisfy
\ref{diff1}--\ref{diff3}. Then
$\Delta_{Q|Q^\star}\leq C_{\mathrm{E}}/l^\star$, with
$C_{\mathrm{E}}$ as in \eqref{eq:diff_l}.
\end{lemma}

\begin{proof}
Since $G^\star_t=G_t\leq 1$, we have
$\Delta_{Q|Q^\star}\leq\sup_{x\in\R^d}
\|\delta_xM-\delta_x\widehat{M}_{l^\star}\|_{\mathrm{tv}}$,
and the result follows from \eqref{eq:diff_l}.
\end{proof}

\subsubsection{Time-uniform error bounds}

\begin{proposition}\label{prop:compare_diff}
Let $G_t$ and $M_{t+1}$ be as
defined in \eqref{eq:G_diff}, and let
$\eta_1\in\mathcal{P}(\R^d)$ have a positive density
with respect to $\lambda_d$. Assume that the functions $b$ and
$\sigma$ satisfy \ref{diff1}--\ref{diff3} and that
$\sup_{t\geq 1}\|y_{t+1}-y_t\|<\infty$. Then, there exists a constant
$C_{\eta_1}\in[1,\infty)$ such that, for any integer $l^\star\geq 1$, if
$G^\star_t$ and
$M^\star_{t+1}$ are as defined in
\eqref{eq:star_diff}, then
\begin{align*}
\sup_{t\geq 1}\|\eta^\star_t-\eta_t\|_{\mathrm{tv}}
\vee
\sup_{t\geq 1}
\|\widehat{\eta}^\star_t-\widehat{\eta}_t\|_{\mathrm{tv}}
\leq C_{\eta_1}\,\vartheta(1/l^\star)
\end{align*}
and, if $l^\star\geq C_{\eta_1}$, then, for any
$t\geq 1$,
\begin{align*}
\Big\|\int_{\R^d}x\,\widehat{\eta}^\star_t(\dd x)
-\int_{\R^d}x\,\widehat{\eta}_t(\dd x)\Big\|
\leq C_{\eta_1}\,\vartheta(1/l^\star).
\end{align*}
\end{proposition}

\begin{proof}
By Section~\ref{sub:SDE}, \ref{A3}, and hence
\ref{A2}, holds for the model \eqref{eq:G_diff} and some sets $D_t$, and under the assumption imposed on $\eta_1$  we necessarily have $\eta_1(D_1)>0$. The first bound then
follows from Theorem~\ref{thm:compare},
Lemma~\ref{lemma:loc_diff} and
\eqref{eq:vartheta_scale}; the constant does not depend on $l^\star$, since the
constant of Theorem~\ref{thm:compare} does not depend
on the approximating model. The second estimate is
checked as in the proof of
Proposition~\ref{prop:extended}, using the fact that $G^\star_t=G_t$ and that $H_t(x)=x$ in
\eqref{eq:G_diff}.
\end{proof}

\begin{remark}\label{rem:Diff}
Under \ref{diff1}--\ref{diff3}, for any $l\geq 1$ the Markov transition
$\widehat{M}_{l}$ has a density $m_l$ with
respect to $\lambda_d$
\citep[see e.g.][Theorem~6]{guyon2006euler}. By
\citet[][Theorem~2.1]{menozzi2010some},
\eqref{eq:densSDE} holds with $m$ replaced by
$m_l$, for some constants
$(c_1,c_2)\in(0,1)^2$ independent of $l$, and by
\citet[][Theorem~6]{guyon2006euler},
\eqref{eq:derive_sde} holds with $m$ replaced by
$m_l$, for some constant $c\in(0,1)$ independent of $l$. Thus, following the discussion in Section~\ref{sub:SDE}, when we have
\[
\sup_{t\geq 1}\|y_{t+1}-y_t\|<\infty,
\]
 then
\ref{A3} and \ref{B1}   hold for the approximating
model $(G^\star_t,M^\star_{t+1})$. This property is used in Section~\ref{sec:PF}.
\end{remark}

\subsection{Nudged state-space models\label{sub:nudging}}

\subsubsection{Description of the models}

Let $r\geq d\geq 1$ be integers, let
$\Sigma_G\in\mathcal{S}_r^+$ and
$\Sigma_M\in\mathcal{S}_d^+$, and let
$H\colon\R^d\to\R^r$ and $F\colon\R^d\to\R^d$ be
continuous functions. We consider the state-space
model \eqref{eq:SSM} with  
\begin{align}\label{eq:SSM_nudge}
H_t=H,\quad F_{t+1}=F,\quad
\Sigma_{G_t}=\Sigma_G,\quad
\Sigma_{M_{t+1}}=\Sigma_M,\quad
U_t\iid\mathcal{N}_r(0,I_r),\quad
V_{t+1}\iid\mathcal{N}_d(0,I_d),
\quad t\geq 1 
\end{align}
and with observations $y_t$ in $\R^r$. To handle model
misspecification, one may modify the model so that the
observed data become more likely, and, to do so, below we consider the
nudging approach of
\citet{gonzalez2025nudging}. Assume that $H$ is
continuously differentiable, denote by
$J_H(x)\in\R^{r\times d}$ its Jacobian matrix at
$x\in\R^d$, and let, for any $t\geq 1$ and
$x\in\R^d$,
\begin{align*}
v_t(x)\coloneqq J_H(x)^\top\Sigma_G^{-1}\big(y_t-H(x)\big)
=-\frac{1}{2}\nabla_x
\Big[\big(y_t-H(x)\big)^\top\Sigma_G^{-1}\big(y_t-H(x)\big)\Big],
\end{align*}
so that $x+\chi\,v_t(x)$, $\chi\geq 0$, is a gradient
step towards states that better explain $y_t$. For
$\chi\in[0,\infty)$, let $H_{1,\chi}\coloneqq H$ and,
for any $t\geq 1$ and $x\in\R^d$, let
\begin{align*}
H_{t+1,\chi}(x)\coloneqq H\big(x+\chi\,v_{t+1}(x)\big)
\qquad\text{and}\qquad
F_{t+1,\chi}(x)\coloneqq F\big(x+\chi\,v_t(x)\big).
\end{align*}
The nudged state-space model of \citet{gonzalez2025nudging} is
the state-space model \eqref{eq:SSM} in which $H_t$
and $F_{t+1}$ are replaced by $H_{t,\chi}$ and
$F_{t+1,\chi}$. The functions $v_t$ determine the direction
of the correction and $\chi$ its magnitude. The parameter $\chi$ is chosen according to the degree of misspecification; and it is expected that a small value for $\chi$ leads to a nudged model close to the original one; hence, there is a  need to control the dependence of the nudged model on $\chi$.

For $\chi\in[0,\infty)$ and $t\geq 1$, let
$G_{t,\chi}\colon\R^d\to(0,1]$ and $M_{t+1,\chi}$ be
defined by
\begin{align*}
G_{t,\chi}(x)\coloneqq
\exp\Big(-\frac{1}{2}\big(y_t-H_{t,\chi}(x)\big)^\top
\Sigma_G^{-1}\big(y_t-H_{t,\chi}(x)\big)\Big)
\qquad\text{and}\qquad
\delta_xM_{t+1,\chi}\coloneqq
\mathcal{N}_d\big(F_{t+1,\chi}(x),\Sigma_M\big).
\end{align*}
For $\chi\in[0,\infty)$, we denote by $\eta_{t,\chi}$ and
$\widehat{\eta}_{t,\chi}$ the prediction and updated
models associated with the potential functions
$G_{t,\chi}$ and the Markov transitions
$M_{t+1,\chi}$, with $\eta_{1,\chi}=\eta_1$. By
Section~\ref{sub:motive_example}, the flow
$\eta_{t+1,\chi}$ (respectively,
$\widehat{\eta}_{t,\chi}$) are the predictive (respectively, filtering) distributions under the nudged state-space model
with parameter $\chi$; the case $\chi=0$ is the original
model. To compare the nudged and the original models,
with the notation of this section, we fix
$\chi^\star\in[0,\infty)$ and let, for any $t\geq 1$,
\begin{align}\label{eq:Model_nudge}
G_t\coloneqq G_{t,0},\qquad
M_{t+1}\coloneqq M_{t+1,0},\qquad
G^\star_t\coloneqq G_{t,\chi^\star},\qquad
M^\star_{t+1}\coloneqq M_{t+1,\chi^\star},
\end{align}
so that $\eta_t=\eta_{t,0}$,
$\widehat{\eta}_t=\widehat{\eta}_{t,0}$,
$\eta^\star_t=\eta_{t,\chi^\star}$ and
$\widehat{\eta}^\star_t=\widehat{\eta}_{t,\chi^\star}$.

\subsubsection{Local perturbation analysis}

We assume that there exist some constants
$L\in[1,\infty)$ and $\alpha\in(0,1]$ such that
\begin{align}\label{eq:assume_nudge}
\|H(x)-H(x')\|\leq L\|x-x'\|
\qquad\text{and}\qquad
\|F(x)-F(x')\|\leq L\|x-x'\|^\alpha,
\qquad\forall x,x'\in\R^d.
\end{align}
Increasing $L$ if necessary, we also assume that
$\sigma_{\max}(\Sigma_G^{-1})\leq L$ and
$\sup_{x\in\R^d}\|J_H(x)^\top\Sigma_G^{-1}\|\leq L$;
the latter assumption is possible since $\|J_H(x)\|\leq L$ by
\eqref{eq:assume_nudge}. We let
$l_G\coloneqq\sigma_{\min}(\Sigma_G^{-1})$ and,
for any $t\geq 1$ and $x\in\R^d$,
$e_t(x)\coloneqq\|y_t-H(x)\|$, so that
\begin{align}\label{eq:bound_vt}
\|v_t(x)\|\leq L\,e_t(x)
\qquad\text{and}\qquad
G_t(x)\leq\exp\Big(-\frac{l_G}{2}\,e_t(x)^2\Big).
\end{align}

Let $t\geq 1$ and $x\in\R^d$ be fixed. Then, since
\begin{align*}
\delta_xQ^\star_{t+1}-\delta_xQ_{t+1}
=(G^\star_t(x)-G_t(x))\,\delta_xM^\star_{t+1}
+G_t(x)\,(\delta_xM^\star_{t+1}-\delta_xM_{t+1}),
\end{align*}
we have
\begin{align}\label{eq:decomp}
\big\|\delta_xQ^\star_{t+1}-\delta_xQ_{t+1}\big\|_{\mathrm{tv}}
\leq\big|G^\star_t(x)-G_t(x)\big|
+G_t(x)\,
\big\|\delta_xM^\star_{t+1}-\delta_xM_{t+1}\big\|_{\mathrm{tv}}.
\end{align}

\emph{First term of \eqref{eq:decomp}.} The first term
vanishes for $t=1$. For $t\geq 2$, let
$H^\star_t(x)\coloneqq H(x+\chi^\star v_t(x))$. Then, writing
$y_t-H^\star_t(x)=(y_t-H(x))+(H(x)-H^\star_t(x))$, we have
\begin{equation}\label{eq:n1}
\begin{split}
\big(y_t-H^\star_t(x)\big)^\top\Sigma_G^{-1}
\big(y_t-H^\star_t(x)\big)
&=\big(y_t-H(x)\big)^\top\Sigma_G^{-1}\big(y_t-H(x)\big)
+R_t(x),\\
R_t(x)&\coloneqq
\big(H(x)-H^\star_t(x)\big)^\top\Sigma_G^{-1}
\big(H(x)-H^\star_t(x)\big)\\
&\qquad\hskip2cm+2\big(H(x)-H^\star_t(x)\big)^\top\Sigma_G^{-1}
\big(y_t-H(x)\big).
\end{split}
\end{equation}
By \eqref{eq:assume_nudge} and \eqref{eq:bound_vt},
$\|H(x)-H^\star_t(x)\|\leq L\chi^\star\|v_t(x)\|
\leq L^2\chi^\star e_t(x)$, so that
\begin{equation}\label{eq:n2}
|R_t(x)|
\leq L\big(L^2\chi^\star e_t(x)\big)^2
+2L\,\big(L^2\chi^\star e_t(x)\big)\,e_t(x)
\leq 2\,C_\chi\,\chi^\star\,e_t(x)^2,
\qquad
C_\chi\coloneqq L^3\big(L^2\chi^\star+1\big).
\end{equation}
Since $G^\star_t(x)/G_t(x)=e^{-R_t(x)/2}$ and
$|1-e^z|\leq|z|e^{|z|}$ for any $z\in\R$,
\eqref{eq:bound_vt} and \eqref{eq:n2} give
\begin{equation}\label{eq:n3}
\begin{split}
\big|G_t(x)-G^\star_t(x)\big|
&=G_t(x)\,\big|1-e^{-R_t(x)/2}\big|
\leq C_\chi\,\chi^\star\,e_t(x)^2\,
e^{C_\chi\chi^\star e_t(x)^2}\,G_t(x)\\
&\leq C_\chi\,\chi^\star\,e_t(x)^2\,
\exp\Big(-\frac{1}{2}\big(\lambda_G-2C_\chi\chi^\star\big)
e_t(x)^2\Big).
\end{split}
\end{equation}
Let $\bar{\chi}>0$ be the unique solution of
$4\bar{\chi}L^3(L^2\bar{\chi}+1)=l_G$, and assume
that $\chi^\star\leq\bar{\chi}$. Then
$C_\chi\leq l_G/(4\bar{\chi})$ and
$2C_\chi\chi^\star\leq l_G/2$, so that, by
\eqref{eq:n3},
\begin{align}\label{eq:Bound_nG}
\big|G_t(x)-G^\star_t(x)\big|
\leq\chi^\star\,\frac{l_G}{4\bar{\chi}}
\sup_{u\geq 0}u\,e^{-l_Gu/4}
=C_{\mathrm{N},1}\,\chi^\star,
\qquad
C_{\mathrm{N},1}\coloneqq\frac{1}{e\,\bar{\chi}}.
\end{align}

\emph{Second term of \eqref{eq:decomp}.} By
\eqref{eq:TV_gauss}, \eqref{eq:assume_nudge} and
\eqref{eq:bound_vt}, with
$C_{\mathcal{N}}=(2\pi\,\sigma_{\min}(\Sigma_M))^{-1/2}$,
\begin{align*}
\big\|\delta_xM^\star_{t+1}-\delta_xM_{t+1}\big\|_{\mathrm{tv}}
\leq C_{\mathcal{N}}\,\big\|F\big(x+\chi^\star v_t(x)\big)-F(x)\big\|
\leq L\,C_{\mathcal{N}}\,\big(\chi^\star\|v_t(x)\|\big)^\alpha
\leq L^{1+\alpha}C_{\mathcal{N}}\,\big(\chi^\star e_t(x)\big)^\alpha.
\end{align*}
Hence, using \eqref{eq:bound_vt} once more,
\begin{align}\label{eq:Bound_nM}
G_t(x)\,
\big\|\delta_xM^\star_{t+1}-\delta_xM_{t+1}\big\|_{\mathrm{tv}}
\leq L^{1+\alpha}C_{\mathcal{N}}\,(\chi^\star)^\alpha
\sup_{z\geq 0}z^\alpha e^{-\lambda_Gz^2/2}
=C_{\mathrm{N},2}\,(\chi^\star)^\alpha,
\end{align}
with
$C_{\mathrm{N},2}\coloneqq L^{1+\alpha}C_{\mathcal{N}}
(\alpha/(e l_G))^{\alpha/2}$.

By using the above calculations, we obtain the following result.
\begin{lemma}\label{lemma:loc_nudge}
Let $G_t$, $M_{t+1}$,
$G^\star_t$ and
$M^\star_{t+1}$ be as defined in
\eqref{eq:Model_nudge}, and assume that $H$ is
continuously differentiable and that
\eqref{eq:assume_nudge} holds for some constants
$L\in[1,\infty)$ and $\alpha\in(0,1]$. The, there exists a constant
$C_{\mathrm{N}}\in(0,\infty)$, that only depends on $(L,\alpha,\Sigma_G,\Sigma_M)$,
such that for any $\chi^\star\in[0,\infty)$ we have
\[
\Delta_{Q|Q^\star}\leq C_{\mathrm{N}}\,(\chi^\star)^\alpha.
\]

\end{lemma}

\begin{proof}
If $\chi^\star\leq\bar{\chi}$, then
\eqref{eq:decomp}, \eqref{eq:Bound_nG} and
\eqref{eq:Bound_nM} give
\[
\Delta_{Q|Q^\star}\leq C_{\mathrm{N},1}\chi^\star
+C_{\mathrm{N},2}(\chi^\star)^\alpha\leq
(C_{\mathrm{N},1}\bar{\chi}^{1-\alpha}+C_{\mathrm{N},2})(\chi^\star)^\alpha.
\]
If $\chi^\star>\bar{\chi}$, then
$\Delta_{Q|Q^\star}\leq 1<
(\chi^\star/\bar{\chi})^\alpha$. The lemma holds with
\[
C_{\mathrm{N}}\coloneqq\max\{C_{\mathrm{N},1}\bar{\chi}^{1-\alpha}
+C_{\mathrm{N},2},\,\bar{\chi}^{-\alpha}\},
\]
 and $\bar{\chi}$
depends only on $L$ and $\Sigma_G$.
\end{proof}

\subsubsection{Time-uniform error bounds}

\begin{proposition}\label{prop:compare_nudge}
Let $\eta_1\in\mathcal{P}(\R^d)$ have a positive density
with respect to $\lambda_d$. Assume that $H$ is
continuously differentiable and $L$-bi-Lipschitz for
some $L\in[1,\infty)$, that \eqref{eq:assume_nudge}
holds for some $\alpha\in(0,1]$, that
$y_t\in H(\R^d)$ for any $t\geq 1$, and that
\begin{equation}\label{cond-H-1}
\sup_{t\geq 1}\|H^{-1}(y_{t+1})-F(H^{-1}(y_t))\|<\infty.
\end{equation}
Then, there exists a constant $C_{\eta_1}\in[1,\infty)$ such
that, for any $\chi\in[0,\infty)$,
\begin{align*}
\sup_{t\geq 1}\|\eta_{t,\chi}-\eta_{t,0}\|_{\mathrm{tv}}
\vee
\sup_{t\geq 1}
\|\widehat{\eta}_{t,\chi}-\widehat{\eta}_{t,0}\|_{\mathrm{tv}}
\leq C_{\eta_1}\,\vartheta\big(\chi^\alpha\big).
\end{align*}
\end{proposition}

\begin{proof}
The model $(G_t,M_{t+1})$ is of the form
\eqref{eq:Gt_def}--\eqref{eq:Mt_def2}, and satisfies
\ref{Model1} (with $\zeta_G(z)=\zeta_M(z)=e^{-z/2}$
and a constant $L'\geq L$ depending only on
$(d,L,\Sigma_G,\Sigma_M)$) as well as \ref{Obs}. Since
$\|F(x_1)-F(x_2)\|\leq L\|x_1-x_2\|^\alpha$, part~\ref{i-prop:example}
of Proposition~\ref{prop:example} applies with
$\varphi(z)=Lz^\alpha$, so that \ref{A3}, and hence
\ref{A2}, holds for some sets $D_t$. Under the assumption imposed on $\eta_1$, we necessarily have 
$\eta_1(D_1)>0$. Then, the result of the proposition   follows from
Theorem~\ref{thm:compare} and
Lemma~\ref{lemma:loc_nudge}, applied with
$\chi^\star=\chi$, and from \eqref{eq:vartheta_scale};
the constant does not depend on $\chi$, since the
constant of Theorem~\ref{thm:compare} does not depend
on the approximating model.
\end{proof}

\begin{remark}
By Lemma~\ref{lemma:Obs2}, bounded observations
satisfy the condition (\ref{cond-H-1})
of Proposition~\ref{prop:compare_nudge}.
\end{remark}

 Proposition \ref{prop:compare_nudge} shows that the nudged filter converges to the original one as
$\chi\to 0$, at rate $\vartheta(\chi^\alpha)$,
uniformly in time. Observe that, since $G_{t,\chi}\neq G_{t,0}$ when
$\chi>0$, this relies on
the estimate \eqref{eq:one_step_Psi} for the updated
model.

\section{Error bounds for particle filters\label{sec:PF}}

In this section, we set
$\setX=\R^d\cap[\underline{x},\bar{x}]$ for some
hyper-rectangle
$[\underline{x},\bar{x}]\subseteq\bar{\R}^d$ (e.g.~$\setX=\R^d$), and we
fix some $\eta_1\in\mathcal{P}(\setX)$. We also assume that $G_t(x)>0$ for any $t\geq 1$ and
$x\in\setX$, so that \ref{A0} holds with $W_t=\setX$ and
$\Psi_t(\mu)$ and $\Phi_{t+1}(\mu)$
are well defined for every $\mu\in\mathcal{P}(\setX)$. As in
Section~\ref{sec:main}, we let
$\eta_{t+1}=\Phi_{t+1}(\eta_t)$ and
$\widehat{\eta}_t=\Psi_t(\eta_t)$ for any $t\geq 1$.

\subsection{Description of the algorithm\label{sub:algoPF}}

Let $N\geq 1$ be some number of particles. The
(bootstrap) particle filter is defined by the Markov
chain
$X^N_t\coloneqq(X^{N,1}_t,\dots,X^{N,N}_t)$,
$t\geq 1$, taking its values in $\setX^N$ and defined as follows. First,
$X^{N,1}_1,\dots,X^{N,N}_1\iid\eta_1$. Then, for any $t\geq 1$, given $X^N_t$, the particles
$X^{N,1}_{t+1},\dots,X^{N,N}_{t+1}$ are independent
with common distribution
\begin{align*}
\Phi_{t+1}(\eta^N_t)
=\sum_{n=1}^N
\frac{G_t(X^{N,n}_t)}{\sum_{m=1}^N G_t(X^{N,m}_t)}\,
M_{t+1}\big(X^{N,n}_t,\cdot\big),
\qquad\text{where}\quad
\eta^N_t\coloneqq\frac{1}{N}\sum_{n=1}^N\delta_{X^{N,n}_t}.
\end{align*}

Each new particle is therefore obtained by
selecting one of the current particles with a
probability proportional to its potential value
(selection step, a.k.a.\ multinomial resampling), and
then by moving it according to the Markov transition
$M_{t+1}$ (mutation step). The random probability measures $\eta^N_t$ and
$\widehat{\eta}^N_t\coloneqq\Psi_t(\eta^N_t)$ are the particle approximations of the prediction
distribution $\eta_t$ and of the updated distribution
$\widehat{\eta}_t$. Since $G_t>0$, \ref{A0} holds with
$W_t=\setX$, so that
$\eta^N_t\in\mathcal{P}_{W_t}(\setX)=\mathcal{P}(\setX)$ for any
$t\geq 1$. Thus, the results of Section~\ref{sec:main} can be applied to
the sequence $\eta^N_t$, for any
realization of the algorithm.

The above algorithm is known under a variety of names. In
signal processing and Bayesian statistics, it is called
particle filter \citep{delmoral1996nonlinear,
del1997nonlinear}, bootstrap filter
\citep{gordon1993novel}, Monte Carlo filter
\citep{kitagawa1996monte} or sequential Monte Carlo
method; in applied probability, interacting particle
system, genetic algorithm with mutation and selection
transitions, or population Monte Carlo
\citep{delmoral1996nonlinear,DelMoral:book,del2013mean,
delmoral2012concentration}.
In computational physics, Monte Carlo methods based on
the Feynman--Kac formula go back to
\citet{metropolis1949monte}, \citet{donsker1950sampling}
and \citet{kalos1962monte}, and are known as diffusion
or quantum Monte Carlo methods, with variants such as
fixed-node \citep{reynolds1982fixed}, pure diffusion
\citep{caffarel1986treatment,caffarel1988development}
and reptation \citep{baroni1999reptation} Monte Carlo.
The version with a constant number of walkers, which
coincides with the above algorithm, is called stochastic
reconfiguration Monte Carlo
\citep{hetherington1984observations,sorella1998green,
sorella2000green,assaraf2000diffusion}. We refer to
\citet{foulkes2001quantum} and
\citet{mareschal2021early} for reviews, and to
\citet{delmoral2003particle},
\citet{delmoral2004particle},
\citet{rousset2006control},
\citet{cances2006quantum},
\citet{elmakrini2007diffusion} and
\citet{caffarel2025mathematical} for mathematical
analyses in this context.

To simplify  the presentation, we restrict our attention to the bootstrap particle filter. However, as explained in Section \ref{sub:PF_extension}, the results presented below apply to a broader class of algorithms, including particle filters based on some low variance resampling schemes and mean filed particle models.

\subsection{Local perturbations\label{sub:local_PF}}

For any $t\geq 1$, we let $\F^N_t$ be the
$\sigma$-algebra generated by $X^N_1,\dots,X^N_t$, and
we let $\F^N_0\coloneqq\{\emptyset,\Omega\}$ and
$\eta^N_0\coloneqq\eta_1$. Since $\Phi_1$ is the
identity mapping on $\mathcal{P}(\setX)$, the local
perturbations \eqref{eq:delta_s} of the sequence
$\mu_t=\eta^N_t$ are the random signed measures
\[
\delta^N_t\coloneqq\delta_t\big(\eta^N_{t-1},\eta^N_t\big)
=\eta^N_t-\Phi_t(\eta^N_{t-1}),
\qquad t\geq 1,
\]
and, given $\F^N_{t-1}$, the measure $\eta^N_t$ is the empirical
measure of $N$ independent draws from
$\Phi_t(\eta^N_{t-1})$.

This local error cannot be measured with the total
variation norm. Indeed, $\eta^N_t$ has a finite support,
while $\Phi_t(\eta^N_{t-1})$ typically has a density
with respect to the Lebesgue measure; in this case,
taking the support of $\eta^N_t$ as the set $A$ in the
definition of $\|\cdot\|_{\mathrm{tv}}$ shows that
$\|\delta^N_t\|_{\mathrm{tv}}=1$
almost surely. On the other hand, the Kolmogorov distance between an
empirical distribution and its sampling distribution
can be estimated using
Dvoretzky--Kiefer--Wolfowitz-type inequalities.

Specifically, below we will use the following Dvoretzky--Kiefer--Wolfowitz-type inequality.
\begin{lemma}\label{lemma:DKW}
There exists a constant $c_d>0$ such that, for any integer $N\geq 1$ and any collection of probability measures $\{\mu_n\}_{n=1}^N\subset\mathcal{P}(\R^d)$, we have, for all $\kappa>0$,
\begin{align*}
\P\Big(\Big\|\frac{1}{N}\sum_{n=1}^N\delta_{X_n}-\frac{1}{N}\sum_{n=1}^N\mu_n\Big\|_{\mathrm{K}}>\kappa\Big)\leq 2 e^{-c_d N\kappa^2} 
\end{align*}
 where the random variables $\{X_n\}_{n=1}^N$ are independent and such that $X_n\sim\mu_n$ for all $n\in\{1,\dots,N\}$.
\end{lemma}
\begin{proof}
Recall that the Orlicz norm of a non-negative
random variable $Y$ associated with the function $\psi(u)=2^{u^2}-1$ is
$\pi_\psi(Y)=\inf\{a>0:\E[e^{Y^2/a^2}]\leq 2\}$. If
$\pi_\psi(Y)\leq b$ for some $b>0$, then $\E[e^{Y^2/c^2}]\leq 2$ for any
$c>b$, and, by letting $c\downarrow b$ and using the monotone convergence theorem, it follows that  $\E[e^{Y^2/b^2}]\leq 2$. Therefore, by Markov's inequality
\begin{align}\label{eq:orlicz}
\P\Big(Y\geq b\sqrt{x+\log 2}\Big)
=\P\big(e^{Y^2/b^2}\geq 2e^{x}\big)\leq e^{-x},\quad\forall x\geq 0.
\end{align}

We now let $\mathcal{C}_d=\big\{\ind_{(-\infty,x]}:\,x\in\R^d\}$ and $\delta^N=\frac{1}{N}\sum_{n=1}^N\delta_{X_n}-\frac{1}{N}\sum_{n=1}^N\mu_n$. By \citet[][Theorem~11.2.6]{del2013mean}, the Orlicz norm   of the random variable 
$\sqrt{N}\sup_{f\in\mathcal{C}_d}
|\delta^N(f)|=\sqrt{N}\|\delta^N\|_{\mathrm{K}}$  associated with the function $\psi(u)=2^{u^2}-1$ is bounded by
\begin{align*}
C_d\coloneqq\max\Big\{1,\,12^2\int_0^2
\sqrt{\log\big(8+\mathcal{N}(\mathcal{C}_d,\epsilon)^2\big)}
\,\dd\epsilon\Big\}
\end{align*}
where $\mathcal{N}(\mathcal{C}_d,\epsilon)$ is the supremum over $\mu\in\mathcal{P}(\R^d)$ of the minimal
number of  $\mathsf{L}^2(\mu)$-balls of radius
$\epsilon$  needed to cover $\mathcal{C}_d$. By applying \eqref{eq:orlicz} with  $Y \coloneqq \sqrt{N}\|\delta^N\|_{\mathrm{K}}$ and $b:=C_d$, we obtain that
\begin{align*}
\P\Big(\|\delta^N\|_{\mathrm{K}}\geq N^{-1/2} C_d\sqrt{x+\log(2)}\Big)\leq e^{-x},\,\,\forall x\geq 0\,\,\Leftrightarrow\,\, \P\Big(\|\delta^N\|_{\mathrm{K}}\geq\kappa\Big)\leq 2e^{-\kappa^2 N/C_d^2},\quad\forall\kappa>0
\end{align*}
and to conclude the proof of the lemma it remains to show that $C_d<\infty$.

By  \citet[][Example~2.6.1 and
Theorem~2.6.4]{vandervaart1996weak}, there exists a universal constant $c>0$  such that $\mathcal{N}(\mathcal{C}_d,\epsilon)\leq c (d+1)(a e)^{d+1}\epsilon^{-2d}$, and it follows from \citet[][Remark  11.2.7 and (11.15)]{del2013mean}, that $C_d\leq c' d^{1/2}$ for some universal constant $c'>0$. This ends the proof of the lemma.
\end{proof}

In the special case where $d=1$ and the random variables $X_n$ are identically distributed, the result of Lemma \ref{lemma:DKW} holds with $c_d=2$ by \citet[][Corollary~1 and Comment~2]{massart1990tight}. We also note that  \citet[][Lemma~4.1]{naaman2021tight} corresponds to the result stated in Lemma \ref{lemma:DKW}, with $c_d=2$ and with the multiplicative constant 2 replaced by $d\log(N+1)$.

By combining Lemma \ref{lemma:DKW} with Hoeffding's
inequality for a given test function, we check that the local $\delta_t^N$
errors are uniformly small on the time horizon $T$ with
high probability. Specifically, for any $q\in(0,1]$ and integers $T\geq 1$ and $N\geq 1$, we let
\begin{align}\label{eq:delta_pT}
\varepsilon_{N,T,q}\coloneqq
\sqrt{\frac{1+\log(T/q)}{N}}
\end{align}
and prove the following result.
\begin{lemma}\label{lemma:thmPF}
There exists a constant $C_d\in[1,\infty)$, depending
only on $d$, such that, for any $q\in(0,1]$, any
integers $T\geq 1$ and  $N\geq 1$, and any (deterministic) functions
$f_1,\dots,f_T\in\mathcal{B}_1(\setX)$,
\begin{align*}
\P\Big(\max_{1\leq t\leq T}
\big|\delta^N_t(f_t)\big|
\vee
\max_{1\leq t\leq T}
\big\|\delta^N_t\big\|_{\mathrm{K}}
\leq C_d\,\varepsilon_{N,T,q}\Big)\geq 1-q.
\end{align*}
\end{lemma}

\begin{proof}
Let $c'_d\coloneqq\min\{2,c_d\}$, with $c_d>0$  as in Lemma~\ref{lemma:DKW}, and let $C_d\coloneqq\max\{1,(\log(4)/c'_d)^{1/2}\}$  and $\kappa_{N,T,q}\coloneqq C_d\,\varepsilon_{N,T,q}$.

For any $t\geq 1$, given $\mathcal{F}^N_{t-1}$, the measure $\eta^N_t$ is the empirical measure of $N$
independent draws from $\Phi_t(\eta^N_{t-1})$. Hence, by Hoeffding's inequality (recall that $f_t$ takes
its values in $[0,1]$) and by Lemma~\ref{lemma:DKW}, we have
\[
\P\big(|\delta^N_t(f_t)|>\kappa\,\big|\,\mathcal{F}^N_{t-1}\big)\leq 2e^{-2N\kappa^2}
\quad\text{and}\quad
\P\big(\|\delta^N_t\|_{\mathrm{K}}>\kappa\,\big|\,\mathcal{F}^N_{t-1}\big)\leq 2e^{-c_dN\kappa^2},\quad\forall \kappa>0,\quad\forall t\geq 1.
\]
By taking expectations and by
using a union bound over $t\in\{1,\dots,T\}$ and over the two events, it follows that
\begin{align}\label{eq:kappa_in}
\P\Big(\max_{1\leq t\leq T}
\big|\delta^N_t(f_t)\big|
\vee
\max_{1\leq t\leq T}
\big\|\delta^N_t\big\|_{\mathrm{K}}
>\kappa\Big)\leq 4T e^{-c'_dN\kappa^2},\quad\forall \kappa>0.
\end{align}
The result of the lemma then follows by applying \eqref{eq:kappa_in} with $\kappa=\kappa_{N,T,q}$, and noting that
\[
4Te^{-c'_d N\kappa_{N,T,q}^2}=4Te^{-c'_dC_d^2(1+\log(T/q))}\leq 4T e^{-\log(4)}\,(q/T)^{ c'_d C_d^2}\leq q,
\]
where the first inequality holds since  $c'_d C_d^2\geq\log(4)$ while the second ienquality holds since $c'_d C_d^2\geq 1$ and $q/T\leq 1$.
\end{proof}

\subsection{High-probability trajectory bounds}

\begin{theorem}\label{thm:PF}
Assume that \ref{A3} and \ref{B1} hold and that
$\eta_1(D_1)>0$, with the sets $D_t$ as in
\ref{A3}. Then, there exist constants
$C_{\eta_1}\in(0,\infty)$ and $c_{\eta_1}\in(0,1]$ such
that, for any $q\in(0,1]$ and any integers
$T\geq 1$ and $N\geq 1$, and with the function $\vartheta$ as defined  in \eqref{eq:vartheta}, the following statements hold.
\begin{enumerate}[label=(\roman*)]
\item\label{PF_i} With probability at least $1-q$,
\begin{align*}
\max_{1\leq t\leq T}\|\eta^N_t-\eta_t\|_{\mathrm{K}}
\vee
\max_{1\leq t\leq T}
\|\widehat{\eta}^N_t-\widehat{\eta}_t\|_{\mathrm{K}}
\leq C_{\eta_1}\,\vartheta\big(\varepsilon_{N,T,q}\big).
\end{align*}
\item\label{PF_ii} For any (deterministic) functions
$f_1,\dots,f_T\in\mathcal{B}_1(\setX)$, with
probability at least $1-q$,
\begin{align*}
\max_{1\leq t\leq T}\big|\eta^N_t(f_t)-\eta_t(f_t)\big|
\vee
\max_{1\leq t\leq T}
\big|\widehat{\eta}^N_t(f_t)-\widehat{\eta}_t(f_t)\big|
\leq C_{\eta_1}\,\vartheta\big(\varepsilon_{N,T,q}\big).
\end{align*}
\item\label{PF_iii} If
$\varepsilon_{N,T,q}\leq c_{\eta_1}$, then, for any
(deterministic) measurable functions
$f_1,\dots,f_T\colon\setX\to\R$ for which there exist
constants $a_1,\dots,a_T\in\R$ such that
$(f_t+a_t)G_t\in\mathcal{B}_1(\setX)$ for any
$t\in\{1,\dots,T\}$, with probability at least $1-q$,
\begin{align*}
\max_{1\leq t\leq T}
\big|\widehat{\eta}^N_t(f_t)-\widehat{\eta}_t(f_t)\big|
\leq C_{\eta_1}\,\vartheta\big(\varepsilon_{N,T,q}\big).
\end{align*}
\end{enumerate}
\end{theorem}

\begin{proof}
Let $q\in(0,1]$ and $T,N\geq 1$, and write
$\varepsilon\coloneqq\varepsilon_{N,T,q}$. In what follows we use the following three facts. First,
since
$1+\log(2)+z\leq 2(1+z)$ for any $z\geq 0$, we have
$\varepsilon_{N,T,q/2}\leq\sqrt{2}\,\varepsilon$. Second, by \eqref{eq:vartheta_scale},
$\vartheta(cz)\leq c\,\vartheta(z)$ for any
$c\geq 1$ and $z\geq 0$. Third, $z\leq\vartheta(z)/\log(2)$ for any $z\in[0,1]$, and
$\vartheta(z)\geq\log( 2)$ for any $z\geq 1$; thus any quantity bounded by $C\varepsilon$ and by $1$ is bounded by
$(C\vee 1)\,\vartheta(\varepsilon)/\log(2)$.

To prove the theorem we apply Theorem~\ref{thm:main} and
Proposition~\ref{prop:filter} with
$\bm{\mu}=(\eta^N_t)_{t\geq 1}$, which is possible recalling that $\eta^N_t\in\mathcal{P}_{W_t}(\setX)$, and with
$\varpi(z)\equiv 1$. With the conventions of
Section~\ref{sub:local_PF}, $$\mu_1-\eta_1
=\eta^N_1-\Phi_1(\eta^N_0)=\eta^N_1-\eta_1$$ and,  by part~\ref{norm_ii} of Theorem~\ref{thm:norm},
\begin{align}\label{eq:PF_Delta}
\tilde{\Delta}_{T,\bm{\mu}|\eta_1,\varpi}\leq
\max_{1\leq t\leq T}
\big\|\delta^N_t\big\|_{\mathrm{K}}.
\end{align}

\emph{Part~\ref{PF_i}.} Let $E_T$ be the event of
Lemma~\ref{lemma:thmPF} (with $f_t\equiv 0$ and
probability parameter $q$), so that $\P(E_T)\geq 1-q$.
On $E_T$, by \eqref{eq:PF_Delta},
Theorem~\ref{thm:main} and part~\ref{norm_i} of
Theorem~\ref{thm:norm},
\[
\max_{t\leq T}\|\eta^N_t-\eta_t\|_{\mathrm{K}}\leq
2\max_{t\leq T}\|\eta^N_t-\eta_t\|_{\mathcal{V}_\setX}
\leq 2C_{\eta_1,\varpi}\,\vartheta(C_d\varepsilon),
\]
 and, by
part~\ref{i-prop:filter} of Proposition~\ref{prop:filter}, the same
bound holds for
$\max_{t\leq T}\|\widehat{\eta}^N_t-\widehat{\eta}_t\|_{\mathrm{K}}$
up to a multiplicative constant. The result stated in part~\ref{PF_i} of the theorem follows
follows.

\emph{Part~\ref{PF_ii}.} Let $E_T$ and $\tilde{E}_T$ be the
events of Lemma~\ref{lemma:thmPF} associated with the
functions $\{f_t\}_{t=1}^T\subset\mathcal{B}_1(\setX)$ and $\{G_tf_t\}_{t=1}^T$,
respectively (note that $G_tf_t\in\mathcal{B}_1(\setX)$),
both with probability parameter $q/2$, so that
$\P(E_T\cap \tilde{E}_T)\geq 1-q$. On $E_T\cap \tilde{E}_T$, all the
local errors in the statement of
Lemma~\ref{lemma:thmPF}, for both families of
functions, are bounded by $\sqrt{2}\,C_d\,\varepsilon$. By
the second inequality of Theorem~\ref{thm:main},
applied with $f=f_t$ and with $f=G_tf_t$, and by
\eqref{eq:PF_Delta}, for any $t\leq T$ and
$g_t\in\{f_t,G_tf_t\}$,
\begin{align*}
|\eta^N_t(g_t)-\eta_t(g_t)|\leq\sqrt{2}\,C_d\,\varepsilon
+C_{\eta_1,\varpi}\,\vartheta\big(\sqrt{2}\,C_d\,\varepsilon\big).
\end{align*}
This yields the estimate for the prediction measures
(using the three facts above), and, together with
part~\ref{ii-prop:filter} of Proposition~\ref{prop:filter}  and the
bound stated in part~\ref{PF_i} of the theorem on
$\|\eta^N_t-\eta_t\|_{\mathcal{V}_\setX}$, gives the bound
for the updated measures.

\emph{Part~\ref{PF_iii}.} Let $g_t\coloneqq(f_t+a_t)G_t$
and let $E_T$ and $\tilde{E}_T$ be as in the proof of
part~\ref{PF_ii}, with $G_tf_t$ replaced by $g_t$. On
$E_T\cap \tilde{E}_T$,
$\max_{t\leq T}\|\eta^N_t-\eta_t\|_{\mathcal{V}_\setX}
\leq C_{\eta_1,\varpi}\vartheta(\sqrt{2}C_d\varepsilon)$,
which is at most the constant $c_{\eta_1}$ of
Proposition~\ref{prop:filter} provided that $\varepsilon$ is
small enough, since $\vartheta$ is continuous and
$\vartheta(0)=0$. Part~\ref{K_unbounded} of
Proposition~\ref{prop:filter} and the bound on
$|\eta^N_t(g_t)-\eta_t(g_t)|$ obtained in the proof of
part~\ref{PF_ii} now yield the result, after decreasing $c_{\eta_1}$ and
increasing $C_{\eta_1}$ if necessary.
\end{proof}

\subsection{Extensions\label{sub:PF_extension}}

Above, we restricted our attention to the bootstrap particle filter, but the    results presented in Theorem ~\ref{thm:PF}  apply to a broader class of algorithms. Specially, they apply as soon as, for any $t\geq 1$, given $X^N_t$, the particles
$X^{N,n}_{t+1}$, $n\in\{1,\dots,N\}$, are independent and such that $\E[\eta_{t+1}^N(f)|\F^N_t]=\Phi_{t+1}(\eta_t^N)$. Indeed, in addition to the bounds derived in Section \ref{sec:main}, the proof of these two theorems  only use Lemma \ref{lemma:DKW} and Hoeffding's inequality, which both only assume independence.  

For instance,  the conclusion of Theorem~\ref{thm:PF}   hold for a particle filter algorithm based on stratified resampling. To define this algorithm, let  $U_t^N:=(U_t^{N,1},\dots,U_t^{N,N})$ be a sequence of independent uniform random variables on $[0,1]^N$, and, for any $t\geq 1$ and $n\in\{1,\dots,N\}$, let
\begin{align*}
V_t^{N,n}:=\frac{(n-1)+U_t^{N,n}}{N}.
\end{align*}
Then, for  particle filter with stratified resampling, the distributions $\eta_t^N$ and $\hat{\eta}_t^N$ are  defined as above, with $X^{N,1}_1,\dots,X^{N,N}_1\iid\eta_1$  and where, for all $t\geq 1$ and  $n\in\{1,\dots,N\}$, 
\begin{align*}
X_{t+1}^n\sim M_{t+1}(X_t^{A_t^n},\cdot)\quad \text{ with }\quad A_t^n:=\min\bigg\{ m\in\{1,\dots,N\}:\, V_t^{N,n}\leq \sum_{k=1}^m \frac{G_t(X_t^{N,k})}{\sum_{i=1}^N G_t(X_t^{N,i})}\bigg\}.
\end{align*}

It is also not difficult to see that, at any time $t\geq 1$, a  particle filter  based on Baker resampling \citep[a.k.a.~residual resampling; see e.g.][Section 9.5]{chopin2020introduction} generates particles $X^{N,n}_{t+1}$, $n\in\{1,\dots,N\}$  which are, given $X^N_t$,  independent and such that $\E[\eta_{t+1}^N(f)|\F^N_t]=\Phi_{t+1}(\eta_t^N)(f)$ for any functions $f\in\mathcal{B}_1(\setX)$. Hence, the conclusions of Theorem \ref{thm:PF} also apply to this algorithm.

We finally  stress that the results of this section also apply to
mean field particle models defined by conditionally
independent particles. Let $K_{t+1,\mu}(x,\dd x')$,
$\mu\in\mathcal{P}(\setX)$, be Markov transitions,
measurable in $(\mu,x)$, satisfying the compatibility
condition
\begin{align}\label{eq:compat_K}
\mu K_{t+1,\mu}=\Phi_{t+1}(\mu),
\qquad\mu\in\mathcal{P}(\setX),\quad t\geq 1,
\end{align}
and assume that, given $X^N_t$, the particles
$X^{N,n}_{t+1}$, $n\in\{1,\dots,N\}$, are independent
with
\begin{align*}
X^{N,n}_{t+1}\sim K_{t+1,\eta^N_t}(X^{N,n}_t,\cdot),
\qquad n\in\{1,\dots,N\}.
\end{align*}
The bootstrap filter corresponds to
$K_{t+1,\mu}(x,\cdot)=\Phi_{t+1}(\mu)$. When
$G_t\leq 1$, another choice is the acceptance-rejection
(or killing/regeneration) selection
\begin{align*}
K_{t+1,\mu}(x,\cdot)=G_t(x)\,M_{t+1}(x,\cdot)
+\big(1-G_t(x)\big)\,\Phi_{t+1}(\mu),
\end{align*}
for which \eqref{eq:compat_K} follows from
$\mu(G_tM_{t+1})=\mu(G_t)\,\Phi_{t+1}(\mu)$; see
\citet[][Section~4.4]{del2013mean} for these mean field
particle models and other choices of $K_{t+1,\mu}$.
Given $\F^N_t$,
the particles $X^{N,n}_{t+1}$ are independent with
respective distributions
$K_{t+1,\eta^N_t}(X^{N,n}_t,\cdot)$, whose average is
$\Phi_{t+1}(\eta^N_t)$ by \eqref{eq:compat_K}. Therefore,  Theorem~\ref{thm:PF} applies to this class of algorithms. 

\subsection{Discussion\label{sub:discussionPF}}

\subsubsection{Dependence of the error bounds on $T$ and $N$}

The dependence on $(T,N)$ in part~\ref{PF_i} of Theorem \ref{thm:PF} is optimal, up
to the factor $\log(1+1/\varepsilon_{N,T,q})$. To prove this claim, let
$d=1$, $\setX=\R$, $G_t\equiv 1$ and
$M_{t+1}(x,\dd x')=\nu(\dd x')$ for some probability measure
$\nu\in\mathcal{P}(\R)$ with a positive continuous
density with respect to $\lambda_1$, and let $\eta_1=\nu$. Then \ref{A3} and \ref{B1} hold, $\eta_t=\nu$
for any $t\geq 1$, and the  measures $\eta^N_t$ are
independent empirical measures of $N$ independent
draws from $\nu$. As $N\to\infty$ with $T$ fixed,
$\sqrt{N}\max_{1\leq t\leq T}\|\eta^N_t-\nu\|_{\mathrm{K}}$
converges in law to the maximum of $T$ independent
random variables with the Kolmogorov distribution,
whose tail is equivalent to $2e^{-2x^2}$ as
$x\to\infty$. This maximum is of order
$\sqrt{\log(T)/2}$ with high probability as
$T\to\infty$, whence
$\max_{1\leq t\leq T}\|\eta^N_t-\eta_t\|_{\mathrm{K}}$
cannot be of order smaller than
$\sqrt{\log(T)/N}$.

The additional factor
$\log(1+1/\varepsilon_{N,T,q})$ involved in
bounds in Theorem \ref{thm:PF} is inherited from the
perturbation bounds of Section~\ref{sec:main}
(see Remark~\ref{rem:log_factor}). We do not know whether
this factor can be removed under our conditions.

\subsubsection{Dependence of the error bounds on $T$ and $q$}

Let $\rho\in(0,1)$. Then, by part~\ref{PF_i} of Theorem~\ref{thm:PF}, with probability
at least $1-q$, the Kolmogorov distance between the
particle approximations and their targets is at most
$\rho$ at all times $t\in\{1,\dots,T\}$ if
$\varepsilon_{N,T,q}$ is of order
$\rho/\log(1/\rho)$, that is, if 
\begin{align}\label{eq:lower_N}
N\geq  C'_{\eta_1}\big(1+\log(1+\rho^{-1})\big)^2 \rho^{-2}\big(1+\log(T/q)\big)
\end{align}
with $C'_{\eta_1}>0$ some constant  depending only on $d$ and on the constant $C_{\eta_1}$ in Theorem \ref{thm:PF}. Hence, according to \eqref{eq:lower_N},  for a given precision $\rho$, the number of particles should grow  logarithmically in $T$
and $1/q$. This dependence of $N$ to the parameters $T$ and $q$ is sharp. Indeed,  by   \citet[][Proposition 3.5]{gerber2026safety} and Proposition \ref{prop:example}, it follows that, for a given precision $\rho\in(0,1/2)$ and under our conditions \ref{A3} and \ref{B1}, the minimum value of $N$ needed to ensure that
\begin{align*}
\P\Big(\max_{1\leq t\leq T}\|\eta^N_t-\eta_t\|_{\mathrm{K}}
\vee
\max_{1\leq t\leq T}
\|\widehat{\eta}^N_t-\widehat{\eta}_t\|_{\mathrm{K}}\leq \rho\Big)\geq 1-q
\end{align*}
cannot be  of order smaller than $\log(T/q)$.

 The same  lower bound \eqref{eq:lower_N} on $N$ applies for the result stated in part~\ref{PF_iii} of Theorem~\ref{thm:PF}, which covers unbounded
functions such as, for the class of state-space models introduced in Section \ref{sec:examples}, the components of $H_t$ (see Remark \ref{rem:unbounded}),
and hence the filtering mean when $H_t(x)=A x$ for some full rank column matrix $A$ (e.g.~$H_t(x)=x$).

\subsubsection{The examples of Section~\ref{sub:expl_ex}}

Section~\ref{sub:expl_ex} provides examples of models for which conditions \ref{A3} and \ref{B1} hold, and thus for which Theorem \ref{thm:PF} applies. We focus here on the partially observed diffusions model introduced in  Section~\ref{sub:SDE}. 

In this model, it is challenging to sample from $M_{t+1}(x,\dd x')$, but the
bootstrap particle filter can be easily deployed on the Euler--Maruyama approximation of the model presented
Section~\ref{sub:Euler}, with $l^\star\geq 1$ discretization
steps per unit of time. By Remark \ref{rem:Diff},  conditions
\ref{A3} and \ref{B1} hold for this approximate model
if the functions $b$ and $\sigma$ satisfy
\ref{diff1}--\ref{diff3} and
$\sup_{t\geq 1}\|y_{t+1}-y_t\|<\infty$, without
requiring the observations to be bounded. Let
$\eta^{N,\star}_t$ denote the particle
approximation of the prediction distributions
$\eta^\star_t$ of the approximate model. By Remarks~\ref{rem:constant} and \ref{rem:Diff}, the constant of
Theorem~\ref{thm:PF} for the approximate model does
not depend on $l^\star$. Then, by using Proposition~\ref{prop:compare_diff},
Theorem~\ref{thm:PF}, 
the inequality
$\|\cdot\|_{\mathrm{K}}\leq\|\cdot\|_{\mathrm{tv}}$ and
the triangle inequality, with probability at least
$1-q$ we have
\begin{align*}
\max_{1\leq t\leq T}\|\eta^{N,\star}_t-\eta_t\|_{\mathrm{K}}
\leq C_{\eta_1}\Big(\vartheta\big(\varepsilon_{N,T,q}\big)
+\vartheta(1/l^\star)\Big),
\end{align*}
and similarly for the updated distributions. The two sources of errors are balanced when $N$ is of order
$(l^\star)^2$, up to logarithmic factors. Since the
cost of the algorithm is of order $N\,l^\star\,T$, it follows from \eqref{eq:lower_N} that
reaching a Kolmogorov error $\rho$ over the horizon $T$
costs of order $\rho^{-3}\,T$ operations, up to
logarithmic factors.

\subsubsection{Comparison with existing results}

To the best of our
knowledge, high-probability
estimates similar to those presented in Theorem \ref{thm:PF} have only been obtained in \citet{delmoral2012concentration},
\citet[][Corollary~14.5.7]{del2013mean},
 and
\citet[][Theorem 3.3]{gerber2026safety}, for restricted classes of
models. Specifically, under uniform regularity conditions on the
Feynman--Kac semigroup, which require in particular
bounded ratios of the potential functions and summable
Dobrushin coefficients, \citet[][Corollary
and~14.5.7]{del2013mean} provides    estimates   of order $\sqrt{\log(T)/N}$ and with explicit constants;
these conditions however exclude the non-compact models
considered here. The result in \citet[][Theorem 3.3]{gerber2026safety} is identical to part \ref{PF_i} of Theorem \ref{thm:PF}, but only applies to one dimensional linear Gaussian state-space models with observations $y_t=0$. By contrast, as illustrated in Section \ref{sec:examples}, Theorem \ref{thm:PF} is valid to a broad class of models, including models defined on non-compact state-spaces and models with unstable Markov transitions.

The first non-asymptotic (in $N$) time-uniform convergence estimates for
particle filters were obtained in
\citet[][Theorem~3.1]{delmoral2001stability} and in
\citet[][Theorem~2.11]{del2000branching}, where uniform
$\mathsf{L}^p$-error bounds are also established for
empirical processes indexed by classes of functions,
under uniform mixing conditions on the Markov
transitions and bounded potential functions; see also
\citet{delmoral2000convergence} for Glivenko--Cantelli
and Donsker theorems for these empirical processes.  Deriving a trajectory estimate from a fixed-time
$\mathsf{L}^p$ estimate  requires a
union bound over $t$, which yields a
dependence of order $(T/q)^{1/p}$. By contrast, in Theorem \ref{thm:PF},   $T$ and $q$
enter $\varepsilon_{N,T,q}$ only through
$\sqrt{\log(T/q)}$.

\section{Concluding remarks\label{sec:concludion}}

We conclude with some directions for future research.

First, the perturbation bounds of Section~\ref{sec:main} apply to any sequence of probability measures
$\bm{\mu}=(\mu_t)_{t\geq 1}$ approximating the flow $ \eta_t$, and only require controlling the
largest local perturbation \eqref{eq:maxPerturbation}. In this work, they have been applied to
approximate Feynman--Kac models, including linearized Kalman filters (Section~\ref{sec:compareFK}), and to
particle filters (Section~\ref{sec:PF}). A natural next step is to use them to study other filtering
algorithms, such as the extended Kalman filter, for which the linearization points are computed from the
approximate filter itself, and the ensemble Kalman filter. For these algorithms, the sequence $\bm{\mu}$ is
not the flow of a fixed approximating Feynman--Kac model, and the main task is to bound the local
perturbations $d_2(\mu_t,\Phi_t(\mu_{t-1}))$. 

Second, the second part of \ref{B1}, which requires the integrals
$\int_\setX\mathcal{V}_\setX(q_{t+1}(\cdot,x'))\,\bar{\varsigma}_{t+1}(\dd x')$ to be bounded uniformly in
$t$, is restrictive when $\setX$ is unbounded: it requires the functions $q_{t+1}(\cdot,x')$ to vanish at
infinity in an integrable way and, in typical models (e.g.\ when $\setX=\R$ and
$m_{t+1}(x,x')=m(x'-x)$ for some continuous density $m$), it prevents the functions $G_t$ from being
bounded below away from zero. This condition is only used to prove
Lemma~\ref{lemma:Q_smooth}, on which the Lipschitz estimate of Lemma~\ref{lemma:Lipschitz} relies, and
future research should aim at relaxing \ref{B1}.

Third, the bounds obtained in this work are expressed in terms of the total variation norm and of the
Kolmogorov distance, and therefore mainly concern bounded test functions; only some specific unbounded
functions can be handled, as discussed for instance in Remark \ref{rem:unbounded}. To control general
unbounded functions,  it would be of interest to replace the total
variation norm by a weighted norm $\|\cdot\|_V$ associated with some function $V\geq 1$, as in the Lyapunov
approaches discussed in Section~\ref{sub:literature}, and the Kolmogorov distance by a Wasserstein distance.

Finally, for the class of state-space models introduced in Section~\ref{sec:examples}, the functions $H_t$
are required to be bi-Lipschitz, which implies that the dimension $r$ of the observations is at least
the dimension $d$ of the state space. Whether the perturbation bounds of Section~\ref{sec:main} hold for
these models when $r<d$ is an open question.


\bibliographystyle{apalike}
\bibliography{Kolmogorov-biblio}
\appendix
\addcontentsline{toc}{section}{Appendix}
 \addtocontents{toc}{\protect\setcounter{tocdepth}{1}}

\section{Index of notation\label{sec:notation_index}}

The table lists the symbols used in several sections,
with the place where they are introduced.

{\small
\renewcommand{\arraystretch}{1.15}
\setlength{\emergencystretch}{2em}
\begin{longtable}{@{}p{0.25\textwidth}p{0.58\textwidth}p{0.11\textwidth}@{}}
\toprule
Symbol & Description & Reference \\
\midrule
\endfirsthead
\toprule
Symbol & Description & Reference \\
\midrule
\endhead
\midrule
\multicolumn{3}{r@{}}{\emph{Continued on next page}} \\
\endfoot
\bottomrule
\endlastfoot

\multicolumn{3}{@{}l}{\textbf{Spaces, measures and
function classes}} \\
$\mathcal{P}(\setX)$, $\mathcal{M}(\setX)$
 & Probability measures, resp.\ finite signed
   measures, on $(\setX,\mathcal{E})$
 & \S\ref{sub:setup} \\
$\mathcal{B}_1(\setX)$
 & Measurable functions $f\colon\setX\to[0,1]$
 & \S\ref{sub:setup} \\
$\setX=\R^d\cap[\underline{x},\bar{x}]$
 & State space when it is a possibly unbounded
   hyper-rectangle (e.g.\ $\setX=\R^d$)
 & \eqref{eq:setX} \\
 $\lambda_d$
 & Lebesgue measure on $\R^d$
 & \S\ref{sub:setup} \\
$\varepsilon$-pseudo-decreasing
 & $f(z)\geq\varepsilon f(z')$ whenever $z\leq z'$
 & \S\ref{sub:basic_notation} \\
$\partial^u f$, $\mathcal{C}_{\mathrm{mixed}}(\setX)$
 & Mixed partial derivative w.r.t.\ the coordinates
   in $u$; functions with continuous mixed first-order
   derivatives
 & \S\ref{sub:basic_notation} \\
$\mathcal{C}_{\mathrm{mixed}}(\R^d;r,g)$
 & $\R^r$-valued functions in
   $\mathcal{C}_{\mathrm{mixed}}$ with
   $|\partial^u f_i|\leq g$
 & \S\ref{sub:basic_notation} \\
$L$-bi-Lipschitz
 & $L^{-1}\|x-x'\|\leq\|f(x)-f(x')\|\leq L\|x-x'\|$
 & \S\ref{sub:basic_notation} \\
$\mathcal{N}_d(m,\Sigma)$
 & Gaussian distribution on $\R^d$ with mean $m$ and
   covariance $\Sigma$
 & \S\ref{sub:basic_notation} \\
\addlinespace

\multicolumn{3}{@{}l}{\textbf{Hardy--Krause variation
and norms}} \\
$\mathcal{V}_{[a,b]}(f)$, $\mathcal{V}_\setX(f)$
 & HK-variation of $f$ on a bounded hyper-rectangle,
   resp.\ on the possibly unbounded set $\setX$
 & \eqref{eq:HK}, \eqref{eq:genHK} \\
$\mathcal{B}_{\mathrm{HK}}(\setX)$,
$\|\cdot\|_{\mathrm{HK}}$
 & Functions with $\mathcal{V}_\setX(f)<\infty$, normed
   by $\|f\|_\infty+\mathcal{V}_\setX(f)$
 & \eqref{eq:HK_norm} \\
$\Alt_w(f;R)$, $\Alt(f;R)$, $\tilde{\mathcal{V}}_{[\tilde{a},\tilde{b}]}(f)$
 & Alternating sums of the values of $f$ at the vertices of a face of $R$, resp.\ of $R$; variation in the sense of Vitali
 & \eqref{eq:alternating} \\
$\|\cdot\|_{\mathrm{tv}}$
 & Total variation norm
 & \S\ref{sub:basic_notation} \\
$\|\cdot\|_{\mathrm{K}}$
 & Kolmogorov norm
   $\sup_{b\in\setX}|\nu((-\infty,b]\cap\setX)|$
 & \eqref{eq:K-norm} \\
$\|\cdot\|_{\mathcal{V}_\setX}$
 & HK dual norm on $\mathcal{M}(\setX)$, equivalent to
   $\|\cdot\|_{\mathrm{K}}$ on $\mathcal{P}(\setX)$
 & \S\ref{sub:New_norm} \\
$\mathcal{F}_2(\setX)\subseteq\mathcal{F}_1(\setX)$,
$d_i$
 & Classes of test functions and pseudo-metrics
   $d_i(\mu_1,\mu_2)=
   \sup_{f\in\mathcal{F}_i}|\mu_1(f)-\mu_2(f)|$
 & \S\ref{sub:assumptions_gen} \\
\addlinespace
\multicolumn{3}{@{}l}{\textbf{The Feynman--Kac
model}} \\
$G_t$, $M_{t+1}$
 & Potential functions in $\mathcal{B}_1(\setX)$ and
   Markov transitions
 & \S\ref{sub:setup} \\
$Q_{t+1}$, $Q_{s,t}$
 & $Q_{t+1}(x,\dd y)=G_t(x)M_{t+1}(x,\dd y)$ and
   $Q_{s,t}=Q_{s+1}\cdots Q_t$
 & \S\ref{sub:setup} \\
$\eta_t$, $\widehat{\eta}_t$
 & Prediction and updated 
   Feynman--Kac measures,
   $\eta_t=\gamma_t/\gamma_t(1)$
 & \eqref{eq:def_eta_t} \\
$\mathcal{P}_{W_t}(\setX)$,
$\mathcal{P}_{W_t|1}(\setX)$
 & $\{\mu:\mu(W_t)>0\}$ and
   $\mathcal{P}_{W_t}(\setX)\cap\mathcal{P}_1(\setX)$
 & \S\ref{sub:setup}, \S\ref{sub:assumptions_gen} \\
$\Psi_t$
 & Updating map
   $\Psi_t(\mu)(\dd x)=G_t(x)\mu(\dd x)/\mu(G_t)$
 & \eqref{eq:def_Phi} \\
$\Phi_{t+1}$, $\Phi_{s,t}$
 & The map $\Phi_{t+1}(\mu)=\Psi_t(\mu)M_{t+1}$ and its
   semigroup $\Phi_t\circ\cdots\circ\Phi_{s+1}$
 & \eqref{eq:def_Phi} \\
$C_{t,\varrho}$
 & Super-level set $\{x:G_t(x)\geq\varrho\}$
 & \S\ref{sub:primitive} \\
\addlinespace

\multicolumn{3}{@{}l}{\textbf{Perturbation
analysis}} \\
$\bm{\mu}=(\mu_t)_{t\geq1}$, $\widehat{\mu}_t$
 & Approximations of $\eta_t$, and
   $\widehat{\mu}_t=\Psi_t(\mu_t)$
 & \S\ref{sub:setup} \\
$\delta_t(\mu,\nu)$
 & Local perturbation $\nu-\Phi_t(\mu)$, with $\Phi_1$
   the identity
 & \eqref{eq:delta_s} \\
$\Delta_{T,\bm{\mu},\|\cdot\|}$,
$\Delta_{T,\bm{\mu}|\eta}$
 & Largest local perturbation up to time $T$ w.r.t. $\|\cdot\|$, resp.\ $d_2$
 & \eqref{eq:maxPerturbation} \\
$\varpi$, $\Delta_{T,\bm{\mu}|\eta,\varpi}$
 & Weight function and largest weighted local
   perturbation
 & \eqref{eq:Delta_phi} \\
$\vartheta(z)$
 & Rate function $z\log(1+z^{-1})$, with
   $\vartheta(0)=0$
 & \eqref{eq:vartheta}\\
\addlinespace

\multicolumn{3}{@{}l}{\textbf{Assumptions and
state-space models}} \\
$D_t$, $\underline{q}$, $\xi$, $\varrho$
 & Sets and constants of the local minorization
   condition \ref{A1}
 & \ref{A1} \\
$m_{t+1}$, $q_{t+1}$
 & Density of $M_{t+1}$ and
   $q_{t+1}(x,x')=G_t(x)m_{t+1}(x,x')$
 & \ref{B1} \\
$H_t$, $F_{t+1}$
 & Observation and state transition functions of
   \eqref{eq:SSM}
 & \eqref{eq:functions} \\
$\zeta_{G_t}$, $\zeta_{M_{t+1}}$
 & Radial profiles of the noise densities
 & \eqref{eq:functions} \\
$\zeta_G$, $\zeta_M$, $\varepsilon$, $\nu_G$,
$\nu_M$, $L$
 & Lower bounds, pseudo-monotonicity, tail decay and
   uniform constants in \ref{Model1}--\ref{Model}
 & \ref{Model1} \\
\addlinespace

\multicolumn{3}{@{}l}{\textbf{Applications}} \\
$G^\star_t$, $M^\star_{t+1}$, $\eta^\star_t$,
$\widehat{\eta}^\star_t$
 & Approximating Feynman--Kac model and its prediction
   and updated measures
 & \eqref{eq:model3} \\
$\Delta_{Q|Q^\star}$
 & $\sup_t\sup_x
   \|Q^\star_{t+1}(x,\cdot)-Q_{t+1}(x,\cdot)\|_{\mathrm{tv}}$
 & \eqref{eq:delta} \\
$\eta^N_t$, $\widehat{\eta}^N_t$
 & Particle filter approximations of $\eta_t$ and
   $\widehat{\eta}_t$ with $N$ particles
 & \S\ref{sub:algoPF} \\
$\delta^N_t$
 & Local perturbation $\eta^N_t-\Phi_t(\eta^N_{t-1})$
   of the particle filter
 & \S\ref{sub:local_PF} \\
$\varepsilon_{N,T,q}$
 & High-probability bound on the particle filter's
   local perturbations
 & \eqref{eq:delta_pT} \\

\end{longtable}
}

\section{Proofs\label{app:proofs}}

\subsection{Proofs of the results in Section~\ref{sec:notation}\label{sec:proofs}}

\subsubsection{Proof of Proposition~\ref{prop:order}\label{p-prop:order}}

\begin{proof}

\noindent\textbf{Proof of part \ref{basic_HK1}.}
Let $[\tilde{a},\tilde{b}]\subseteq[a,b]$, so that
$a\leq\tilde{a}\leq\tilde{b}\leq b$.
Recall the conventions on
hyper-rectangles with a side reduced to a point
introduced in Section~\ref{sub:HK-var}. In what follows, for any disjoint sets
$u,w\subseteq\{1,\dots,d\}$ such that $u\cup w\neq\emptyset$, we identify
$\R^{|u\cup w|}$ with the set of vectors whose coordinates are indexed by the
elements of $u\cup w$ (in increasing order) and, with a slight abuse of notation,
for any hyper-rectangles $[\alpha,\beta]\subseteq\R^{|u|}$ and
$[\tilde{\alpha},\tilde{\beta}]\subseteq\R^{|w|}$ we let
\begin{align*}
[\alpha,\beta]\times[\tilde{\alpha},\tilde{\beta}]\coloneqq\big\{x\in\R^{|u\cup w|}:\ \alpha\leq x^u\leq \beta\text{ and } \tilde{\alpha}\leq x^{w}\leq \tilde{\beta}\big\}.
\end{align*}
All the inequalities below are between elements of $[0,\infty]$.

\emph{Step 1: change of anchor.}
Let $m\geq 1$ and $\varphi\colon\R^m\to\R$. Then, by
\citet[][Proposition~6]{owen2005multidimensional}, for any  $c\leq e$ in $\R^m$ we have
\begin{equation*}
\varphi(c)=\varphi(e)+\sum_{\emptyset\neq u\subseteq\{1,\dots,m\}}
(-1)^{|u|}\sum_{w\subseteq u}(-1)^{|w|}\,
\varphi\big(\langle c^w,e\rangle_w\big).
\end{equation*}
By \eqref{eq:alternating}, the inner sum is equal to $\Alt_u(\varphi;[c,e])$, while
$\Alt_\emptyset(\varphi;[c,e])=\varphi(e)$. Therefore,
\begin{equation}\label{eq:anchor_change}
\varphi(c)=\sum_{u\subseteq\{1,\dots,m\}}(-1)^{|u|}\,
\Alt_u\big(\varphi;[c,e]\big).
\end{equation}
Moreover, if $\varphi$ does not depend on its $i$-th coordinate for some $i\in\{1,\dots,m\}$, then
$\Alt_u(\varphi;[c,e])=0$ for any set $u\ni i$, since in \eqref{eq:alternating} the terms associated with
the sets $s$ and $s\cup\{i\}$, $s\subseteq u\setminus\{i\}$, cancel out.

We now let $u\subseteq\{1,\dots,d\}$ be  a non-empty set
and $z\in[\tilde{a}^u,\tilde{b}^u]$, we let $\varphi_z(y)\coloneqq f(\langle z,y\rangle_u)$ for any
$y\in\R^d$, and, for any $w\subseteq -u$, we let
\begin{equation*}
\Lambda_{u,w}(z)\coloneqq
\Alt_{w}\big(\varphi_z;[\tilde{b},b]\big)=\sum_{s\subseteq w}(-1)^{|s|}\,
f\big(\langle z,\langle\tilde{b}^s,b\rangle_s\rangle_u\big),
\end{equation*}
so that, in particular, $\Lambda_{u,\emptyset}(z)=f(\langle z,b\rangle_u)$. Since the function $\varphi_z$
does not depend on the coordinates $y^u$, by applying \eqref{eq:anchor_change} with $m=d$,
$\varphi=\varphi_z$, $c=\tilde{b}$ and $e=b$, only the sets $w\subseteq -u$ contribute to the sum,
and we obtain
\begin{equation}\label{eq:f_Lambda}
f(\langle z,\tilde{b}\rangle_u)=\varphi_z(\tilde{b})
=\sum_{w\subseteq -u}(-1)^{|w|}\,\Lambda_{u,w}(z).
\end{equation}
Since $\Alt(\cdot\,;R)$ is linear with respect to the function for any
hyper-rectangle $R$, it follows from \eqref{eq:f_Lambda} and from the triangle
inequality that, for any $P\in\mathsf{P}([\tilde{a}^u,\tilde{b}^u])$,
\begin{equation*}
\sum_{R\in P}\big|\Alt\big(f(\langle\cdot,\tilde{b}\rangle_u);R\big)\big|
\leq\sum_{w\subseteq -u}\sum_{R\in P}\big|\Alt(\Lambda_{u,w};R)\big|
\leq\sum_{w\subseteq -u}
\tilde{\mathcal{V}}_{[\tilde{a}^u,\tilde{b}^u]}
(\Lambda_{u,w}).
\end{equation*}
Taking the supremum over $P$ (the next inequality being trivial if a side of
$[\tilde{a}^u,\tilde{b}^u]$ is reduced to a point, since its left-hand side is then null by convention) yields
\begin{equation*}
\tilde{\mathcal{V}}_{[\tilde{a}^u,\tilde{b}^u]}
\big(f(\langle\cdot,\tilde{b}\rangle_u)\big)
\leq\sum_{w\subseteq -u}
\tilde{\mathcal{V}}_{[\tilde{a}^u,\tilde{b}^u]}
(\Lambda_{u,w}).
\end{equation*}
By summing this inequality over all non-empty sets $u\subseteq\{1,\dots,d\}$ and using the
definition \eqref{eq:HK} of $\mathcal{V}_{[\tilde{a},\tilde{b}]}(f)$ (anchored at $\tilde{b}$), we obtain
\begin{align}\label{eq:T_u_bound}
\mathcal{V}_{[\tilde{a},\tilde{b}]}
(f)=\sum_{\emptyset\neq u\subseteq\{1,\dots,d\}}\tilde{\mathcal{V}}_{[\tilde{a}^u,\tilde{b}^u]}
\big(f(\langle\cdot,\tilde{b}\rangle_u)\big)\leq  \sum_{\emptyset\neq u\subseteq\{1,\dots,d\}}\sum_{w\subseteq -u}
\tilde{\mathcal{V}}_{[\tilde{a}^u,\tilde{b}^u]}(\Lambda_{u,w}).
\end{align}

\emph{Step 2: composition of  alternating sums.} As above, let $u\subseteq\{1,\dots,d\}$ be  a non-empty set and
$w\subseteq-u$, and let $v\coloneqq u\cup w$ and   $R=[\alpha,\beta]  \subseteq [\tilde{a}^u,\tilde{b}^u]$ be a hyper-rectangle.
Then, by the definition of $\Lambda_{u,w}$ and by \eqref{eq:alternating},
\begin{equation*}
\Alt(\Lambda_{u,w};R)
=\sum_{t\subseteq u}\sum_{s\subseteq w}(-1)^{|t|+|s|}\,f(x_{t,s}) 
\end{equation*}
where $x_{t,s}\in\R^d$ denotes the vector such that $x_{t,s}^{t}=\alpha^{t}$,
$x_{t,s}^{u\setminus t}=\beta^{u\setminus t}$, $x_{t,s}^{s}=\tilde{b}^{s}$
and $x_{t,s}^{-(u\cup s)}=b^{-(u\cup s)}$ (with $\alpha^t$ and $\beta^{u\setminus t}$
denoting the corresponding coordinates of $\alpha,\beta\in\R^{|u|}$). The points
$\{x_{t,s}^{v},\,t\subseteq u,\,s\subseteq w\}$ are exactly the vertices of the hyper-rectangle
$R\times[\tilde{b}^w,b^w]$, the vertex $x_{t,s}^v$ being obtained by taking the
coordinates of the lower corner $(\alpha,\tilde{b}^w)$ in the set $t\cup s$. Therefore,
by \eqref{eq:alternating},
\begin{equation}\label{eq:Delta_product}
\Alt(\Lambda_{u,w};R)
=\Alt\big(f(\langle\cdot,b\rangle_v);
R\times[\tilde{b}^w,b^w]\big),
\end{equation}
from which we conclude that
\begin{equation*}
\sum_{R\in P}\big|\Alt(\Lambda_{u,w};R)\big|
=\sum_{R\in P}\Big|\Alt
\big(f(\langle\cdot,b\rangle_v);
R\times[\tilde{b}^w,b^w]\big)\Big|,\quad \forall P\in\mathsf{P}([\tilde{a}^u,\tilde{b}^u]).
\end{equation*}

\emph{Step 3: summation over the faces.}
Fix a non-empty set $v\subseteq\{1,\dots,d\}$ and let
\begin{align*}
\mathcal{S}_v\coloneqq\Big\{(u,w):\ u,w\subseteq\{1,\dots,d\},\ u\neq\emptyset,\ u\cap w=\emptyset,\ u\cup w=v\Big\}.
\end{align*}
In addition, for any $(u,w)\in \mathcal{S}_v$ let
$\mathcal{R}_{u,w|v}\coloneqq
[\tilde{a}^u,\tilde{b}^u]\times[\tilde{b}^w,b^w]$ and note that
$\mathcal{R}_{u,w|v}\subseteq [a^v,b^v]$ since $a\leq\tilde{a}\leq\tilde{b}\leq b$.

Let $(u,w)\in \mathcal{S}_v$. If $\tilde{a}_i=\tilde{b}_i$ for some $i\in u$ then
$\tilde{\mathcal{V}}_{[\tilde{a}^u,\tilde{b}^u]}(\Lambda_{u,w})=0$ by convention, while if
$\tilde{b}_i=b_i$ for some $i\in w$ then $\Lambda_{u,w}\equiv 0$, since
$\Lambda_{u,w}(z)$ is then an alternating sum over a hyper-rectangle with a side
reduced to a point, and thus again
$\tilde{\mathcal{V}}_{[\tilde{a}^u,\tilde{b}^u]}(\Lambda_{u,w})=0$. In both cases, the inequality
\eqref{eq:inc1} stated below trivially holds. Otherwise, $\tilde{a}^u<\tilde{b}^u$ and $\tilde{b}^w<b^w$,
and
\begin{align}\label{eq:inc0}
\{R\times[\tilde{b}^w,b^w]:\ R\in P\}\in\mathsf{P}(\mathcal{R}_{u,w|v}),\quad\forall  P\in\mathsf{P}([\tilde{a}^u,\tilde{b}^u]).
\end{align}
Consequently,
\begin{equation}\label{eq:inc1}
\begin{split}
\tilde{\mathcal{V}}_{[\tilde{a}^u,\tilde{b}^u]}(\Lambda_{u,w})=\sup_{P\in \mathsf{P}([\tilde{a}^u,\tilde{b}^u])} \sum_{R\in P}
\Big|\Alt(\Lambda_{u,w};R)\Big|&=\sup_{P\in \mathsf{P}([\tilde{a}^u,\tilde{b}^u])} \sum_{R\in P}
\Big|\Alt\big(f(\langle\cdot,b\rangle_v);
R\times[\tilde{b}^w,b^w]\big)\Big|\\
&\leq\sup_{P'\in \mathsf{P}(\mathcal{R}_{u,w|v})}\sum_{R'\in P'}
\Big|\Alt\big(f(\langle\cdot,b\rangle_v);
R'\big)\Big|\\
&=\tilde{\mathcal{V}}_{\mathcal{R}_{u,w|v}}
\big(f(\langle\cdot,b\rangle_v)\big)
\end{split}
\end{equation}
where the second equality holds by \eqref{eq:Delta_product} and the inequality holds by \eqref{eq:inc0}.

To proceed further, we observe that, for any elements $(u,w)$ and $(u',w')$ of $\mathcal{S}_v$ such that $(u,w)\neq (u',w')$,  the two hyper-rectangles $\mathcal{R}_{u,w|v}$ and $\mathcal{R}_{u',w'|v}$   have disjoint interiors. Indeed, if
$(u,w)\neq(u',w')$ then $u\neq u'$ (since $w=v\setminus u$ and $w'=v\setminus u'$), and thus there is a coordinate
$i\in u\cap w'$ or $i\in w\cap u'$, and the $i$-th
sides $[\tilde{a}_i,\tilde{b}_i]$ and
$[\tilde{b}_i,b_i]$ overlap at most at the point
$\tilde{b}_i$.
Assume that $a^v<b^v$ (otherwise, every $\mathcal{R}_{u,w|v}$ has a side reduced to a point,
so that both sides of \eqref{eq:inc2} below are null by convention), and let $G_v\in\mathsf{P}([a^v,b^v])$
be the grid splitting generated, in each coordinate $i\in v$, by the points
$a_i\leq\tilde{a}_i\leq\tilde{b}_i\leq b_i$ (repeated points being removed).
Each hyper-rectangle $\mathcal{R}_{u,w|v}$ with no side reduced to a point is a cell of $G_v$,
since in each coordinate $i\in v$ its side is either $[\tilde{a}_i,\tilde{b}_i]$ or $[\tilde{b}_i,b_i]$,
and, by the above observation, distinct elements of $\mathcal{S}_v$ correspond to distinct cells. Since the
hyper-rectangles $\mathcal{R}_{u,w|v}$ with a side reduced to a point have a null variation by convention, 
it follows from \citet[][Lemma 1]{owen2005multidimensional} that
\begin{align}\label{eq:inc2}
\sum_{(u,w)\in\mathcal{S}_v}\tilde{\mathcal{V}}_{\mathcal{R}_{u,w|v}}\big(f(\langle\cdot,b\rangle_v)\big)\leq
\sum_{C\in G_v}\tilde{\mathcal{V}}_{C}\big(f(\langle\cdot,b\rangle_v)\big)=
\tilde{\mathcal{V}}_{[a^v,b^v]}\big(f(\langle\cdot,b\rangle_v)\big).
\end{align}

\emph{Conclusion.} The proof of part \ref{basic_HK1}   of the proposition is completed by noting that
\begin{align*}
\mathcal{V}_{[\tilde{a},\tilde{b}]}
(f)\leq  \sum_{\emptyset\neq u\subseteq\{1,\dots,d\}}\sum_{w\subseteq -u}
\tilde{\mathcal{V}}_{[\tilde{a}^u,\tilde{b}^u]}(\Lambda_{u,w})&=  \sum_{\emptyset\neq v\subseteq\{1,\dots,d\}}\sum_{(u,w)\in\mathcal{S}_v}\tilde{\mathcal{V}}_{[\tilde{a}^u,\tilde{b}^u]}(\Lambda_{u,w})\\
&\leq \sum_{\emptyset\neq v\subseteq\{1,\dots,d\}}\sum_{(u,w)\in\mathcal{S}_v}\tilde{\mathcal{V}}_{\mathcal{R}_{u,w|v}}
\big(f(\langle\cdot,b\rangle_v)\big)\\
&\leq \sum_{\emptyset\neq v\subseteq\{1,\dots,d\}}\tilde{\mathcal{V}}_{[a^v,b^v]}
\big(f(\langle\cdot,b\rangle_v)\big)\\
&=\mathcal{V}_{[a,b]}(f)
\end{align*}
where the first inequality holds by \eqref{eq:T_u_bound}, the first equality holds by regrouping the terms
according to $v=u\cup w$ (the map $(u,w)\mapsto u\cup w$ being a bijection from
$\{(u,w):\ \emptyset\neq u\subseteq\{1,\dots,d\},\ w\subseteq -u\}$ onto
$\bigcup_{\emptyset\neq v\subseteq\{1,\dots,d\}}\mathcal{S}_v$ with fibres $\mathcal{S}_v$),
the second inequality holds by \eqref{eq:inc1}, the third inequality holds by \eqref{eq:inc2},
and the last equality holds by the definition \eqref{eq:HK} of $\mathcal{V}_{[a,b]}(f)$.

 \noindent \textbf{Proof of part \ref{basic_HK2}  and of the ``if'' direction of
part   \ref{basic_HK3}.}
Assume that there exists a (possibly empty) set
$v\subsetneq\{1,\dots,d\}$ such that
$f(x)=f(\langle x^v,b\rangle_v)$ for any
$x\in[a,b]$, that is, $f$ depends on $x$ only
through $x^v$. By
\citet[][Proposition~8]{owen2005multidimensional},
we then have
$\tilde{\mathcal{V}}_{[a^u,b^u]}
\big(f(\langle\cdot,b\rangle_u)\big)=0$
for any sets $u\subseteq\{1,\dots,d\}$ such that
$u\cap(-v)\neq\emptyset$, and therefore
\begin{equation}\label{eq:lemma_lower}
\mathcal{V}_{[a,b]}(f)
=\sum_{\emptyset\neq u\subseteq v}
\tilde{\mathcal{V}}_{[a^u,b^u]}
\big(f(\langle\cdot,b\rangle_u)\big).
\end{equation}
If $v\neq\emptyset$, the right-hand side of
\eqref{eq:lemma_lower} is, by definition
\eqref{eq:HK}, the HK-variation of the function
$f(\langle\cdot,b\rangle_v)$ on $[a^v,b^v]$ anchored at $b^v$, since
$f\big(\langle\langle z,b^v\rangle_u,b\rangle_v\big)=f(\langle z,b\rangle_u)$ for any non-empty set $u\subseteq v$ and any $z\in\R^{|u|}$, so that
$\mathcal{V}_{[a,b]}(f)=
\mathcal{V}_{[a^v,b^v]}
\big(f(\langle\cdot,b\rangle_v)\big)$,
which ends the proof of part \ref{basic_HK2}. If $v=\emptyset$, that is, if
$f(x)=f(b)$ for any $x\in[a,b]$, the sum in
\eqref{eq:lemma_lower} is empty, and hence
$\mathcal{V}_{[a,b]}(f)=0$ since empty sums equal
zero by convention. This proves that
$\mathcal{V}_{[a,b]}(f)=0$ whenever $f$ is constant
on $[a,b]$.

\noindent \textbf{Proof of the ``only if'' direction of part \ref{basic_HK3}}. 
It remains to show that $\mathcal{V}_{[a,b]}(f)=0$
implies that $f$ is constant on $[a,b]$. By
\citet[][Proposition~6]{owen2005multidimensional},
for any $x\in[a,b]$ we have
\begin{equation}\label{eq:B1_HK}
f(x)=f(b)+\sum_{\emptyset\neq u\subseteq\{1,\dots,d\}}
(-1)^{|u|}\sum_{w\subseteq u}(-1)^{|w|}\,
f\big(\langle x^w,b\rangle_w\big).
\end{equation}
Let $x\in[a,b]$ and let $u\subseteq\{1,\dots,d\}$ be
non-empty. Since
$\langle x^w,b\rangle_w
=\big\langle(\langle x^w,b\rangle_w)^u,b\big\rangle_u$
for any $w\subseteq u$, the inner sum in
\eqref{eq:B1_HK} involves only the function
$f(\langle\cdot,b\rangle_u)$ evaluated at the
$2^{|u|}$ vertices of the hyper-rectangle
$R_{x,u}\coloneqq[x^u,b^u]\subseteq[a^u,b^u]$. More
precisely,
\begin{equation*}
\sum_{w\subseteq u}(-1)^{|w|}\,
f\big(\langle x^w,b\rangle_w\big)
=\Alt
\big(f(\langle\cdot,b\rangle_u);R_{x,u}\big).
\end{equation*}
If $x_i=b_i$ for some $i\in u$, then $R_{x,u}$ has a side reduced to a point and the
left-hand side of the next inequality is null. Otherwise, we have $a^u\leq x^u<b^u$, and $R_{x,u}$ is a cell
of the grid splitting of $[a^u,b^u]$ generated, in each coordinate $i\in u$, by the points
$a_i\leq x_i<b_i$ (the point $a_i$ being removed when $x_i=a_i$). In both cases, it follows that
\begin{equation*}
\Big|\sum_{w\subseteq u}(-1)^{|w|}\,
f\big(\langle x^w,b\rangle_w\big)\Big|
\leq\sup_{P\in\mathsf{P}([a^u,b^u])}
\sum_{R\in P}
\Big|\Alt
\big(f(\langle\cdot,b\rangle_u);R\big)\Big|
=\tilde{\mathcal{V}}_{[a^u,b^u]}
\big(f(\langle\cdot,b\rangle_u)\big).
\end{equation*}
Together with \eqref{eq:B1_HK}, this shows that
\begin{align}\label{eq:lms2}
|f(x)-f(b)|
\leq\sum_{\emptyset\neq u\subseteq\{1,\dots,d\}}
\tilde{\mathcal{V}}_{[a^u,b^u]}
\big(f(\langle\cdot,b\rangle_u)\big)
=\mathcal{V}_{[a,b]}(f),
\qquad\forall x\in[a,b].
\end{align}
Consequently, if $\mathcal{V}_{[a,b]}(f)=0$, then
$f(x)=f(b)$ for any $x\in[a,b]$, which completes the
proof of part   \ref{basic_HK3} and of the proposition.
\end{proof}
 
\subsection{Proofs of the results in Section~\ref{sec:preliminaries}}

\subsubsection{Proof of Lemma~\ref{lemma:V_well_def} \label{p-lemma:V_well_def}}

\begin{proof}
To prove the first assertion of the lemma we let $(n_k)_{k\geq 1}$  and $(m_k)_{k\geq 1}$ be two sequences  in $\mathbb{N}$ such that
\begin{align*}
\lim_{k\rightarrow\infty}\mathcal{V}_{[a_{n_k},b_{n_k}]}(f)=\limsup_{n\rightarrow\infty}\mathcal{V}_{[a_{n},b_{n}]}(f)=:l,\quad  \lim_{k\rightarrow\infty}\mathcal{V}_{[a_{m_k},b_{m_k}]}(f)=\liminf_{n\rightarrow\infty}\mathcal{V}_{[a_{n},b_{n}]}(f)=:\underline{l}.
\end{align*}
Next, we let $k\geq 1$ be arbitrary  and   $p_k\in\mathbb{N}$ be such that the inclusion $[a_{n_k},b_{n_k}]\subseteq  [a_{m_p},b_{m_p}]$ holds for any $p\geq p_k$. Then,  by using part \ref{basic_HK1} of Proposition~\ref{prop:order}, we have
\begin{align*}
\mathcal{V}_{ [a_{m_p},b_{m_p}]}(f)\geq \mathcal{V}_{ [a_{n_k},b_{n_k}]}(f),\quad\forall p\geq p_k
\end{align*}
implying that $\underline{l}\geq \mathcal{V}_{ [a_{n_k},b_{n_k}]}(f)$. Since $k\geq 1$ is arbitrary, we conclude that $\underline{l}\geq l$, and thus $l=\underline{l}$.

To prove the second assertion of the lemma, we let  $(a'_n)_{n\geq 1}\in\mathcal{L}(\setX)$ and  $(b'_n)_{n\geq 1}\in\mathcal{U}(\setX)$ be arbitrary, and  $l'=\lim_{n\rightarrow\infty}\mathcal{V}_{ [a'_n,b'_n]}(f)$. Without loss of generality, we assume that $l'\geq l$, and below we   show    that $l'\leq l$. To this end, we let $c\in(0,l')$,  $n_c\in\mathbb{N}$  be such that $\mathcal{V}_{[a'_{n_c},b'_{n_c}]}(f)\geq c$, and we note that the properties of the sequences $(a_n)_{n\geq 1}$, $(a'_n)_{n\geq 1}$, $(b_n)_{n\geq 1}$ and $(b'_n)_{n\geq 1}$ ensure that there exists an $m_c\in\mathbb{N}$ for which the inclusion $[a'_{n_c},b'_{n_c}]\subseteq  [a_{m},b_{m}]$ holds for any $m\geq m_c$. Then,  by using part \ref{basic_HK1} of Proposition~\ref{prop:order}, we obtain that   
\[
l=\lim_{m\rightarrow\infty} \mathcal{V}_{[a_{m},b_{m}]}(f)\geq  \mathcal{V}_{ [a'_{n_c},b'_{n_c}]}(f)\geq c
\]
 and thus, since the constant $c\in(0,l')$ is arbitrary, it follows that $l \geq l'$. This shows that $l=l'$ and  the proof of the lemma is complete.
\end{proof}

\subsubsection{Proof of Proposition~\ref{prop:set_BHK}\label{p-prop:set_BHK}}

\begin{proof}
Below we let $f\in\mathcal{B}_{\mathrm{HK}}(\setX)$ be fixed, and we note first that if $\underline{x}=\bar{x}$   then $f(x)=f(\bar{x})$ for any $x\in\setX$, in which case the function $f$ is trivially bounded and measurable. 

We now   prove the proposition assuming that $\underline{x}<\bar{x}$, starting by    showing that the function  $f$ is  measurable. To this end, recall that, by convention, if $\setX\subseteq\R^d$ then  $\mathcal{E}$ is the Borel $\sigma$-algebra on $\setX$. Thus, if $\setX=[0,1]^d$ then the function $f$ is measurable   by \citet[][Theorem 3.1]{aistleitner2017functions}. Assume now that $\setX=[\underline{x},\bar{x}]$ for some arbitrary  $\underline{x}<\bar{x}$, and let $\tilde{f}=f\circ \phi_{\underline{x},\bar{x}}$ with the function $\phi_{\underline{x},\bar{x}}:[0,1]^d\rightarrow[\underline{x},\bar{x}]$   defined by
\begin{align}\label{eq:phi_fn}
\phi_{\underline{x},\bar{x}}(x)=  \underline{x}+ x(\bar{x}-\underline{x}),\quad x\in[0,1]^d.
\end{align}
By \citet[][Proposition 10]{owen2005multidimensional}, we have 
\[
\mathcal{V}_{[0,1]^d}(\tilde{f})=\mathcal{V}_{[\underline{x},\bar{x}]}(f)<\infty
\]
 and thus, by  \citet[][Theorem 3.1]{aistleitner2017functions}, the function $\tilde{f}:[0,1]^d\rightarrow\R$ is   measurable. Noting that $f=\tilde{f}\circ \phi_{\underline{x},\bar{x}}^{-1}$, it follows that $f$ is the composition of two   measurable   functions  and   is  therefore  measurable. 

We now consider the  case where the set $\setX$ is unbounded.  In this scenario,  we let $(a_n)_{n\geq 1}\in\mathcal{L}(\setX)$ and   $(b_n)_{n\geq 1}\in\mathcal{U}(\setX)$ be arbitrary, and, for any $n\geq 1$, we let   $g_n: [a_n,b_n]\rightarrow\R$ be such that $g_n(x)=f(x)$ for any $x\in   [a_n,b_n]$  and $f_n:\setX\rightarrow \R$ be   such that $f_n(x)=f(x)\ind_{[a_n,b_n]}(x)$ for any $x\in\setX$. In addition,  we let $n'\in\mathbb{N}$ be such that $a_n<b_n$ for any $n\geq n'$; note that such an  integer $n'$ exists since we are assuming that $\underline{x}<\bar{x}$. By Lemma~\ref{lemma:V_well_def}, we have 
\[
\lim_{n\rightarrow\infty}\mathcal{V}_{  [a_n,b_n]}(f)=\mathcal{V}_\setX(f)<\infty,
\]
 and thus  $\mathcal{V}_{  [a_n,b_n]}(f)<\infty$ for any $n\geq 1$.
Therefore, for any $n\geq n'$   the function $g_n$  is measurable, as proved above, which implies that  the function  $f_n$ is also measurable   for any $n\geq n'$. Noting that $f(x)=\lim_{n\rightarrow\infty} f_n(x)$ for any $x\in\setX$, it follows that the function $f$ is     also   measurable. This concludes to show that $f$ is measurable   whenever the hyper-rectangle $\setX$ is such that $\underline{x}<\bar{x}$.

To show that $\|f\|_\infty<\infty$,   we let $(a_n)_{n\geq 1}$, $(b_n)_{n\geq 1}$ and $(f_n)_{n\geq 1}$ be as above. Then, owing to \eqref{eq:lms2}, we have
\begin{align}\label{eq:lms22}
|f(x)-f(b_n)|\leq  \mathcal{V}_{[a_n,b_n]}(f),\quad\forall x\in  [a_n,b_n],\quad\forall n\geq 1
\end{align}
from which we conclude that  
\begin{align*}
\limsup_{n\rightarrow\infty}|f(x)-f(b_n)|\leq \limsup_{n\rightarrow\infty}\mathcal{V}_{[a_n,b_n]}(f)=\mathcal{V}_\setX(f)<\infty,\quad\forall x\in\setX 
\end{align*}
and thus $\sup_{n\geq 1}|f(b_n)|<\infty$. Together with \eqref{eq:lms22}, this shows that
\begin{align*}
\sup_{x\in\setX}|f(x)|\leq \sup_{n\geq 1}|f(b_n)|+\mathcal{V}_\setX(f)<\infty
\end{align*}
and  the proof of the proposition in the scenario where   $\underline{x}<\bar{x}$ is complete.

To conclude the proof of the proposition it remains to consider the case where there exists a   non-empty set $u\subsetneq\{1,\dots,d\}$ such that $\underline{x}_i<\bar{x}_i$ for any $i\in u$ and such that $\underline{x}_i=\bar{x}_i$ for  any $i\in -u$. To this aim, we let $f_{u}:\setX_u\rightarrow\R$ be such that  
\[
f_{u}(z)=f\big(\langle z,\bar{x}\rangle_u\big)\quad\mbox{\rm for all $z\in\setX_u$.}
\]
 As   proved above,  if $\tilde{\setX}=\R^k\cap[a,b]$ for some integer $k\geq 1$ and some set $[a,b]\subseteq\bar{\R}^k$ such that $a<b$,  then any function $g\in\mathcal{B}_{\mathrm{HK}}(\tilde{\setX})$ is bounded and measurable. Thus, the function $f_u$ is bounded and measurable  and, by noting that $f(x)=f_u(x^{u})$ for any $x\in \setX$,  we   conclude that the function $f$  is also  bounded and measurable. This ends the proof of the proposition.
\end{proof}

\subsubsection{Proof of parts \ref{HK_5}--\ref{HK_7} of Proposition~\ref{prop:V_HK}\label{p-prop:V_HK}}

\begin{proof}

\noindent \textbf{Proof of   \ref{HK_5}.} Let $f,g:\setX\rightarrow\R$ and   assume first that $\setX=[\underline{x},\bar{x}]$. If $\underline{x}= \bar{x}$ then it follows from part \ref{basic_HK3} of Proposition~\ref{prop:order} that $\mathcal{V}_{\setX}(fg)=0$, in which case the result stated in \ref{HK_5} trivially holds. We now consider the situation where $\underline{x}\neq \bar{x}$, noting that, by part \ref{basic_HK2} of Proposition~\ref{prop:order}, to establish \ref{HK_5} in this scenario it  is enough to consider the case where $\underline{x}<\bar{x}$.  Thus, we assume that $\underline{x}<\bar{x}$, and we let $\tilde{f}=f\circ \phi_{\underline{x},\bar{x}}$ and $\tilde{g}=g\circ \phi_{\underline{x},\bar{x}}$, with the function $\phi_{\underline{x},\bar{x}}:[0,1]^d\rightarrow[\underline{x},\bar{x}]$  defined  in \eqref{eq:phi_fn}. Then,
\begin{align}\label{eq:sub_mult}
\mathcal{V}_{\setX}(fg)=\mathcal{V}_{[0,1]^d}(\tilde{f}\tilde{g})&\leq \frac{1}{3}\big(3\mathcal{V}_{[0,1]^d}(\tilde{f})+\|\tilde{f}\|_\infty\big)\big(3\mathcal{V}_{[0,1]^d}(\tilde{g})+\|\tilde{g}\|_\infty\big)\\
&=\frac{1}{3}\big(3\mathcal{V}_{\setX}(f)+\|f\|_\infty\big)\big(3\mathcal{V}_\setX(g)+\|g\|_\infty\big)
\end{align}
where the two equalities hold by \citet[][Proposition 10]{owen2005multidimensional} and where the inequality holds  by \citet[][Corollary  4.3 and Lemma 5.2]{aistleitner2017functions}.
 
Next, to prove \ref{HK_5} in the case where $\setX$ is unbounded,  we let  $(a_n)_{n\geq 1}\in\mathcal{L}(\setX)$ and $(b_n)_{n\geq 1}\in\mathcal{U}(\setX)$ be arbitrary. Then, as we just proved,
\begin{align*}
\mathcal{V}_{ [a_n,b_n]}( fg)\leq \frac{1}{3}\big(3\mathcal{V}_{ [a_n,b_n]}(f)+\|f\|_\infty\big)\big(3\mathcal{V}_{ [a_n,b_n]}(g)+\|g\|_\infty\big),\quad\forall n\geq 1 
\end{align*}
impying that
\begin{align*}
\mathcal{V}_{\setX}( fg)&=\limsup_{n\rightarrow\infty}\mathcal{V}_{[a_n,b_n]}( fg)\\
&\leq  \frac{1}{3} \big(3 \limsup_{n\rightarrow\infty}\mathcal{V}_{[a_n,b_n]}(f)+\|f\|_\infty\big)\big(3\limsup_{n\rightarrow\infty} \mathcal{V}_{[a_n,b_n]}(g)+\|g\|_\infty\big)\\
&= \frac{1}{3}\big(3 \mathcal{V}_{\setX}(f)+\|f\|_\infty\big)\big(3\mathcal{V}_{\setX}(g)+\|g\|_\infty\big)
\end{align*}
where the two equalities hold by Lemma~\ref{lemma:V_well_def}. This ends the proof of part \ref{HK_5} of the proposition.

\noindent \textbf{Proof of   \ref{HK_6}.} Let  $b\in\setX$ be such that $\setX\cap(-\infty,b]\subsetneq \setX$. Then, we first note that $\ind_{\setX\cap(-\infty,b]}(\bar{x})=0$  since  $\setX\cap(-\infty,b]\subsetneq \setX$, and thus, by \eqref{eq:lms2}, we have
\begin{align}\label{eq:Eq_V}
\ind_{\setX\cap(-\infty,b]}(x)\leq \ind_{\setX\cap(-\infty,b]}(\bar{x})+\mathcal{V}_{\setX}(\ind_{\setX\cap(-\infty,b]})=\mathcal{V}_{\setX}(\ind_{\setX\cap(-\infty,b]}),\quad\forall x\in\setX
\end{align}
implying that $\mathcal{V}_{\setX}(\ind_{\setX\cap(-\infty,b]})\geq 1$.
 
To proceed further we first   assume   that $\setX=[\underline{x},\bar{x}]$, and we let $u\subseteq\{1,\dots,d\}$ be such that $\underline{x}_i<\bar{x}_i$ if and only if $i\in u$. Then $u\neq\emptyset$ since we are assuming that the set $\setX$ is  such that   $\setX\cap(-\infty,b]\subsetneq \setX$ for some $b\in\setX$, and let $c\in [0,1]^d$ be such that   $c_i=(b_i-\underline{x}_i)/(\bar{x}_i-\underline{x}_i)\in[0,1]$ for any $i\in u$. Then,
\begin{align}\label{eq:V1}
\mathcal{V}_{\setX}(\ind_{\setX\cap(-\infty,b]})=\mathcal{V}_{[\underline{x},\bar{x}]}(\ind_{[\underline{x},b]})=\mathcal{V}_{[\underline{x}^u,\bar{x}^u]}(\ind_{[\underline{x}^u,b^u]})=\mathcal{V}_{[0,1]^{|u|}}(\ind_{[0,c]})\geq 1
\end{align}
where the second equality holds by part \ref{basic_HK2} of Proposition~\ref{prop:order}, where the third equality holds by \citet[][Proposition 10]{owen2005multidimensional}, and where the inequality holds by \eqref{eq:Eq_V}. On the other hand, it follows from   \citet[][Corollary  4.3]{aistleitner2017functions} that $\mathcal{V}_{[0,1]^{|u|}}(\ind_{[0,c]})\leq 1$ which, together with \eqref{eq:V1}, implies that $\mathcal{V}_{\setX}(\ind_{\setX\cap(-\infty,b]})=1$ if $\setX=[\underline{x},\bar{x}]$.

Next, to show that \ref{HK_6} also holds if  the set $\setX$ is unbounded, we let $(a_n)_{n\geq 1}\in \mathcal{L}(\setX)$ and $(b_n)_{n\geq 1}\in \mathcal{U}(\setX)$ be arbitrary, and we let   $n_1\in\mathbb{N}$ be such that $a_n\leq  b\leq b_n$ for any $n\geq n_1$; observe that such an integer $n_1$ indeed exists. As proved above, for any $n\geq n_1$ we have $\mathcal{V}_{[a_n,b_n]}(\ind_{[a_n,b]})=1$, which, together with Lemma~\ref{lemma:V_well_def}, implies that  
\[
\mathcal{V}_\setX\big(\ind_{\setX\cap(-\infty,b]}\big)=\lim_{n\rightarrow\infty}\mathcal{V}_{[a_n,b_n]}(\ind_{[a_n,b]})=1.
\]

\noindent \textbf{Proof of \ref{HK_7}.}
Recall that $\mathcal{E}$ is the Borel $\sigma$-algebra on $\setX$, so that $f$ is
$(\mathcal{E}\otimes\mathcal{E})$-measurable, and note that the assumption $g(x)\in\R$ means that
$f(x,\cdot)$ is $\nu$-integrable for any $x\in\setX$. In addition, since $\nu$ is $\sigma$-finite, it
has the same null sets as some probability measure on $(\setX,\mathcal{E})$, and thus any universally
measurable subset of $\setX$ is $\bar{\nu}$-measurable. In particular, $\int_\setX\phi\,\dd\bar{\nu}$ is
well defined for any universally measurable function $\phi\colon\setX\to[0,\infty]$.

\emph{Step 1: Vitali variation on a bounded hyper-rectangle.}
Let $k\geq 1$, $[\tilde{a},\tilde{b}]\subset\R^k$ be a bounded hyper-rectangle and
$\tilde{f}\colon[\tilde{a},\tilde{b}]\times\setX\to\R$ be a
$(\mathcal{B}([\tilde{a},\tilde{b}])\otimes\mathcal{E})$-measurable function such that $\tilde{f}(z,\cdot)$
is $\nu$-integrable for any $z\in[\tilde{a},\tilde{b}]$, and let
\[
\tilde{g}(z)\coloneqq \int_\setX \tilde{f}(z,x')\,\nu(\dd x'),\qquad z\in[\tilde{a},\tilde{b}].
\]
We show that the function $x'\mapsto\tilde{\mathcal{V}}_{[\tilde{a},\tilde{b}]}(\tilde{f}(\cdot,x'))$ is
universally measurable and that
\begin{equation}\label{eq:compareVitalli}
\tilde{\mathcal{V}}_{[\tilde{a},\tilde{b}]}(\tilde{g})\leq \int_\setX \tilde{\mathcal{V}}_{[\tilde{a},\tilde{b}]}\big(\tilde{f}(\cdot,x')\big)\,\bar{\nu}(\dd x').
\end{equation}
If $[\tilde{a},\tilde{b}]$ has a side reduced to a point, then
$\tilde{\mathcal{V}}_{[\tilde{a},\tilde{b}]}(\tilde{g})=0$ and
$\tilde{\mathcal{V}}_{[\tilde{a},\tilde{b}]}(\tilde{f}(\cdot,x'))=0$ for any $x'\in\setX$ by convention,
and both claims trivially hold. We therefore assume that $\tilde{a}<\tilde{b}$.

First, for any hyper-rectangle $R=[\hat{a},\hat{b}]\subseteq[\tilde{a},\tilde{b}]$, since $\Alt(\tilde{g};R)$
is a finite linear combination of values of $\tilde{g}$ and since the functions
$\tilde{f}(z,\cdot)$ are $\nu$-integrable,
\begin{equation}\label{eq:compareVitalli1}
\Alt\big(\tilde{g};R\big)=\sum_{v\subseteq\{1,\dots,k\}}(-1)^{|v|}\int_\setX \tilde{f}\big(\langle \hat{a}^v,\hat{b}\rangle_v, x'\big)\,\nu(\dd x')
=\int_\setX \Alt\big(\tilde{f}(\cdot,x');R\big)\,\nu(\dd x').
\end{equation}

Next, we show the measurability claim. For any
$\bm{m}=(m_1,\dots,m_k)\in\mathbb{N}^k$, let
\begin{align*}
S_{\bm{m}}\coloneqq\prod_{i=1}^{k}\Big\{(s_{1,i},\dots,s_{m_i-1,i})\in\R^{m_i-1}:\
\tilde{a}_i< s_{1,i}<\dots <s_{m_i-1,i}<\tilde{b}_i\Big\},
\end{align*}
with the convention that the $i$-th factor is a set with a single element when $m_i=1$, and, for any
$s\in S_{\bm{m}}$, let
\begin{align*}
P_{\bm{m}}(s)\coloneqq\Big\{\textstyle\times_{i=1}^{k}
\big[s_{j_i,i},s_{j_i+1,i}\big]:\
j_i\in\{0,\dots,m_i-1\},\ i\in\{1,\dots,k\}\Big\},\quad\text{with }(s_{0,i}, s_{m_i,i})\coloneqq(\tilde{a}_i,\tilde{b}_i).
\end{align*}
Then $\mathsf{P}([\tilde{a},\tilde{b}])=\{P_{\bm{m}}(s):\ \bm{m}\in\mathbb{N}^k,\ s\in S_{\bm{m}}\}$. Let
$\bm{m}\in\mathbb{N}^k$. The set $S_{\bm{m}}$ is (homeomorphic to) an open subset of a Euclidean space, and is
therefore a Polish space. Moreover, the vertices of the hyper-rectangles in $P_{\bm{m}}(s)$ are continuous
functions of $s$, and $\tilde{f}$ is jointly measurable, so that the mapping
\begin{align*}
(s,x')\in S_{\bm{m}}\times\setX\mapsto \Lambda_{\bm{m}}(s,x')\coloneqq\sum_{R\in P_{\bm{m}}(s)}\big|\Alt(\tilde{f}(\cdot,x');R)\big|
\end{align*}
is $(\mathcal{B}(S_{\bm{m}})\otimes\mathcal{E})$-measurable. Hence, by \citet[][Corollary 2.13]{Crauel2003}, the function
\begin{align*}
\phi_{\bm{m}}\colon x'\in\setX\mapsto\sup_{s\in S_{\bm{m}}}\Lambda_{\bm{m}}(s,x')\in[0,\infty]
\end{align*}
is universally measurable. Since $\mathbb{N}^k$ is countable, it follows that the function
\begin{align*}
x'\in\setX\mapsto \tilde{\mathcal{V}}_{[\tilde{a},\tilde{b}]}\big(\tilde{f}(\cdot,x')\big)=\sup_{\bm{m}\in\mathbb{N}^k}\phi_{\bm{m}}(x')
\end{align*}
is universally measurable.

Finally, let $P\in\mathsf{P}([\tilde{a},\tilde{b}])$. By \eqref{eq:compareVitalli1}, and since $\bar{\nu}$
extends $\nu$,
\begin{align*}
\sum_{R\in P}\big|\Alt(\tilde{g};R)\big|
=\sum_{R\in P}\Big|\int_\setX \Alt\big(\tilde{f}(\cdot,x');R\big)\,\nu(\dd x')\Big|
&\leq\int_\setX\sum_{R\in P}\big|\Alt\big(\tilde{f}(\cdot,x');R\big)\big|\,\bar{\nu}(\dd x')\\
&\leq \int_\setX \tilde{\mathcal{V}}_{[\tilde{a},\tilde{b}]}\big(\tilde{f}(\cdot,x')\big)\,\bar{\nu}(\dd x'),
\end{align*}
and \eqref{eq:compareVitalli} follows by taking the supremum over $P\in\mathsf{P}([\tilde{a},\tilde{b}])$.

\emph{Step 2: HK-variation.}
Let $(a_n)_{n\geq 1}\in\mathcal{L}(\setX)$ and $(b_n)_{n\geq 1}\in\mathcal{U}(\setX)$ be monotone (see
Section~\ref{sub:basic_notation}), and let $n_0\geq 1$ be such that $a_n\leq b_n$ for any $n\geq n_0$,
so that
\begin{align}\label{eq:an_bn}
[a_n,b_n]\subseteq[a_{n+1},b_{n+1}]\subseteq\setX,\quad\forall n\geq n_0.
\end{align}
Let $n\geq n_0$ and $\emptyset\neq u\subseteq\{1,\dots,d\}$. For any $z\in[a^u_n,b^u_n]$ we have
$\langle z,b_n\rangle_u\in[a_n,b_n]\subseteq\setX$, and the function
$(z,x')\mapsto f(\langle z,b_n\rangle_u,x')$ is jointly measurable on $[a^u_n,b^u_n]\times\setX$, as the
composition of $f$ with a measurable mapping. Thus, Step~1 applies with $k=|u|$,
$[\tilde{a},\tilde{b}]=[a^u_n,b^u_n]$, $\tilde{f}(z,x')=f(\langle z,b_n\rangle_u,x')$ and
$\tilde{g}(z)=g(\langle z,b_n\rangle_u)$. By summing the resulting inequalities over $u$ and by using the
definition \eqref{eq:HK} of the HK-variation, we obtain that the function
$x'\mapsto \mathcal{V}_{[a_n,b_n]}\big(f(\cdot,x')\big)$ is universally measurable and that
\begin{align}\label{eq:V_lim1}
\mathcal{V}_{[a_n,b_n]}(g)\leq \int_\setX  \mathcal{V}_{[a_n,b_n]}\big(f(\cdot,x')\big)\,\bar{\nu}(\dd x'),\quad\forall n\geq n_0.
\end{align}
By \eqref{eq:an_bn} and part \ref{basic_HK1} of Proposition~\ref{prop:order}, for any $x'\in\setX$ the
sequence $\big(\mathcal{V}_{[a_n,b_n]}(f(\cdot,x'))\big)_{n\geq n_0}$ is non-decreasing and, by
Lemma~\ref{lemma:V_well_def}, it converges to $\mathcal{V}_\setX(f(\cdot,x'))$. In particular, the function
$x'\mapsto\mathcal{V}_\setX(f(\cdot,x'))$ is universally measurable, as the pointwise limit of
universally measurable functions, which proves the first assertion of part \ref{HK_7}. Finally, by
Lemma~\ref{lemma:V_well_def}, \eqref{eq:V_lim1} and the monotone convergence theorem,
\begin{align*}
\mathcal{V}_{\setX}(g)=\lim_{n\rightarrow\infty}\mathcal{V}_{[a_n,b_n]}(g)\leq \lim_{n\rightarrow\infty}\int_\setX  \mathcal{V}_{[a_n,b_n]}\big(f(\cdot,x')\big)\,\bar{\nu}(\dd x')
=\int_\setX\mathcal{V}_{\setX}\big(f(\cdot,x')\big)\,\bar{\nu}(\dd x'),
\end{align*}
which proves the second assertion of part \ref{HK_7}. The proof of the proposition is complete.
\end{proof}

 \subsubsection{Proof of Proposition~\ref{prop:Banach}\label{p-prop:Banach}}
 
 \begin{proof}
 
 We first assume that $\setX=[\underline{x},\bar{x}]$. If $\underline{x}=\bar{x}$  then  $\mathcal{V}_\setX(f)=0$ by part \ref{basic_HK3} of Proposition~\ref{prop:order}, in which case the result of the proposition trivially holds. We assume now that $\underline{x}\neq \bar{x}$ and note that, by part \ref{basic_HK2} of Proposition~\ref{prop:order}, to prove the result of the proposition in this context it is enough to consider the case where  $\underline{x}<\bar{x}$. In this scenario, the result of the proposition follows by first noting that,   for any hype-rectangle $[a,b]\subset\R^d$ such that $a<b$,  we have   \citep[][Proposition 10]{owen2005multidimensional} 
\begin{align}\label{eq:use_ab}
\mathcal{V}_{[a,b]}(f)=\mathcal{V}_{[0,1]^d}(\tilde{f}),\quad \text{where $\tilde{f}=f\circ \phi_{a,b}$  with  $\phi_{a,b}$ as defined by in \eqref{eq:phi_fn}},\quad\forall f\in \mathcal{B}_{\mathrm{HK}}([a,b]),
\end{align} 
 and then by using  \citet[][Corollary  4.3 and Lemma 5.1]{aistleitner2017functions}.

We now prove the result of the proposition in the case where the set $\setX=\R^d\cap [\underline{x},\bar{x}]$ is unbounded. By using again part \ref{basic_HK2} of Proposition~\ref{prop:order}, we can check    that to establish the result of the proposition in  this scenario it is enough to consider the case where $\underline{x}<\bar{x}$. 

To prove the proposition when $\setX$ is unbounded and $\underline{x}<\bar{x}$, let $(f_k)_{k\geq 1}$ be a Cauchy sequence in $(\mathcal{B}_{\mathrm{HK}}(\setX),\|\cdot\|_{\mathrm{HK}})$, and let $(a_n)_{n\geq 1}\in \mathcal{L}(\setX)$ and $(b_n)_{n\geq 1}\in \mathcal{U}(\setX)$ be such that $[a_n,b_n]\subseteq[a_{n+1},b_{n+1}]$ for any $n\geq 1$. Since $\|\cdot\|_\infty\leq\|\cdot\|_{\mathrm{HK}}$, the sequence $(f_k)_{k\geq 1}$ converges uniformly to some bounded function $f\colon\setX\to\R$. Let $\epsilon>0$ and let $K_\epsilon\geq 1$ be such that $\mathcal{V}_\setX(f_k-f_l)\leq\epsilon$ for any $k,l\geq K_\epsilon$. For any $n\geq 1$ and any $k,l\geq K_\epsilon$, part \ref{basic_HK1} of Proposition~\ref{prop:order} and Lemma~\ref{lemma:V_well_def} give $\mathcal{V}_{[a_n,b_n]}(f_l-f_k)\leq\mathcal{V}_\setX(f_l-f_k)\leq\epsilon$. Since $f_l-f_k\to f-f_k$ uniformly as $l\to\infty$, the lower semicontinuity of the HK-variation on bounded hyper-rectangles \citep[][Proposition~2.3 and Corollary~4.3]{aistleitner2017functions}, combined with \eqref{eq:use_ab}, yields
\begin{align*}
\mathcal{V}_{[a_n,b_n]}(f-f_k)\leq\liminf_{l\to\infty}\mathcal{V}_{[a_n,b_n]}(f_l-f_k)\leq\epsilon,\qquad\forall n\geq 1,\quad\forall k\geq K_\epsilon.
\end{align*}
Letting $n\to\infty$ and using Lemma~\ref{lemma:V_well_def}, we obtain $\mathcal{V}_\setX(f-f_k)\leq\epsilon$ for any $k\geq K_\epsilon$. Hence, by part~\ref{HK_1} of Proposition~\ref{prop:V_HK}, $\mathcal{V}_\setX(f)\leq\mathcal{V}_\setX(f_{K_\epsilon})+\epsilon<\infty$, so that $f\in\mathcal{B}_{\mathrm{HK}}(\setX)$, and $\|f-f_k\|_{\mathrm{HK}}\leq\|f-f_k\|_\infty+\epsilon$ for any $k\geq K_\epsilon$. Thus $f_k\to f$ in $(\mathcal{B}_{\mathrm{HK}}(\setX),\|\cdot\|_{\mathrm{HK}})$, which proves the completeness.
\end{proof}
 
 \subsubsection{Proof of Proposition~\ref{prop:variationHK} }\label{p-prop:variationHK}
 \begin{proof}
Let $f\in\mathcal{C}_{\mathrm{HK}}(\setX)$, let
$(b_n)_{n\geq 1}\in\mathcal{U}(\setX)$ be such that
\begin{align}\label{eq:CHK_cond}
\lim_{n\to\infty}V_{\setX\cap(-\infty,b_n]}(f)
=V_\setX(f),
\end{align}
and let $(a_n)_{n\geq 1}\in\mathcal{L}(\setX)$ be
arbitrary. Let $n_0\geq 1$ be such that
$a_n\leq b_n$ for any $n\geq n_0$; then, for any
$n\geq n_0$, $[a_n,b_n]$ is a bounded hyper-rectangle
contained in $\setX$. By
Lemma~\ref{lemma:V_well_def}, the limit
$\mathcal{V}_\setX(f)=
\lim_{n\to\infty}\mathcal{V}_{[a_n,b_n]}(f)$ exists
in $[0,\infty]$, and it therefore suffices to show
that
\begin{align}\label{eq:goal_varHK}
\lim_{n\to\infty}\mathcal{V}_{[a_n,b_n]}(f)=V_\setX(f).
\end{align}
Since $V_\setX(f)<\infty$ by definition of
$\mathcal{C}_{\mathrm{HK}}(\setX)$, this will also
show that $f\in\mathcal{B}_{\mathrm{HK}}(\setX)$.

\emph{Step 1: HK-variation on bounded
hyper-rectangles.}
Let $n\geq n_0$ and let
$\emptyset\neq u\subseteq\{1,\dots,d\}$. Since
$f\in\mathcal{C}_{\mathrm{mixed}}(\setX)$, the function
$z\mapsto f(\langle z,b_n\rangle_u)$ has a continuous
mixed derivative
$z\mapsto\partial^u f(\langle z,b_n\rangle_u)$ on
$[a_n^u,b_n^u]$. Thus, by \eqref{eq:smoothV} and \eqref{eq:HK}, 
\begin{align}\label{eq:W_n}
\mathcal{V}_{[a_n,b_n]}(f)
=\sum_{\emptyset\neq u\subseteq\{1,\dots,d\}}
\int_{[a_n^u,b_n^u]}
\big|\partial^u f\big(\langle z,b_n\rangle_u\big)\big|
\,\lambda_{|u|}(\dd z),
\end{align}
where, with the convention of
Section~\ref{sub:HK-var}, a term is zero if
$[a_n^u,b_n^u]$ has a side reduced to a point.
Similarly, since $\setX\cap(-\infty,b_n]
=\R^d\cap[\underline{x},b_n]$ is bounded above, so
that $\mathcal{I}_{\setX\cap(-\infty,b_n]}=\emptyset$,
the definition of $V_{\setX\cap(-\infty,b_n]}(f)$ gives
\begin{align}\label{eq:V_bn}
V_{\setX\cap(-\infty,b_n]}(f)
=\sum_{\emptyset\neq u\subseteq\{1,\dots,d\}}
\int_{(\setX\cap(-\infty,b_n])_u}
\big|\partial^u f\big(\langle z,b_n\rangle_u\big)\big|
\,\lambda_{|u|}(\dd z).
\end{align}

\emph{Step 2: upper bound.}
For any $n\geq n_0$ and
$\emptyset\neq u\subseteq\{1,\dots,d\}$, we have
$[a_n^u,b_n^u]\subseteq(\setX\cap(-\infty,b_n])_u$, and
the integrands in \eqref{eq:W_n} and \eqref{eq:V_bn}
are the same non-negative function. Therefore
\[
\mathcal{V}_{[a_n,b_n]}(f)\leq
V_{\setX\cap(-\infty,b_n]}(f)
\]
 for any $n\geq n_0$,
and \eqref{eq:CHK_cond} yields
\begin{align}\label{eq:upper_varHK}
\limsup_{n\to\infty}\mathcal{V}_{[a_n,b_n]}(f)
\leq V_\setX(f).
\end{align}

\emph{Step 3: lower bound.}
Let $u\subseteq\{1,\dots,d\}$ be non-empty and such
that $\mathcal{I}_\setX\subseteq u$. For
$i\notin u$ we have $i\notin\mathcal{I}_\setX$, so
that $\bar{x}_i<\infty$ and, by definition of
$\mathcal{U}(\setX)$, $b_{n,i}=\bar{x}_i$ for any
$n\geq 1$. Consequently,
\begin{align*}
\langle z,b_n\rangle_u=\langle z,\bar{x}\rangle_u,
\qquad\forall z\in\R^{|u|},\quad\forall n\geq 1,
\end{align*}
that is, the face of $[a_n,b_n]$ anchored at $b_n$
and associated with $u$ lies in the face of $\setX$
anchored at $\bar{x}$. The corresponding term of
\eqref{eq:W_n} is therefore
\begin{align*}
\int_{\R^{|u|}}\ind_{[a_n^u,b_n^u]}(z)\,
\big|\partial^u f\big(\langle z,\bar{x}\rangle_u\big)\big|
\,\lambda_{|u|}(\dd z),
\end{align*}
in which the integrand no longer depends on $n$
except through the indicator function. Now let
$z\in\setX_u$ and $i\in u$. If $\underline{x}_i>-\infty$,
then $a_{n,i}=\underline{x}_i\leq z_i$ for any $n$,
while if $\underline{x}_i=-\infty$, then
$a_{n,i}\to-\infty$; similarly, $b_{n,i}=\bar{x}_i\geq
z_i$ for any $n$ if $\bar{x}_i<\infty$, while
$b_{n,i}\to\infty$ otherwise. Thus
$\ind_{[a_n^u,b_n^u]}(z)=1$ for any $n$ large enough,
and thus $\lim_{n\to\infty}\ind_{[a_n^u,b_n^u]}(z)
=\ind_{\setX_u}(z)$ for any $z\in\setX_u$. By Fatou's
lemma,
\begin{align*}
\liminf_{n\to\infty}\int_{[a_n^u,b_n^u]}
\big|\partial^u f\big(\langle z,\bar{x}\rangle_u\big)\big|
\,\lambda_{|u|}(\dd z)
\geq\int_{\setX_u}
\big|\partial^u f\big(\langle z,\bar{x}\rangle_u\big)\big|
\,\lambda_{|u|}(\dd z).
\end{align*}
Since all the terms in \eqref{eq:W_n} are
non-negative, we may drop those for which
$\mathcal{I}_\setX\not\subseteq u$ and sum the last
inequality over the remaining sets $u$, which gives
\begin{align}\label{eq:lower_varHK}
\liminf_{n\to\infty}\mathcal{V}_{[a_n,b_n]}(f)
\geq
\sum_{\substack{\mathcal{I}_\setX\subseteq u
\subseteq\{1,\dots,d\}\\ u\neq\emptyset}}
\int_{\setX_u}
\big|\partial^u f\big(\langle z,\bar{x}\rangle_u\big)\big|
\,\lambda_{|u|}(\dd z)
=V_\setX(f).
\end{align}

\emph{Conclusion.}
Together, \eqref{eq:upper_varHK} and
\eqref{eq:lower_varHK} prove \eqref{eq:goal_varHK}.
Thus $\mathcal{V}_\setX(f)=V_\setX(f)<\infty$, and in
particular $f\in\mathcal{B}_{\mathrm{HK}}(\setX)$,
which completes the proof.
\end{proof}

\subsubsection{Proof of Proposition~\ref{prop:Multivariate_KH1}\label{p-prop:Multivariate_KH1}}

To prove Proposition~\ref{prop:Multivariate_KH1} we first need to extend  the conclusions of Lemma~\ref{lemma:aistleitner} to functions defined on an arbitrary bounded hyper-rectangle $[a,b]\subset\R^d$.
\begin{lemma}\label{lemma:discrepancy2}
Let $[a,b]\subset\R^d$  be a hyper-rectangle and $f\in \mathcal{B}_{\mathrm{HK}}([a,b])$. Then,
\begin{align*}
 |\mu_1(f)-\mu_2(f)|\leq \mathcal{V}_{[a,b]}(f)\,\|\mu_1-\mu_2\|_{\mathrm{K}},\quad\forall \mu_1,\mu_2\in\mathcal{P}([a,b]).
\end{align*}
\end{lemma}
\begin{proof}

We let  $\mu_1,\mu_2\in\mathcal{P}([a,b])$ be arbitrary   and, in a  first step, we show that the result of the lemma  holds when $a<b$. To do so, we let  $\phi_{a.b}:[0,1]^d\rightarrow[a,b]$  be as defined in \eqref{eq:phi_fn} with $\underline{x}=a$ and $\bar{x}=b$,  we let $\tilde{f}=f\circ \phi_{a,b}$,  and, for $i=1,2$,  we let $\tilde{\mu}_i$ denote the pushforward of $\mu_i$ by  $\phi_{a,b}^{-1}$.

Then, on the one hand, by \citet[][Proposition 10]{owen2005multidimensional}  we have 
\[
\mathcal{V}_{[0,1]^d}(\tilde{f})=\mathcal{V}_{[a,b]}(f)
\]
 and, on the other hand,  we have  $\|\mu_1-\mu_2\|_{\mathrm{K}}=\|\tilde{\mu}_1-\tilde{\mu}_2\|_{\mathrm{K}}$. Thus,
\begin{align*}
|\mu_1(f)-\mu_2(f)|=|\tilde{\mu}_1(\tilde{f})-\tilde{\mu}_2(\tilde{f})|\leq \mathcal{V}_{[0,1]^d}(\tilde{f}) \|\tilde{\mu}_1-\tilde{\mu}_2\|_{\mathrm{K}}=\mathcal{V}_{[a,b]}(f)\|\mu_1-\mu_2\|_{\mathrm{K}}
\end{align*}
where the first inequality holds by Lemma~\ref{lemma:aistleitner}. 

To conclude the proof of the lemma we now assume  that there exists a  set $v\subsetneq \{1,\dots,d\}$ such that $a_i=b_i$ if and only if $i\in -v$. If $v=\emptyset$ then $|\mu_1(f)-\mu_2(f)|=0$ and the result of the lemma trivially holds. We now assume that $v\neq\emptyset$  and let $f_v:[a^{v},b^{v}]\rightarrow\R$ be such that 
\[
f_v(z)=f(\langle z,b\rangle_{v})\quad\mbox{\rm for all $z\in [a^{v},b^{v}]$.}
\]
 Then $f(x)=f_v(x^{v})$ for any $x\in [a,b]$ and, letting $X\sim\mu_1$ and $Y\sim \mu_2$, we let $\mu_{1,v}$ denote the distribution of $X^{v}$ and $\mu_{2,v}$ denote the distribution of $Y^{v}$. 

 As   proved above,  if for some $k\geq 1$ the set $[\tilde{a},\tilde{b}]\subset\R^k$ is such that $\tilde{a}<\tilde{b}$, then the inequality  
\[
|\tilde{\mu}_1(g)-\tilde{\mu}_2(g)|\leq \mathcal{V}_{[\tilde{a},\tilde{b}]}(g)\|\tilde{\mu}_1-\tilde{\mu}_2\|_{\mathrm{K}}
\]
 holds for any functions  $g\in\mathcal{B}_{\mathrm{HK}}([\tilde{a},\tilde{b}])$ and any probability measures $\tilde{\mu}_1,\tilde{\mu}_2\in\mathcal{P}([\tilde{a},\tilde{b}])$. Therefore,
\begin{align*}
|\mu_1(f)-\mu_2(f)|=|\mu_{1,v}(f_{v})-\mu_{2,v}(f_{v})|\leq \mathcal{V}_{[a^{v},b^{v}]}(f_v)\|\mu_{1,v}-\mu_{2,v}\|_{\mathrm{K}}&\leq \mathcal{V}_{[a^{v},b^{v}]}(f_v)\|\mu_{1}-\mu_{2}\|_{\mathrm{K}}\\
&=\mathcal{V}_{[a,b]}(f)\|\mu_{1}-\mu_{2}\|_{\mathrm{K}}
\end{align*}
where the last equality holds by part \ref{basic_HK2} of Proposition~\ref{prop:order}. The proof of the lemma is complete.
\end{proof}

To prove Proposition~\ref{prop:Multivariate_KH1} we will also need the following result:
\begin{lemma}\label{lem:technical}
Let  $\mu\in\mathcal{P}(\setX)$ and $A\in\mathcal{E}$ be such that $\mu(A)>0$, and let $\mu_A(\dd x)=\mu(\dd x)\ind_A(x)/\mu(A)$. Then, $\sup_{B\in\mathcal{E}}\big|\mu_A(B)-\mu(B)\big|=  \mu(A^c)$.
\end{lemma}
\begin{proof}

This result is folklore; we include a proof.

Let $B\in \mathcal{E}$ and note that $\mu_A(B)-\mu(B)=b(1-a)/a-c$ where
\begin{align*}
a=\mu(A),\quad b=\mu(A\cap B),\quad c=\mu(A^c\cap B).
\end{align*}
Then, noting that $0\leq b\leq a$ and $0\leq c\leq 1-a$, we have
\begin{align*}
-(1-a)\leq -c\leq b(1-a)/a-c= \mu_A(B)-\mu(B)\leq b(1-a)/a\leq 1-a
\end{align*}
showing that $\sup_{B\in\mathcal{E}} |\mu_A(B)-\mu(B) |\leq 1-a=1-\mu(A)$. The proof of the lemma is completed upon noting that $\sup_{B\in\mathcal{E}} |\mu_A(B)-\mu(B) |\geq |\mu_A(A)-\mu(A) |=1-\mu(A)$.
\end{proof}

We now prove Proposition~\ref{prop:Multivariate_KH1}:

\begin{proof}[Proof of Proposition~\ref{prop:Multivariate_KH1}]

In what follows we let $\mu_1,\mu_2\in\mathcal{P}(\setX)$ and $f\in \mathcal{B}_{\mathrm{HK}}(\setX)$ be fixed, and we let $\epsilon\in(0,1)$ be arbitrary. Then, by Lemma~\ref{lemma:V_well_def}, since  $\mathcal{V}_\setX(f)<\infty$  there exist    $ a_\epsilon, b_\epsilon \in\setX$ such that $a_\epsilon\leq b_\epsilon$ and such that 
\begin{align}\label{eq:sum}
\mathcal{V}_{A_\epsilon}(f)\leq \mathcal{V}_\setX(f)+\epsilon,\quad \mu_1(A_\epsilon)\geq 1-\epsilon,\quad\mu_2(A_\epsilon)\geq 1-\epsilon,\quad A_\epsilon:=[a_\epsilon,b_\epsilon] 
\end{align} 
and below we let $\mu_{1,A_\epsilon}$ and $\mu_{2,A_\epsilon}$ denote the restriction of the probability measures $\mu_1$ and $\mu_2$ to the set $A_\epsilon$. 

With this notation in place, we have the decomposition
\begin{align*}
\mu_1(f)-\mu_2(f )&=\mu_{1,A_\epsilon}(f)-\mu_{2,A_\epsilon }(f)+\mu_1(f\ind_{A_\epsilon})\Big(1-\frac{1}{\mu_1(A_\epsilon)}\Big)\\
&-\mu_2(f\ind_{A_\epsilon})\Big(1-\frac{1}{\mu_2(A_\epsilon)}\Big)+\mu_1(f\ind_{A_\epsilon^c})-\mu_2(f\ind_{A_\epsilon^c})
\end{align*}
which implies that
\begin{equation}\label{eq:dis0}
\begin{split}
|\mu_1(f)-\mu_2(f)|&\leq |\mu_{1,A_\epsilon}(f)-\mu_{2,A_\epsilon}(f)|+2\|f\|_\infty\bigg( \frac{\mu_1(A_\epsilon^c)}{\mu_1(A_\epsilon)}+  \frac{\mu_2(A_\epsilon^c)}{\mu_2(A_\epsilon)}\bigg).
\end{split}
\end{equation}
Moreover, by  Lemma~\ref{lemma:discrepancy2},  
\begin{align}\label{eq:dis1}
|\mu_{1,A_\epsilon}(f)-\mu_{2,A_\epsilon}(f)|&\leq \mathcal{V}_{A_\epsilon}(f)\|\mu_{1,A_\epsilon}-\mu_{2,A_\epsilon}\|_{\mathrm{K}}
\end{align}
where, by Lemma~\ref{lem:technical},  
\begin{align}\label{eq:dis3}
\|\mu_{1,A_\epsilon}-\mu_{2,A_\epsilon}\|_{\mathrm{K}}\leq \|\mu_{1}-\mu_{2}\|_{\mathrm{K}}+\mu_1(A_\epsilon^c)+\mu_2(A_\epsilon^c).
\end{align}
By combining \eqref{eq:sum}-\eqref{eq:dis3},  we obtain that
\begin{equation}\label{eq:dis4}
\begin{split}
|\mu_1(f)-\mu_2(f)|&\leq   \mathcal{V}_\setX(f)\|\mu_{1}-\mu_{2}\|_{\mathrm{K}}+ \epsilon \Big(1+2\mathcal{V}_\setX(f)+2\epsilon +\frac{4 \|f\|_\infty}{1-\epsilon}\Big) 
\end{split}
\end{equation}
where, by Proposition~\ref{prop:set_BHK},   $\|f\|_\infty<\infty$  since $\mathcal{V}_\setX(f)<\infty$. Therefore, recalling that in \eqref{eq:dis4} the constant $\epsilon\in(0,1)$ is arbitrary, it follows that  
\[
|\mu_1(f)-\mu_2(f)|\leq \mathcal{V}_\setX(f)\|\mu_{1}-\mu_{2}\|_{\mathrm{K}},
\]
 and the proof of the Proposition is complete.    
\end{proof}

\subsection{Proofs of the results in Section~\ref{sub:assumptions_gen}}

\subsubsection{Proof of Lemma~\ref{lemma:main_Gen}\label{p-lemma:main_Gen}}

\begin{proof}

Below, we let the   functions $\varphi_1$ and $\varphi_2$ and the constant $\delta\in(0,1)$  be as  in \ref{Astar}, and we let   $\bm{\mu}=(\mu_t)_{t\geq 1}$ be an  arbitrary sequence in $\mathcal{P}(\setX)$ such that $\mu_t\in\mathcal{P}_{W_t|1}(\setX)$ for any $t\geq 1$. In addition, in what follows we use the convention that $\mu_0=\eta_1$, and  we recall that $\Phi_1$ is the identity mapping on $\mathcal{P}(\setX)$ while $\Phi_{t,t}=\Phi_1$ for any integers $t\geq 1$.

First, under \ref{Astar}, we have
\begin{equation}\label{eq:part1}
\begin{split}
 d_1\Big(\Phi_{p,t}(\mu_{p}),\Phi_{p,t}\big(\Phi_{p}(\mu_{p-1})\big) \Big)&\leq \varphi_1\Big(\mu_p(f_p), \Phi_{p}(\mu_{p-1})(f_p)  \Big) e^{-\delta(t-p)},\quad\forall t>p\geq 1
 \end{split}
\end{equation}
as well as
\begin{equation}\label{eq:part2}
\begin{split}
d_1\Big(\Phi_{p,t}(\mu_{p}),\Phi_{p,t}\big(\Phi_{p}(\mu_{p-1})\big)\Big)&\leq   \varphi_2\big(\Phi_{p}(\mu_{p-1})(f_p)\big)\, d_2\big(\mu_p,\Phi_{p}(\mu_{p-1})\big),\quad\forall t>p\geq 1.
\end{split}
\end{equation}
 
Next, by the telescoping decomposition \eqref{eq:decomp_mu}, introduced in \citet[][Chapter 7]{DelMoral:book} for studying particle filter algorithms, for any $t\geq 1$ and any function $f\in\mathcal{F}_1(\setX)$, we have 
\begin{align}\label{eq:part3}
|\mu_{t}(f)-\eta_t(f)|\leq |\mu_t(f)-\Phi_{t}(\mu_{t-1})(f)|+\sum_{p=1}^{t-1}d_1\Big(\Phi_{p,t}(\mu_{p}),\Phi_{p,t}\big(\Phi_{p}(\mu_{p-1})\big)\Big).
\end{align}

To proceed further, we let $T\geq 1$ be fixed   and we let  $\bar{\varphi}:[0,\infty)\rightarrow (0,\infty]$ be defined by
\begin{align*}
\bar{\varphi}(z)= \max\Big\{1,\varphi_1(  z,z),\, \varphi_2 ( z )\Big\},\quad z\in[0,\infty).
\end{align*}
Then, using \eqref{eq:part1} for $p\leq t-k$, \eqref{eq:part2} for the at most $k-1$ indices $t-k<p<t$, and $d_2(\mu_t,\Phi_t(\mu_{t-1}))\leq\Delta_{T,\bm{\mu}|\eta_1}$ with $\bar{\varphi}\geq 1$, for any integer  $k\geq 1$ we have  
\begin{equation}\label{eq:part4}
\begin{split}
 \sup_{1\leq t\leq T}\bigg(d_2\big(\mu_t,\Phi_t(\mu_{t-1})\big)+ &\sum_{p=1}^{t-1 }d_1\Big(\Phi_{p,t}(\mu_{p}),\Phi_{p,t}\big(\Phi_{p}(\mu_{p-1})\big)\Big)\bigg)\\
& \leq \bar{\varphi}\big(\iota_{T,\bm{\mu}|\eta_1}\big)\sup_{1\leq t\leq T} \bigg( \sum_{p=1}^{t-k}   e^{-\delta(t-p)}+k\Delta_{T,\bm{\mu}|\eta_1}\bigg)\\
&\leq  \bar{\varphi}\big(\iota_{T,\bm{\mu}|\eta_1}\big)\Big( \frac{ e^{-\delta k}}{1-e^{-\delta}}   +k\Delta_{T,\bm{\mu}|\eta_1}\Big). 
\end{split}
\end{equation}

We assume henceforth that  $0<\Delta_{T,\bm{\mu}|\eta_1}\leq 1$ (the result of the lemma being trivial if $\Delta_{T,\bm{\mu}|\eta_1}=0$), and use \eqref{eq:part4} to show that
\begin{equation}\label{eq:part44}
\begin{split}
&\sup_{1\leq t\leq T}\bigg(d_2\big(\mu_t,\Phi_t(\mu_{t-1})\big)+\sum_{p=1}^{t-1 }d_1\Big(\Phi_{p,t}(\mu_{p}),\Phi_{p,t}\big(\Phi_{p}(\mu_{p-1})\big)\Big)\bigg)\leq \varphi\big(\iota_{T,\bm{\mu}|\eta_1}\big) \vartheta\big(\Delta_{T,\bm{\mu}|\eta_1}\big)
\end{split}
\end{equation}
with the function  $\varphi:[0,\infty)\rightarrow [1,\infty]$  as defined in \eqref{eq:def_varphi}.

To do so we assume first that $\Delta_{T,\bm{\mu}|\eta_1}<\delta$. In this scenario, \eqref{eq:part44} follows by applying \eqref{eq:part4} with $k=  \lceil \log(\delta/\Delta_{T,\bm{\mu}|\eta_1})/\delta\rceil$,   noting that, for this definition of $k$, 
\begin{equation*}
\begin{split}
\frac{ e^{-\delta k}}{1-e^{-\delta}}   +k\Delta_{T,\bm{\mu}|\eta_1}& \leq \frac{1 }{\delta(1-e^{-\delta})}\Big(2\Delta_{T,\bm{\mu}|\eta_1} +\Delta_{T,\bm{\mu}|\eta_1} \log\big(\Delta^{-1}_{T,\bm{\mu}|\eta_1}\big)\Big)\\
&\leq \frac{4}{\log(2)\delta(1-e^{-\delta})\,}\vartheta\big(\Delta_{T,\bm{\mu}|\eta_1}\big).
\end{split}
\end{equation*}
On the other hand, if $\Delta_{T,\bm{\mu}|\eta_1}\geq \delta$ then \eqref{eq:part44} follows by applying \eqref{eq:part4}   with $k=1$, noting that
\begin{align*}
 \frac{e^{-\delta}}{1-e^{-\delta}}   + \Delta_{T,\bm{\mu}|\eta_1} \leq    \frac{ 2 }{\delta(1-e^{-\delta})}\,     \Delta_{T,\bm{\mu}|\eta_1}\leq \frac{2}{\log(2)\delta(1-e^{-\delta})\,}\vartheta\big(\Delta_{T,\bm{\mu}|\eta_1}\big),
\end{align*}
where the last inequality uses $z\leq\vartheta(z)/\log 2$ for $z\in[0,1]$.

Therefore, recalling that $\mathcal{F}_2(\setX)\subseteq\mathcal{F}_1(\setX)$, it follows from \eqref{eq:part3} and \eqref{eq:part44} that 
\begin{align*}
d_2\big(\mu_{t},\eta_t\big)&=\sup_{f\in\mathcal{F}_2(\setX)}|\mu_{t}(f)-\eta_t(f)|\leq \varphi\big(\iota_{T,\bm{\mu}|\eta_1}\big) \vartheta\big(\Delta_{T,\bm{\mu}|\eta_1}\big)
\end{align*}
which shows   the first assertion of the lemma.

The  second assertion of the lemma directly follows from  \eqref{eq:part3} and \eqref{eq:part44}, and the proof of the lemma is therefore complete.
\end{proof}

\subsubsection{Proof of Theorem~\ref{thm:main_Gen}\label{p-thm:main_Gen}}

\begin{proof}

Below, we let  the function $\varphi_2$ be as in \ref{Astar} and the function $\varphi:[0,\infty)\rightarrow[1,\infty]$ be as defined in \eqref{eq:def_varphi}, and we  let $\bm{\mu}=(\mu_t)_{t\geq 1}$ be an arbitrary sequence in $\mathcal{P}(\setX)$ such that $\mu_t\in\mathcal{P}_{W_t|1}(\setX)$ for any $t\geq 1$. In addition, for any   integers $T\geq 1$, we let quantities $\Delta_{T,\bm{\mu}|\eta_1}$ and $\iota_{T,\bm{\mu}|\eta_1}$ be as defined in  Section~\ref{sub:perturbation_gen}, and we let $f_0=f_1$ and $\mu_0=\eta_1$.  Finally, to simplify the notation in what follows, for any $T\geq 1$ we let $\Delta_{T}=\Delta_{T,\bm{\mu}|\eta_1,\varpi}$.

To prove the result of the theorem, we note first that
\begin{align}\label{eq:B_delta}
\Delta_{T,\bm{\mu}|\eta_1}\leq \frac{\Delta_{T}}{\inf_{1\leq t\leq T}\varpi\big(\mu_{t-1}(f_{t-1})\big)}\leq \frac{\Delta_{T}}{ \varpi\big( \iota_{T,\bm{\mu}|\eta_1}\big)},\quad\forall T\geq 1.
\end{align}
Therefore, letting $\tilde{\varphi}:[0,\infty)\rightarrow[1,\infty]$ be such that $\tilde{\varphi}(z)=\varphi(z)/\varpi(z)$ for any $z\in[0,\infty)$, and recalling that $\|\varpi\|_\infty\leq 1$, it follows from Lemma~\ref{lemma:main_Gen}  and \eqref{eq:B_delta} that, for any $T\geq 1$ such that $\Delta_{T}\leq\varpi(\iota_{T,\bm{\mu}|\eta_1})$, we have 
\begin{align}\label{eq:thm_main_e1}
\max_{1\leq t\leq T}d_2\big(\mu_t,\eta_t\big)\leq \tilde{\varphi}\big(\iota_{T,\bm{\mu}|\eta_1}\big) \vartheta\big(\Delta_{T}\big)
\end{align}
while, for any functions $f\in\mathcal{F}_1(\setX)$,
\begin{align}\label{eq:thm_main_e2}
\big|\mu_{t}(f) - \eta_t(f)\big|\leq   \big|\mu_{t}(f)-\Phi_{t}(\mu_{t-1})(f)\big|+ \tilde{\varphi}\big(\iota_{T,\bm{\mu}|\eta_1}\big)   \vartheta\big(\Delta_{T}\big),\quad\forall t\in\{1,\dots,T\}.
\end{align}

For any $t\geq 1$, we  let 
\[
\tilde{\iota}_{t}=\min_{s\in\{1,\dots,t\}}\min\{\mu_s(f_s),\Phi_{s+1}(\mu_{s})(f_{s+1})\}.
\]
 Then,   recalling that we have  $\eta_1(f_1)>0$ by assumption, and noting that the function $\tilde{\varphi}$ is non-increasing, it follows from  \eqref{eq:thm_main_e1} and \eqref{eq:thm_main_e2} that,   to prove the theorem, it is enough to establish the following result: 
\begin{align}\label{eq:thm}
\exists  (\xi,\upsilon)\in(0,\infty)^2  \text{ s.t.  for all $T\geq 1$  we have  }   \Delta_{T}\leq \xi\implies \tilde{\varphi}(\tilde{\iota}_{T})\leq \upsilon.
\end{align}
 
To show \eqref{eq:thm}, we let $\varrho=\inf_{t\geq 1}\Phi_{1,t}(\eta_1)(f_{t})$, noting that  $\varrho>0$ under \ref{Astar} since $\eta(f_1)>0$. Then,    for any $t\geq 1$, we have
\begin{equation}\label{eq:lowePhiT}
\begin{split}
\Phi_{t+1}(\mu_t)(  f_{t+1})&\geq \eta_{t+1}(f_{t+1})- d_2\big(\Phi_{t+1}(\mu_t),\eta_{t+1}\big)\\
&\geq \eta_{t+1}(f_{t+1})-\varphi_2\big(\eta_t(f_t)\big) d_2\big(\mu_t, \eta_t\big)\geq \varrho-\varphi_2(\varrho)\,d_2\big( \mu_t,\eta_t\big) 
\end{split}
\end{equation}
where the first      inequality uses the fact that $f_{t+1}\in\F_2(\setX)$    under \ref{Astar}, where the second inequality uses the Lipschitz condition \eqref{eq:lipschitz_gen} imposed on the model by \ref{Astar}, and where the last inequality uses the fact that the function $\varphi_2$ is non-increasing. We also observe that, under \ref{Astar},  
\begin{align*}
\mu_t(f_{t}) \geq  \eta_t(f_{t})- d_2(\mu_t, \eta_t),\quad\forall t\geq 1
\end{align*}
which, together with \eqref{eq:lowePhiT}, implies that
\begin{align}\label{eq:thm11}
d_2(\mu_t,\eta_t) \leq c_\varrho:=\min\Big\{1,\frac{\varrho}{2},\frac{\varrho}{2  \varphi_2(\varrho)}\Big\}\implies \min\Big\{\mu_t(f_t),\Phi_{t+1}(\mu_t)(f_{t+1})\Big\}\geq \frac{\varrho}{2 },\quad\forall t\geq 1. 
\end{align}
Since the function $\tilde{\varphi}$ is non-increasing and such that $\tilde{\varphi}(z)=\infty\Leftrightarrow z=0$, it follows from \eqref{eq:thm11} that to show \eqref{eq:thm} it is enough to prove the following result:
\begin{align}\label{eq:thm2}
\exists   \xi\in(0,1)  \text{ s.t. for all $T\geq 1$  we have }\Delta_{T}\leq  \xi\implies\, \sup_{1\leq t\leq T}d_2(\mu_t,\eta_t\big) \leq   c_\varrho.
\end{align}

To this aim, we show by induction on $T$ that \eqref{eq:thm2} holds with
\begin{align}\label{eq:CC}
\xi= \varpi(\varrho/2) \bigg(\frac{c_\varrho}{2 C_* \tilde{\varphi}(\varrho/2)}\bigg)^2 \quad\text{ where } C_*=\max\Big\{1,\sup_{x\in[0,1]} x^{1/2}\log(1+1/x)\Big\}<\infty.
\end{align}
Noting that $\Delta_{1}= \varpi(\eta_1(f_1)) d_2(\mu_1,\eta_1)$, it follows that
\begin{align*}
d_2(\mu_1,\eta_1)=\frac{\Delta_{1}}{\varpi(\eta_1(f_1))} \leq \frac{\Delta_{1}}{\varpi(\varrho)}
\end{align*}
and thus, since $\xi\leq\varpi(\varrho/2) c_\varrho\leq\varpi(\varrho) c_\varrho$, the implication in \eqref{eq:thm2} holds for $T=1$. 

Assume now that the implication in \eqref{eq:thm2} holds for $T=s$ for some $s\geq 1$, and that $\Delta_{s+1}\leq\xi$. Then $\Delta_s\leq\xi$ and, by the inductive hypothesis and \eqref{eq:thm11}, $\iota_{s+1,\bm{\mu}|\eta_1}\geq\varrho/2$, so that $\Delta_{s+1}\leq\xi\leq\varpi(\varrho/2)\leq\varpi(\iota_{s+1,\bm{\mu}|\eta_1})$ and \eqref{eq:thm_main_e1} applies with $T=s+1$. Noting that, by \eqref{eq:thm_main_e1},  
\begin{align*}
 \sup_{t\in\{1,\dots,s+1\}} d_2(\mu_t,\eta_t)&\leq \tilde{\varphi}\big(\min\big\{\tilde{\iota}_s,\eta_1(f_1)\big\}\big)\vartheta\big(\Delta_{s+1}\big)
\end{align*}
where $\eta_1(f_1)\geq \varrho$, it follows that, under the inductive hypothesis and using \eqref{eq:thm11}, we have the implication
\begin{equation}\label{eq:work1}
\begin{split}
 \Delta_{s}\leq \xi\implies \sup_{t\in\{1,\dots,s+1\}} d_2(\mu_t,\eta_t)&\leq   \tilde{\varphi}(\varrho/2)\,\vartheta\big(\Delta_{s+1}\big).
\end{split}
\end{equation}

To proceed further, observe that 
\begin{align}\label{eq:work2}
 \tilde{\varphi}(\varrho/2)\vartheta\big(\Delta_{s+1}\big)\leq c_\varrho\Leftrightarrow  \vartheta\big(\Delta_{s+1}\big)\leq \frac{c_\varrho}{\tilde{\varphi}(\varrho/2)}
\end{align}
and note that
\begin{equation}\label{eq:work3}
\Delta_{s+1}\leq \xi \implies C_*\Delta_{s+1}^{1/2}\leq \varpi(\varrho/2)^{1/2}  \frac{c_\varrho}{\tilde{\varphi}(\varrho/2)}\leq    \frac{c_\varrho}{\tilde{\varphi}(\varrho/2)}.
\end{equation}
The constant $C_*\in[1,\infty)$ defined in \eqref{eq:CC} being  such that $\vartheta\big(\Delta_{s+1}\big)\leq C_*\Delta_{s+1}^{1/2}$, it follows from \eqref{eq:work3} that the following implication holds:
\begin{equation}\label{eq:work4}
\begin{split}
\Delta_{s+1}\leq \xi&\implies  \vartheta\big(\Delta_{s+1}\big)\leq  \frac{c_\varrho}{     \tilde{\varphi}(\varrho/2)}.
\end{split}
\end{equation}
By combining \eqref{eq:work1} and \eqref{eq:work4}, we obtain that
\begin{align*}
\Delta_{s+1}\leq \xi\implies \sup_{t\in\{1,\dots,s+1\}}d_2\big(\mu_t,\eta_t\big)\leq c_\varrho
\end{align*}
showing that  the implication in \eqref{eq:thm2} holds for $T=s+1$. This ends the proof of \eqref{eq:thm2} and thus of the theorem.
\end{proof}

\subsubsection{Proof of Lemma~\ref{lemma:filter_Gen}\label{p-lemma:filter_Gen}}

\begin{proof}
Below, we let   $\mu_t$   be a sequence in $\mathcal{P}(\setX)$ such that $\mu_t\in\mathcal{P}_{W_t|1}(\setX)$ for any $t\geq 1$, and, for any $t\geq 1$, we let $\mathcal{F}_{G_t|b}(\setX)$  denote the set of measurable functions $f:\setX\rightarrow\R$   such that $f G_t\in\mathcal{B}_1(\setX)$. In addition, in what follows the constant $c_\star\in (0,\infty)$ is as in the statement of the lemma.

Let $t\geq 1$ and $f\in\mathcal{F}_{G_t|b}(\setX)$. Since
$\mu_t,\eta_t\in\mathcal{P}_{W_t}(\setX)$ and $G_t>0$ on $W_t$, the quantities
$\mu_t(G_t)$ and $\eta_t(G_t)$ are positive and, since
$G_tf\in\mathcal{B}_1(\setX)$, the quantities $\mu_t(G_tf)$ and
$\eta_t(G_tf)$ belong to $[0,1]$. Writing
\begin{align*}
\Psi_t(\mu_t)(f)-\widehat{\eta}_t(f)
=\frac{\mu_t(G_tf)}{\mu_t(G_t)}-\frac{\eta_t(G_tf)}{\eta_t(G_t)}
=\frac{\mu_t(G_tf)-\eta_t(G_tf)}{\eta_t(G_t)}
+\mu_t(G_tf)\Big(\frac{1}{\mu_t(G_t)}-\frac{1}{\eta_t(G_t)}\Big),
\end{align*}
and noting that
$\mu_t(G_tf)\big(\mu_t(G_t)^{-1}-\eta_t(G_t)^{-1}\big)
=\Psi_t(\mu_t)(f)\,\big(\eta_t(G_t)-\mu_t(G_t)\big)/\eta_t(G_t)$,
we obtain
\begin{align}\label{eq:up_filter_a}
\big|\Psi_t(\mu_t)(f)-\widehat{\eta}_t(f)\big|
\leq\frac{|\mu_t(G_tf)-\eta_t(G_tf)|
+|\Psi_t(\mu_t)(f)|\,|\mu_t(G_t)-\eta_t(G_t)|}{\eta_t(G_t)}.
\end{align}
We now bound $|\Psi_t(\mu_t)(f)|$. Since $0\leq G_tf\leq 1$,
we have $0\leq\Psi_t(\mu_t)(f)\leq 1/\mu_t(G_t)$. If, in
addition, $\|f\|_\infty\leq 1$, then $0\leq G_tf\leq G_t$, so
that $0\leq\Psi_t(\mu_t)(f)\leq 1$. Thus
\begin{align*}
|\Psi_t(\mu_t)(f)|\leq\frac{1}{D_t(f)},
\qquad
D_t(f)\coloneqq 1-\big(1-\mu_t(G_t)\big)
\ind_{(1,\infty]}(\|f\|_\infty)
=\begin{cases}1,&\|f\|_\infty\leq 1,\\
\mu_t(G_t),&\|f\|_\infty>1.\end{cases}
\end{align*}
Since $G_t\leq 1$, we have $D_t(f)\in(0,1]$, so that the first
term in the numerator of \eqref{eq:up_filter_a} is also bounded
by $|\mu_t(G_tf)-\eta_t(G_tf)|/D_t(f)$. Hence
\begin{align*}
\big|\Psi_t(\mu_t)(f)-\widehat{\eta}_t(f)\big|
\leq\frac{|\mu_t(G_tf)-\eta_t(G_tf)|+|\mu_t(G_t)-\eta_t(G_t)|}
{\eta_t(G_t)\,D_t(f)}.
\end{align*}
Multiplying the numerator and the denominator by $c_\star$, and
using $c_\star|\mu_t(G_t)-\eta_t(G_t)|
=|\mu_t(c_\star G_t)-\eta_t(c_\star G_t)|$, we obtain the first
inequality in \eqref{eq:up_filter} below. Finally, since
$c_\star G_t\in\mathcal{F}_2(\setX)$ and
$\mu_t,\eta_t\in\mathcal{P}_{W_t|1}(\setX)\subseteq\mathcal{P}_2(\setX)$,
the definition of $d_2$ gives
$|\mu_t(c_\star G_t)-\eta_t(c_\star G_t)|\leq d_2(\mu_t,\eta_t)$. Hence,
\begin{equation}\label{eq:up_filter}
\begin{split}
\big |\Psi_t(\mu_t)(f)-\widehat{\eta}_t(f)\big|&\leq \frac{c_\star|\mu_t(G_t f)-\eta_t(G_t f)|+|\mu_t(c_\star G_t)-\eta_t(c_\star G_t)|}{c_\star \eta_t(G_t)\big[1-(1-\mu_t(G_t))\ind_{(1,\infty]}(\|f\|_\infty)\big]} \\
&\leq   \frac{c_\star|\mu_t(G_t f)-\eta_t(G_t f)|+d_2(\mu_t,  \eta_t)}{c_\star \eta_t(G_t)\big[1-(1-\mu_t(G_t))\ind_{(1,\infty]}(\|f\|_\infty)\big]}
\end{split}
\end{equation}
and the first assertion of the lemma now comes from the fact that $\varrho:=\inf_{t\geq 1}\eta_t(G_t)>0$. 

To prove the second assertion of the lemma, we let $t\geq 1$ and $f:\setX\rightarrow\R$ be a measurable function such that the function $(f+a_t)G_t\in\mathcal{B}_1(\setX)$ for some constant $a_t\in\R$. In this context, and letting $h_t=f+a_t$, we have
\begin{align*}
\big |\Psi_t(\mu_t)(f)-\widehat{\eta}_t(f)\big|=\big |\Psi_t(\mu_t)(h_t)-\widehat{\eta}_t(h_t)\big|\leq   \frac{c_\star|\mu_t(G_t h_t)-\eta_t(G_t h_t)|+d_2(\mu_t,  \eta_t)}{c_\star \varrho \mu_t(G_t)}
\end{align*}
where the second inequality holds by \eqref{eq:up_filter}, noting that $h_t\in  \mathcal{F}_{G_t|b}(\setX)$ and using the fact that $\eta_t(G_t)\geq \varrho$ by assumption. Then, the second assertion of the lemma follows upon noting that,   since $c_\star G_t\in\F_2(\setX)$ by assumption, we have the implication
\begin{align*}
d_2(\mu_t,\eta_t)\leq \varrho\frac{ c_\star}{2}\implies \mu_t(G_t)\geq \eta_t(G_t)- \frac{1}{c_\star}d_2(\mu_t,\eta_t)\geq\frac{\varrho}{2}.
\end{align*}
 This ends the proof of the lemma.
\end{proof}

\subsection{Proofs of the results in Sections~\ref{sub:primitive}--\ref{sub:K_estimates}}

\subsubsection{Proof of Proposition~\ref{prop:A1p-A2}\label{p-prop:A1p-A2}}

\begin{proof}

Below,    we let the constants $(\bar{\varrho},c_{\bar{\varrho}},\varrho)\in(0,1)^3$ be as in the statement of the proposition, and we start by proving the first assertion of the proposition.
 
To this aim, we first recall that,  for any integers $t\geq s\geq 1$, any $\eta\in\mathcal{P}_{W_s}(\setX)$ and $ f\in\mathcal{B}_1(\setX)$ we have (recall that $\Phi_{s,i}(\eta)\in\mathcal{P}_{W_i}(\setX)$ for any $i\geq s$ by \ref{A0})
\begin{align}\label{eq:identity0}
\eta Q_{s,t}(f)=\Phi_{s,t}(\eta)(f)~\prod_{s\leq i<t}\Phi_{s,i}(\eta)(G_{i})
\end{align}
Since $\|G_t\|_\infty\leq 1$ by assumption, for any integers $t\geq s\geq 1$  we have
\begin{align}\label{eq:lower_G}
\min_{s\leq i\leq t}\Phi_{s,i}(\eta)(G_i)\geq \prod_{s\leq j\leq t}\Phi_{s,j}(\eta)(G_{j}). 
\end{align}

We now  let $\upsilon\in(0,1]$ and $\mu\in\mathcal{P}_{W_1}(\setX)$ be fixed and such that $\mu(C_{1,\bar{\varrho}})\geq \upsilon$. By  assumption, we have 
\[
\mu Q_{1,t+1}(1)\geq c_{\bar{\varrho}}^{t}~\mu(C_{1,\bar{\varrho}})\geq c_{\bar{\varrho}}^{t}\,\upsilon.
\]

Thus, by using \eqref{eq:identity0},  we obtain that
\begin{equation}\label{eq:lower}
\begin{split}
\prod_{1\leq i\leq t}\Phi_{1,i}(\mu)(G_i)=\mu Q_{1,t+1}(1)\geq c_{\bar{\varrho}}^{t}\,\upsilon&\implies \prod_{1\leq i\leq t}\frac{\Phi_{1,i}(\mu)(G_i)}{c_{\bar{\varrho}}}\geq \upsilon>0,\quad\forall t\geq 1. 
\end{split}
\end{equation}
Noting that 
\begin{align*}
\limsup_{t\rightarrow\infty} \Phi_{1,t}(\mu)(G_t)<c_{\bar{\varrho}} \implies \limsup_{t\rightarrow\infty}\prod_{1\leq i\leq t}\frac{\Phi_{1,i}(\mu)(G_i)}{c_{\bar{\varrho}}}=0,
\end{align*}
it follows from \eqref{eq:lower} that $\limsup_{t\rightarrow\infty}\Phi_{1,t}(\mu)(G_t)\geq c_{\bar{\varrho}}$.

We now let $\epsilon\in (0,1)$ be such that  $\varrho<\epsilon c_{\bar{\varrho}}$; note that such an $\epsilon>0$ exists since  $\varrho<c_{\bar{\varrho}}$.  As proved above, we have $\limsup_{t\rightarrow\infty}\Phi_{1,t}(\mu)(G_t)\geq c_{\bar{\varrho}}$ and thus  there exists a strictly increasing sequence $(t_k)_{k\geq 1}$  in $\mathbb{N}\setminus\{1\}$ such that $\Phi_{1,t_k}(\mu)(G_{t_k})\geq \epsilon  c_{\bar{\varrho}}$ for any $k\geq 1$ and such that  
\begin{align}\label{eq:epsilon_cdelta}
\Phi_{1,t}(\mu)(G_{t})< \epsilon c_{\bar{\varrho}},\quad\text{for any $t>1$ such that $t\not\in (t_k)_{k\geq 1}$}.
\end{align}

To proceed further we will need the following result:
\begin{align}\label{eq:implication}
\eta(G_t)\geq \ell_1\implies  \eta(C_{t,\ell_2})\geq  \frac{\ell_1-\ell_2}{1-\ell_2}>0,\quad\forall \ell_1\in(0,1),\quad\forall \ell_2\in(0,\ell_1),\quad\forall\eta\in\mathcal{P}(\setX),\quad\forall t\geq 1.
\end{align}
To result stated in \eqref{eq:implication}  can be checked by using the fact that, for any  $t\geq 1$, $\eta\in\mathcal{P}(\setX)$,   $\ell_1\in(0,1)$ and any $\ell_2\in(0,\ell_1)$, we have
\begin{align*}
\eta(G_t)=\eta(G_t\ind_{C_{t,\ell_2}})+\eta(G_t\ind_{C^c_{t,\ell_2}})\leq \eta(C_{t,\ell_2})+\ell_2(1-\eta(C_{t,\ell_2}))=(1-\ell_2)\eta(C_{t,\ell_2})+\ell_2 
\end{align*}
where
\begin{align*}
(1-\ell_2)\eta(C_{t,\ell_2})+\ell_2< \ell_1\Leftrightarrow \eta(C_{t,\ell_2})<\frac{\ell_1-\ell_2}{1-\ell_2}.
\end{align*}
Consider the parameters 
\[
\kappa=(\epsilon c_{\bar{\varrho}}-\varrho)/(1-\varrho)>0\quad\mbox{\rm and} \quad c_{\varrho,\bar{\varrho}}=\inf_{t\geq 1}\inf_{x\in C_{t,\varrho}}Q_{t+1}(x,C_{t+1,\bar{\varrho}})>0.
\]
 Then, for any   $k\geq 1$,  we have 
\begin{equation}\label{eq:lower2}
\begin{split}
\prod_{t_k\leq i<t_{k+1}}\Phi_{1,i}(\mu)(G_i) &=\prod_{t_k\leq i<t_{k+1}}\Phi_{t_k,i}\big(\Phi_{1,t_k}(\mu)\big)(G_i)=\Phi_{1,t_k}(\mu) Q_{t_k,t_{k+1}}(1)\\
&\geq  c_{\bar{\varrho}}^{t_{k+1}-t_k-1}~ c_{\varrho,\bar{\varrho}}~\Phi_{1,t_k}(\mu)(C_{t_k,\varrho})\geq  c_{\bar{\varrho}}^{t_{k+1}-t_k}~\kappa~ c_{\varrho,\bar{\varrho}}/c_{\bar{\varrho}}
\end{split}
\end{equation}
where the second equality holds by  \eqref{eq:identity0}, where the first inequality uses the fact that   
\begin{align*}
\eta Q_{t_k,t_{k+1}}(1)\geq c_{\bar{\varrho}}^{t_{k+1}-t_k-1} \eta Q_{t_k+1}(C_{t_k+1,\bar{\varrho}})\geq c_{\bar{\varrho}}^{t_{k+1}-t_k-1}\, c_{\varrho,\bar{\varrho}}~\eta(C_{t_k,\varrho}),\quad\forall\eta\in\mathcal{P}(\setX)
\end{align*}  
and where the last inequality holds by applying \eqref{eq:implication} with $\ell_1=\epsilon c_{\bar{\varrho}}$ and with $\ell_2=\varrho$,   recalling that we have $\Phi_{1,t_k}(\mu)(G_{t_k})\geq \epsilon c_{\bar{\varrho}}$.

To conclude the proof of the first assertion of the proposition, we let $\tau_k\coloneqq t_{k+1}-t_k$ for any $k\geq 1$ and show that
\begin{align}\label{eq:Delta_max}
\sup_{k\geq 1}\tau_k\leq \bar{\tau}\coloneqq 1+\frac{\log(\kappa c_{\varrho,\bar{\varrho}})}{\log(\epsilon)}.
\end{align}
To show \eqref{eq:Delta_max} assume that $\tau_{k'}>\bar{\tau}$ for some integer $k'\geq 1$. In this scenario $t_{k'+1}-t_{k'}\geq 2$ and, since by \eqref{eq:epsilon_cdelta}  we have $\Phi_{1,t}(\mu)(G_{t})< \epsilon c_{\bar{\varrho}}$ for any $t\in\{t_k'+1,\dots,t_{k'+1}-1\}$, it follows that
\begin{align}\label{eq:prod_k2}
\prod_{t_{k'}\leq i<t_{k'+1}}\frac{\Phi_{1,i}(\mu)(G_i)}{c_{\bar{\varrho}}}\leq \frac{1}{c_{\bar{\varrho}}}\prod_{t_{k'}< i<t_{k'+1}}\frac{\Phi_{1,i}(\mu)(G_i)}{c_{\bar{\varrho}}}< \frac{1}{c_{\bar{\varrho}}}~\epsilon^{\,\tau_{k'}-1}<\frac{ \kappa\, c_{\varrho,\bar{\varrho}}}{c_{\bar{\varrho}}} 
\end{align}
while, for any $k\geq 1$ by \eqref{eq:lower2},  
\begin{align}\label{eq:prod_k}
\prod_{t_k\leq i<t_{k+1}}\frac{\Phi_{1,i}(\mu)(G_i)}{c_{\bar{\varrho}}}\geq \frac{\kappa\, c_{\varrho,\bar{\varrho}}}{c_{\bar{\varrho}}}.
\end{align}
Since  \eqref{eq:prod_k2} contradicts \eqref{eq:prod_k}, we conclude that \eqref{eq:Delta_max} holds.

By using \eqref{eq:lower_G}, \eqref{eq:lower2} and \eqref{eq:Delta_max}, and recalling that $t_1\geq 2$ and that $\inf_{k\geq 1}\Phi_{1,t_k}(\mu)(G_{t_k})\geq \epsilon  c_{\bar{\varrho}}$, we obtain that
\begin{align}\label{eq:zeta_inf}
\inf_{t\geq t_1}\Phi_{1,t}(\mu)(G_t)\geq  \zeta_1\coloneqq \exp\Big(\min\Big\{\log(\epsilon\, c_{\bar{\varrho}}),\bar{\tau}\log(c_{\bar{\varrho}})+\log(\kappa c_{\varrho,\bar{\varrho}})\Big\}\Big).
\end{align}
On the other hand, by  \eqref{eq:epsilon_cdelta}
  \begin{align*}
    \prod_{1\leq i<t_1}\frac{\Phi_{1,i}(\mu)(G_i)}{c_{\bar{\varrho}}}\leq \frac{1}{c_{\bar{\varrho}}}\prod_{2\leq i<t_1}\frac{\Phi_{1,i}(\mu)(G_i)}{c_{\bar{\varrho}}}\leq c_{\bar{\varrho}}^{-1}\,\epsilon^{t_1-2}
   \end{align*}
   and thus, by using \eqref{eq:lower},  we conclude that $t_1\leq \tau:=\lfloor 2+ \log(c_{\bar{\varrho}}\upsilon)/\log(\epsilon)\rfloor $. By combining \eqref{eq:lower_G} and \eqref{eq:lower}, we obtain that
\begin{align*}
\inf_{1\leq t<t_1}\Phi_{1,t}(\mu)(G_t)\geq \zeta_2\coloneqq c_{\bar{\varrho}}^{\tau}\upsilon 
\end{align*}
which, together with \eqref{eq:zeta_inf}, shows that $\inf_{t\geq  1}\Phi_{1,t}(\mu)(G_t)\geq \zeta:=\min\{\zeta_1,\zeta_2\}$. In this context, if follows from    \eqref{eq:implication} that 
 \begin{align*}
 \inf_{t\geq 1}\Phi_{1,t}(\mu)(C_{t,\varepsilon})\geq\frac{\zeta-\varepsilon}{1-\varepsilon},\quad  \forall\varepsilon\in(0,\zeta) 
\end{align*} 
showing that the implication stated in the proposition holds e.g.~with  $\varepsilon_\upsilon=\zeta/2$ and $\kappa_{\upsilon}=\zeta/(2-\zeta)$.

To prove the second assertion of proposition, we let $D_{t}$ be a sequence in $\mathcal{E}$ such that, for some constant $\upsilon\in(0,1)$, we have
\begin{align}\label{eq:use}
\inf_{t\geq 1}\inf_{x\in C_{t,\varepsilon_\upsilon}} Q_{t+1}(x, D_{t+1})>0
\end{align}
with  $\varepsilon_\upsilon\in(0,1)$ as in the first assertion of the proposition.

Next, we let $\mu\in\mathcal{P}_{W_1}(\setX)$ be such that $\mu(D_{1})>0$ and such that $\mu(C_{1,\bar{\varrho}})\geq \upsilon$; observe that such a measure $\mu$ exists since $\upsilon<1$. Then, 
\begin{align*}
\inf_{t\geq 1}\Phi_{1,t+1}(\mu)(D_{t+1})&\geq \inf_{t\geq 1}\Phi_{1,t}(\mu)\big(Q_{t+1}(D_{t+1})\big)\\
&\geq \Big(\inf_{t\geq 1}\Phi_{1,t}(\mu)(C_{t,\varepsilon_\upsilon})\Big)\Big(\inf_{t\geq 1}\inf_{x\in C_{t,\varepsilon_\upsilon}} Q_{t+1}(x,D_{t+1})\Big)>0
\end{align*}
where the first inequality uses the fact that $\sup_{t\geq 1}\|G_t\|_\infty\leq 1$ and where the last inequality holds by \eqref{eq:use} and the result stated in the first part of the proposition. Together with the fact that we have $\mu(D_{1})>0$, this shows that $\inf_{t\geq 1}\Phi_{1,t}(\mu)(D_{t})>0$ and the proof of the proposition is complete.
\end{proof}

\subsubsection{Proof of Lemma~\ref{lemma:forgetting}\label{p-lemma:forgetting}}

\begin{proof}
As already mentioned, the proof of the lemma follows closely the proofs of Lemma 13 and Proposition 5 in \citet{MLE_SSM}.

In what follows, for any functions $f,h:\setX\rightarrow\R$  we let $(f\otimes h):\setX^2\rightarrow\R$ be the function such that 
\[
(f\otimes h)( x_1,x_2)=f(x_1)~h(x_2)
\]
 for any $(x_1,x_2)\in\bar{\setX}\coloneqq \setX^2$, and we let the constants $(\varrho,\xi)\in(0,1)\times(1,\infty)$, the measurable sets $D_t$ and the measures $\nu_{t+1}$ be as in \ref{A1}. In addition, the result of the lemma trivially holds if $t=s$, and in what  follows we let $t> s\geq 1$ and $\mu_1,\mu_2\in\mathcal{P}_{W_s}(\setX)$ be fixed.

As a first step to prove the lemma, we let  $f,h\in\mathcal{B}_1(\setX)$ and  study  the quantity
\begin{align*}
\Theta_{s,t}(f,h)\coloneqq  \mu_1 Q_{s,t}(f)~ \mu_2 Q_{s,t}(h)- \mu_1 Q_{s,t}(h)~ \mu_2 Q_{s,t}(f).
\end{align*}
To this aim, for any  $j> i\geq 1$ and  any $(x,y)\in\bar{\setX}$, we let  
\[
\bar{Q}_{i+1}\big((x,y),\dd(x',y')\big)=Q_{i+1}(x,\dd x')~ Q_{i+1}(y,\dd y')
\]
 and   
\begin{align*}
 \bar{Q}_{i,j}\big((x,y),\dd(x',y')\big)=Q_{i,j}(x,\dd x')~ Q_{i,j}(y,\dd y')=(\bar{Q}_{i+1}\cdots\bar{Q}_j)\big((x,y),\dd(x',y')\big)
 \end{align*}
so that
 \begin{align}\label{eq:barQ1}
 \Theta_{s,t}(f,h)= (\mu_1\otimes\mu_2)\bar{Q}_{s,t}\big(f\otimes h-h\otimes f\big).
 \end{align}
 
 To proceed further,    for any $i\geq 1$ we let   $C_i=C_{i,\varrho}$,  $\bar{\nu}_{i+1}=\nu_{i+1}\otimes\nu_{i+1}$   and   $\bar{C}_i=C_{i}\times C_{i}$. Then, under \ref{A1}, 
 \begin{align}\label{eq:barQ}
 \frac{1}{\xi^2}~ \bar{\nu}_{i+1}(\bar{A}\cap\bar{C}_{i+1})\leq \bar{Q}_{i+1}( \bar{x},\bar{A}\cap\bar{C}_{i+1})\leq  \xi^2 ~\bar{\nu}_{i+1}(\bar{A}\cap\bar{C}_{i+1}),\quad\forall \bar{x}\in\bar{C}_i,\quad\forall\bar{A}\in\bar{\mathcal{E}}\coloneqq \mathcal{E}\otimes\mathcal{E}.
 \end{align}
Then, by using \eqref{eq:barQ1}-\eqref{eq:barQ}  and following line-by-line the computations in  \citet[][proof of Proposition 5]{MLE_SSM}, it can be shown that
\begin{align}\label{eq:L1}
\Theta_{s,t}(f,h)=(\mu_1\otimes\mu_2)\bar{L}_{s,t}\big(f\otimes h-h\otimes f\big),\quad \bar{L}_{s,t}\coloneqq \bar{L}_{s+1}\cdots \bar{L}_t
\end{align}
where, for any $i\geq 1$, the unnormalized kernel $\bar{L}_{i+1}$ acting on $(\bar{\setX},\bar{\mathcal{E}})$ is such that
\begin{align*}
\bar{L}_{i+1}(\bar{x},\bar{A})=\bar{Q}_{i+1}(\bar{x},\bar{A})-\ind_{\bar{C}_i}(\bar{x}) \frac{1}{\xi^2}\,\bar{\nu}_{i+1}(\bar{A}\cap \bar{C}_{i+1})\geq 0,\quad \forall (\bar{x},\bar{A})\in\bar{\setX}\times \bar{\mathcal{E}}.
\end{align*}
Using \eqref{eq:barQ}, we verify that
\begin{align}\label{eq:L2}
\bar{L}_{i+1}(\bar{x},\bar{A}\cap \bar{C}_{i+1})&\leq  r_\xi\,  \bar{Q}_{i+1}(\bar{x},\bar{A}\cap \bar{C}_{i+1}),\quad r_\xi:=1- \frac{1}{\xi^4}\in(0,1),\quad\forall  (\bar{x},\bar{A})\in\bar{C}_i\times\bar{\mathcal{E}},\quad\forall i\geq 1
\end{align}
and thus, letting $\bar{\phi}=f\otimes h-h\otimes f$, it follows from \eqref{eq:L1} and \eqref{eq:L2} that
\begin{align}\label{eq:L3}
\big|\Theta_{s,t}(f,h)\big|\leq \int_{\bar{\setX}^{t-s+1}} |\bar{\phi}(\bar{x}_{t})|\,  r_\xi^{\sum_{i=s}^{t-1}\ind_{\bar{C}_i\times\bar{C}_{i+1}}(\bar{x}_i,\bar{x}_{i+1})} (\mu_1\otimes\mu_2)(\dd \bar{x}_s) \prod_{s\leq i<t} \bar{Q}_{i+1}(\bar{x}_i,\dd\bar{x}_{i+1}).
\end{align}

We now let $\beta\in (0,1)$ be arbitrary and note that, by using Lemma 17 in \citet{douc2009forgetting},  for any  elements $\{\bar{x}_i\}_{i=s}^{t}$ of $\bar{\setX}$ we have the implication 
\begin{align*}
\sum_{s\leq i<t}\ind_{\bar{C}_i\times\bar{C}_{i+1}}(\bar{x}_i,\bar{x}_{i+1})< \lfloor \beta (t-s)\rfloor\implies \sum_{s\leq i<t}\ind_{\bar{C}_i^c}(\bar{x}_i)\geq a_{\beta,s,t}\coloneqq \frac{t-s-\lfloor \beta (t-s)\rfloor}{2}.
\end{align*}
By using this latter result and \eqref{eq:L3}, we obtain that
\begin{align}\label{eq:Delta}
\big|\Theta_{s,t}(f,h)\big|&\leq r_\xi^{\lfloor \beta (t-s)\rfloor} \int_{\bar{\setX}^{t-s+1}} |\bar{\phi}(\bar{x}_{t})|  (\mu_1\otimes\mu_2)(\dd \bar{x}_s) \prod_{s\leq i<t} \bar{Q}_{i+1}(\bar{x}_i,\dd\bar{x}_{i+1})+ \Gamma_{s,t}
\end{align}
where 
\begin{align*}
\Gamma_{s,t}=\int_{\bar{\setX}^{t-s+1}} \ind\Big(\sum_{s\leq i<t}\ind_{\bar{C}_i^c}(\bar{x}_i)\geq a_{\beta,s,t}\Big) ~ (\mu_1\otimes\mu_2)(\dd \bar{x}_s) \prod_{s\leq i<t} \bar{Q}_{i+1}(\bar{x}_i,\dd\bar{x}_{i+1}).
\end{align*}

Next, for any $i\geq 1$, $\|Q_{i+1}(1)\|_\infty=\|G_i\|_\infty\leq 1$  and we let   $\bar{G}_i(\bar{x})=G_i(x_1)G_i(x_2)$ for any $\bar{x}=(x_1,x_2)\in\bar{\setX}$ and any $i\geq 1$. Then we have
\begin{equation}\label{eq:Gamma_s}
\begin{split}
\Gamma_{s,t}=&\int_{\bar{\setX}^{t-s+1}} \ind\Big(\sum_{s\leq i<t}\ind_{\bar{C}_i^c}(\bar{x}_i)\geq a_{\beta,s,t}\Big)~ \varrho^{ \sum_{s\leq i<t}\ind_{\bar{C}_i^c}(\bar{x}_i) }  ~(\mu_1\otimes\mu_2)(\dd \bar{x}_s) \prod_{s\leq i<t} \frac{\bar{Q}_{i+1}(\bar{x}_i,\dd\bar{x}_{i+1})}{\varrho^{\ind_{\bar{C}_i^c}(\bar{x}_i)}}\\
&\leq \varrho^{a_{\beta,s,t}}\int_{\bar{\setX}^{t-s+1}}    (\mu_1\otimes\mu_2)(\dd \bar{x}_s) \prod_{s\leq i<t} \frac{\bar{Q}_{i+1}(\bar{x}_i,\dd\bar{x}_{i+1})}{\varrho^{\ind_{\bar{C}_i^c}(\bar{x}_i)}}\leq \varrho^{a_{\beta,s,t}}\bigg(\sup_{i\geq 1}\sup_{\bar{x}\in\setX^2} \frac{\bar{G}_i(\bar{x})}{\varrho^{\ind_{\bar{C}_i^c}(\bar{x})}}\bigg)^{t-s}.
\end{split}
\end{equation}
To control the last term in the above equation, we  let $i\geq 1$ and $\bar{x}=(x_1,x_2)\in\bar{\setX}$ be arbitrary. Then 
\begin{align*}
\bar{x}\in\bar{C}_i\implies \frac{\bar{G}_i(\bar{x})}{\varrho^{\ind_{\bar{C}_i^c}(\bar{x})}}=\bar{G}_i(\bar{x})\leq 1 
\end{align*}
while, recalling that $x_1 \in C^c_i\implies G_i(x_1)< \varrho$,
\begin{align*}
x_1\in C^c_i\implies \frac{\bar{G}_i(\bar{x})}{\varrho^{\ind_{\bar{C}_i^c}(\bar{x})}}=\frac{G_i(x_1)G_i(x_2)}{\varrho^{\ind_{\bar{C}_i^c}(\bar{x})}}=\frac{G_i(x_1)G_i(x_2)}{\varrho}\leq \frac{G_i(x_1)}{\varrho }\leq 1.
\end{align*}
The above calculations show that
\begin{align*}
\bigg(\sup_{i\geq 1}\sup_{\bar{x}\in\setX^2} \frac{\bar{G}_i(\bar{x})}{\varrho^{\ind_{\bar{C}_i^c}(\bar{x})}}\bigg)^{t-s}\leq 1,
\end{align*}
which, together with \eqref{eq:Gamma_s}, yields  $\Gamma_{s,t}\leq \varrho^{a_{\beta,s,t}}$.  By plugging this estimate in   \eqref{eq:Delta}, we obtain that
\begin{align}\label{eq:boundDelta}
\big|\Theta_{s,t}(f,h)\big|&\leq  r_\xi^{\lfloor \beta (t-s)\rfloor} \Big(\mu_1 Q_{s,t}(f)~\mu_2 Q_{s,t}(h)+\mu_1 Q_{s,t}(h)~\mu_2 Q_{s,t}(f)\Big)+ \varrho^{a_{\beta,s,t}}
\end{align}
and this ends our analysis of the quantity $\Theta_{s,t}(f,h)$ for two arbitrary functions $f,h\in\mathcal{B}_1(\setX)$.

To prove the result of the lemma, we   let $g\in\mathcal{B}_1(\setX)$ be a  arbitrary  and observe first that
\begin{equation}\label{eq:main}
\begin{split}
|\Phi_{s,t}(\mu_1)(g)-\Phi_{s,t}(\mu_2)(g)|&=\frac{|\Theta_{s,t}(g,\ind_{\setX})|}{\mu_1 Q_{s,t}(1)\mu_2 Q_{s,t}(1)}\\
&\leq \frac{r_\xi^{\lfloor \beta (t-s)\rfloor} \big(\mu_1 Q_{s,t}(g)~\mu_2 Q_{s,t}(1)+\mu_1 Q_{s,t}(1)~\mu_2 Q_{s,t}(g)\big)+\varrho^{a_{\beta,s,t}}}{\mu_1 Q_{s,t}(1)~\mu_2 Q_{s,t}(1)}\\
&\leq 2 r_\xi^{\lfloor \beta (t-s)\rfloor}+\frac{\varrho^{a_{\beta,s,t}}}{\mu_1 Q_{s,t}(1)~\mu_2 Q_{s,t}(1)}\leq   2 r_\xi^{\lfloor \beta (t-s)\rfloor}+\frac{\varrho^{a_{\beta,s,t}}}{\mu_1(D_s)~\mu_2(D_s)\underline{q}^{\,t-s}}
\end{split}
\end{equation}
with the constant $\underline{q}\in(0,1)$ as defined in \ref{A1} and where the first inequality holds by applying  \eqref{eq:boundDelta} with $f=g$ and with $h=\ind_\setX$.

To complete the proof, we let $k=t-s$ to simplify the notation. Then
\begin{equation}\label{eq:beta}
\begin{split}
\varrho^{\frac{k-\lfloor \beta k\rfloor}{2}} \underline{q}^{-k}&=\exp\Big(-\frac{k-\lfloor \beta k\rfloor}{2}\log(1/\varrho)+k\log(1/\underline{q})\Big)\\
&\leq \exp\Big(-\frac{k-  \beta k-1}{2}\log(1/\varrho)+k\log(1/\underline{q})\Big)\\
&\leq\varrho^{-1} \exp\Big(-k\Big(\frac{(1-\beta)\log(1/\varrho)}{2}-\log(1/\underline{q})\Big)\Big)
\end{split}
\end{equation}
where 
\begin{align}\label{eq:beta2}
\frac{(1-\beta)\log(1/\varrho)}{2}-\log(1/\underline{q})>0\Leftrightarrow \log(1/\varrho)> 2 \frac{ \log(1/\underline{q})}{1-\beta}\Leftrightarrow \varrho<   \underline{q}^{\frac{2}{1-\beta}}.
\end{align}
By assumption, $\varrho<\underline{q}^2$ and thus, since in the above calculations the constant $\beta\in(0,1)$ is arbitrary, it follows from \eqref{eq:beta}-\eqref{eq:beta2} that we can choose this constant sufficiently close to  zero to ensure that we have $\varrho^{\frac{k-\lfloor \beta k\rfloor}{2}} \underline{q}^{-k}\leq \varrho^{-1} e^{-\kappa k}$ for any $k\geq 1$ and some $\kappa>0$. The result of the lemma  for the case where $t>s$ then follows from \eqref{eq:main}, and the proof of the lemma is therefore complete.
\end{proof}

\subsubsection{Proof of Lemma~\ref{lemma:lower_phi}\label{p-lemma:lower_phi}}

\begin{proof}
Let $\mu^*\in\mathcal{P}_{W_1}(\setX)$ be as in
\ref{A2}, let
\begin{equation}\label{def-epsi-ref-ap}
\epsilon\coloneqq\inf_{t\geq 1}\Phi_{1,t}(\mu^*)(D_t)>0,
\end{equation}
let $(c,\delta)$ be as in
Lemma~\ref{lemma:forgetting} and let
$\underline{q}\in(0,1)$ be as in \ref{A1}; recall that
\ref{A1} holds under \ref{A2}. We prove the following
slightly more general statement: for any $s\geq 1$
and $\mu\in\mathcal{P}_{W_s}(\setX)$ with
$\mu(D_s)>0$,
\begin{align}\label{eq:lower_st}
\inf_{t\geq s}\Phi_{s,t}(\mu)(D_t)\geq
\min\Big\{\mu(D_s)\,\underline{q}^{\,\bar{k}(\mu(D_s))},
\frac{\epsilon}{2}\Big\}>0,
\end{align}
where $\bar{k}\colon(0,1]\to\{0,1,\dots\}$ is defined by
\begin{align}\label{eq:kmu_def}
\bar{k}(z)\coloneqq\max\Big\{0,\Big\lceil\frac{1}{\delta}
\log\Big(\frac{2c}{\epsilon^2z}\Big)\Big\rceil\Big\},
\qquad z\in(0,1].
\end{align}
Observe that the first assertion of the lemma is \eqref{eq:lower_st}
with $s=1$.

Let $s\geq 1$, let $\mu\in\mathcal{P}_{W_s}(\setX)$ be
such that $\mu(D_s)>0$, and let
$\mu^*_s\coloneqq\Phi_{1,s}(\mu^*)$, so that
$\mu^*_s(D_s)\geq\epsilon$ and
$\Phi_{s,t}(\mu^*_s)=\Phi_{1,t}(\mu^*)$ for any
$t\geq s$. By Lemma~\ref{lemma:forgetting}, for any
$t\geq s$,
\begin{align*}
\Phi_{s,t}(\mu)(D_t)\geq\Phi_{s,t}(\mu^*_s)(D_t)
-\frac{c\,e^{-\delta(t-s)}}{\mu(D_s)\,\mu^*_s(D_s)}
\geq\epsilon-\frac{c\,e^{-\delta(t-s)}}{\epsilon\,\mu(D_s)},
\end{align*}
and the right-hand side is at least $\epsilon/2$ when
$t\geq s+\bar{k}(\mu(D_s))$. On the other hand, since
$\|G_i\|_\infty\leq 1$ for any $i\geq 1$, we have
$\mu Q_{s,t}(1)\leq 1$ and, by \ref{A1a},
\begin{align*}
\Phi_{s,t}(\mu)(D_t)\geq\mu Q_{s,t}(\ind_{D_t})
\geq\mu(D_s)\,\underline{q}^{\,t-s},
\qquad\forall t\geq s,
\end{align*}
which is at least
$\mu(D_s)\underline{q}^{\,\bar{k}(\mu(D_s))}$ when
$t<s+\bar{k}(\mu(D_s))$. This ends the proof of
\eqref{eq:lower_st}.

Finally, by \ref{A1a}, we have
\[
G_t(x)\geq Q_{t+1}(x,D_{t+1})\geq\underline{q}
\]
 for any
$x\in D_t$, so that
\begin{align}\label{eq:l_G}
\inf_{t\geq s}\Phi_{s,t}(\mu)(G_t)\geq
\underline{q}~\inf_{t\geq s}\Phi_{s,t}(\mu)(D_t),
\qquad\forall s\geq 1,
\end{align}
which ends the proof of the second assertion of the lemma.
\end{proof}
 
\subsubsection{Proof of Lemma~\ref{lemma:ratio}\label{p-lemma:ratio}}

\begin{proof}

In  what follows, we let  the constants $(c,\delta)\in(0,\infty)^2$ be as in Lemma~\ref{lemma:forgetting}, the constant  $\underline{q}\in(0,1)$ be as  in \ref{A1} (recall that \ref{A1} holds under \ref{A2}), the probability measure  $\mu^*\in\mathcal{P}_{W_1}(\setX)$  be as in \ref{A2}, and we let $\epsilon$ as in (\ref{def-epsi-ref-ap}).

Next, to prove the result of the lemma, we let $t\geq s\geq 1$ be two integers,  $\mu_1\in\mathcal{P}(\setX)$ be arbitrary, and $\mu_2\in\mathcal{P}_{W_s}(\setX)$ be such that $\mu_2(D_s)>0$. In addition, we let $x_s\in D_s$  and $\tilde{\mu}_1\coloneqq (\mu_1+\delta_{x_s})/2\in\mathcal{P}(\setX)$.

Then, 
\begin{align}\label{eq:identity2}
\frac{\mu_1 Q_{s,t}(1)}{\mu_2 Q_{s,t}(1)}\leq 2\frac{\tilde{\mu}_1 Q_{s,t}(1)}{\mu_2 Q_{s,t}(1)}=2\prod_{s\leq i<t}\frac{\Phi_{s,i}(\tilde{\mu}_1)(G_{i})}{\Phi_{s,i}(\mu_2)(G_{i})} 
\end{align}
where the equality holds by   \eqref{eq:identity0} and where, by Lemma~\ref{lemma:forgetting} and noting that $\tilde{\mu}_1(D_s)\geq 1/2$,
\begin{align}\label{eq:identity3}
\frac{\Phi_{s,i}(\tilde{\mu}_1)(G_{i})}{\Phi_{s,i}(\mu_2)(G_{i})}\leq 1+\frac{2c\, e^{-\delta(i-s)}}{\mu_2(D_s)\Phi_{s,i}(\mu_2)(G_{i})},\quad\forall i\geq s.
\end{align}

To proceed further, we let $\bar{k}:[0,1]\rightarrow(0,\infty]$ be as defined in \eqref{eq:kmu_def}  and $\kappa:[0,1]\rightarrow [0,\infty]$ be defined by
\begin{align*}
\kappa(z)= z \underline{q} \min\Big\{z\,\underline{q}^{\, \bar{k}(z)},\frac{\epsilon}{2} \Big\},\quad z\in (0,1].
\end{align*}
With this notation in place, by using \eqref{eq:lower_st} and \eqref{eq:l_G}, we obtain that
\begin{align*}
\mu_2(D_s)\Phi_{s,i}(\mu_2)(G_{i})\geq\kappa\big(\mu_2(D_s)\big),\quad\forall i\geq s 
\end{align*}
which,  together with    \eqref{eq:identity2} and \eqref{eq:identity3}, shows that  
\begin{align*}
\frac{\mu_1 Q_{s,t}(1)}{\mu_2 Q_{s,t}(1)}&\leq  2 \prod_{s\leq i<t}\bigg(1+\frac{2 c\, e^{-\delta(i-s)}}{\kappa(\mu_2(D_s))}\bigg) \leq 2\exp\bigg(\frac{2c  }{\kappa(\mu_2(D_s)) (1-e^{-\delta})}\bigg).
\end{align*}
The proof of the lemma is then completed upon noting that the function $\kappa$  is increasing and such that $\kappa(z)=0\Leftrightarrow z=0$.
\end{proof}

\subsubsection{Proof of Lemma~\ref{lemma:Lipschitz_tv}\label{p-lemma:Lipschitz_tv}}

\begin{proof}
In what follows, we let the function $\varphi:[0,1]\rightarrow[1,\infty]$ be as in Lemma~\ref{lemma:ratio}, and, to prove the lemma,  we let  $t\geq s\geq 1$ be two integers,  $\mu_1 \in\mathcal{P}_{W_s}(\setX)$ and $x_s\in D_s$ be arbitrary, and $ \mu_2 \in\mathcal{P}_{W_s}(\setX)$ be such that $\mu_2(D_s)>0$.

In this context, by Lemma~\ref{lemma:ratio}, we have
\begin{align}\label{eq:Bound_Qs}
Q_{s,t}(f)(x)=\delta_x Q_{s,t}(f)\leq  \delta_x Q_{s,t}(1)\leq \varphi(1)~\delta_{x_s} Q_{s,t}(1),\quad\forall x\in\setX,\quad\forall f\in\mathcal{B}_1(\setX)
\end{align}
and thus, for any functions  $f\in\mathcal{B}_1(\setX)$,
\begin{equation*}
\begin{split}
|\Phi_{s,t}(\mu_1)(f)-\Phi_{s,t}(\mu_2)(f)|&\leq \frac{|\mu_1 Q_{s,t}(f)-\mu_2 Q_{s,t}(f)|+|\mu_1 Q_{s,t}(1)-\mu_2 Q_{s,t}(1)|}{\mu_2 Q_{s,t}(1)  }\\
&\leq 2\varphi(1) \|\mu_1-\mu_2\|_{\mathrm{tv}}\,\frac{\delta_{x_s}Q_{s,t}(1)}{\mu_2 Q_{s,t}(1)}\leq 2\varphi(1) \|\mu_1-\mu_2\|_{\mathrm{tv}}\,\varphi(\mu_2(D_s))
\end{split}
\end{equation*}
where the last inequality holds by Lemma~\ref{lemma:ratio}. The proof of the lemma is complete.

\end{proof}

\subsubsection{Proof of Lemma~\ref{lemma:Lipschitz}\label{p-lemma:Lipschitz}}

To prove Lemma~\ref{lemma:Lipschitz} we will need the following result:

\begin{lemma}\label{lemma:Q_smooth}
Assume that condition \ref{B1} holds. Then, for any integers $t> s\geq 1$ and any functions $f\in\mathcal{B}_1(\setX)$,  
\begin{align*}
 \mathcal{V}_\setX\big(Q_{s,t}(f)\big)\leq \int_{\setX} Q_{s+1,t}(1)(x')~ \mathcal{V}_\setX\big(q_{s+1}(\cdot ,x')\big)\bar{\varsigma}_{s+1}(\dd x')
\end{align*}
with  the functions $(q_{i+1})_{i\geq 1}$ and the measures $(\bar{\varsigma}_{i+1})_{i\geq 1}$ as in \ref{B1}.  
\end{lemma}
\begin{proof}

In what follows, we let $t> s\geq 1$ and $f\in\mathcal{B}_1(\setX)$ be arbitrary. Then,  noting that
\begin{align*}
Q_{s,t}(f)(x)=\int_{\setX} Q_{s+1,t}(f)(x') ~q_{s+1}(x,x')\bar{\varsigma}_{s+1}(\dd x'),\quad\forall x\in\setX,
\end{align*}
it follows that
\begin{align*}
\mathcal{V}_\setX\big(Q_{s,t}(f)\big) \leq \int_{\setX}\mathcal{V}_\setX\Big(Q_{s+1,t}(f)(x')~ q_{s+1}(\cdot,x')\Big)\bar{\varsigma}_{s+1}(\dd x')&=\int_{\setX}Q_{s+1,t}(f)(x') \mathcal{V}_\setX\big(q_{s+1}(\cdot,x')\big)\bar{\varsigma}_{s+1}(\dd x')\\
&\leq \int_{\setX}Q_{s+1,t}(1)(x') \mathcal{V}_\setX\big(q_{s+1}(\cdot,x')\big)\bar{\varsigma}_{s+1}(\dd x') 
\end{align*}
where the first inequality and the equality  holds parts \ref{HK_7} and \ref{HK_2} of Proposition~\ref{prop:V_HK}, respectively. The proof of the lemma is complete.
\end{proof}

\begin{remark}\label{rem:Q_smooth}
Applying the result of Lemma~\ref{lemma:Q_smooth} with $t=s+1$ and $f\equiv 1$ shows that, under \ref{B1}, we have $\sup_{t\geq 1}\mathcal{V}_\setX(G_t)<\infty$.
\end{remark}

We now prove Lemma~\ref{lemma:Lipschitz}:

\begin{proof}[Proof of Lemma~\ref{lemma:Lipschitz}]

In what follows, we let  the functions $q_{t+1}$ and the measures $\lambda_{t+1}$ be as in \ref{B1}, the constant $\underline{q}\in(0,1)$   be as in \ref{A1} (recall that \ref{A1} holds under \ref{A2}), and we let the function  $\varphi:[0,1]\rightarrow[1,\infty]$ be as in Lemma~\ref{lemma:ratio}. In addition, we let $(\mathcal{B}_{\mathrm{HK}}(\setX),\|\cdot\|_{\mathrm{HK}})$ be the Banach space of functions defined in Section~\ref{sub:HK_gen}.

We now let $t> s\geq 1$ be two integers. Then, under \ref{B1} and by Lemma~\ref{lemma:Q_smooth}, for any functions $f\in\mathcal{B}_1(\setX)$ we have  $Q_{s,t}(f)\in \mathcal{B}_{\mathrm{HK}}(\setX)$. Next, to prove the result of the lemma, we let $\mu_1 \in\mathcal{P}_{W_s}(\setX)$ and $x_s\in D_s$ be arbitrary, and we let $\mu_2\in\mathcal{P}_{W_s}(\setX)$ be such that $\mu_2(D_s)>0$.  In addition, for  any integers $i\geq 1$,  we  let $\zeta_{i+1}\coloneqq 1+\int_\setX  \mathcal{V}_\setX\big(q_{i+1}(\cdot,x')\big)\bar{\varsigma}_{i+1}(\dd x')$ and $\pi_{i+1}\in\mathcal{P}(\setX)$ be the probability measure  defined by
\begin{align*}
\pi_{i+1}(A)&=\frac{\delta_{x_s}(A)+\int_A \mathcal{V}_\setX\big(q_{i+1}(\cdot,x')\big)\bar{\varsigma}_{i+1}(\dd x')}{\zeta_{i+1}},\quad A\in\mathcal{E}.
\end{align*}

With this notation in place, by using \eqref{eq:Bound_Qs} and Lemma~\ref{lemma:Q_smooth}, we have
\begin{align*}
\|Q_{s,t}(f)\|_{\mathrm{HK}}\leq   \varphi(1)~\delta_{x_s}Q_{s,t}(1) +(1+\zeta_{s+1}) ~\pi_{s+1} Q_{s+1,t}(1)<\infty,\quad\forall f\in\mathcal{B}_1(\setX)
\end{align*}
and thus, for any functions $f\in\mathcal{B}_1(\setX)$,
\begin{equation}\label{eq:KK1}
\begin{split}
|\Phi_{s,t}(\mu_1)(f)-\Phi_{s,t}(\mu_2)(f)|&\leq \frac{|\mu_1 Q_{s,t}(f)-\mu_2 Q_{s,t}(f)|+|\mu_1 Q_{s,t}(1)-\mu_2 Q_{s,t}(1)|}{\mu_2 Q_{s,t}(1)  }\\
&\leq 2\big(\varphi(1)+ 1+\zeta_{s+1}\big)  \|\mu_1-\mu_2\|_{\mathcal{V}_{\setX}}\bigg(\frac{\delta_{x_s}Q_{s,t}(1)}{\mu_2 Q_{s,t}(1)}+\frac{  \pi_{s+1} Q_{s+1,t}(1)}{\mu_2 Q_{s,t}(1)}\bigg)\\
&\leq 2\big(\varphi(1)+ 1+\zeta_{s+1}\big)  \|\mu_1-\mu_2\|_{\mathcal{V}_{\setX}}\bigg(\varphi\big(\mu_2(D_s)\big)+\frac{  \pi_{s+1} Q_{s+1,t}(1)}{\mu_2 Q_{s,t}(1)}\bigg)
\end{split}
\end{equation}
where the last inequality holds by Lemma~\ref{lemma:ratio}.

Moreover, $\mu_2(G_s)=\mu_2 Q_{s+1}(1)\geq \underline{q}\mu_2(D_s)$, and thus
\begin{equation}\label{eq:KKK1}
\begin{split}
\frac{ \pi_{s+1} Q_{s+1,t}(1)}{\mu_2 Q_{s,t}(1)}&=\frac{ \pi_{s+1} Q_{s+1,t}(1)}{\mu_2(G_s)\Phi_{s+1}(\mu_2) Q_{s+1,t}(1)}\\
&\leq\frac{1}{\underline{q}\,\mu_2(D_s)}  \frac{ \pi_{s+1} Q_{s+1,t}(1)}{\Phi_{s+1}(\mu_2) Q_{s+1,t}(1)}
\leq \frac{\varphi\big(\Phi_{s+1}(\mu_2)(D_{s+1})\big)}{\underline{q}\,\mu_2(D_s)}
\end{split}
\end{equation}
where the second inequality holds by Lemma~\ref{lemma:ratio}.

Finally,       $\zeta:=\sup_{i\geq 1}\zeta_{i+1}<\infty$  under \ref{B1}, while, since $\|G_s\|_\infty\leq 1$, we have   
\[
\Phi_{s+1}(\mu_2)(D_{s+1})\geq \mu_2 Q_{s+1}(D_{s+1})\geq \underline{q}\mu_2(D_s).
\]
  Therefore,  recalling that the  function $\varphi$   is  non-increasing, by using \eqref{eq:KK1} and \eqref{eq:KKK1} we obtain that
\begin{align*}
\|\Phi_{s,t}(\mu_1) -\Phi_{s,t}(\mu_2)\|_{\mathrm{tv}}\leq 2\big(\varphi(1)+ 1+\zeta\big)  \|\mu_1-\mu_2\|_{\mathcal{V}_{\setX}}\bigg(\varphi\big(\mu_2(D_s)\big)+\frac{ \varphi\big(\underline{q} \mu_2(D_s)\big)}{\underline{q} \mu_2(D_s)}\bigg) 
\end{align*}
and the result of the lemma follows.
\end{proof}

\subsection{Proofs of the results in Section~\ref{sec:examples}\label{p-prop:example}}

Throughout this appendix, $G_t$ and
$M_{t+1}$ are as defined in
\eqref{eq:Gt_def}--\eqref{eq:Mt_def2}, and we let
\begin{align}\label{eq:qt_def}
q_{t+1}(x,x')\coloneqq G_t(x)\,m_{t+1}(x,x'),
\qquad x,x'\in\R^d,\quad t\geq 1.
\end{align}
When \ref{Obs1} holds, we let
$x_t\coloneqq H_t^{-1}(y_t)$ for any $t\geq 1$ and
\begin{align}\label{eq:Kt_def}
K_{t,\delta}\coloneqq[x_t-\delta,\,x_t+\delta]
=\{x\in\R^d:\ \|x-x_t\|_\infty\leq\delta\},
\qquad\delta\in(0,\infty),\quad t\geq 1.
\end{align}
Finally, for $k\geq 1$ and $\nu>0$, we let
$c_{k,\nu}\coloneqq\int_{\R^k}(1+\|u\|^2)^{-(k+\nu)/2}
\lambda_k(\dd u)<\infty$.

\subsubsection{Technical lemmas}\label{sub:diff_opp}

\begin{lemma}\label{lemma:immersion}
Let $r\geq d\geq 1$, $\nu\in(0,\infty)$, $L\in[1,\infty)$,
let $f\colon\R^d\to\R^r$ be $L$-bi-Lipschitz, and let
$\phi\colon[0,\infty)\to[0,\infty)$ be a measurable
function such that
\[
C_\phi\coloneqq\sup_{z\geq 0}\left(\phi(z)(1+z)^{(d+\nu)/2}\right)
<\infty.
\]
 Then, for any $y\in f(\R^d)$:
\begin{enumerate}[label=(\roman*)]
\item\label{i-lemma:immersion} $\int_{\R^d}\phi\big(\|y-f(x)\|^2\big)\,
\lambda_d(\dd x)\leq C_\phi\,L^d\,c_{d,\nu}$;
\item\label{ii-lemma:immersion} if $d>1$, then, for any
$(b_n)_{n\geq 1}\in\mathcal{U}(\R^d)$ and any non-empty
sets $u\subsetneq\{1,\dots,d\}$,
\begin{align*}
\lim_{n\to\infty}\int_{\R^{|u|}}
\phi\Big(\big\|y-f\big(\langle z,b_n\rangle_u\big)\big\|^2\Big)
\,\lambda_{|u|}(\dd z)=0.
\end{align*}
\end{enumerate}
\end{lemma}

\begin{proof}
Let $x_y\coloneqq f^{-1}(y)$. Since $f$ is
$L$-bi-Lipschitz,
$\|y-f(x)\|\geq\|x-x_y\|/L$ for any $x\in\R^d$, and,
since $z\mapsto(1+z)^{-(d+\nu)/2}$ is decreasing,
\begin{align}\label{eq:immersion_bound}
\phi\big(\|y-f(x)\|^2\big)
\leq C_\phi\Big(1+\frac{\|x-x_y\|^2}{L^2}\Big)^{-\frac{d+\nu}{2}},
\qquad\forall x\in\R^d.
\end{align}
Part~\ref{i-lemma:immersion} follows by integrating \eqref{eq:immersion_bound}
and using the change of variables $x=x_y+Lu$. For
part~\ref{ii-lemma:immersion}, let $z\in\R^{|u|}$ and
$R_n\coloneqq\|b_n^{-u}-x_y^{-u}\|$. Then we have
\[
\|\langle z,b_n\rangle_u-x_y\|^2=\|z-x_y^u\|^2+R_n^2,
\]
so that, by \eqref{eq:immersion_bound},
\begin{align*}
\phi\Big(\big\|y-f\big(\langle z,b_n\rangle_u\big)\big\|^2\Big)
\leq C_\phi\Big(1+\frac{\|z-x_y^u\|^2+R_n^2}{L^2}\Big)^{-\frac{d+\nu}{2}}
\leq C_\phi\Big(1+\frac{\|z-x_y^u\|^2}{L^2}\Big)^{-\frac{d+\nu}{2}}.
\end{align*}
The right-hand side is integrable on $\R^{|u|}$ since
$d+\nu>|u|$, and the middle term converges to zero for
any $z\in\R^{|u|}$ since $R_n\to\infty$ (recall that
$b_{n,i}\to\infty$ for any $i$ and that $-u\neq\emptyset$).
Part~\ref{ii-lemma:immersion} follows by dominated convergence.
\end{proof}

\begin{lemma}\label{lemma:derivative}
Let $k_1,k_2\geq 1$ be integers, and let $\varepsilon\in(0,1]$
and $C_1\in[1,\infty)$. Then, there exists a constant
$C_2\in(0,\infty)$, depending only on
$(k_1,k_2,\varepsilon,C_1)$, such that the following
holds. Let $g\colon[0,\infty)\to[0,\infty)$ be $k_1$
times continuously differentiable and
$\psi\colon[0,\infty)\to(0,\infty)$ be
$\varepsilon$-pseudo-decreasing, such that
\begin{align}\label{eq:bound_g_fn}
(1+z)^{i/2}\,\big|g^{(i)}(z)\big|\leq\psi(z),
\qquad\forall z\in[0,\infty),\quad
\forall i\in\{1,\dots,k_1\}.
\end{align}
In addition, let
$\varphi\colon\R^{k_1}\to[1,\infty)$, let
$h\in\mathcal{C}_{\mathrm{mixed}}(\R^{k_1};k_2,\varphi)$
and let $A\in\mathcal{S}^+_{k_2}$ be such that
\[
1/C_1\leq\sigma_{\min}(A)\leq\sigma_{\max}(A)\leq C_1.
\]
Then, the function
$f(v,w)\coloneqq g\big((w-h(v))^\top A(w-h(v))\big)$,
$(v,w)\in\R^{k_1}\times\R^{k_2}$  satisfies, for any
non-empty set $u\subseteq\{1,\dots,k_1\}$,
\begin{align*}
\big|\partial_1^uf(v,w)\big|
\leq C_2\,\varphi(v)^{|u|}\,
\psi\big(\sigma_{\min}(A)\|w-h(v)\|^2\big),
\qquad\forall(v,w)\in\R^{k_1}\times\R^{k_2}.
\end{align*}
\end{lemma}

\begin{proof}
Fix $w\in\R^{k_2}$, let $e(v)\coloneqq w-h(v)$ and $Q(v)\coloneqq e(v)^\top Ae(v)$ for any $v\in\R^{k_1}$,
so that $f(\cdot,w)=g\circ Q$; all derivatives are taken with respect to $v$, and
$\partial^\emptyset e\coloneqq e$. Since $\sigma_{\min}(A)\|e\|^2\leq Q\leq C_1\|e\|^2$ and $C_1\geq1$, we
have $1+\|e\|\leq C_1^{1/2}(1+Q^{1/2})$. Let $B\subseteq\{1,\dots,k_1\}$ be non-empty. By the Leibniz rule,
$\partial^BQ=\sum(\partial^{B_1}e)^\top A\,\partial^{B_2}e$, where the sum is over the $2^{|B|}$ ordered
pairs $(B_1,B_2)$ of disjoint sets such that $B_1\cup B_2=B$. Since $\|\partial^{B'}e\|\leq k_2^{1/2}\varphi$
for any non-empty set $B'$, $\|A\|\leq C_1$ and $\varphi\geq 1$ (and since $|B|\geq2$ when $B_1$ and $B_2$ are
both non-empty), each term of this sum is bounded by $C_1k_2\varphi^{|B|}(1+\|e\|)$, and thus
\begin{align}\label{eq:dQ_bound}
|\partial^BQ(v)|\leq 2^{|B|}C_1^{3/2}k_2\,\varphi(v)^{|B|}\big(1+Q(v)^{1/2}\big),\qquad\forall v\in\R^{k_1}.
\end{align}
In particular, $Q\in\mathcal{C}_{\mathrm{mixed}}(\R^{k_1})$. By the Fa\`a di Bruno formula for mixed
partial derivatives \citep{hardy2006combinatorics}, which follows from the chain and product rules by
induction on $|u|$, we have $f(\cdot,w)\in\mathcal{C}_{\mathrm{mixed}}(\R^{k_1})$ and, for any non-empty
set $u\subseteq\{1,\dots,k_1\}$,
\[
\partial^u_1f(v,w)=\sum_{\pi\in\Pi(u)}g^{(|\pi|)}\big(Q(v)\big)\prod_{B\in\pi}\partial^BQ(v),\qquad\forall v\in\R^{k_1},
\]
where $\Pi(u)$ denotes the set of partitions of $u$ into non-empty blocks and $|\pi|$ the number of
blocks of $\pi$. Let $\pi\in\Pi(u)$ and $j\coloneqq|\pi|\leq|u|\leq k_1$. Since $\sum_{B\in\pi}|B|=|u|$ and
$(1+z^{1/2})^j\leq 2^j(1+z)^{j/2}$ for any $z\geq0$, it follows from \eqref{eq:dQ_bound} and
\eqref{eq:bound_g_fn} that the term associated with $\pi$ is bounded by
\[
2^{|u|}\big(C_1^{3/2}k_2\big)^j\varphi(v)^{|u|}\big(1+Q(v)^{1/2}\big)^j\big|g^{(j)}\big(Q(v)\big)\big|
\leq\big(4C_1^{3/2}k_2\big)^{k_1}\varphi(v)^{|u|}\,\psi\big(Q(v)\big).
\]
Finally, $\Pi(u)$ has at most $\mathrm{B}_{k_1}$ elements, where $\mathrm{B}_{k_1}$ denotes the Bell number,
and
\begin{align*}
\psi(Q(v))\leq\varepsilon^{-1}\psi(\sigma_{\min}(A)\|e(v)\|^2)
\end{align*}
 since
$\sigma_{\min}(A)\|e(v)\|^2\leq Q(v)$ and $\psi$ is $\varepsilon$-pseudo-decreasing. Hence, the result of
the lemma holds with $C_2\coloneqq\varepsilon^{-1}\mathrm{B}_{k_1}(4C_1^{3/2}k_2)^{k_1}$.
\end{proof}

\begin{lemma}\label{lemma:Prelim1_A2}
Assume that \ref{Model1} and \ref{Obs1} hold. Then
$\inf_{t\geq 1}\|G_t\|_\infty\geq\zeta_G(0)$ and, for
any $\varrho\in(0,\zeta_G(0))$, there exist constants
$0<\delta_1\leq\delta_2<\infty$ such that, for any $t\geq 1$, we have
\[
K_{t,\delta_1}\subseteq C_{t,\varrho}
\subseteq K_{t,\delta_2}.
\]

\end{lemma}

\begin{proof}
Let $t\geq 1$, $x\in\R^d$ and
$Z_t(x)\coloneqq(y_t-H_t(x))^\top\Sigma_{G_t}^{-1}
(y_t-H_t(x))$, so that $G_t(x)=\zeta_{G_t}(Z_t(x))$.
Since $y_t=H_t(x_t)$, conditions \ref{mat}--\ref{lips}
give
\begin{align}\label{eq:Zt_bounds}
L^{-3}\|x-x_t\|^2\leq Z_t(x)\leq L^3\|x-x_t\|^2.
\end{align}
In particular, $G_t(x_t)=\zeta_{G_t}(0)\geq\zeta_G(0)$
by \ref{zeta}, which ends the proof of the first assertion. Now
let $\varrho\in(0,\zeta_G(0))$. If
$x\in C_{t,\varrho}$, then, by \ref{upper_GtMt},
$\varrho\leq L(1+Z_t(x))^{-(d+\nu_G)/2}$, so that
$Z_t(x)\leq z_\varrho\coloneqq(L/\varrho)^{2/(d+\nu_G)}$
and, by \eqref{eq:Zt_bounds}, we have
\[
\|x-x_t\|_\infty\leq\|x-x_t\|\leq
\delta_2\coloneqq(L^3z_\varrho)^{1/2}.
\]
 Conversely,
since $\zeta_G$ is continuous and
$\zeta_G(0)>\varrho$, there exists $z^\star>0$ such
that $\zeta_G\geq\varrho$ on $[0,z^\star]$. If
$x\in K_{t,\delta_1}$ with
$\delta_1\coloneqq(z^\star/(dL^3))^{1/2}$, then
$Z_t(x)\leq L^3d\delta_1^2=z^\star$ by
\eqref{eq:Zt_bounds}, and hence
$G_t(x)\geq\zeta_G(Z_t(x))\geq\varrho$ by \ref{zeta}.
Finally, $\delta_1\leq\delta_2$ since
$\emptyset\neq K_{t,\delta_1}\subseteq K_{t,\delta_2}$.
\end{proof}

\begin{lemma}\label{lemma:Prelim1_A1}
Assume that \ref{Model1} and \ref{Obs} hold, and that
\begin{align}\label{eq:omega_def}
\omega(\delta)\coloneqq\sup_{t\geq 1}\,
\sup_{x\in K_{t,\delta}}
\big\|F_{t+1}(x)-F_{t+1}(x_t)\big\|<\infty,
\qquad\forall\delta\in(0,\infty).
\end{align}
Then, for any $\delta,\delta'\in(0,\infty)$,
\begin{align}\label{eq:lower_xixi}
\xi_{\delta,\delta'}\coloneqq\inf_{t\geq 1}\,
\inf_{(x,x')\in K_{t,\delta}\times K_{t+1,\delta'}}
q_{t+1}(x,x')>0.
\end{align}
In addition, \eqref{eq:omega_def} holds under each of the
four conditions on $F_{t+1}$ in
Propositions~\ref{prop:example}
and~\ref{prop:example2}.
\end{lemma}

\begin{proof}
Let $t\geq 1$ and
$(x,x')\in K_{t,\delta}\times K_{t+1,\delta'}$. By
\eqref{eq:Zt_bounds}, $Z_t(x)\leq L^3d\delta^2$, so
that, by \ref{zeta} and since $\zeta_G$ is
$\varepsilon$-pseudo-decreasing,
$G_t(x)\geq\varepsilon\,\zeta_G(L^3d\delta^2)$. Next,
by \ref{Obs2} and \eqref{eq:omega_def},
\begin{align*}
\|x'-F_{t+1}(x)\|
\leq\|x'-x_{t+1}\|+\|x_{t+1}-F_{t+1}(x_t)\|
+\|F_{t+1}(x_t)-F_{t+1}(x)\|
\leq R\coloneqq d^{1/2}\delta'+\Xi+\omega(\delta).
\end{align*}
Since $\sigma_{\min}(\Sigma_{M_{t+1}})\geq 1/L$, we have $\big(x'-F_{t+1}(x)\big)^\top
\Sigma_{M_{t+1}}^{-1}\big(x'-F_{t+1}(x)\big)\leq LR^2$ so that, by \ref{zeta},
\begin{align*}
\zeta_{M_{t+1}}\Big(\big(x'-F_{t+1}(x)\big)^\top
\Sigma_{M_{t+1}}^{-1}\big(x'-F_{t+1}(x)\big)\Big)\geq \varepsilon\,\zeta_M(LR^2).
\end{align*}
 In addition, by \ref{mat} and \ref{upper_GtMt}, we have
$C_{M_{t+1}}\det(\Sigma_{M_{t+1}})^{1/2}\leq
L^{1+d/2}c_{d,\nu_M}$ and thus
\begin{align*}
q_{t+1}(x,x')\geq
\frac{\varepsilon^2\,\zeta_G(L^3d\delta^2)\,
\zeta_M(LR^2)}{L^{1+d/2}\,c_{d,\nu_M}}>0,
\end{align*}
which ends the proof of \eqref{eq:lower_xixi}.

To show the last assertion of the lemma, note first that
\[
\|x-x_t\|\leq d^{1/2}\delta\quad\mbox{\rm for all
$x\in K_{t,\delta}$.}
\]
Then, under the condition of part~\ref{i-prop:example}  
of Proposition~\ref{prop:example},
$\omega(\delta)\leq\varphi(d^{1/2}\delta)$ while,  under
\eqref{eq:bound_derive_ht0}, since
$\|H_t(x_t)\|=\|y_t\|\leq\Xi$ and $\varphi$ is
coordinate-wise non-decreasing, we have
$\omega(\delta)\leq\varphi(d^{1/2}\delta,\Xi)$. Under
the condition of part~\ref{ii-prop:example}   of
Proposition~\ref{prop:example}, $F_{t+1}$ is
$(dL)$-Lipschitz, so that the condition of part~~\ref{i-prop:example} of the proposition 
holds. Finally, under \eqref{eq:bound_derive_ht}, condition
\eqref{eq:bound_derive_ht0} holds by
Lemma~\ref{lemma:tech_th} below.
\end{proof}

\begin{lemma}\label{lemma:tech_th}
Assume that \ref{Model} holds and that, for any
$t\geq 1$,
$F_{t+1}\in\mathcal{C}_{\mathrm{mixed}}(\R^d;d,\infty)$
satisfies \eqref{eq:bound_derive_ht}. Then,
\eqref{eq:bound_derive_ht0} holds with the
coordinate-wise non-decreasing function
\begin{align*}
\varphi(z_1,z_2)\coloneqq 2L\,d\,z_1\big(1+z_2+Lz_1\big)^{\alpha},
\qquad z_1,z_2\in[0,\infty),
\qquad\alpha\coloneqq\frac{\nu_G(1-1/L)}{d}.
\end{align*}
\end{lemma}

\begin{proof}
Let $t\geq 1$ and $v,w\in\R^d$, and let
$v_s\coloneqq v+s(w-v)$ for $s\in[0,1]$. Since $H_t$
is $L$-Lipschitz, we have
\[
\|H_t(v_s)\|\leq\|H_t(v)\|+L\|w-v\|,
\]
 and hence, by
\eqref{eq:bound_derive_ht},
\begin{align*}
\big\|F_{t+1}(w)-F_{t+1}(v)\big\|
\leq\|w-v\|\sup_{s\in[0,1]}\|\nabla F_{t+1}(v_s)\|
\leq 2L\,d\,\|w-v\|\big(1+\|H_t(v)\|+L\|w-v\|\big)^{\alpha},
\end{align*}
where we used that the operator norm of the Jacobian
matrix of $F_{t+1}$ at $x$ is at most
$d\max_{i}\|\partial^{\{i\}}F_{t+1}(x)\|$, and that
$1+b^\alpha\leq 2(1+b)^\alpha$ for $b\geq 0$.
\end{proof}

\subsubsection{Proof of Propositions~\ref{prop:example}
and~\ref{prop:example2}}

Propositions~\ref{prop:example}
and~\ref{prop:example2} follow from
Lemma~\ref{lemma:proofA1}, which establishes \ref{A3},
and Lemmas~\ref{lemma:Prop1_showB1}
and~\ref{lemma:Prop1_showB1_ex2}, which establish
\ref{B1}.

\begin{lemma}\label{lemma:proofA1}
Assume that \ref{Model1} and \ref{Obs} hold, and that
the functions $F_{t+1}$ satisfy one of the
four conditions in Propositions~\ref{prop:example}
and~\ref{prop:example2}. Let $\delta\in(0,\infty)$ and
$D_t\coloneqq K_{t,\delta}$ for any $t\geq 1$. Then:
\begin{enumerate}[label=(\roman*)]
\item \ref{A1a} holds with
$\underline{q}\coloneqq\min\{1/2,\,\xi_{\delta,\delta}(2\delta)^d\}$;
\item for every $\ell\in(0,\zeta_G(0))$, there exists a constant
$\xi_\ell\in(1,\infty)$ such that \ref{A1b} holds at
level $\ell$ with the constant $\xi_\ell$ and the measure
$\nu_{t+1}=\lambda_d$ for any $t\geq 1$;
\item for any $0<\ell_2\leq\ell_1<\zeta_G(0)$,
\begin{align*}
\inf_{t\geq 1}\inf_{x\in C_{t,\ell_2}}
Q_{t+1}(x,C_{t+1,\ell_1})>0
\qquad\text{and}\qquad
\inf_{t\geq 1}\inf_{x\in C_{t,\ell_2}}
Q_{t+1}(x,D_{t+1})>0.
\end{align*}
\end{enumerate}
Consequently, \ref{A3} holds with the sets
$D_t=K_{t,\delta}$.
\end{lemma}

\begin{proof}
By Lemma~\ref{lemma:Prelim1_A1},
$\xi_{\delta,\delta'}>0$ for any
$\delta,\delta'\in(0,\infty)$. Part~(i) of the lemma then follows from
\[
Q_{t+1}(x,D_{t+1})=\int_{K_{t+1,\delta}}
q_{t+1}(x,x')\,\dd x'\geq\xi_{\delta,\delta}(2\delta)^d,
\qquad x\in K_{t,\delta}.
\]

To prove (ii) let $\ell\in(0,\zeta_G(0))$. By
Lemma~\ref{lemma:Prelim1_A2}, there exists a
$\kappa\in(0,\infty)$ such that
$C_{t,\ell}\subseteq K_{t,\kappa}$ for any $t\geq 1$.
Thus, for any $t\geq 1$, $x\in C_{t,\ell}$ and
$A\in\mathcal{E}$,
\[
Q_{t+1}(x,A\cap C_{t+1,\ell})\geq\xi_{\kappa,\kappa}\,
\lambda_d(A\cap C_{t+1,\ell}).
\]
 On the other hand,
since $G_t\leq 1$ and, by \ref{mat}--\ref{upper_GtMt}, we have
\[
m_{t+1}(x,x')\leq L\big(C_{M_{t+1}}
\det(\Sigma_{M_{t+1}})^{1/2}\big)^{-1}\leq L^{1+d/2}/c_M,\quad\forall x,x'\in\setX
\]
with $c_M\coloneqq\int_{\R^d}\zeta_M(\|u\|^2)\,\dd u$, it follows that 
\[
Q_{t+1}(x,A\cap C_{t+1,\ell})\leq
(L^{1+d/2}/c_M)\,\lambda_d(A\cap C_{t+1,\ell}).
\]
Part~(ii) then holds with
$\xi_\ell\coloneqq\max\{2,1/\xi_{\kappa,\kappa},
L^{1+d/2}/c_M\}$.

To prove part (iii), let $0<\ell_2\leq\ell_1<\zeta_G(0)$. By
Lemma~\ref{lemma:Prelim1_A2}, there exist constants
$\kappa_1,\kappa_2\in(0,\infty)$ such that for any
$t\geq 1$ we have
\[
K_{t,\kappa_1}\subseteq C_{t,\ell_1}\quad\mbox{\rm and}\quad
C_{t,\ell_2}\subseteq K_{t,\kappa_2}.
\]
 Thus, for $x\in C_{t,\ell_2}$, we have
\[
Q_{t+1}(x,C_{t+1,\ell_1})\geq
Q_{t+1}(x,K_{t+1,\kappa_1})\geq
\xi_{\kappa_2,\kappa_1}(2\kappa_1)^d\quad\mbox{\rm and}\quad
Q_{t+1}(x,D_{t+1})\geq\xi_{\kappa_2,\delta}(2\delta)^d.
\]

Finally, by applying parts (i) and (ii) of the lemma with
$\ell<\min\{\zeta_G(0),\underline{q}^2\}$, it follows that \ref{A1}
holds with the sets $D_t=K_{t,\delta}$. In addition, by part (iii) of the lemma,
Proposition~\ref{prop:A1p-A2} applies with any
$\bar{\varrho}\in(0,\zeta_G(0))$ and with the sets
$D_t=K_{t,\delta}$, which yields \ref{A2}. Since
$\mathcal{V}_{\R^d}(\ind_{K_{t,\delta}})\leq 3^d-1$ by
part~\ref{HK_4} of Proposition~\ref{prop:V_HK},
\ref{A3} holds. This ends the proof of the lemma.
\end{proof}

\begin{lemma}\label{lemma:Prop1_showB1}
Assume that \ref{Model1} and \ref{Obs1} hold, that
$q_{t+1}(\cdot,x')\in\mathcal{C}_{\mathrm{mixed}}(\R^d)$
for any $t\geq 1$ and $x'\in\R^d$, and that there
exist constants $C\in(0,\infty)$ and $\nu\in(0,\infty)$
such that, with $\psi(z)\coloneqq(1+z)^{-(d+\nu)/2}$,
\begin{align}\label{eq:ineq_qt}
\big|\partial^u_1q_{t+1}(x,x')\big|
\leq C\,\psi\big(\|y_t-H_t(x)\|^2\big)\,
\psi\big(\|x'-F_{t+1}(x)\|^2\big)
\end{align}
for any $t\geq 1$, any $x,x'\in\R^d$, and any non-empty sets
$u\subseteq\{1,\dots,d\}$. Then, \ref{B1} holds with
$\varsigma_{t+1}=\lambda_d$ for any $t\geq 1$.
\end{lemma}

\begin{proof}
Let $t\geq 1$ and $x'\in\R^d$. Since
$\setX=\R^d$, we have
$\mathcal{I}_\setX=\{1,\dots,d\}$, so that
\[
V_\setX\big(q_{t+1}(\cdot,x')\big)=\int_{\R^d}
|\partial_1^{\{1,\dots,d\}}q_{t+1}(x,x')|\,\dd x.
\]
Then, by using Tonelli's theorem, \eqref{eq:ineq_qt} and by part \ref{i-lemma:immersion} of
Lemma~\ref{lemma:immersion},  
\begin{align*}
\int_{\R^d}V_\setX\big(q_{t+1}(\cdot,x')\big)\,\dd x'
\leq C\int_{\R^d}\psi\big(\|y_t-H_t(x)\|^2\big)\,\dd x
\int_{\R^d}\psi\big(\|u\|^2\big)\,\dd u
\leq C\,L^d\,c_{d,\nu}^2.
\end{align*}
It remains to show that
$\mathcal{V}_\setX(q_{t+1}(\cdot,x'))=
V_\setX(q_{t+1}(\cdot,x'))$, for which, by
Proposition~\ref{prop:variationHK}, it suffices to
show that $q_{t+1}(\cdot,x')\in
\mathcal{C}_{\mathrm{HK}}(\setX)$. Let
$(b_n)_{n\geq 1}\in\mathcal{U}(\R^d)$. By using \eqref{eq:ineq_qt} and the fact that $\psi\leq 1$, as well as part \ref{i-lemma:immersion} of Lemma~\ref{lemma:immersion}, we can apply the dominated convergence theorem to obtain that, in the
definition of $V_{\setX\cap(-\infty,b_n]}(q_{t+1}(\cdot,x'))$,
the term associated with $u=\{1,\dots,d\}$ converges
to $V_\setX(q_{t+1}(\cdot,x'))$ as $n\to\infty$

If $d>1$ and
$\emptyset\neq u\subsetneq\{1,\dots,d\}$, then, since
$\psi\leq 1$, \eqref{eq:ineq_qt} gives
\begin{align*}
\int_{\R^{|u|}}\big|\partial^u_1q_{t+1}
\big(\langle z,b_n\rangle_u,x'\big)\big|\,\dd z
\leq C\int_{\R^{|u|}}\psi\Big(\big\|y_t-H_t
\big(\langle z,b_n\rangle_u\big)\big\|^2\Big)\,\dd z,
\end{align*}
which converges to zero as $n\rightarrow\infty$, by part \ref{ii-lemma:immersion} of
Lemma~\ref{lemma:immersion}. Thus
$\lim_{n\to\infty}V_{\setX\cap(-\infty,b_n]}
(q_{t+1}(\cdot,x'))=V_\setX(q_{t+1}(\cdot,x'))<\infty$,
and the proof of the lemma is complete.
\end{proof}

\begin{lemma}\label{lemma:Prop1_showB1_ex2}
Assume that \ref{Model} and \ref{Obs1} hold, and that
either (a) $F_{t+1}\in\mathcal{C}_{\mathrm{mixed}}
(\R^d;d,L)$ for any $t\geq 1$, or (b)
$\sup_{t\geq 1}\|y_t\|<\infty$ and, for any $t\geq 1$, the function 
$F_{t+1}\in\mathcal{C}_{\mathrm{mixed}}(\R^d;d,\infty)$
satisfies \eqref{eq:bound_derive_ht}. Then
\eqref{eq:ineq_qt} holds with $\nu=\min\{\nu_G,\nu_M\}$
in case~(a) and with $\nu=\min\{\nu_G/L,\nu_M\}$ in
case~(b), so that \ref{B1} holds.
\end{lemma}

\begin{proof}
Let $\psi_G(z)\coloneqq L(1+z)^{-(d+\nu_G)/2}$ and
$\psi_M(z)\coloneqq L(1+z)^{-(d+\nu_M)/2}$, which are
non-increasing. Let $t\geq 1$ and let $u$ be a
non-empty subset of $\{1,\dots,d\}$. By the Leibniz
rule,
\begin{align*}
\partial^u_1q_{t+1}(x,x')=\sum_{v\subseteq u}
\partial^vG_t(x)\;\partial_1^{u\setminus v}m_{t+1}(x,x').
\end{align*}
By \ref{upper_GtMt} and \ref{cond_zeta}, and
Lemma~\ref{lemma:derivative} applied with $g=\zeta_{G_t}$,
$\psi=\psi_G$, $h=H_t$, $\varphi\equiv L$, $w=y_t$ and
$A=\Sigma_{G_t}^{-1}$, together with
$\sigma_{\min}(\Sigma_{G_t}^{-1})\geq 1/L$ and $\psi_G(z/L)\leq L^{(d+\nu_G)/2}\psi_G(z)$ for any $z\geq0$, there
exists a constant $C$ independent of $t$ such that
$|\partial^vG_t(x)|\leq C\psi_G(\|y_t-H_t(x)\|^2)$ for
any $v\subseteq u$. Similarly, since $m_{t+1}$ is a
bounded multiple of $\zeta_{M_{t+1}}$, applying
Lemma~\ref{lemma:derivative} with $g=\zeta_{M_{t+1}}$,
$\psi=\psi_M$, $h=F_{t+1}$, $w=x'$,
$A=\Sigma_{M_{t+1}}^{-1}$ (using similarly that $\psi_M(z/L)\leq L^{(d+\nu_M)/2}\psi_M(z)$) and
$\varphi(x)\coloneqq\max\{1,\max_{\emptyset\neq u'}
\|\partial^{u'}F_{t+1}(x)\|\}$ gives,  for any
$v'\subseteq u$, the estimate
\[
|\partial_1^{v'}m_{t+1}(x,x')|\leq C\varphi(x)^{d}
\psi_M(\|x'-F_{t+1}(x)\|^2).
\]
 Thus
\begin{align}\label{eq:ineq_qt_ex2}
\big|\partial^u_1q_{t+1}(x,x')\big|
\leq C\,\varphi(x)^d\,\psi_G\big(\|y_t-H_t(x)\|^2\big)\,
\psi_M\big(\|x'-F_{t+1}(x)\|^2\big).
\end{align}
In case~(a), $\varphi\leq L$, and \eqref{eq:ineq_qt}
follows from \eqref{eq:ineq_qt_ex2} with
$\nu=\min\{\nu_G,\nu_M\}$. In case~(b), let
$\alpha\coloneqq\nu_G(1-1/L)/d$ and
$\Xi'\coloneqq\sup_{t\geq 1}\|y_t\|$. By
\eqref{eq:bound_derive_ht} and since
$\|H_t(x)\|\leq\Xi'+\|y_t-H_t(x)\|$, we have
\[
\varphi(x)\leq C(1+\|y_t-H_t(x)\|)^{\alpha}.
\]
 Since
$d\alpha=\nu_G-\nu_G/L$, for any $z\geq 0$ we have
\[
(1+z^{1/2})^{d\alpha}\psi_G(z)\leq
C(1+z)^{-(d+\nu_G/L)/2}, 
\]
 and
\eqref{eq:ineq_qt} follows from
\eqref{eq:ineq_qt_ex2} with $\nu=\min\{\nu_G/L,\nu_M\}$.
In both cases, $q_{t+1}(\cdot,x')\in
\mathcal{C}_{\mathrm{mixed}}(\R^d)$ under the
assumptions of the lemma, and \ref{B1} follows from
Lemma~\ref{lemma:Prop1_showB1}.
\end{proof}

\subsubsection{Proof of Remark~\ref{rem:constant}\label{p-rem:constant}}

\begin{proof}
Throughout this subsection, $\mathfrak{P}$ denotes the
list of quantities on which the constants are allowed
to depend in Remark~\ref{rem:constant}, that is,
\[
\mathfrak{P}=(d,r,\varepsilon,\nu_G,\nu_M,L,\Xi,
\zeta_G,\zeta_M,\varphi,a,\varpi)
\]
in case~\ref{rem1}, and the same list without $\varphi$ in
case~\ref{rem2}. We say that a constant is
\emph{admissible} if it can be chosen as a function
of $\mathfrak{P}$ only, and that a function is
admissible if it is determined by $\mathfrak{P}$. The
proof has two steps. We first observe that the proofs
of the results of Section~\ref{sec:main} use the model
only through a few structural constants, and remain
valid when these constants are replaced by bounds
(Step~1). We then show that, for the models of the
remark, admissible bounds are available for all these
structural constants (Step~2). We use the sets
$D_t=K_{t,1}$, with $K_{t,\delta}$ as in
\eqref{eq:Kt_def}.

\paragraph{Step 1: structural constants.}
Assume that \ref{A2} holds for some sets
$D_t$. We call \emph{structural constants}
any numbers $\underline{q}\in(0,1)$,
$\xi\in(1,\infty)$ and $\varrho\in(0,\underline{q}^2)$
for which \ref{A1} holds, together with a number
$\epsilon_\star\in(0,1]$ such that
\[
\inf_{t\geq 1}\Phi_{1,t}(\mu^\star)(D_t)\geq\epsilon_\star
\]
for some $\mu^\star\in\mathcal{P}_{W_1}(\setX)$. In the
setting \eqref{eq:scenarioB}, we add two numbers
$V_D,Z\in[1,\infty)$ such that
\[
1+\sup_{t\geq 1}\mathcal{V}_\setX(\ind_{D_t})\leq V_D\quad\mbox{\rm
and}\quad
\sup_{t\geq 1}\big(1+\int_\setX\mathcal{V}_\setX(q_{t+1}(\cdot,x'))
\,\bar{\varsigma}_{t+1}(\dd x')\big)\leq Z 
\]
with the measures $\varsigma_{t+1}$ as in \ref{B1}. 

An inspection of
the proofs of Appendices~\ref{p-lemma:main_Gen}
to~\ref{p-lemma:Lipschitz} gives the following facts.

(a) \emph{Lemma~\ref{lemma:forgetting}.} In its proof,
the constant $\beta\in(0,1)$ only needs to
satisfy $\varrho<\underline{q}^{2/(1-\beta)}$; since
$\varrho<\underline{q}^2$, one can take
$\beta\coloneqq\frac{1}{2}\big(1-2\log\underline{q}/\log\varrho\big)$.
With $r_\xi\coloneqq 1-\xi^{-4}$ and
\[
\kappa\coloneqq\frac{1}{2}(1-\beta)\log(1/\varrho)
-\log(1/\underline{q})>0,
\]
the results stated in \eqref{eq:main} and
\eqref{eq:beta} then give the lemma with
\begin{align*}
\delta_{\mathrm{f}}\coloneqq\min\Big\{\frac{1}{2},\,
\beta\log(1/r_\xi),\,\kappa\Big\}
\qquad\text{and}\qquad
c_{\mathrm{f}}\coloneqq\frac{2}{r_\xi}+\frac{1}{\varrho},
\end{align*}
since $2r_\xi^{\lfloor\beta k\rfloor}\leq
(2/r_\xi)e^{-\beta k\log(1/r_\xi)}$ and
$\mu_1(D_s)\mu_2(D_s)\leq 1$. 

(b) \emph{Lemma~\ref{lemma:lower_phi}.} Its proof shows
that, for any $s\geq 1$ and
$\mu\in\mathcal{P}_{W_s}(\setX)$,
\begin{align}\label{eq:lower_quant}
\inf_{t\geq s}\Phi_{s,t}(\mu)(D_t)\geq\ell\big(\mu(D_s)\big)
\quad\text{and}\quad
\inf_{t\geq s}\Phi_{s,t}(\mu)(G_t)\geq\underline{q}\,
\ell\big(\mu(D_s)\big),
\qquad
\ell(z)\coloneqq\min\Big\{z\,\underline{q}^{\,\bar{k}(z)},
\frac{\epsilon_\star}{2}\Big\},
\end{align}
with $\bar{k}$ as in \eqref{eq:kmu_def}, computed with
$(c_{\mathrm{f}},\delta_{\mathrm{f}},\epsilon_\star)$.
The function $\ell$ is non-decreasing on $[0,1]$
(because $\bar{k}$ is non-increasing) and positive on
$(0,1]$.

(c) \emph{Lemmas~\ref{lemma:ratio},
\ref{lemma:Lipschitz_tv} and~\ref{lemma:Lipschitz}.}
The functions constructed in their proofs are
explicit functions of
$(\underline{q},c_{\mathrm{f}},\delta_{\mathrm{f}},\epsilon_\star)$,
and, for Lemma~\ref{lemma:Lipschitz}, of $Z$.
Consequently, \ref{Astar} holds with
$\delta=\delta_{\mathrm{f}}$ and functions
$\varphi_1,\varphi_2$ that only depend on the
structural constants: in the setting
\eqref{eq:scenarioA}, with $f_t=\ind_{D_t}$,
$\varphi_1(z,z')=c_{\mathrm{f}}/(zz')$ and $\varphi_2$
given by Lemma~\ref{lemma:Lipschitz_tv}; in the setting
\eqref{eq:scenarioB}, with 
\[
f_t=V_D^{-1}\ind_{D_t},\qquad
\varphi_1(z,z')=c_{\mathrm{f}}/(V_D^2zz')\quad\mbox{\rm and}\quad
\varphi_2(z)
\]
 given by Lemma~\ref{lemma:Lipschitz}
evaluated at $V_Dz$.

(d) \emph{Theorem~\ref{thm:main_Gen}.} Its proof
produces the constants
$c_{\eta_1,\varpi}=\xi$ in \eqref{eq:CC} and
$C_{\eta_1,\varpi}=\tilde{\varphi}(\varrho_\eta/2)$,
where $\varrho_\eta\coloneqq\inf_{t\geq 1}\eta_t(f_t)$.
This quantity is only used through the inequalities
$\eta_t(f_t)\geq\varrho_\eta$, $t\geq 1$, and $\xi$ is
non-decreasing, and $\tilde{\varphi}(\cdot/2)$
non-increasing, in $\varrho_\eta$. The proof therefore
remains valid with $\varrho_\eta$ replaced by any
positive lower bound on $\inf_{t\geq 1}\eta_t(f_t)$.
By \eqref{eq:lower_quant} with $s=1$ and $\mu=\eta_1$,
and since $\ell$ is non-decreasing, such a lower bound
is $\ell(a)$ in the setting \eqref{eq:scenarioA} and
$\ell(a)/V_D$ in the setting \eqref{eq:scenarioB},
if $\eta_1(D_1)\geq a$. In
Theorems~\ref{thm:main_TV} and~\ref{thm:main}, the
condition $\Delta_{T,\bm{\mu}|\eta_1,\varpi}\leq c_{\eta_1,\varpi}$ is removed by
replacing $C_{\eta_1,\varpi}$ with
$\max\{C_{\eta_1,\varpi},1/\vartheta(c_{\eta_1,\varpi})\}$,
since $d_2(\mu_1,\mu_2)\leq 1$ for all $\mu_1,\mu_2\in\mathcal{P}(\setX)$.

(e) \emph{Lemma~\ref{lemma:filter_Gen}.} Its proof
shows that its constants only depend on $c_\star$ and
on a positive lower bound on
$\inf_{t\geq 1}\eta_t(G_t)$, which by
\eqref{eq:lower_quant} can be taken equal to
$\underline{q}\,\ell(a)$. In the setting
\eqref{eq:scenarioA}, $c_\star=1$. In the setting
\eqref{eq:scenarioB}, since
$\mathcal{V}_\setX(G_t)\leq Z$ by
Lemma~\ref{lemma:Q_smooth} (applied with $t=s+1$ and
$f\equiv 1$) and $\|G_t\|_\infty\leq 1$, we have
$\|G_t\|_{\mathrm{HK}}\leq 1+Z$, and one can take
$c_\star=(1+Z)^{-1}$. The same applies to
Propositions~\ref{prop:filter_TV} and~\ref{prop:filter},
whose constants are those of
Lemma~\ref{lemma:filter_Gen}.

In summary, all the constants of
Theorems~\ref{thm:main_TV} and~\ref{thm:main} and of
Propositions~\ref{prop:filter_TV}
and~\ref{prop:filter}, as well as a lower bound on
$\inf_{t\geq 1}\eta_t(G_t)$, are determined by
$(\underline{q},\xi,\varrho,\epsilon_\star,\varpi,a)$,
and also by $(V_D,Z)$ in the setting \eqref{eq:scenarioB}.
It remains to show that admissible structural
constants exist.

\paragraph{Step 2: admissible structural constants.}
Let the assumptions of the remark hold. We use the
notation introduced at the beginning of Appendix~\ref{p-prop:example}, in particular the points
$x_t=H_t^{-1}(y_t)$, the sets $K_{t,\delta}$ defined in \eqref{eq:Kt_def} and the constants $c_{k,\nu}$;
note that $E_t=K_{t,1}$ for any $t\geq 1$.

\emph{Lower bounds on the kernel.} By the proof of
Lemma~\ref{lemma:Prelim1_A1}, for any
$\delta,\delta'\in(0,\infty)$,
\begin{align*}
\xi_{\delta,\delta'}\geq\underline{\xi}(\delta,\delta')
\coloneqq\frac{\varepsilon^2\,\zeta_G(L^3d\delta^2)\,
\zeta_M\big(L\,R(\delta,\delta')^2\big)}
{L^{1+d/2}\,c_{d,\nu_M}},
\qquad
R(\delta,\delta')\coloneqq d^{1/2}\delta'+\Xi
+\bar{\omega}(\delta),
\end{align*}
where $\bar{\omega}(\delta)$ is the bound on the quantity
$\omega(\delta)$ defined in \eqref{eq:omega_def} obtained at the end of that proof, namely:
\begin{itemize}
\item $\bar{\omega}(\delta)\coloneqq\varphi(d^{1/2}\delta)$ under the assumptions of
part~\ref{i-prop:example} of Proposition~\ref{prop:example};
\item $\bar{\omega}(\delta)\coloneqq\varphi(d^{1/2}\delta,\Xi)$ under the assumptions of
part~\ref{i-prop:example2} of Proposition~\ref{prop:example2};
\item $\bar{\omega}(\delta)\coloneqq d^{3/2}L\,\delta$ under the assumptions of
part~\ref{ii-prop:example} of Proposition~\ref{prop:example}, since the functions $F_{t+1}$ are then
$(dL)$-Lipschitz;
\item $\bar{\omega}(\delta)\coloneqq 2L\,d^{3/2}\delta\,\big(1+\Xi+Ld^{1/2}\delta\big)^{\alpha}$, with
$\alpha\coloneqq\nu_G(1-1/L)/d$, under the assumptions of
part~\ref{ii-prop:example2} of Proposition~\ref{prop:example2}, by Lemma~\ref{lemma:tech_th}.
\end{itemize}
In case~\ref{rem1}, the first two functions depend on the function $\varphi$, which belongs to
$\mathfrak{P}$, while, in case~\ref{rem2}, the last two functions only depend on $(d,L,\nu_G,\Xi)$. Thus,
in both cases, $\underline{\xi}$ is an admissible function.

\emph{Level sets.} By the proof of
Lemma~\ref{lemma:Prelim1_A2}, for any $\ell\in(0,1]$
and $t\geq 1$, $C_{t,\ell}\subseteq K_{t,\delta_2(\ell)}$
with $\delta_2(\ell)\coloneqq
\big(L^3(L/\ell)^{2/(d+\nu_G)}\big)^{1/2}$ (the proof
of this inclusion does not use $\ell<\zeta_G(0)$), and,
for $\ell\in(0,\zeta_G(0))$,
$K_{t,\delta_1(\ell)}\subseteq C_{t,\ell}$ with
$\delta_1(\ell)\coloneqq(z^\star_\ell/(dL^3))^{1/2}$ and
$z^\star_\ell\coloneqq\sup\{z\geq 0:\
\zeta_G\geq\ell\text{ on }[0,z]\}>0$. Both functions
are admissible.

\emph{Constants in \ref{A1}.} By the proof of
Lemma~\ref{lemma:proofA1} with $\delta=1$, \ref{A1}
holds with the sets $D_t=K_{t,1}$ and with the admissible constants
\begin{align*}
\underline{q}\coloneqq\min\big\{1/2,\,
2^d\,\underline{\xi}(1,1)\big\},
\qquad
\varrho\coloneqq\frac{1}{2}\min\big\{\zeta_G(0),
\underline{q}^2\big\},
\qquad
\xi\coloneqq\max\Big\{2,\,
\frac{1}{\underline{\xi}(\delta_2(\varrho),\delta_2(\varrho))},\,
\frac{L^{1+d/2}}{c_M}\Big\},
\end{align*}
where $c_M\coloneqq\int_{\R^d}\zeta_M(\|u\|^2)\,\dd u$,
and with $\nu_{t+1}=\lambda_d$ for any $t\geq 1$.

\emph{The constant $\epsilon_\star$.} We apply
Proposition~\ref{prop:A1p-A2}, with a constant $\varrho'$ defined below in place of the constant
$\varrho$ of its statement. Its proof only uses the
two constants $c_{\bar{\varrho}}$ and $c_{\varrho',\bar{\varrho}}$ appearing in it through the lower
bounds $Q_{t+1}(x,C_{t+1,\bar{\varrho}})\geq
c_{\bar{\varrho}}$ and
$Q_{t+1}(x,C_{t+1,\bar{\varrho}})\geq
c_{\varrho',\bar{\varrho}}$, for $x$ in
$C_{t,\bar{\varrho}}$ and $C_{t,\varrho'}$
respectively, and through the condition $\varrho'<c_{\bar{\varrho}}$. It therefore remains valid when
$c_{\bar{\varrho}}$ and $c_{\varrho',\bar{\varrho}}$ are replaced by smaller positive numbers $c'$ and $c''$,
provided that $\varrho'<c'$.
Let $\bar{\varrho}\coloneqq\zeta_G(0)/2$. By the
inclusions above, for any $t\geq 1$,
\begin{align*}
\inf_{x\in C_{t,\bar{\varrho}}}
Q_{t+1}(x,C_{t+1,\bar{\varrho}})\geq
c'\coloneqq\big(2\delta_1(\bar{\varrho})\big)^d\,
\underline{\xi}\big(\delta_2(\bar{\varrho}),
\delta_1(\bar{\varrho})\big),
\end{align*}
and, with $\varrho'\coloneqq\frac{1}{2}
\min\{\bar{\varrho},c'\}$,
\begin{align*}
\inf_{x\in C_{t,\varrho'}}
Q_{t+1}(x,C_{t+1,\bar{\varrho}})\geq
c''\coloneqq\big(2\delta_1(\bar{\varrho})\big)^d\,
\underline{\xi}\big(\delta_2(\varrho'),
\delta_1(\bar{\varrho})\big).
\end{align*}
Since $\varrho'<c'$, with $\upsilon\coloneqq 1/2$ and the constant $\epsilon\in(0,1)$
of that proof chosen as $\epsilon\coloneqq(1+\varrho'/c')/2$, which satisfies $\varrho'<\epsilon c'$, the proof
of Proposition~\ref{prop:A1p-A2} yields constants
$\varepsilon_\upsilon=\zeta/2$ and
$\kappa_\upsilon=\zeta/(2-\zeta)$, where $\zeta$ is an
explicit function of $(c',c'',\varrho',\upsilon)$ and
hence admissible. Let
$\mu^\star\coloneqq\delta_{x_1}$, with
$x_1=H_1^{-1}(y_1)$; note that $\mu^\star\in\mathcal{P}_{W_1}(\setX)$ since, the functions $G_t$ and
$m_{t+1}$ being positive, \ref{A0} holds with $W_t=\R^d$ for any $t\geq 1$. Then $\mu^\star(D_1)=1$ and, since
$G_1(x_1)\geq\zeta_G(0)>\bar{\varrho}$ by
Lemma~\ref{lemma:Prelim1_A2},
$\mu^\star(C_{1,\bar{\varrho}})=1\geq\upsilon$. By the
second part of that proof and the inclusion
$C_{t,\varepsilon_\upsilon}\subseteq
K_{t,\delta_2(\varepsilon_\upsilon)}$,
\begin{align*}
\inf_{t\geq 1}\Phi_{1,t}(\mu^\star)(D_t)\geq
\epsilon_\star\coloneqq\min\Big\{1,\,\kappa_\upsilon\,
2^d\,\underline{\xi}\big(\delta_2(\varepsilon_\upsilon),1\big)\Big\},
\end{align*}
which is admissible. In particular, \ref{A2} holds with
the sets $D_t$.

\emph{The constants $V_D$ and $Z$ (case~\ref{rem2}).} By
part~\ref{HK_4} of Proposition~\ref{prop:V_HK},
$\mathcal{V}_{\R^d}(\ind_{D_t})\leq 3^d-1$, so that
$V_D\coloneqq 3^d$ is admissible. By the proofs of
Lemmas~\ref{lemma:Prop1_showB1}
and~\ref{lemma:Prop1_showB1_ex2}, \eqref{eq:ineq_qt}
holds with a constant $C$ given by
Lemma~\ref{lemma:derivative} (applied with the non-increasing functions $\psi_G$ and $\psi_M$, that is,
with $\varepsilon=1$), which depends only on
$(d,r,L)$, on $c_M$ and, in case~(b) of
Lemma~\ref{lemma:Prop1_showB1_ex2}, on $(\nu_G,\Xi)$,
and with $\nu=\min\{\nu_G,\nu_M\}$ or
$\nu=\min\{\nu_G/L,\nu_M\}$. Thus
$Z\coloneqq 1+C\,L^d\,c_{d,\nu}^2$ is admissible.

Combining Steps~1 and~2 proves the remark.
\end{proof}

\subsubsection{Proof of Lemma \ref{lemma:bounded_pert}\label{p-lemma:bounded_pert}}
\begin{proof}
Consider the integral operator 
\[
Q^g_{t+1}(x,\dd x'):=g_t(x)M_{t+1}(x,\dd x')\quad
\mbox{\rm and set}\quad C^g_{t,\varrho}:=\{g_t\geq\varrho\}.
\]
Also denote
by $Q_{t+1}$ and $C_{t,\varrho}$ the corresponding
objects for $G_t$. Since
$c\,g_t\leq G_t\leq g_t$, we have
\[c\,Q^g_{t+1}\leq Q_{t+1}\leq Q^g_{t+1}\quad \text{and}\quad
C^g_{t,\varrho/c}\subseteq C_{t,\varrho}
\subseteq C^g_{t,\varrho}\] for any $t\geq 1$ and
$\varrho\in(0,1]$.

To show part (i) of the lemma note first that, by Lemma~\ref{lemma:proofA1}
(Appendix~\ref{p-prop:example})  applied to the model
with potential functions $g_t$, there
exist a constant $\ell_0\in(0,1)$ and bounded hyper-rectangles
$D_t$ such that: \ref{A1a} holds for
some $\underline{q}\in(0,1)$; for every level
$\ell\in(0,\ell_0)$, \ref{A1b} holds at level
$\ell$ for some constant $\xi_\ell\in(1,\infty)$; and, for any
$0<\ell_2\leq\ell_1<\ell_0$,
\[
\inf_{t\geq 1}\inf_{x\in C^g_{t,\ell_2}}
Q^g_{t+1}(x,C^g_{t+1,\ell_1})>0\quad\mbox{\rm and}\quad
\inf_{t\geq 1}\inf_{x\in C^g_{t,\ell_2}}
Q^g_{t+1}(x,D_{t+1})>0.
\]
We now show that, for the same sets $D_t$, condition \ref{A1} also holds when the functions $g_t$ are replaced by the functions $G_t$. First, we have
\[\inf_{t\geq 1}\inf_{x\in D_t}Q_{t+1}(x,D_{t+1})
\geq c\,\underline{q}.\] Next, let
$\varrho\in(0,\min\{\ell_0,c^2\underline{q}^2\})$,
$t\geq 1$, $x\in C_{t,\varrho}$ and
$A\in\mathcal{E}$, and let
\[
A'\coloneqq A\cap C_{t+1,\varrho}\subseteq
C^g_{t+1,\varrho}.
\]
For any
$x\in C^g_{t,\varrho}$, \ref{A1b} at level $\varrho$
for $g_t$ gives
\[\xi_\varrho^{-1}\nu_{t+1}(A')\leq Q^g_{t+1}(x,A')\leq
\xi_\varrho\,\nu_{t+1}(A'),
\] and therefore
\[(c/\xi_\varrho)\,\nu_{t+1}(A')\leq Q_{t+1}(x,A')\leq
\xi_\varrho\,\nu_{t+1}(A').
\] Hence \ref{A1} holds for the model with potential functions
$G_t$, with the constants
$(c\,\underline{q},\xi_\varrho/c,\varrho)$.

Next, for $0<\ell_2\leq\ell_1$ with
$\ell_1/c<\ell_0$, we have
\begin{align*}
\inf_{x\in C_{t,\ell_2}}Q_{t+1}(x,C_{t+1,\ell_1})
\geq c\inf_{x\in C^g_{t,\ell_2}}
Q^g_{t+1}\big(x,C^g_{t+1,\ell_1/c}\big)
\end{align*}
as well as
\[
\inf_{x\in C_{t,\ell_2}}Q_{t+1}(x,D_{t+1})\geq
c\inf_{x\in C^g_{t,\ell_2}}Q^g_{t+1}(x,D_{t+1}),
\]
and both lower bounds are positive uniformly in
$t\geq 1$. Proposition~\ref{prop:A1p-A2} therefore
applies to the model with potential functions and yields \ref{A2}
with the sets $D_t$. Since these sets are
hyper-rectangles, \ref{A3} follows from part~\ref{HK_4}
of Proposition~\ref{prop:V_HK}.

To prove part (ii) of the lemma, consider the densities
\[
q_{t+1}(x,x')\coloneqq G_t(x)m_{t+1}(x,x')\quad\mbox{\rm
and}\quad q^g_{t+1}(x,x')\coloneqq g_t(x)m_{t+1}(x,x'),\qquad x,x'\in\R^d,\quad t\geq 1.
\]
Under the assumptions of part~\ref{ii-prop:example} of
Proposition~\ref{prop:example} (resp.\ of part~\ref{ii-prop:example2} of
Proposition~\ref{prop:example2}), case~(a) (resp.\ case~(b)) of
Lemma~\ref{lemma:Prop1_showB1_ex2} applies to the model with potential
functions $g_t$. An inspection of its proof shows that the bound
\eqref{eq:ineq_qt_ex2}, and hence the bound \eqref{eq:ineq_qt}, holds for the
functions $q^g_{t+1}$ for any set $u\subseteq\{1,\dots,d\}$, including
$u=\emptyset$: in this case the sum given by the Leibniz rule reduces to the
term $v=\emptyset$, and the bounds on $|\partial^vg_t|$ and on
$|\partial_1^{v'}m_{t+1}|$ used in that proof hold for $v=v'=\emptyset$
by \ref{upper_GtMt}. Therefore, there exist constants $C,\nu\in(0,\infty)$
such that, with $\psi(z)\coloneqq(1+z)^{-(d+\nu)/2}$,
\begin{align}\label{eq:qg_all_u}
\big|\partial_1^uq^g_{t+1}(x,x')\big|\leq C\,\psi\big(\|y_t-H_t(x)\|^2\big)\,
\psi\big(\|x'-F_{t+1}(x)\|^2\big)
\end{align}
for any $t\geq 1$, any $x,x'\in\R^d$ and any set $u\subseteq\{1,\dots,d\}$.
Let $t\geq 1$ and $x'\in\R^d$. Since $b_t\in\mathcal{C}_{\mathrm{mixed}}(\R^d)$ and
$q^g_{t+1}(\cdot,x')\in\mathcal{C}_{\mathrm{mixed}}(\R^d)$, the Leibniz rule shows that
$q_{t+1}(\cdot,x')\in\mathcal{C}_{\mathrm{mixed}}(\R^d)$ and that, for any non-empty set
$u\subseteq\{1,\dots,d\}$,
\begin{align*}
\partial_1^uq_{t+1}(x,x')=\sum_{v\subseteq u}\partial^vb_t(x)\;
\partial_1^{u\setminus v}q^g_{t+1}(x,x'),\qquad\forall x\in\R^d,
\end{align*}
where $\partial^\emptyset b_t=b_t$ takes its values in $[c,1]$ and
$|\partial^vb_t|\leq L_b$ for any $v\neq\emptyset$. Hence, by
\eqref{eq:qg_all_u}, the functions $q_{t+1}$ satisfy \eqref{eq:ineq_qt} with the
constant $C$ replaced by $2^dL_bC$. Since the proof of
Lemma~\ref{lemma:Prop1_showB1} only uses \eqref{eq:ineq_qt} and the fact that
the functions $H_t$ are $L$-bi-Lipschitz, which holds under \ref{Model}, this lemma
shows that \ref{B1} holds for the model with potential functions $G_t$, with
$\varsigma_{t+1}=\lambda_d$ for any $t\geq 1$.
\end{proof}

\end{document}